\documentclass[12pt]{article}

\usepackage{amssymb,amsmath,verbatim,mathtools,needspace,enumitem,etoolbox,graphicx,physics,microtype,afterpage,bm,xparse}
\usepackage{bm}
\usepackage{latexsym}
\usepackage{amsmath}
\usepackage{amssymb}
\usepackage{amsthm}
\usepackage{amscd}
\usepackage[mathscr]{eucal}
\usepackage{fullpage}
\usepackage{graphicx}
\usepackage{subfigure}
\usepackage{psfrag}
\usepackage{rotating}
\usepackage{hyperref}
\usepackage{cancel}
\usepackage{float}

\usepackage{color}

\newcommand{\R}{\mathbb{R}}

\definecolor{ggreen}{cmyk}{0.7,     0,      0.9,      0}
\definecolor{viol}{cmyk}{0.3,1,0,0}
\definecolor{myred}{cmyk}{0.1, 1, 0.5, 0}
\definecolor{bblue}{rgb}{0.2, 0.29996, 0.8 }

\theoremstyle{plain}

\newtheorem{theorem}{Theorem}[section]

\newtheorem{remark}{Remark}
\newtheorem{examples}{Example}
\newtheorem{proposition}{Proposition}[section]

\newtheorem{definition}[theorem]{Definition}

\begin{document}


\title{\bf Spherically symmetric\\ elastic bodies in general relativity}

\author{
	\sc Artur Alho$^{1}$\thanks{Electronic address:{\tt
			artur.alho@tecnico.ulisboa.pt}}\,, Jos\'e Nat\'ario$^{1}$\thanks{Electronic address:{\tt
			jnatar@math.ist.utl.pt}}\,, Paolo Pani$^{2}$\thanks{Electronic
		address:{\tt paolo.pani@uniroma1.it}}\,  and Guilherme Raposo$^{3}$\thanks{Electronic
		address:{\tt graposo@ua.pt}}\\
	$^{1}${\small\em Center for Mathematical Analysis, Geometry and Dynamical Systems,}\\
	{\small\em Instituto Superior T\'ecnico, Universidade de Lisboa,}\\
	{\small\em Av. Rovisco Pais, 1049-001 Lisboa, Portugal.}\\
	$^{2}${\small\em Dipartimento di Fisica, Sapienza Universit\`a di Roma \& INFN Roma1,}\\
	{\small\em Piazzale Aldo Moro 5, 00185, Roma, Italy}\\
	$^{3}${\small\em Centre for Research and Development in Mathematics and Applications (CIDMA),}\\
	{\small\em Campus de Santiago, 3810-183 Aveiro, Portugal}
	}

\maketitle


\begin{abstract}
\noindent
The purpose of this review it to present a renewed perspective of the problem of self-gravitating elastic bodies under spherical symmetry. It is also a companion to the papers [Phys.\ Rev.\ D{\bf 105}, 044025 (2022)], [Phys.\ Rev.\ D{\bf 106}, L041502 (2022)], and [arXiv:2306.16584 [gr-qc]], where we introduced a new definition of spherically symmetric elastic bodies in general relativity, and applied it to investigate the existence and physical viability, including radial stability, of static self-gravitating elastic balls. We focus on elastic materials that generalize fluids with polytropic, linear, and affine equations of state, and discuss the symmetries of the energy density function, including homogeneity and the resulting scale invariance of the TOV equations. By introducing invariant characterizations of physically admissible initial data, we numerically construct mass-radius-compactness diagrams, and conjecture about the maximum compactness of stable physically admissible elastic balls.
\end{abstract}

\newpage

\tableofcontents

\newpage

\section{Introduction}
\label{sec:int}
Self-gravitating matter distributions in general relativity~(GR) satisfy the field equations \cite{1915SPAW.......844E}
\begin{subequations}\label{EE}
	\begin{align}
	& G_{\mu\nu}=8\pi T_{\mu\nu}\qquad (\mu,\nu=0,1,2,3), \label{Einstein}\\
	& \nabla_\mu T^{\mu}_{\nu}=0,\qquad \quad (\nu=0,1,2,3), \label{T}\\
	& g_{\mu\nu}u^{\mu}u^{\nu}=-1, \label{u}
	\end{align}
\end{subequations}
where $G_{\mu\nu}=R_{\mu\nu}-\frac{1}{2}R\,g_{\mu\nu}$ is the \emph{Einstein tensor}, $R_{\mu\nu}$ is the \emph{Ricci tensor}, and $R$ is and the \emph{Ricci scalar} of the spacetime metric $g_{\mu\nu}$. The matter content is described by the \emph{stress-energy tensor} $T_{\mu\nu}$, and $u^{\mu}$ is the \emph{mean 4-velocity field} of the particles making up the body, assumed to be a unit future-directed timelike vector field. Given $u^{\mu}$, the tensor $h_{\mu\nu}=g_{\mu\nu}+u_{\mu} u_{\nu}$, satisfying $h_{\mu\nu}u^{\nu}=0$, is a Riemannian metric on the tangent subspaces orthogonal to $u^\mu$, and $h^{\mu}_{\nu}=g^{\mu\sigma}h_{\sigma\nu}$ is the orthogonal projection operator on these subspaces. Assuming that there are no energy currents (for instance heat flow) in the frame of the particle's mean motion, the stress-energy tensor can be uniquely decomposed as
\begin{equation}\label{DefinitionT}
T_{\mu\nu}=\rho u_\mu u_\nu +\sigma_{\mu\nu}, \qquad \qquad \sigma_{\mu\nu}u^{\nu}=0,
\end{equation}
where  $\rho$ is the matter \emph{energy density}, and $\sigma_{\mu\nu}$ is the symmetric \emph{stress tensor}. The conservation equations~\eqref{T}, which are actually integrability conditions for~\eqref{Einstein}, yield
%
%
\begin{subequations}\label{ConsEq}
	\begin{align}
	&u^\mu\partial_\mu\rho+\rho h^{\mu\nu}\nabla_{(\mu}u_{\nu)}+\sigma^{\mu\nu}\nabla_{(\mu}u_{\nu)}=0, \label{ConsEn}\\
	&\rho u^\mu\nabla_\mu u^{\nu}+\sigma^{\nu}_{\mu}u^\lambda\nabla_\lambda u^{\mu}+h^{\mu}_{\sigma}h^{\nu}_{\lambda}\nabla_\mu \sigma^{\sigma\lambda}=0 \label{ConsMom}.
	\end{align}	
\end{subequations}
The above system is closed by postulating an \emph{equation of state} (EoS) relating the stress tensor $\sigma_{\mu\nu}$ with the energy density $\rho$. 
%
In the field of astrophysics, compact objects are typically idealized as self-gravitating  perfect fluids, where the \emph{isotropic stress (or pressure)} $p_\mathrm{iso}$, defined by the relation $\sigma_{\mu\nu} = p_\mathrm{iso} h_{\mu\nu}$, prevents gravitational collapse. The theory of thermodynamics states that, for \emph{adiabatic} and \emph{isentropic} fluids,
%
%
%
there exists a $C^{3}$ function $\widehat{\rho}:[0,+\infty)\rightarrow\mathbb{R}$, and a $C^2$ function $\widehat{p}_\mathrm{iso}:[0,+\infty)\rightarrow\mathbb{R}$, with $\rho=\widehat{\rho}(n)$ and 
\begin{equation}\label{FluidEoS}
\widehat{p}_\mathrm{iso}(n)=n\frac{d\widehat{\rho}(n)}{dn}-\widehat{\rho}(n),
\end{equation}
where $n$ is the \emph{particle (baryon) number density}. Equation~\eqref{ConsEn} is then equivalent to the particle number conservation law
\begin{equation}\label{nCons}
\nabla_\mu(n u^\mu)=0,
\end{equation}
whereas~\eqref{ConsMom} reduces to the relativistic Euler equation.

Physical compressible fluids should satisfy $\widehat{\rho}(n)>0$ and $\widehat{p}_\mathrm{iso}(n)>0$ for $n>0$. Moreover, the fluid \emph{adiabatic sound speed} $c_\mathrm{s}$, given by
\begin{equation}
c^2_\mathrm{s}\equiv\frac{dp_\mathrm{iso}}{d\rho}=\frac{n\frac{d\widehat{p}_\mathrm{iso}}{dn}}{\widehat{\rho}(n)+\widehat{p}_\mathrm{iso}(n)},
\end{equation}
is assumed to satisfy the \emph{strict hyperbolicity condition} of the system formed by the number density conservation and the relativistic Euler equations, and also the \emph{causality condition}, given by the two inequalities
\begin{equation}\label{velcaus}
0<c_\mathrm{s}\leq 1.
\end{equation}
%
%
%
%
%
This implies that $\widehat{p}_\mathrm{iso}(n)$ is strictly monotonically increasing, and so has an inverse function $n=n(p_\mathrm{iso})$, so that $\rho$ can be written as function of $p_\mathrm{iso}$. 
%
%
%
The two most popular families of relativistic fluids consists of models for which the \emph{adiabatic index}
\begin{equation}\label{AdFluid}
\gamma(n)\equiv\frac{n}{\widehat{p}_\mathrm{iso} }\frac{d\widehat{p}_\mathrm{iso}}{dn}>0
\end{equation}
is constant, and therefore
\begin{equation}
\widehat{p}_\mathrm{iso}(n)
=\mathrm{K} n^{\gamma},\qquad \mathrm{K}>0.
\end{equation}
Plugging~\eqref{AdFluid} into~\eqref{FluidEoS} leads to the following two examples of EoS, see e.g.  Appendix~B in Tooper~\cite{Top1965}:
%
%
\begin{examples}[Polytropic equation of state~\cite{Top1965}]\label{Ex1}
	The relativistic polytropic EoS consists of the 2-parameter family  given by
    \begin{equation} \label{eq:relpoly}
\rho = \mathrm{C} p^{\frac{1}{\gamma}}_\mathrm{iso}+\frac{p_\mathrm{iso}}{\gamma-1}=\mathrm{C} \mathrm{K}^{\frac{1}{\gamma}}n+\frac{\mathrm{K}}{\gamma-1}n^{\gamma}\qquad (\gamma\neq 1),
\end{equation}
	where $\mathrm{C}>0$ is a constant. The squared sound speed is
	\begin{equation}
c^2_\mathrm{s}=\frac{\gamma-1}{1 + \mathrm{C}(1-\frac{1}{\gamma})\mathrm{K}^{\frac{1}{\gamma}-1}n^{1-\gamma}}\qquad (\gamma\neq 1).
	\end{equation}
 The adiabatic index $\gamma$ is related to the polytropic index $\mathrm{n}$ (not to be confused with the particle number density $n$) by $\gamma=1+\frac{1}{\mathrm{n}}$, i.e. for $\gamma\in(0,+\infty)\setminus \{1\}$, the polytropic index $\mathrm{n}$ ranges over two disconnected intervals: $\mathrm{n}\in(-\infty,-1)$ for $\gamma\in(0,1)$, and $\mathrm{n}\in(0,+\infty)$ for $\gamma\in (1,+\infty)$. 
\end{examples}
\begin{remark}\label{Rem1}
    The case $\gamma=1$ corresponds to the limiting isothermal case $\mathrm{n}\rightarrow+\infty$, with
    \begin{equation}
    \rho= \ln{\left(p_\mathrm{iso}/\mathrm{C}\right)}p_\mathrm{iso}=\ln{\left(\mathrm{K} n/\mathrm{C}\right)}\mathrm{K}n  \qquad (\gamma=1)
    \end{equation}
    and sound speed
    \begin{equation}
  c^2_\mathrm{s}=   \frac{1}{1+\ln{(\mathrm{K}n/\mathrm{C})}} \qquad (\gamma=1).   
    \end{equation}
    The limiting case $\gamma\rightarrow +\infty$ corresponds to the incompressible limit $\mathrm{n}\rightarrow 0$, with constant energy density and infinite sound speed,
    \begin{equation}
  c^2_\mathrm{s}= +\infty \qquad (\gamma\rightarrow +\infty).   
    \end{equation}
	Since the energy per particle $\widehat{e}(n)=\widehat{\rho}(n)/n$ has a finite limit as $n\rightarrow0$ only for $\gamma>1$, and the sound speed satisfies the strict hyperbolicity and the causality conditions for all $n$ if and only if $1<\gamma\leq2$ $(\mathrm{n}\geq1)$, it is somewhat  natural to restrict the relativistic polytropes to this parameter range. However, we remark that neutron star cores are often approximated by an EoS with $\mathrm{n}\approx1/2$ $(\gamma\approx 3)$.
\end{remark}
\begin{examples}[Linear equation of state~\cite{Bon64,harrison1965gravitation}]\label{Ex2}
	The linear EoS consists of the 1-parameter family with constant sound speed 
	$c^2_\mathrm{s}=\gamma-1$ given by
	\begin{equation}
	\rho =\frac{p_\mathrm{iso}}{(\gamma-1)}=\frac{\mathrm{K}}{\gamma-1} n^{\gamma}\qquad (\gamma\neq 1),
	\end{equation}
obtained by taking the limit $\mathrm{C}\rightarrow 0$ in the relativistic polytropic EoS.
The strict hyperbolicity and causality conditions on the sound speed imply that $\gamma\in(1,2]$.
\end{examples}
\begin{remark}
	When $\gamma=1$, the linear EoS corresponds to a pressureless or dust fluid with $p_\mathrm{iso}=0$ and $\rho=\varrho$, where $\varrho=\mathfrak{m}n$ is the \emph{(baryon) mass density}, with $\mathfrak{m}$ the rest mass of the dust particles (baryons). In this case $c^2_\mathrm{s}=0$.
\end{remark}
Another popular model is the affine EoS: 
\begin{examples}[Affine equation of state]\label{Ex3}
	The affine EoS consists of the 2-parameter family with constant sound speed $c^2_\mathrm{s}=\gamma-1$ given by
	\begin{equation}
	\rho=\rho_0+\frac{p_\mathrm{iso}}{(\gamma-1)}=\frac{\gamma-1}{\gamma}\rho_0+\frac{\mathrm{K}}{\gamma-1} n^{\gamma} \qquad (\gamma\neq 1),
	\end{equation}
	with $\rho_0$ a positive constant, so that
	\begin{equation}
		\widehat{p}_\mathrm{iso}(n) = -\frac{\gamma-1}{\gamma}\rho_0+\mathrm{K} n^{\gamma}.
	\end{equation}
The strict hyperbolicity and causality conditions on the sound speed imply that $\gamma\in(1,2]$.
\end{examples}
\begin{remark}
	These types of matter models where $\rho$ remains positive when $p_\mathrm{iso}$ vanishes are often called \emph{liquids}. When $\gamma=2$ (that is, when the speed of sound equals the speed of light), the affine EoS corresponds to Christodoulou's hard phase model~\cite{Christodoulou:1995}. This class of equations of state includes the MIT bag model describing quark stars~\cite{Bodmer:1971we,Witten:1984rs,Farhi:1984qu}.
\end{remark}
Static asymptotically flat relativistic perfect fluid solutions are believed to be necessarily spherically symmetric, a conjecture known as the \emph{fluid ball conjecture}. This conjecture has been proved to hold for solutions with positive energy density, $\rho>0$, and nonnegative squared speed of sound, $c^2_\mathrm{s}\geq0$, by Masood-ul-Alam~\cite{Masood2007GReGr}. 

The polytropic fluids models of Example~\ref{Ex1} were introduced by Tooper~\cite{Top1965} in 1965, and were analyzed by Heinzle, R\"{o}hr \& Uggla~\cite{Heinzle_2003}. They are known to admit solutions with a regular center of symmetry and finite radius and mass for $\gamma>6/5$ ($0<\mathrm{n}<5$), see also~\cite{Heinzle_2002,Nambo:2021} and~\cite{Ramming2013SphericallySE,NILSSON2000292,RS91,Pfister_2011} for related power-law type equations of state. Fluids with linear EoS are scale-invariant, as shown by Taub \& Cahil~\cite{TC1971}. These models were analyzed by Collins~\cite{Collins:1985}, who showed that there are no solutions with a regular center and finite radius. On the other hand, affine equations of state, for which $\rho\rightarrow\rho_0>0$ as $p_\mathrm{iso}\rightarrow0$, always have solutions with finite radius, as shown by Rendall \& Schmidt~\cite{RS91} and Nilsson \& Uggla~\cite{NILSSON2000278}. The particular case of Christodoulou's hard phase material was investigated in detail by Fournodavlos \& Schlue~\cite{Fourn:2019}.

Once an EoS is given, spherically symmetric steady states of the Einstein-Euler equations consist of a 1-parameter family of solutions parameterized by the value of the baryon number density at the center of symmetry, $n_c$, or, equivalently, the central density $\rho_c=\widehat{\rho}(n_c)$, or the central pressure $p_c=\widehat{p}_\mathrm{iso}(n_c)$. An important problem concerns the stability of such perfect fluid ball solutions under radial perturbations. In 1965, Harrison et al.~\cite{harrison1965gravitation} conjectured that, in contrast with the Newtonian setting, spherically symmetric perfect fluid balls solutions are radially unstable for high central pressures corresponding to configurations beyond the configuration of maximum mass. A rigorous proof of the onset of stability/instability, and its relation with the turning point principle in the mass-radius $\mathcal{M}(\mathcal{R})$ curve, were recently proved in~\cite{Hadi2021StaIns} and~\cite{Hadi2021TurningPP}. The axisymmetric Einstein equations, linearized around the spherically symmetric steady states, have also been analyzed in~\cite{Kind_1993b}.

In the spherically symmetric, time-dependent case, pressureless fluids are described by explicit solutions known as the Lem\^aitre-Tolman-Bondi solutions~\cite{Lem33,Tol34,Bon47}, which were analyzed in detail by Christodoulou~\cite{Chr84}. A global understanding of the state space of self-similar solutions can be found in the work by Carr et al.~\cite{BJCarr_2001}. Of particular relevance are the solutions containing naked singularities of Ori \& Piran~\cite{PhysRevLett.59.2137,PhysRevD.42.1068}, whose stability has been recently addressed by Guo, Hadzic \& Jang~\cite{GHJ21}. Self-similar solutions for fluids undergoing diffusion have also been analyzed by Alho \& Calogero~\cite{AC17}, generalizing previous results on radiating exteriors~\cite{BOS89,FST96}. A local existence results for the initial boundary value problem was proved by Kind \& Ehlers~\cite{Kind_1993a}, and spherically symmetric shock development has been analyzed by Christodoulou \& Lisibach~\cite{Christodoulou2016}.

A perfect fluid model is the typical approximation to describe the interior of cold stars, because degenerate fermions behave as a weakly interacting gas at relatively small densities. However, nuclear interactions and QCD effects become crucial inside relativistic stars, and one might expect that the perfect fluid idealization will eventually break down, at least to some extent. Thus, solid phases of matter may be relevant for astrophysical compact objects, which is indeed the situation in the crust of a neutron star~\cite{Chamel:2008ca,Suleiman:2021hre}. Furthermore, the behaviour of matter inside a neutron star is still poorly known, especially in the core~\cite{Lattimer:2004pg}. This motivates exploring generalizations of the fluid models.
Exploring new models of self-gravitating objects is also relevant in the context of exotic compact objects that could mimic the properties of black holes~\cite{Cardoso:2019rvt}. Modelling these properties is relevant for tests of the nature of compact objects with gravitational wave~\cite{Maggio:2021ans} and electromagnetic observations. Generalizations of the fluid model would have richer phenomenology than fluid stars and also different from black holes. This includes observable quantities such as the mass-radius diagram and the moment of inertia, but also the multipolar structure, the tidal deformability, and the linear response to external perturbations, which have direct consequences for gravitational wave astronomy, in particular for the coalescence and merger of a binary system of compact objects.

A natural generalization of perfect fluids is to consider elastic materials~\cite{CarterQuintana,Beig:2002pk,Brown:2020pav}. Beside offering a more accurate description of the stellar interior~\cite{Rajagopal_2006,Rajagopal_2006b,Alford_2008}, elasticity might play a crucial role in constructing consistent models of exotic compact objects and black hole mimickers within and beyond GR~\cite{Cardoso:2019rvt,Carballo-Rubio:2018jzw}. Static spherically symmetric elastic configurations have been considered in the past by Park~\cite{Par00}, Karlovini \& Samuelsson~\cite{Karlovini:2002fc}, Frauendiener \& Kabobel~\cite{FK07}, and Andr\'easson \& Calogero~\cite{Andreasson:2014lka}. In addition to these works, elasticity has also been studied perturbatively to model the crust of a neutron  star~\cite{Chamel:2008ca,Suleiman:2021hre}, and linear radial perturbations of relativistic elastic balls have been treated by Karlovini, Samuelsson \& Zarroug in~\cite{Karlovini:2003xi} using a Lagrangian approach. 

Comparing to the perfect fluid case, very little is known when the matter content is described by elastic materials. The main obstacle is that relativistic elasticity is a geometrical Lagrangian field theory, where the configuration of a body is described by a projection map $\bm{\Pi}$ from the spacetime Lorentzian manifold $(\mathcal{S}, \bm{g})$ to a 3-dimensional Riemannian material space $(\mathcal{B}, {\bm \gamma})$, and it is often not clear how the configuration map relates to the matter fields, making it hard to close the system of equations in a natural way. This difficulty has been recently overcome in the spherically symmetric Newtonian setting by Alho \& Calogero~\cite{Alho:2019fup}, where a new Eulerian definition of elastic bodies consisting of balls or shells was introduced, opening the way to prove existence and uniqueness of steady states solutions for several elastic equations of state of power-law type~\cite{Alho:2018mro} (see also~\cite{Cal21} for the time-dependent equations of motion, where self-similar solutions were constructed). 

This paper introduces a new Eulerian definition of spherically symmetric elastic balls in general relativity (already used in the companion short papers~\cite{Alho:2021sli,Alho:2022bki,Alho:2023mfc}). This allows us to:
\begin{itemize}
	\item[i)] Give a new characterization of physically admissible initial data for static spherically symmetric elastic objects (Definition~\ref{IDstatic}), and apply it to investigate the existence and physical viability of ball solutions, including radial stability;
	\item[ii)] Discuss (spherically symmetric) elastic generalizations of the fluid equations of state of Examples~\ref{Ex1},~\ref{Ex2},  and~\ref{Ex3}, including scale invariance. 
	\item[iii)] Conjecture about the maximum compactness of spherically symmetric physically admissible and radially stable elastic static configurations.
\end{itemize}
%
The formalism and general framework introduced in this paper open up the possibility of analyzing other interesting problems, such as the existence and properties of relativistic  spherically symmetric elastic self-similar solutions, or spherically symmetric elastic shock development.

The paper is organized as follows: Section~\ref{sec:ballsmatter} reviews spherically symmetric matter distributions in General Relativity, with emphasis on the perfect fluid example. We introduce the precise definition of ball solutions,  their steady states, and linear radial perturbations thereof. It ends with a discussion of maximum compactness bounds and the Buchdahl limit. In Section~\ref{ElasticBalls} we give a new definition of spherically symmetric relativistic elastic balls, which is used to write the Einstein-elastic equations in spherical symmetry as first order system of partial differential equations. This section also contains the important definitions and properties of the reference state and the isotropic state, as well as of the different wave speeds. These are then used in Section~\ref{SS} to analyze the static case, and to introduce a definition of physically admissible initial data. In Section~\ref{RS} we deduce the linearized spherically symmetric Einstein-elastic equations around the steady state solutions, and in Section~\ref{Newtonian} we discuss the Newtonian limit. 
Section~\ref{Polytropes} discusses elastic polytropes that generalize the polytropic fluids of Example~\ref{Ex1}, including scale invariance in the Newtonian limit. In Section~\ref{LinEOS} we introduce elastic materials with constant wave speeds in the isotropic state, which generalize the linear equations of state of Example~\ref{Ex2}, including relativistic scale invariance, and the affine equations of state of Example~\ref{Ex3}. In Section~\ref{LinearES} we analyze elastic materials with constant longitudinal wave speeds, which also generalize fluids with linear and affine equations of state (the list of the various material models studied in this work, as well as their main properties, is given in Table~\ref{table} below). Appendix~\ref{App:A} contains an overview of relativistic elasticity, including wave speeds, the definitions of the reference state and the isotropic state, and examples of elastic materials, with emphasis given to quasi-Hookean materials. Appendix~\ref{App:B} contains a detailed derivation of the Einstein-elastic equations in spherical symmetry.

\begin{table}[H]
\begin{tabular}{|c|l|c|l|}
\hline
\textbf{Acronym} & \multicolumn{1}{c|}{\textbf{Model}} & \textbf{Reference} & \multicolumn{1}{c|}{\textbf{Features}}  \\ \hline
QP               & Quadratic Polytropic                          & Sec.~\ref{sec:QM}                 
& \begin{tabular}[c]{@{}l@{}}Generalization of polytropic fluid EoS:\\ - pre-stressed reference state, \\
- simplest model \end{tabular}   \\ \hline
NSI              & Newtonian Scale Invariant           & Sec.~\ref{NSI}                 
& \begin{tabular}[c]{@{}l@{}}Generalization of polytropic fluid EoS: \\ - pre-stressed reference state, \\ - scale-invariant in the Newtonian limit, \\ - includes quadratic model with $\mathrm{n}=1$ \end{tabular}  \\ \hline
LIS              & Linear Isotropic State              & Sec.~\ref{LIS}                
& \begin{tabular}[c]{@{}l@{}} Generalization of linear fluid EoS: \\  - pre-stressed reference state \\ - relativistically scale-invariant, \\ - constant wave speeds at the isotropic state\end{tabular} \\ \hline
AIS              & Affine Isotropic State              & Sec.~\ref{AIS}                
& \begin{tabular}[c]{@{}l@{}} Generalization of affine fluid EoS: \\
- stress-free reference state, \\ - constant wave speeds at the isotropic state\end{tabular}  \\ \hline
ACS             & Affine Constant Speed                                & Sec.~\ref{ACS}                 
& \begin{tabular}[c]{@{}l@{}} Generalization of affine fluid EoS: \\
- stress-free reference state, \\ - constant longitudinal wave speeds \end{tabular}            \\ \hline
LCS             & Linear Constant Speed                                & Sec.~\ref{LCS}                
&  \begin{tabular}[c]{@{}l@{}} Generalization of linear fluid EoS: \\ 
- pre-stressed reference state, \\ - constant longitudinal wave speeds \end{tabular}             \\ \hline
\end{tabular}\label{table}
\caption{Summary table of the elastic material models considered in this work.}
\end{table}

\subsubsection*{Notation and conventions:} 
We use geometrized units, where the speed of light $c$ and the Newtonian gravitational constant $G$ are both set to one, $c=G=1$, so that all quantities are measured in some power of the lenght unit $[L]$. The symbol ${}^\prime$ denotes differentiation of a real function of one real variable with respect to its argument. The symbols for the main quantities used in this work, as well as their geometrized dimensions, are the following:
\begin{itemize}
\item $\rho$ -- energy density $[L^{-2}]$;
\item $\varrho$ -- baryonic mass density (related to volume deformations) $[L^{-2}]$;
\item $\varsigma$ -- average baryonic mass (related to shear deformations) $[L^{-2}]$;
\item $\epsilon$ -- relativistic stored energy function $[L^{-2}]$;
\item $w$ -- deformation potential energy density $[L^{-2}]$;
\item $p_{\mathrm{rad}}$ -- radial pressure $[L^{-2}]$;
\item $p_{\mathrm{tan}}$ -- tangential pressure $[L^{-2}]$;
\item $p_{\mathrm{iso}}$ -- isotropic pressure $[L^{-2}]$;
\item $n$ -- particle number density $[L^{-3}]$;
\item $\delta$ -- normalized particle number density (related to volume deformations) $[L^{0}]$;
\item $\eta$ -- normalized average number of particles (related to shear deformations) $[L^{0}]$;
\item $\gamma$ -- adiabatic index $[L^{0}]$;
\item ${\rm n}$ -- polytropic index $[L^{0}]$;
\item ${\rm s}$ -- shear index $[L^{0}]$;
\item $\mathcal{K}$ -- polytropic scale factor $[L^{\frac{2}{\rm n}}]$;
\item $\lambda$ -- first Lam\'e parameter $[L^{-2}]$;
\item $\mu$ -- second Lam\'e parameter (shear modulus) $[L^{-2}]$;
\item $\nu$ -- Poisson ratio $[L^{0}]$;
\item $K$ -- bulk modulus $[L^{-2}]$;
\item $E$ -- Young modulus $[L^{-2}]$;
\item $L$ -- p-wave modulus $[L^{-2}]$;
\item $c_\mathrm{L}$ -- speed of longitudinal waves in the radial direction $[L^{0}]$;
\item $\tilde{c}_\mathrm{L}$ -- speed of longitudinal waves in the tangential direction $[L^{0}]$;
\item $c_\mathrm{T}$ -- speed of transverse waves in the radial direction $[L^{0}]$;
\item $\tilde{c}_\mathrm{T}$ -- speed of transverse waves in the tangential direction, oscillating radially $[L^{0}]$;
\item $\tilde{c}_\mathrm{TT}$ -- speed of transverse waves in the tangential direction, oscillating tangentially $[L^{0}]$;
\item $m$ -- Hawking mass $[L]$;
\item $\phi$ -- gravitational potential $[L^{0}]$;
\item $v$ -- radial velocity $[L^{0}]$;
\item $\mathcal{M}$ -- mass of the matter ball $[L]$;
\item $\mathcal{R}$ -- radius of the matter ball $[L]$;
\item $\mathcal{C}$ -- compactness of the matter ball $[L^{0}]$.
\end{itemize}

\newpage

\section{Self-gravitating relativistic balls of matter}
\label{sec:ballsmatter}
%
%
The Einstein equations in spherical symmetry are discussed in detail in Appendix~\ref{App:B}, both in Schwarzschild coordinates (Eulerian description) and in comoving coordinates (Lagrangian description). 
In spherical symmetry, the group $SO(3)$ acts by isometries on the space-time manifold, and the orbits are either round $2$-spheres or fixed points. Spherically symmetric bodies are then described by the four scalars $(\rho,p_\mathrm{rad},p_\mathrm{tan},v)$, functions of the coordinates on the quotient manifold, where $\rho$ is the \emph{energy density}, $p_\mathrm{rad}$ and $p_\mathrm{tan}$ are the \emph{radial} and \emph{tangential pressures}, and $v$ is the \emph{radial velocity}.

In Schwarzschild coordinates, the \emph{radial area function} $r$  (defined such that the area of the $SO(3)$ orbits is $4\pi r^2$) is taken as a coordinate on the quotient manifold. It is usually supplemented by a time coordinate $t$ defined  on the quotient manifold by requiring its level sets to be orthogonal to the level sets of $r$ in the quotient metric. Notice that $t$ is defined up to a rescaling of the form $t'=f(t)$, for $f$ an invertible smooth function. Hence, the spherically symmetric metric takes the form
\begin{equation}
g_{\mu\nu}dx^{\mu}dx^{\nu}=-e^{2\phi(t,r)}dt^2+\frac{dr^2}{1-\frac{2m(t,r)}{r}}+r^2(d\theta^2+\sin^2{\theta}d\varphi^2),
\end{equation}
where $\phi(t,r)$ is the \emph{relativistic gravitational potential} and $m(t,r)$ is the \emph{Hawking (or Misner-Sharp) mass function}. The matter 4-velocity is given by
\begin{equation}
u^{\mu}\partial_\mu=e^{-\phi}\langle v\rangle\partial_t+v\partial_r,\qquad \langle v\rangle(t,r)\equiv\left(1+\frac{v^2(t,r)}{1-\frac{2m(t,r)}{r}}\right)^{1/2},
\end{equation}
and the Cauchy tensor is given in Schwarzschild coordinates by
\begin{equation}
    \sigma_{\mu\nu}dx^{\mu}dx^{\nu}=p_\mathrm{rad}(t,r)\left(\frac{e^{2\phi}v^2}{1-\frac{2m}{r}}dt^2-\frac{2e^{\phi}\langle v\rangle v}{1-\frac{2m}{r}} dtdr+\frac{\langle v\rangle^2}{1-\frac{2m}{r}}dr^2\right)+p_\mathrm{tan}(t,r)r^2(d\theta^2+\sin^2{\theta}d\varphi^2).
\end{equation}
The Einstein equations~\eqref{EE} yield a single equation for the relativistic gravitational potential,
\begin{equation}\label{RelPot}
\partial_r\phi = \frac{1}{1-\frac{2m}{r}}\left[\frac{m}{r^2}+4\pi r\left(p_\mathrm{rad}+\frac{\rho+p_\mathrm{rad}}{1-\frac{2m}{r}}v^2\right)\right],
\end{equation}
whereas the Hawking mass solves
\begin{equation}\label{HwkMass}
\partial_t m = -4\pi r^2 (\rho+p_\mathrm{rad}) e^{\phi} v \langle v\rangle,\qquad 	\partial_r m = 4\pi r^2 \left(\rho +\frac{\rho+p_\mathrm{rad}}{1-\frac{2m}{r}}v^2\right).
\end{equation}
%
%
Note that there is no equation for $\partial_t \phi$, as the time-dependence of $\phi$ is not fixed at this point (there is still gauge freedom in the choice of the coordinate $t$).
Using the above equations for $\phi$ and $m$, equations~\eqref{ConsEq} for the conservation of energy and momentum yield
\begin{subequations}\label{ConsEM}
	\begin{align}
	&e^{-\phi}\langle v\rangle\partial_t\rho+v\partial_r\rho+(\rho+p_\mathrm{rad})\left(\frac{ve^{-\phi}\langle v\rangle}{1-\frac{2m}{r}+v^2}\partial_t v+\partial_r v\right) = -\frac{2}{r}(\rho+p_\mathrm{tan})v, \label{Cons1}\\
	&\frac{\rho+p_\mathrm{rad}}{1-\frac{2m}{r}+v^2}\left(e^{-\phi}\langle v\rangle\partial_t v+v\partial_r v\right)+\frac{ve^{-\phi}\langle v\rangle}{1-\frac{2m}{r}+v^2}\partial_t p_\mathrm{rad}+\partial_r p_\mathrm{rad}= \nonumber \\ &\qquad\qquad\qquad\qquad\qquad\qquad\qquad=\frac{2}{r} (p_\mathrm{tan}-p_\mathrm{rad})-\frac{\rho+p_\mathrm{rad}}{1-\frac{2m}{r}+v^2}\left(\frac{m}{r^2}+4\pi r p_\mathrm{rad}\right). \label{Cons2}
	\end{align}
\end{subequations}
We are interested in asymptotically flat regular ball solutions, i.e. solutions of~\eqref{RelPot}-\eqref{ConsEM} with a regular center of symmetry and compact support. The precise definition of regular and strongly regular balls of matter is given below.

\begin{definition}[Relativistic Ball of Matter]\label{Def1}
	Given $T>0$ and a function $\mathcal{R}\in C^{1}([0,T],\mathbb{R}^{+})$, define $\Omega=\left\{(t,r):0\leq t\leq T,0\leq r\leq \mathcal{R}(t)\right\}$. A \emph{relativistic self-gravitating ball of matter} with a regular center and boundary $\partial\Omega=\left\{(t,r):0\leq t\leq T, r = \mathcal{R}(t)\right\}$ is given by the set of functions $\phi,m,\rho,p_\mathrm{rad},p_\mathrm{tan},v\in C^{0}(\Omega)\cap C^{1}(\Omega\backslash\{r=0\})$ satisfying~\eqref{RelPot}-\eqref{ConsEM} in $\Omega\backslash \{r=0\}$ such that:
	\begin{itemize}
		\item[(i)] $m(t,0)=0$, and $\lim_{r\rightarrow 0^{+}}\frac{m(t,r)}{r}=0$ for all $t\in[0,T]$;
		\item[(ii)] $\rho(t,r)>0$ for $(t,r)\in\Omega\backslash\partial\Omega$;
		\item[(iii)] $p_\mathrm{tan}(t,0)=p_\mathrm{rad}(t,0)$ for all $t\in[0,T]$;
		\item[(iv)] $p_{\mathrm{rad}}(t,\mathcal{R}(t))=0$ for all $t\in[0,T]$;
		\item[(v)] 	$v(t,0)=0$, and $\displaystyle \frac{e^{\phi}v}{\langle v\rangle} (t,\mathcal{R}(t))=\frac{d\mathcal{R}(t)}{dt}$ for all $t\in[0,T]$.
	\end{itemize}
	Moreover, the ball is said to be:
	\begin{itemize}
	\item
	(radially) \emph{compressed} if $p_{\mathrm{rad}}(t,r)>0$ for all $(t,r)\in \Omega\backslash\partial\Omega$;
	\item
	(radially) \emph{stretched} if $p_{\mathrm{rad}}(t,r)<0$ for all $(t,r)\in \Omega\backslash\partial\Omega$;
	\item
	\emph{strongly regular at the center} if $\phi,m,\rho,p_\mathrm{rad},p_\mathrm{tan},v\in C^{1}(\Omega)$, and, for all $t\in[0,T]$:
 $$\lim_{r\rightarrow 0^{+}}\frac{m(t,r)}{r^2}= \lim_{r\rightarrow 0^{+}}\partial_r m(t,r)= \lim_{r\rightarrow 0^{+}}\partial_r \phi(t,r)=0;$$
        $$\lim_{r\rightarrow 0^{+}}\partial_r\rho(t,r)=\lim_{r\rightarrow 0^{+}}\partial_r p_\mathrm{rad}(t,r) = \lim_{r\rightarrow 0^{+}}\partial_r p_\mathrm{tan}(t,r)=0;$$
        $$ \lim_{r\rightarrow 0^{+}}\partial_r v(t,r)=\lim_{r\rightarrow 0^{+}}\frac{v(t,r)}{r}=\omega(t)\in C^{0}([0,T])$$
        $$\lim_{r\rightarrow 0^{+}}\partial_t \rho(t,r)=-3e^{\phi(t,0)}(\rho(t,0)+p_\mathrm{rad}(t,0))\omega(t);$$
        $$\lim_{r\rightarrow 0^{+}}\partial_t v(t,r)=0; \quad\lim_{r\rightarrow 0^{+}}\partial_t p_\mathrm{tan}(t,r)=\lim_{r\rightarrow 0^{+}}\partial_t p_\mathrm{rad}(t,r).$$
   	\end{itemize}
\end{definition}
Condition $(i)$ guarantees that the metric is regular at the center of symmetry. The physical condition $(ii)$ requires that energy density is positive in the interior of the ball, while $(iii)$ is the isotropic center condition. Condition $(iv)$ is the boundary condition, and $(v)$ implies that the total mass of the ball is conserved, as shown in the following proposition:
\begin{proposition} \label{totalmass}
	The total mass of the ball, $\mathcal{M}=m(t,\mathcal{R}(t))$, is conserved. 
\end{proposition}
\begin{proof}
	The proof follows by computing the total time derivative of the mass along the boundary of the ball. First we note that
	\begin{equation}
	u^{\mu}\partial_\mu m(t,r)=e^{-\phi}\langle v\rangle \partial_t m+v\partial_r m=-4\pi r^2 p_\mathrm{rad} v.
	\end{equation}
	The boundary conditions $(iv)$ and $(v)$ then yield $\displaystyle \frac{d}{dt} m(t,\mathcal{R}(t))=0$.
\end{proof}
	Due to Birkhoff's theorem, in the region $r>\mathcal{R}(t)$ outside the matter support, where $\rho=p_\mathrm{rad}=p_\mathrm{tan}=0$ and $v$ is not defined, the metric is the standard Schwarzschild metric \cite{1916AbhKP1916..189S}, with $m(t,r)=\mathcal{M}$ and $\phi(t,r)=\frac{1}{2}\ln{(1-\frac{2\mathcal{M}}{r})}$.
Therefore, the \emph{asymptotic flatness condition} holds: 
\begin{equation}
\lim_{r\rightarrow+\infty}\phi(t,r)=0,\quad\quad\lim_{r\rightarrow+\infty}\frac{m(t,r)}{r}=0.
\end{equation}
Note that the first condition fixes the residual gauge freedom in the choice of $t$.
Under this condition, the relativistic gravitational potential has the explicit form
\begin{equation}\label{integrophi}
\phi(t,r)= - \int^{+\infty}_{r}\frac{1}{1-\frac{2m}{s}}\left[\frac{m}{s^2}+4\pi s\left(p_\mathrm{rad}+\frac{\rho+p_\mathrm{rad}}{1-\frac{2m}{s}}v^2\right)\right]ds,
\end{equation}
where $\lim_{r\to 0^+}\phi(t,r)=\phi_c(t), t\in[0,T]$, is finite for regular balls due to Proposition~\ref{SReg} below.

The system of equations~\eqref{ConsEM} is valid for any anisotropic spherically symmetric matter model. However, there are too many unknowns, namely $\rho$, $p_\mathrm{rad}$, $p_\mathrm{tan}$ and $v$, for the available equations, i.e., the system is undetermined. In order to close the system of equations, one needs to specify an EoS.

If the matter making up the body consists of a \emph{perfect fluid}, then  $p_\mathrm{rad}=p_\mathrm{tan}=p_\mathrm{iso}$, and $p_\mathrm{iso}>0$ for all $r\in[0,\mathcal{R}(t))$, where $p_\mathrm{iso}$ denotes the \emph{isotropic pressure}
\begin{equation}
p_\mathrm{iso}(t,r)=\frac{1}{3}(p_\mathrm{rad}(t,r)+2p_\mathrm{tan}(t,r)).
\end{equation}
In this case, the choice of EoS amounts to expressing the energy  density $\rho$ as a function of the number particle density $n$, $\rho=\widehat{\rho}(n)$, with $p_\mathrm{iso}=\widehat{p}_\mathrm{iso}(n)$ given by equation~\eqref{FluidEoS}. Equations~\eqref{ConsEM} then become a system of two equations for the two unknowns $n$ and $v$. The equation for the conservation of the particle number current~\eqref{nCons}, which in spherical symmetry is written as
%
\begin{align}\label{Consn}
& e^{-\phi}\langle v\rangle \partial_t n+v\partial_r n+\frac{nv} {1-\frac{2m}{r}+v^2}e^{-\phi}\langle v\rangle\partial_t v +n\partial_r v=-\frac{2}{r}nv,
\end{align}
is a consequence of~\eqref{Cons1} in this case, and~\eqref{Cons2} reduces to the spherically symmetric relativistic Euler equation,
\begin{equation}\label{EulerSS}
    \frac{n}{1-\frac{2m}{r}+v^2}\left(e^{-\phi}\langle v\rangle\partial_t v+v\partial_r v\right)+\frac{v c^2_\mathrm{s}}{1-\frac{2m}{r}+v^2}e^{-\phi}\langle v\rangle\partial_t n+c^2_\mathrm{s}(n)\partial_r n=-n\frac{\left(\frac{m}{r^2}+4\pi r p_\mathrm{iso}\right)}{1-\frac{2m}{r}+v^2}.
\end{equation}
For solutions with a regular center, we have
\begin{equation}
    \lim_{r\rightarrow0^+}n(t,r)=n(t,0)=n_\mathrm{c}(t), \quad t\in[0,T].
\end{equation}
Moreover from
\begin{equation}
   \lim_{r\rightarrow 0^+} \frac{\partial\rho}{\partial r}=\frac{(\widehat{\rho}(n_\mathrm{c}(t))+\widehat{p}_\mathrm{iso}(n_\mathrm{c}(t))}{n_c(t)}\lim_{r\rightarrow0^+}\frac{\partial n}{\partial r},
\end{equation}
the strong regularity condition $\lim_{r\rightarrow 0^+}\partial_r\rho(t,r)=0$ is equivalent to
\begin{equation}\label{SRnr}
    \lim_{r\rightarrow0^+}\partial_rn(r,t)=0, \quad t\in[0,T].
\end{equation}
Similarly, for solutions with a strongly regular center it must hold that
\begin{equation}\label{SRnt}
    \lim_{r\rightarrow 0^{+}}\partial_t n(t,r)=-3n_c(t)e^{\phi(t,0)}\omega(t) \quad t\in[0,T].
\end{equation}
\begin{proposition} \label{totalnumber}
The total number $\mathcal{N}$ of particles inside the ball is conserved. 
\end{proposition}
\begin{proof}
The number of particles inside the ball is given by~\footnote{The total number of particles on a spacelike surface $\Sigma$ is given by $\mathcal{N}=\int_{\Sigma} n u^{\mu} n_{\mu} dV_3$, where $n_{\mu}$ is the unit past-pointing normal vector.}
\begin{equation}
\mathcal{N}(t)=4\pi\int^{\mathcal{R}(t)}_{0}\frac{\langle v\rangle n r^2}{\sqrt{1-\frac{2m}{r}}} dr.
\end{equation}
Since equation~\eqref{Consn} for the conservation of the particle number current can be written as
\begin{equation}\label{Divn}
\partial_t \left(\frac{\langle v\rangle n r^2}{\sqrt{1-\frac{2m}{r}}}\right) + \partial_r \left(\frac{v n e^\phi r^2}{\sqrt{1-\frac{2m}{r}}}\right) = 0,
\end{equation}
we have
\begin{equation}
\mathcal{N}'(t)=4\pi\mathcal{R}'(t)\left(\frac{\langle v\rangle n r^2}{\sqrt{1-\frac{2m}{r}}}\right)(t,\mathcal{R}(t)) - 4\pi\int^{\mathcal{R}(t)}_{0}\partial_r \left(\frac{v n e^\phi r^2}{\sqrt{1-\frac{2m}{r}}}\right) dr = 0,
\end{equation}
in view of the boundary condition $(v)$ in Definition~\ref{Def1}.
\end{proof}
For \emph{elastic} materials, there is the possibility of equilibrium configurations that are stretched (for instance, in~\cite{Costa:2018oim} Costa \& Nat\'ario allow for stretched regions in the Christodoulou hard phase material).
Moreover, there is no \emph{a priori} reason to specify the sign of $p_\mathrm{tan}$, as it might change along the body (although, as we will see, reality of the squared velocity of traverse waves might impose additional physical restrictions on $p_\mathrm{tan}$, see Definition~\ref{Speeds}). Moreover, the matching condition at the boundary of the star to an exterior Schwarzschild is solely $p_\mathrm{rad}(t,\mathcal{R}(t))=0$.
Thus, in general, $\rho$ and $p_\mathrm{tan}$ are discontinuous at the boundary, in contrast with standard fluids such as polytropes, where $\rho(t,\mathcal{R}(t))=0$ (there are, however, fluids -- such as those with an affine EoS -- that do have positive energy density on the boundary). 

In the static case, a theorem by Rendall \& Schimdt~\cite{RS91} provides the conditions for  uniqueness and $C^1$ regularity of relativistic spherically symmetric perfect fluid balls for a given baryon number density at the center $n_\mathrm{c}$. Later, in~\cite{BR93}, this analysis was generalized to anisotropic matter models (however, $p_\mathrm{tan}$ was assumed a priori to be $C^1$ up to the center of symmetry). Such arguments were applied to elastic material models in~\cite{Par00,FK07}. 

\begin{proposition}\label{SReg}
	For regular ball solutions, the following conditions hold:
	\begin{equation}
	\partial_r m (t,0)=0;\qquad \lim_{r\rightarrow 0^{+}}\frac{m(t,r)}{r^2}=0;\qquad \partial_r\phi(t,0)=0.
	\end{equation}
\end{proposition}
\begin{proof}
	From condition $(i)$ in Definition~\ref{Def1} it follows that $\partial_r m(t,0)=0$. Equation~\eqref{HwkMass} and L'H\^opital's rule then imply $\lim_{r\rightarrow 0^{+}}\frac{m(t,r)}{r^2}=0$. Using this fact in~\eqref{RelPot} yields $\partial_r\phi(t,0)=0$.
\end{proof}
In view of the above proposition, given a regular solution to the spherically symmetric Einstein equations with anisotropic matter models, the issue in proving that regular ball solutions are strongly regular at the center consists essentially in showing that the \emph{anisotropic pressure} $q(t,r)$, defined by
\begin{equation}
q(t,r)=p_\mathrm{tan}(t,r)-p_\mathrm{rad}(t,r),
\end{equation}
satisfies~\cite{1974ApJ...188..657B}
\begin{equation*}
\lim_{r\rightarrow 0^+}\frac{q(t,r)}{r}=0\quad t\in[0,T].
\end{equation*}
Therefore perfect fluids balls are are strongly regular at the center if 
\begin{equation}
\frac{d\widehat{p}_\mathrm{iso}}{dn}(n_\mathrm{c}(t))\neq0, \quad t\in[0,T],
\end{equation}
In particular, since $\rho+p_\mathrm{iso}>0$ for all $n>0$, the strict hyperbolicity condition $c^2_\mathrm{s}(n)>0$ implies the last condition.
In~\cite{Alho:2019fup} it was shown that a similar result holds for homogeneous and isotropic elastic materials in the static Newtonian setting. Such methods are easily applicable to the relativistic setting (see Section~\ref{SS}).

\subsection{TOV equation and radial stability of static configurations}
Assuming staticity, $\rho$, $p_\mathrm{rad}$ and $p_\mathrm{tan}$ are time-independent functions of the areal coordinate $r$ only, and the radial velocity vanishes, $v=0$. In this case, the system of equations~\eqref{ConsEM} reduces to a nonlinear integro-differential equation known as the Tolman-Oppenheimer-Volkoff (TOV) equation \cite{1939PhRv...55..374O}:
\begin{equation}\label{TOVeq}
\frac{dp_\mathrm{rad}}{dr} =\frac{2}{r}(p_\mathrm{tan}-p_\mathrm{rad})-\frac{\rho+p_\mathrm{rad}}{1-\frac{2m}{ r}}\left(\frac{m}{r^2}+4\pi r p_{\mathrm{rad}}\right),
\end{equation}
%
where the Hawking mass $m(r)$ and the relativistic gravitational potential $\phi(r)$ are explicitly given by
\begin{equation}\label{TOVmassdefdelta}
m(r)=4\pi\int^{r}_{0}\rho(s)s^2ds^2, \qquad \phi(r)= -\int^{+\infty}_r\frac{1}{1-\frac{2m}{s}}\left(\frac{m}{s^2}+4\pi s p_\mathrm{rad}\right)ds.
\end{equation}
Linear radial perturbations around the steady state solutions $(\mathring{\rho},\mathring{p}_\mathrm{rad},\mathring{p}_\mathrm{tan},\mathring{v}=0)$ are obtained by making the ansatz
$\rho(t,r)=\mathring{\rho}(r)+\rho_\mathrm{E}(t,r)$, $p_\mathrm{rad}(t,r)=\mathring{p}_\mathrm{rad}(r)+(p_\mathrm{rad})_\mathrm{E}(t,r)$, $p_\mathrm{tan}(t,r)=\mathring{p}_\mathrm{tan}(r)+(p_\mathrm{tan})_\mathrm{E}(t,r)$, $v(t,r)=v_\mathrm{E}(t,r)$, $\phi(t,r)=\mathring{\phi}(r)+\phi_\mathrm{E}(t,r)$, $m(t,r)=\mathring{m}(r)+m_\mathrm{E}(t,r)$, where the subscript $\mathrm{E}$ stands for ``Eulerian pertubations" and $\,\mathring{}\,$ denotes background quantities. In what follows, and in order to simplify the notation, we will drop the latter, as it is clear from the context when it would be present (the coefficients of the perturbation equations contain background quantities only).  The linearized system is given by 
\begin{equation}\label{LinPot}
\partial_r \phi_\mathrm{E}=\frac{1}{1-\frac{2m}{r}}\left[4\pi r(p_\mathrm{rad})_\mathrm{E}+\left(\frac{1}{r}+2\frac{d\phi}{dr}\right)\frac{m_\mathrm{E}}{r}\right] ,
\end{equation}
\begin{equation}\label{Linmass}
\partial_t m_\mathrm{E}=-4\pi r^2(\rho+p_\mathrm{rad})e^{\phi}v_\mathrm{E},\qquad \partial_rm_\mathrm{E}=4\pi r^2 \rho_\mathrm{E},
\end{equation}
\begin{subequations}
	\begin{align}
	&e^{-\phi}\partial_t\rho_\mathrm{E}+(\rho+p_\mathrm{rad})\partial_r v_\mathrm{E} =-\left(\frac{d\rho}{dr}+\frac{2}{r}(\rho+p_\mathrm{tan})\right)v_\mathrm{E}, \label{linearized1}\\
	&\frac{e^{-\phi}(\rho+p_\mathrm{rad})}{1-\frac{2m}{r}}\partial_t v_\mathrm{E}+\partial_r (p_\mathrm{rad})_\mathrm{E} = \frac{2}{r}\left((p_\mathrm{tan})_\mathrm{E}-(p_\mathrm{rad})_\mathrm{E}\right)-\frac{\rho+p_\mathrm{rad}}{1-\frac{2m}{r}}\left(\frac{m_\mathrm{E}}{r^2}+4\pi r (p_\mathrm{rad})_\mathrm{E}\right) \nonumber\\
	&\qquad\qquad\qquad\qquad\qquad\qquad\qquad -\frac{\frac{m}{r^2}+4\pi r p_\mathrm {rad}}{1-\frac{2m}{r}}\left(\rho_\mathrm{E}+(p_\mathrm{rad})_\mathrm{E}+\frac{\rho+p_\mathrm{rad}}{1-\frac{2m}{r}}\frac{2m_\mathrm{E}}{r}\right).\label{linearized2}
	\end{align}
\end{subequations}
We wish to solve this system of linear partial differential equations subject to the initial conditions
\begin{equation}(\phi_\mathrm{E}, m_\mathrm{E}, \rho_\mathrm{E}, (p_\mathrm{rad})_\mathrm{E}, (p_\mathrm{tan})_\mathrm{E}, v_\mathrm{E})(0,r)=(0,0,0,0,0,h(r)), \qquad 0<r<\mathcal{R},
\end{equation}
where $h(r)$ is an arbitrary function, and appropriate boundary conditions. The standard procedure is to introduce the \emph{radial displacement field} in the perturbed configuration, 
\begin{equation}
r\rightarrow r+\xi(t,r),
\end{equation}
satisfying
\begin{equation}
\xi(0,r)=0 ,\qquad 0<r<\mathcal{R}.
\end{equation}
We assume that $\xi\in C^{1}(\Omega)\cap C^{2}(\Omega\backslash\{r=0\})$ if the ball of matter is regular, and $\xi\in C^{2}(\Omega)$ if it is strongly regular. This is compatible with the perturbation equations below, given the extra regularity of $\rho$ and $\phi$ resulting from the TOV equations~\eqref{TOVeq} and \eqref{TOVmassdefdelta}.

The linear perturbed radial velocity is given in terms of $\xi$ by
\begin{equation}\label{LinVel}
v_\mathrm{E}(t,r) = e^{-\phi(r)}\partial_t \xi(t,r)
\end{equation}
so that 
\begin{equation}
    \partial_t\xi(0,r)=e^{\phi(r)}h(r),\qquad 0<r<\mathcal{R}.
\end{equation}
The Eulerian perturbation of the mass function in terms of the displacement field is obtained using the first equation in~\eqref{Linmass} and integrating:
\begin{equation}\label{LinMass}
m_\mathrm{E}(t,r)= -4\pi r^2 (\rho+p_\mathrm{rad}) \xi(t,r).
\end{equation}
Using~\eqref{LinVel} in equation~\eqref{linearized1}, and again integrating in $t$, yields
\begin{equation}\label{LinRho}
\rho_\mathrm{E}(t,r) = - (\rho+p_\mathrm{rad})\partial_r\xi(t,r)-\left[\frac{2}{r}(p_\mathrm{tan}-p_\mathrm{rad})+(\rho+p_\mathrm{rad})\left(\frac{2}{r}-\frac{d\phi}{dr}\right)+\frac{d\rho}{dr}\right]\xi(t,r).
\end{equation}
At this point, and in order to close the system of equations, one needs to postulate an EoS relating the pressures with the energy density. Once this is done, the linearized equation~\eqref{linearized2} results in a wave equation for the displacement field $\xi$.

By the characterization of regular ball solutions given in Definition~\ref{Def1}, the linearized equations should be solved subject to the regular center conditions
\begin{equation}\label{BoundCond1}
\lim_{r\rightarrow 0^+}v_\mathrm{E}(t,r)=0,\qquad \lim_{r\rightarrow 0^+}\left(\frac{m}{r}\right)_\mathrm{E}(t,r)=0 ,
\end{equation}
i.e., that the perturbed radial velocity vanishes at the center, and that perturbed metric is regular at the center of symmetry. This 
%
%
yields the conditions
\begin{equation}
\lim_{r\rightarrow 0^{+}} \partial_t\xi(t,r)=0, \qquad \lim_{r\rightarrow 0^{+}} r\xi(t,r)=0.
\end{equation}
Moreover, by Proposition~\ref{SReg}, the above implies
\begin{equation}\label{BoundCond2}
\lim_{r\rightarrow 0^+}\left(\frac{m}{r^2}\right)_\mathrm{E}(t,r)=0 ,
\end{equation}
and therefore
\begin{equation}\label{xiCenterCond}
\lim_{r\rightarrow 0^{+}} \xi(t,r)=0.
\end{equation}
In addition, ball solutions with a strongly regular center satisfy
\begin{equation}
    \lim_{r\rightarrow0^+}\frac{v_\mathrm{E}(t,r)}{r}=\lim_{r\rightarrow0^+}\partial_r v_\mathrm{E}(t,r)=\omega_\mathrm{E}(t),
\end{equation}
which implies that
\begin{equation}\label{sRg}
    \lim_{r\rightarrow 0^+}\frac{\xi(t,r)}{r}=\lim_{r\rightarrow 0^+}\partial_r \xi(t,r)=g(t)\in C^1([0,T]).
\end{equation}
At the surface of the ball, the Lagrangian perturbation of the radial pressure and of the mass function should vanish,~\footnote{See e.g. Karlovini et al~\cite{Karlovini:2003xi} for a discussion on the junction conditions and fast and slow phase transitions.} 
\begin{equation}\label{BoundCond}
(p_\mathrm{rad})_\mathrm{L}(t,\mathcal{R}(t)) = 0,\qquad m_\mathrm{L}(t,\mathcal{R}(t))=0
\end{equation}
(we use a subscript $\mathrm{L}$ to denote Lagrangian perturbations), stating that the particles on the boundary of the ball remain so in the perturbed configuration, and that the mass should be the same as in the background solution, so that the exterior Schwarzschild solution has the same total mass $\mathcal{M}$. To go from the Eulerian to the Lagrangian gauge in Schwarzschild coordinates, we use the change of gauge formula 
\begin{equation}
\Box_{\mathrm{L}}=\Box_{\mathrm{E}}+\xi(t,r)\frac{\partial}{\partial r}
\end{equation}
acting on scalars, see e.g.~\cite{Karlovini:2003xi}.

For background solutions with a perfect fluid, $\widehat{p}_\mathrm{iso}(n)$ and $\widehat{\rho}(n)$ are related by equation~\eqref{FluidEoS}, and the TOV equation~\eqref{TOVeq} results in an ordinary differential equation for the particle number density $n(r)$, which together with the equation for the mass $m(r)$ form a closed first order system describing static spherically symmetric perfect fluid objects. Due to the result in~\cite{RS91}, for each value of the central baryon number density $n_\mathrm{c}$ there exists a unique strongly regular solution of the TOV equation. The solutions with compact support in this 1-parameter family have a well defined total mass $\mathcal{M}(n_\mathrm{c})$ and radius $\mathcal{R}(n_\mathrm{c})$, and so determine a $C^1$ curve in the $\mathcal{M}-\mathcal{R}$ plane, the so-called $\mathcal{M}(\mathcal{R})$ curve.

In this case, relation~\eqref{FluidEoS} yields
\begin{equation}\label{LinEoS}
\rho_\mathrm{E}= \frac{d\widehat{\rho}}{dn} n_\mathrm{E}=(\widehat{p}_\mathrm{iso}+\widehat{\rho})\frac{n_\mathrm{E}}{n}, \qquad (p_\mathrm{iso})_\mathrm{E}=\frac{d\widehat{p}_\mathrm{iso}}{dn} n_\mathrm{E}=(\widehat{p}_\mathrm{iso}+\widehat{\rho}) c^2_\mathrm{s}\frac{n_\mathrm{E}}{n},
\end{equation}
and the linearization of the particle current conservation equation~\eqref{Consn} gives
\begin{equation} \label{LinCons}
e^{-\phi}\partial_t n_\mathrm{E}+n\partial_r v_\mathrm{E}=-\left(\frac{dn}{dr}+\frac{2}{r}n\right) v_\mathrm{E}.
\end{equation}
Now, using equation~\eqref{LinCons} together with~\eqref{LinVel}, we obtain
\begin{equation}
n_\mathrm{E}(t,r)=-n\partial_r\xi(t,r)-n\left[\frac{2}{r}-\frac{d\phi}{dr}+\frac{1}{n}\frac{dn}{dr}\right]\xi(t,r).
\end{equation}
The linearization of the relativistic Euler equation~\eqref{EulerSS} then results in a linear wave equation for the displacement field:
	\begin{align}\label{WEdisp}
	&\frac{e^{-2\phi}}{1-\frac{2m}{r}}\partial^2_t\xi(t,r) = c^2_\mathrm{s}\partial^2_r \xi(t,r) + \left[\frac{d}{dr}\left(c^2_\mathrm{s}-\phi\right)-c^2_\mathrm{s}\left(\frac{d\phi}{dr}-\frac{2}{r}-\frac{4\pi r(\rho+p_\mathrm{iso})}{1-\frac{2m}{r}}\right)\right]\partial_r\xi(t,r) \nonumber\\
 &+\left[2\left(\frac{d\phi}{dr}\right)^2-\frac{d}{dr}\left(c^2_\mathrm{s}\left(\frac{d\phi}{dr}-\frac{2}{r}\right)\right)+\frac{2m}{r^3\left(1-\frac{2m}{r}\right)}-\frac{4\pi r(\rho+p_\mathrm{iso})}{1-\frac{2m}{r}}c^2_\mathrm{s}\left(\frac{d\phi}{dr}-\frac{2}{r}\right)\right]\xi(t,r).
	\end{align}
%
%
%
Note that for solutions with a strongly regular center, the condition $\lim_{r\rightarrow 0^+} \partial_t v_{\rm E}(t,r)=0$ implies
\begin{equation}
\lim_{r\rightarrow 0^+} \partial^2_t \xi(t,r)=0.
\end{equation}
This implies that the right-hand side of~\eqref{WEdisp} must vanish in the limit $r\rightarrow 0^+$; a simple application of L'H\^opital's rule shows that the unbounded terms in this right-hand side cancel the term $c^2_\mathrm{s}\partial^2_r \xi(t,r)$ in this limit.

%
%
For a perfect fluid the boundary conditions at the surface of the star~\eqref{BoundCond} give
\begin{subequations}
\begin{align}
    (p_\mathrm{iso})_\mathrm{L}(t,\mathcal{R}(t)) &= -\frac{e^{\phi}}{\mathcal{R}^2}n(\mathcal{R})\frac{dp_\mathrm{iso}}{dn}(\mathcal{R})\partial_r(r^2e^{-\phi}\xi)(t,\mathcal{R}(t))=0, \label{LagBoundCond} \\ 
m_\mathrm{L}(t,\mathcal{R}(t)) &= -4\pi \mathcal{R}^2 p_\mathrm{iso}(\mathcal{R}) \xi(t,\mathcal{R}(t))=0.
\end{align}
\end{subequations}
The last condition is automatically satisfied, since $p_\mathrm{iso}(\mathcal{R})=0$. Equation~\eqref{LagBoundCond} suggests that one introduce the \emph{renormalized radial displacement field}
\begin{equation}
\zeta(t,r)=r^2e^{-\phi(r)}\xi(t,r),
\end{equation}
so that the wave equation~\eqref{WEdisp}, in the case of a perfect fluid, can be written in the form
\begin{equation}
W_{\mathrm{pf}}(r)\partial^2_t\zeta(t,r)-\partial_r\left(P_{\mathrm{pf}}(r)\partial_r\zeta(t,r)\right)-Q_{\mathrm{pf}}(r)\zeta(t,r)=0,
\end{equation}
where
\begin{align}
& r^2W_{\mathrm{pf}}(r)=\frac{(\widehat{\rho}+\widehat{p}_\mathrm{iso})e^{\phi}}{\left(1-\frac{2m}{r}\right)^{3/2}},\qquad r^2P_{\mathrm{pf}}(r)=\frac{(\widehat{\rho}+\widehat{p}_\mathrm{iso})c^2_\mathrm{s}e^{3\phi}}{\left(1-\frac{2m}{r}\right)^{1/2}}, \\
& r^2Q_{\mathrm{pf}}(r)=\frac{(\widehat{\rho}+\widehat{p}_\mathrm{iso})e^{3\phi}}{\left(1-\frac{2m}{r}\right)^{1/2}}\left[\left(\frac{d\phi}{dr}\right)^2+\frac{4}{r}\frac{d\phi}{dr}-\frac{8\pi \widehat{p}_\mathrm{iso}}{1-\frac{2m}{r}}\right].
\end{align}
The standard procedure for studying linear stability consists in looking for harmonic solutions, i.e., making the ansatz
\begin{equation}
\zeta(t,r)=\zeta_0(r)e^{i\omega t}, \qquad \zeta_0(r)=\zeta(0,r),
\end{equation}
which leads to a self-adjoint boundary value problem involving the squared oscillating frequency $\omega^2$:
\begin{equation}
\frac{d}{dr}\left(P_\mathrm{pf}\frac{d\zeta_0}{dr}\right)+\left[Q_\mathrm{pf}+\omega^2 W_\mathrm{pf}\right]\zeta_0=0,
\end{equation}
\begin{equation}
\lim_{r\rightarrow 0^{+}}\frac{\zeta_0}{r^2}=0,\qquad \lim_{r\rightarrow \mathcal{R}}P_\mathrm{pf}\frac{d\zeta_0}{dr}=0.
\end{equation}
Introducing $\chi_0(r)=
P_\mathrm{pf}d\zeta_0/dr$ leads to a first order system for the variables $(\zeta_0,\chi_0)$. The above equation is in standard Sturm-Liouville form. If $W_\mathrm{pf}>0$, $P_\mathrm{pf}>0$, and $r^2P_\mathrm{pf}$, $r^2\frac{dP_\mathrm{pf}}{dr}$, $r^2Q_\mathrm{pf}$, $r^2W_\mathrm{pf}$  are regular functions in $r\in[0,{\cal R}]$, then standard results from \emph{regular} Sturm-Liouville theory follow. 
The eigenvalues $\omega^{2}$ are real and non-degenerate, and form an ordered infinite sequence $\{\omega^2_j\}$,  $\omega^2_0<\omega^2_1<\omega^2_2<...$, where $j$ is the number of distinct radial nodes of the corresponding eigenfunction in the interval $(0,\mathcal{R})$. The eigenfunctions (also called normal modes) corresponding to different eigenvalues are orthogonal with respect to the weight function $W$. The eigenfunction of the fundamental mode (or ground state), corresponding to the eigenvalue $\omega^2_0$, has no nodes, so that $\zeta_0$ has the same sign in $(0,\mathcal{R})$. If $\omega^2_0>0$ then all $\omega_j$ are real, all solutions are purely oscillatory, and the steady sate is said to be \emph{linearly stable}. If $\omega_0^2<0$ then the corresponding solution grows exponentially, and the steady state is said to be \emph{linearly unstable}.

As first realized in~\cite{harrison1965gravitation}, 
the stability of the modes of a perfect fluid ball solution changes (i.e. $\omega^2_j=0$ for some $j$) only at extrema of the curve $\mathcal{M}(\mathcal{R})$, where $d\mathcal{M}/d\mathcal{R}=0$; the mode becoming unstable at such extrema is even $(j=0,2,...)$ if $d\mathcal{R}/dn_\mathrm{c}<0$, and it is odd $(j=1,3,...)$ if $d\mathcal{R}/dn_\mathrm{c}>0$. Typically, one looks at the stability of the fundamental mode, characterized by $\omega_0$, which vanishes at some central density $n_\mathrm{c}$ yielding the configuration of maximum mass; the higher modes then become unstable for higher central densities still. These results on the spectral stability/instability and the turning point principle, as well as the resulting spiral structure of the $\mathcal{M}(\mathcal{R})$ curve, have been recently proved in~\cite{Hadi2021StaIns,Hadi2021TurningPP}.

\subsection{Bounds on compactness}
{\it What is the maximum compactness that a self-gravitating material object can support within GR?}
This fundamental question was addressed by Buchdahl in 1959, in the context of perfect fluid models~\cite{Buchdahl:1959zz}. He showed that self-gravitating, spherically symmetric, perfect fluid GR solutions satisfy the following bound on the compactness: 
\begin{equation}
\mathcal{C}\equiv \frac{\mathcal{M}}{\mathcal{R}}\leq\frac{4}{9}\approx 0.444.
\end{equation}
On the other hand, the horizon radius of a Schwarzschild black hole is $\mathcal{R}=2\mathcal{M}$, i.e. $\mathcal{C}=1/2$. Thus, as an important consequence, Buchdahl's theorem forbids the existence within GR of fluid objects whose compactness is arbitrarily close to the black hole limit, providing an important cornerstone for tests of the nature of compact objects~\cite{Cardoso:2019rvt}.

However, the bound in Buchdhal's theorem is achieved by an incompressible fluid \cite{1916skpa.conf..424S}, with infinite speed of sound, and therefore breaks causality. Moreover, in GR any physically viable matter model should also satisfy the energy conditions:
\begin{definition}[Energy conditions]\label{EC} 
	In spherical symmetry, the standard energy conditions are given by~\cite{HawkingEllis}: 
	\begin{subequations}
		\begin{align}
		\mathrm{SEC}&:\quad \rho+3p_\mathrm{iso}\geq 0,\quad \rho+p_\mathrm{rad}\geq 0, \quad \rho+p_\mathrm{tan}\geq 0; \\
		\mathrm{WEC}&:\quad \rho\geq 0,\quad \rho+p_\mathrm{rad}\geq 0, \quad \rho+p_\mathrm{tan}\geq 0; \\
		\mathrm{NEC}&:\quad \rho+p_\mathrm{rad}\geq 0, \quad \rho+p_\mathrm{tan}\geq 0; \\
		\mathrm{DEC}&:\quad \rho\geq |p_\mathrm{rad}|, \quad \rho\geq |p_\mathrm{tan}|,
		\end{align}
	\end{subequations}
	for the strong, weak, null, and dominant energy conditions, respectively.
\end{definition}
The incompressible fluid of Buchdahl's bound also violates the dominant energy condition.  
For perfect fluids, it has been recently conjectured that the so-called \emph{causal} Buchdahl bound holds (see~\cite{Urbano:2018nrs,Boskovic:2021nfs}):
\begin{equation}
\mathcal{C}_{\rm PA}^{\rm fluid}\lesssim 0.364\,,
\end{equation}
for physically admissible~(PA) configurations, that is, configurations whose speed of sound does not exceed the speed of light. A further reasonable physical condition is to require stability of the equilibrium configurations. In the fluid case, this constraint is more stringent than causality and imposes 
\begin{equation}
\mathcal{C}_{\rm PAS}^{\rm fluid}\lesssim 0.354
\end{equation}
for physically admissible, radially stable~(PAS) configurations~\cite{Lindblom:1984,Lattimer:2006xb,Urbano:2018nrs,Boskovic:2021nfs}. Both these bounds are attained by the affine fluid model of Example~\ref{Ex3} with $\gamma=2$, i.e.~$c^2_\mathrm{s}=1$.

Generalized versions of Buchdhal's bound, allowing for larger compactness, can be obtained for anisotropic materials (see~\cite{Guven:1999wm, Andreasson:2007ck, Karageorgis:2007cy, Urbano:2018nrs}). Indeed, compact objects made of strongly anisotropic materials (e.g., gravastars~\cite{Mazur:2004fk} and anisotropic stars~\cite{Raposo:2018rjn}) can have higher compactness and a continuous black hole limit $\frac{\mathcal{M}}{\mathcal{R}} \to \frac{1}{2}$~\cite{Pani:2015tga,Uchikata:2015yma,Uchikata:2016qku,Beltracchi:2021lez}.
However, the viability of such ultracompact models is questionable, since they either violate some of the energy conditions~\ref{EC}, or feature superluminal speeds of sound or ad-hoc thin shells within the fluid (see~\cite{Cardoso:2019rvt} for a discussion). On the other hand, viable models like boson stars (which also feature anisotropies in the matter fields) are not significantly more compact than an ordinary perfect fluid neutron star in the static case~\cite{Liebling:2012fv,Boskovic:2021nfs}.

For matter fields satisfying the DEC and with $p_{\rm rad}, p_{\rm tan}\geq0$, the maximum compactness is $\mathcal{C}_{\rm DEC}\lesssim 0.4815$~\cite{Karageorgis:2007cy}.
%
As we shall see, for elastic balls under radial compression, and imposing causality, we obtain solutions that have a maximum compactness consistent with, but not saturating, this upper bound. This happens for the ACS model in Table~\ref{table}. Further imposing radial stability will force physically admissible solutions to have an even smaller maximum compactness.

%
%

\newpage

\section{Relativistic elastic balls}\label{ElasticBalls}
The Einstein equations are closed by postulating an EoS relating the pressures and the density. In the case of a perfect fluid, the theory of thermodynamics states that this relation is given by~\eqref{FluidEoS} . In Appendix~\ref{App:B} we show that for spherically symmetric \emph{homogeneous} and \emph{isotropic} elastic materials similar relations hold, which we summarize in the following definition:
\begin{definition}[Relativistic Elastic Balls]\label{DefElaBall} 
	A self-gravitating  relativistic ball characterized by the functions  $\rho$, $p_\mathrm{rad}$, $p_\mathrm{tan}$ and $v$ is said to be composed of homogeneous and isotropic elastic matter if there exists a $C^3$ function $\widehat{\rho}:[0,+\infty)^2\rightarrow\mathbb{R}$, and $C^2$ functions $\widehat{p}_\mathrm{rad}:[0,+\infty)^2\rightarrow\mathbb{R}$ and $\widehat{p}_\mathrm{tan}:[0,+\infty)^2\rightarrow\mathbb{R}$, such that
\begin{subequations}
	\begin{align}
& \rho(t,r)=\widehat{\rho}(\delta(t,r),\eta(t,r)),\\ 
& p_\mathrm{rad}(t,r)=\widehat{p}_\mathrm{rad}(\delta(t,r),\eta(t,r)),\\ 
& p_\mathrm{tan}(t,r)=\widehat{p}_\mathrm{tan}(\delta(t,r),\eta(t,r)),
\end{align}
\end{subequations}
where
	\begin{equation}\label{DefDeltaEta}
	\delta(t,r)=\frac{n(t,r)}{n_0},\qquad \eta(t,r)=\frac{3}{r^3}\int^r_0 \left(1+\frac{v^2(t,s)}{1-\frac{2m(t,s)}{s}}\right)^{1/2}\frac{\delta(t,s)}{\left(1-\frac{2m(t,s)}{s}\right)^{1/2}}s^2 ds 
	\end{equation}
(with $n_0$ a positive constant representing the number density in a reference state) and
\begin{subequations}\label{DefP}\label{EoS}
	\begin{align}
	\widehat{p}_{\mathrm{rad}}(\delta,\eta)&
	=\delta\partial_\delta\widehat{\rho}(\delta,\eta)-\widehat{\rho}(\delta,\eta),\label{DefPrad} \\
	\widehat{p}_{\mathrm{tan}}(\delta,\eta)&
	= \widehat{p}_{\mathrm{rad}}(\delta,\eta)+\frac{3}{2}\eta \partial_{\eta}\widehat{\rho}(\delta,\eta)\label{DefPtan}.
	\end{align}
\end{subequations}
\end{definition}
\begin{remark}
The constitutive functions $\widehat{p}_\mathrm{rad}(\delta,\eta)$, and $\widehat{p}_\mathrm{tan}(\delta,\eta)$ are independent of the constant $n_0$ while $\widehat{\rho}(\delta,\eta)$ can depend on $n_0$ or not depending if the material is relativistic or ultra-relativistic (see Definition~\ref{defrelstoredenergy} below).
\end{remark}
\begin{remark}
    In the special case of a perfect fluid, the energy density depends only on $\delta$, i.e., $\rho(t,r)=\widehat{\rho}(\delta(t,r))$, and so $\partial_\eta\widehat{\rho}=0$.
\end{remark}
\begin{remark}
	The above definition of $\eta(t,r)$ assumes a flat material metric, and is only valid for ball solutions (see  Appendix~\ref{App:B} for more details). If the material metric is non-flat then in general it is not possible to obtain an explicit formula for $\eta$ in Schwarzschild coordinates. Moreover, for shell solutions the definition of $\eta(t,r)$ contains one more parameter, corresponding to the inner radius of the shell in material space. As a final remark, note that $\eta$ can also be written as
 \begin{equation}
 \eta(t,r)=\frac{\bar{N}(t,r)}{n_0},
 \end{equation}
 where $\bar{N}$ is the average number of particles in a ball of radius $r>0$,
 \begin{equation}
\bar{N}=\frac{3N(t,r)}{4\pi r^3}.
 \end{equation}
\end{remark}
\begin{remark}
	From the above definitions we have the important relation
	\begin{equation}\label{Maxwell}
	3\eta\partial_\eta\widehat{p}_\mathrm{rad}(\delta,\eta)=2\left(\delta\partial_\delta\widehat{q}(\delta,\eta)-\widehat{q}(\delta,\eta)\right),
	\end{equation}
	where
	\begin{equation}
	\widehat{q}(\delta,\eta) = \widehat{p}_\mathrm{tan}(\delta,\eta) - \widehat{p}_\mathrm{rad}(\delta,\eta).
	\end{equation}
	is the \emph{anisotropic pressure}.
\end{remark}
\begin{definition}[Stored energy function, energy per particle]
    The \emph{relativistic stored energy function} is the $C^2$ function given by $\epsilon(t,r)=\widehat{\epsilon}(\delta(t,r),\eta(t,r))$, where $\widehat{\epsilon}(\delta,\eta):(0,\infty)^2\rightarrow\mathbb{R}$ is related to the energy density by
    \begin{equation}
    \widehat{\rho}(\delta,\eta)=\delta \widehat{\epsilon}(\delta,\eta).
    \end{equation}
    The \emph{energy per particle} is then defined as
    $$
    e = \frac{\epsilon}{n_0}.
    $$ 
    \end{definition}

\begin{remark}    
    The radial and tangential pressures are obtained from the relativistic stored energy function as
    \begin{subequations}
    \begin{align}
    \widehat{p}_\mathrm{rad}(\delta,\eta) &= \delta^2\partial_{\delta}\widehat{\epsilon}(\delta,\eta),\\
    \widehat{p}_\mathrm{tan}(\delta,\eta) &= \widehat{p}_{\mathrm{rad}}(\delta,\eta)+\frac{3}{2}\delta\eta \partial_{\eta}\widehat{\epsilon}(\delta,\eta).
    \end{align}
    \end{subequations}
\end{remark}    
\begin{definition}[Isotropic state]\label{IsoState}
	The isotropic state corresponds to a state where $\delta=\eta$. The radial and tangential pressures satisfy the isotropic state conditions
 \begin{equation}\label{ISOCond}
\widehat{p}_\mathrm{rad}(\delta,\delta)=\widehat{p}_\mathrm{tan}(\delta,\delta)
 \end{equation}
 and
\begin{subequations}\label{ISOCond2}
 \begin{align}
& 2\partial_\delta\widehat{p}_\mathrm{rad}(\delta,\delta)+3\partial_\eta\widehat{p}_\mathrm{rad}(\delta,\delta)=2\partial_\delta\widehat{p}_\mathrm{tan}(\delta,\delta), \\
& \partial_\delta\widehat{p}_\mathrm{tan}(\delta,\delta)+3\partial_\eta\widehat{p}_\mathrm{tan}(\delta,\delta)=\partial_\delta\widehat{p}_\mathrm{rad}(\delta,\delta).
 \end{align}
\end{subequations}
\end{definition}
\begin{remark}\label{Isotropicrelations}From the formulas above, it is easily seen that in the isotropic state the energy density, the isotropic pressure and the anisotropic pressure satisfy
	\begin{subequations}
		\begin{align}
		&\delta\partial_\delta\widehat{\rho}(\delta,\delta) =\widehat{\rho}(\delta,\delta)+\widehat{p}_\mathrm{iso}(\delta,\delta),\qquad \qquad \qquad	\partial_\eta\widehat{\rho}(\delta,\delta)=0 ,\label{73a}\\
		&\partial_\delta\widehat{p}_\mathrm{iso}(\delta,\delta)=\partial_\delta\widehat{p}_\mathrm{rad}(\delta,\delta)+\partial_\eta\widehat{p}_\mathrm{rad}(\delta,\delta),\qquad \partial_\eta \widehat{p}_\mathrm{iso}(\delta,\delta)=0, \\
			&\widehat{q}(\delta,\delta)=0, \qquad \qquad \qquad \qquad \qquad \qquad \qquad \partial_\delta\widehat{q}(\delta,\delta)+\partial_\eta\widehat{q}(\delta,\delta)=0.		\label{ISOq}	
		\end{align}
	\end{subequations}	
	Other identities can be obtained by differentiation. The second equation in~\eqref{73a}, for example, yields
	\begin{equation}
	\partial^2_{\delta\eta}\widehat{\rho}(\delta,\delta)+\partial^2_\eta\widehat{\rho}(\delta,\delta)=0 .
	\end{equation}
\end{remark}

\begin{definition}[Reference state]\label{RefState}
		The reference state is an idealized state with $n\equiv n_0$,  $m\equiv 0$ (corresponding to a flat material metric, see Appendix~\ref{sec:materialmetric}), and $v\equiv0$, that is, $({{\delta}},\eta)=(1,1)$. The function $\widehat{\epsilon}(\delta,\eta)$ is assumed to satisfy the \emph{reference state condition} $\widehat{\epsilon}(1,1)=\rho_0$, where $\rho_0>0$ is the constant rest frame energy density, so that 
		\begin{equation}
		\widehat{\rho}(1,1)=\rho_0.
		\end{equation}
		The radial and tangential pressures satisfy the reference state condition
		\begin{equation}\label{PRS}
		\widehat{p}_{\mathrm{rad}}(1,1)=	\widehat{p}_{\mathrm{tan}}(1,1)=p_0.
		\end{equation}
   If $p_0=0$ then the reference state is said to satisfy the \emph{stress-free reference state condition}, while if $p_0\neq0$ it is said to satisfy the \emph{pre-stressed reference state condition}. Moreover, compatibility with linear elasticity requires
		\begin{subequations}\label{DevPreference}
			\begin{align}
			&\partial_\delta\widehat{p}_{\mathrm{rad}}(1,1)=\lambda+2\mu ,\qquad \partial_\eta \widehat{p}_{\mathrm{rad}}(1,1) =-\frac{4}{3}\mu, \\
			&\partial_\delta\widehat{p}_{\mathrm{tan}}(1,1)=\lambda ,\qquad\qquad \partial_\eta \widehat{p}_{\mathrm{tan}}(1,1) =\frac{2}{3}\mu,
			\end{align}
		\end{subequations}
		where $\lambda$ and $\mu$ are constants which can be identified with the \emph{first and second  Lam\'e parameters} of linear elasticity in the stress-free case.
\end{definition}
\begin{remark}
In the reference state, the energy density, the isotropic and the anisotropic pressures satisfy
\begin{subequations}
	\begin{align}
	&\widehat{\rho}(1,1)=\rho_0,\qquad\partial_\delta\widehat{\rho}(1,1)=\rho_0 +p_0,\qquad\quad \partial_\eta\widehat{\rho}(1,1)=0, \\
	&\widehat{p}_\mathrm{iso}(1,1)=p_0,\quad\partial_\delta\widehat{p}_\mathrm{iso}(1,1)=\lambda+\frac{2}{3}\mu,\qquad \partial_\eta \widehat{p}_\mathrm{iso}(1,1)=0, \\
	&\widehat{q}(1,1)=0,\qquad\partial_\delta\widehat{q}(1,1)=-2\mu, \qquad\qquad  \partial_\eta\widehat{q}(1,1)=2\mu.
	\end{align}
\end{subequations}
\end{remark}
\begin{remark}\label{RefStaGauge}
	When the model is stress-free, the reference state $(\delta,\eta)=(1,1)$ is uniquely defined. However, pre-stressed material models do not have a preferred reference state. In this case, a different reference state, compressed or expanded with respect to the original reference state, provides an equivalent description of the material, moving from the parameters $(\rho_0,p_0,\lambda,\mu)$ to new parameters $(\tilde\rho_0,\tilde{p}_0,\tilde\lambda,\tilde{\mu})$. The choice of reference state is thus akin to a gauge choice.
\end{remark}	
%
%
\begin{remark}
	Linear elasticity is fully characterized by two elastic constants, typically defined by the Lam\'e parameters $\lambda$ and $\mu$. Other elastic constants of interest are the \emph{p-wave modulus} $L$, the \emph{Poisson ratio} $\nu$, the \emph{Young modulus} $E$, and the \emph{bulk modulus} $K$, given in terms of the Lam\'e parameters by, respectively,
	\begin{equation}
	L=\lambda+2\mu,\qquad \nu=\frac{\lambda}{2(\lambda+\mu)},\qquad E=\frac{(3\lambda+2\mu)\mu}{\lambda+\mu},\qquad K=\lambda+\frac{2}{3}\mu.
	\end{equation} 
 For perfect fluids, we have $\mu=0$, $\nu=1/2$, $E=0$, and $K=L=\lambda$.
\end{remark}
\begin{remark}[Invariance under renormalization of the reference state] \label{InvRenRefSt} In the pre-stressed case, if we choose a new reference state which is compressed by a factor $f$ (in volume) with respect to an original reference state, then the new variables $(\tilde\delta,\tilde\eta)$ are related to the original variables $(\delta,\eta)$ by
\begin{equation}
\delta = f\tilde\delta, \qquad \eta = f\tilde\eta.
\end{equation}
The expression $\tilde\rho$ of the energy density as a function of the new variables must of course satisfy
\begin{equation}
    \tilde{\rho}(\tilde\delta,\tilde\eta) = \rho(\delta,\eta)=\rho(f\tilde\delta,f\tilde\eta).
\end{equation}
The new elastic parameters $\tilde\lambda$ and $\tilde\mu$ (and others) will in general differ from $\lambda$ and $\mu$, as they are calculated with respect to a different reference state. However, it is usually possible to construct combinations of these parameters which are invariant under the renormalization of the reference state. It is these combinations, rather than the standard elastic parameters, that carry physical meaning.
\end{remark}
\subsection{The deformation potential: relativistic and ultra-relativistic materials}
  \begin{definition}[Relativistic and ultra-relativistic materials]\label{defrelstoredenergy}
    A material is said to be \emph{relativistic} if its stored energy function consists of the sum of the positive constant rest frame mass density, 
    \begin{equation}
    \varrho_0=\mathfrak{m}n_0,    
    \end{equation}
    where $\mathfrak{m}$ is the rest mass of the particles making up the body, and a \emph{potential energy density}  $w(t,r)=\widehat{w}(\delta(t,r),\eta(t,r))$ due to the deformations,
	\begin{equation}\label{RSEF}
	\widehat{\epsilon}(\delta,\eta)	= \varrho_0 + \widehat{w}(\delta,\eta).
	\end{equation}
  A material is said to be \emph{ultra-relativistic} if the relativistic stored energy function $\widehat{\epsilon}(\delta,\eta)$ is independent of the baryonic rest mass  density,
  \begin{equation}\label{URSEF}
  \widehat{\epsilon}(\delta,\eta)	= \widehat{w}(\delta,\eta).
  \end{equation}
 The deformation potential energy is assumed to satisfy the reference state condition
    \begin{equation}
        \widehat{w}(1,1)=w_0,
    \end{equation}
where $w_0$ is a constant potential energy density due to the deformations of the reference state (so that $\rho_0=\varrho_0+w_0$ for relativistic materials, and $\rho_0=w_0$ for ultra-relativistic materials). The \emph{natural reference state condition} requires this constant to vanish for stress-free reference states, that is, $w_0=0$ whenever $p_0=0$.     
\end{definition}
    In~\cite{Alho:2021sli} we gave a slightly  different definition of $\widehat{w}$; however, the definition above is the consistent definition if $\widehat{w}$ is to be considered the potential energy density due to deformations, so that a stress-free (undeformed) reference state has zero potential energy.
\begin{remark}
	If the reference state is stress-free, $p_0=0$, then $(\delta,\eta)=(1,1)$ is a stationary point of $\widehat{w}(\delta,\eta)$. In this case, compatibility with linear elasticity implies that
			\begin{align}
			&\partial^2_\delta\widehat{w}(1,1)=\lambda+2\mu ,\qquad \partial^2_\eta \widehat{w}(1,1) =\frac{4}{3}\mu ,\qquad \partial^2_{\eta\delta}\widehat{w}(1,1)=-\frac{4}{3}\mu,  
			\end{align}
		and $(\delta,\eta)=(1,1)$ is a nondegenerate local minimum of the stored energy function if and only if the shear modulus is positive $(\mu>0)$ and the bulk modulus is also positive $(K>0)$. In the fluid limit, $\mu\rightarrow0^+$, this critical point becomes degenerate.
\end{remark}
\begin{remark}
Ultra-relativistic materials have a pre-stressed reference state, and are obtained from relativistic materials by taking the limit $\varrho_0\rightarrow0$, so that $\rho_0=w_0$. For example, the polytropic fluid of Example~\ref{Ex1} in this formalism is written as
\begin{equation}
\widehat{\rho}(\delta)=\delta\left(\varrho_0+\widehat{w}(\delta)\right), \qquad \widehat{w}(\delta)=\frac{\lambda}{\gamma(\gamma-1)}\delta^{\gamma-1}
\end{equation}
with the identifications
\begin{equation}
\mathrm{K}=\frac{\lambda}{\gamma n^{\gamma}_0},\qquad \mathrm{C}=\left(\frac{\gamma}{\lambda}\right)^{1/\gamma}\varrho_0,  
\end{equation}
and it is a relativistic material with $w_0=\frac{p_0}{\gamma-1}$, where $p_0 = \frac{\lambda}{\gamma}$, and so $\rho_0=\varrho_0+\frac{p_0}{\gamma-1}$. The fluid with linear EoS of Example~\ref{Ex2} is obtained from the polytropic fluid by taking the limit $\mathrm{C}\rightarrow 0$, i.e., $\varrho_0\rightarrow 0$, 
\begin{equation}
\widehat{\rho}(\delta)=\delta\widehat{w}(\delta), \qquad \widehat{w}(\delta)=\frac{\lambda}{\gamma(\gamma-1)}\delta^{\gamma-1}
\end{equation}
and it is ultra-relativistic with $\rho_0=w_0=\frac{p_0}{\gamma-1}$.
\end{remark}
 \begin{remark} \label{UltraRelExtraSym}
     For ultra-relativistic materials, there is an extra symmetry due to the invariance of $\rho$ under rescalings of the rest baryonic mass density,
     \begin{equation}
      \tilde{\varrho}_0    = f \varrho_0,
     \end{equation}
     while keeping $(\delta,\eta)$ fixed. Notice that this is quite different from renormalizing the reference state, which is simply a change in the description of a \emph{fixed} material: here we are modifying the material itself by adding more baryons per unit volume, so that its (gauge-invariant) elastic parameters will change.
 \end{remark}
    \begin{proposition}
The transformation
    \begin{equation}\label{SFPSTransfw}
        \widehat{w}^{(\mathrm{ps})}(\delta,\eta)=\widehat{w}^{(\mathrm{sf})}(\delta,\eta)-\alpha_0-p_0\delta^{-1},
      \end{equation}
      where $\alpha_0=-(w_0+p_0)$ (with $p_0\neq0$ and $w_0=\rho_0-\varrho_0^{(\mathrm{ps})}\neq0$), takes any given stress-free natural reference state material to a 1-parameter family (parameterized by $p_0$) of pre-stressed reference state materials (or vice-versa). This changes the energy density according to
     \begin{equation}\label{TransfRho}
         \widehat{\rho}^{(\mathrm{ps})}(\delta,\eta)=\widehat{\rho}^{(\mathrm{sf})}(\delta,\eta)+\left(-\alpha_0+\varrho^{(\mathrm{ps})}_0-\varrho^{(\mathrm{sf})}_0\right)\delta-p_0,
     \end{equation}
     where $\varrho^{(\mathrm{ps})}_0=\varrho^{(\mathrm{sf})}_0$ for relativistic materials, and $\varrho^{(\mathrm{ps})}_0=0$ for ultra-relativistic materials.
\end{proposition}
\begin{proof}
    Using~\eqref{TransfRho} into~\eqref{DefPrad}, yields
    \begin{equation}
        \widehat{p}^{(\mathrm{ps})}_\mathrm{rad}(\delta,\eta)=\widehat{p}^{(\mathrm{sf})}_\mathrm{rad}(\delta,\eta)+p_0.
    \end{equation}
    Inserting the above equation together with~\eqref{TransfRho} in~\eqref{DefPtan}, yields
    \begin{equation}
        \widehat{p}^{(\mathrm{ps})}_\mathrm{tan}(\delta,\eta)=\widehat{p}^{(\mathrm{sf})}_\mathrm{tan}(\delta,\eta)+p_0.
    \end{equation}
\end{proof}
\begin{definition}[Natural choice of $p_0$]
 The choice of the smallest value of $p_0$ such that the physical condition
     \begin{equation}\label{NatChoiceP0}
      \widehat{p}^{(\mathrm{ps})}_\mathrm{iso}(\delta,\delta)>0,
     \end{equation}
     holds for $\delta>0$ is called the \emph{natural choice} of the reference state pressure.
\end{definition}
\begin{examples}\label{AffineTransLin}
    The (stress-free reference state) affine fluid EoS can be obtained from the (pre-stressed reference state) linear EoS using the above transformation, the fact that $\varrho^{(\mathrm{ps})}_0=0$ for ultra-relativistic materials, and the identity
\begin{equation}
    c^2_\mathrm{s}(1)=\frac{\lambda}{\rho_0+p_0}=\gamma-1,
\end{equation}
    which yields the relation $\rho^{(\mathrm{ps})}_0+p_0=\frac{\lambda}{\gamma-1}=\rho^{(\mathrm{sf})}_0=\varrho^{(\mathrm{sf})}_0$, and therefore
    \begin{equation}
        \widehat{\rho}^{(\mathrm{sf})}(\delta)=\widehat{\rho}^{(\mathrm{ps})}(\delta)+p_0 = \frac{\lambda}{\gamma(\gamma-1)}\delta^{\gamma}+\frac{\lambda}{\gamma}.
    \end{equation}
\end{examples}
\begin{definition}[Stress-free natural reference state power-law deformation potentials~\cite{Alho:2019fup}]\label{PLDP}
 A stress-free natural reference state deformation potential function is said to be of \emph{power-law type} $(n_1,n_2,...,n_m)\in\mathbb{N}^{m}$, $m\geq2$, if it is of the form
\begin{equation}\label{SSDPF}
 \widehat{w}^{(\mathrm{sf})}(\delta,\eta)=\alpha_0+\sum^{m}_{j=1}\eta^{\theta_{j}}\sum^{n_j}_{i=1}\alpha_{ij}\left(\frac{\delta}{\eta}\right)^{\beta_{ij}},\quad  \alpha_0=-\sum^{m}_{j=1}\sum^{n_j}_{i=1}\alpha_{ij},
\end{equation}
 where the exponents $\theta_j, \beta_{ij}\in\mathbb{R}$, $i=1,...,n_j$, $j=1,...,m$ satisfy
\begin{itemize}
\item[(i)] $\theta_1<\theta_2<...<\theta_{m}$, $\beta_{1j}<\beta_{2j}<...<\beta_{n_{j}j}$, for all $j=1,...,m$;
\item[(ii)] if $n_j=1$ then $\theta_j\neq0$ and $\theta_j=\beta_{1j}$;
\item[(iii)] at least one of the exponents $\beta_{ij}$ is different from $0$ and $-1$,
\end{itemize}
and the real coefficients $\alpha_{ij}\neq0$ are a solution of the system of linear equations
 \begin{subequations}
     \begin{align}
         &\sum^{m}_{j=1}\sum^{n_j}_{i=1}\alpha_{ij}\theta_j =0, \quad \sum^{m}_{j=1}\sum^{n_j}_{i=1}\alpha_{ij}\theta^2_j =\lambda+\frac{2}{3}\mu, \quad \sum^{m}_{j=1}\sum^{n_j}_{i=1}\alpha_{ij}\beta^2_{ij}=\lambda+2\mu,  \label{PLSEF1st}\\
         &\sum^{n_j}_{i=1}\alpha_{ij}(\theta_j-\beta_{ij})=0, \quad j=1,...,m. \label{PLSEF2nd}
     \end{align}
 \end{subequations}
\end{definition}
For spherically symmetric power-law type deformation potential functions, the radial and tangential pressures are given by
\begin{subequations}\label{PPL}
\begin{align}
    \widehat{p}_\mathrm{rad}(\delta,\eta) &= \sum^{m}_{j=1}\eta^{1+\theta_j}\sum^{n_j}_{i=1}\alpha_{ij}\beta_{ij}\left(\frac{\delta}{\eta}\right)^{1+\beta_{ij}}, \\
    \widehat{p}_\mathrm{tan}(\delta,\eta) &=\frac{1}{2}\sum^{m}_{j=1}\eta^{1+\theta_j}\sum^{n_j}_{i=1}\alpha_{ij}(3\theta_j-\beta_{ij})\left(\frac{\delta}{\eta}\right)^{1+\beta_{ij}}.
\end{align}
\end{subequations}
Condition~\eqref{SSDPF} on the coefficient $\alpha_0$ corresponds to the natural reference state condition $\widehat{w}(1,1)=0$, while conditions~\eqref{PLSEF1st} arise from the stress-free reference state condition~\eqref{PRS} with $p_0=0$ and the linear elasticity compatibility conditions~\eqref{DevPreference}. The last condition~\eqref{PLSEF2nd} is in turn equivalent to the isotropic state condition~\eqref{ISOCond}. Condition $(i)$ only affects the order of the factors appearing in
the stored energy function (lexicographic order), while condition $(ii)$ is needed for consistency
with~\eqref{PLSEF2nd}. Finally, condition $(iii)$ ensures that the radial pressure is not independent of $\delta$.

This class of spherically symmetric deformation functions has been further analyzed in a subsequent work~\cite{ACL21}. The definition of Lam\'e-type deformation potential functions has been introduced as the subclass for which the system of linear equations~\eqref{PLSEF1st}--\eqref{PLSEF2nd} is determined, and hence the coefficients $\alpha_{ij}$ are uniquely given (in terms of the Lam\'e coefficients $(\lambda,\mu)$) by the exponents $\theta_j$, $\beta_{ij}$. It was also shown that the only admissible Lam\'e-type $(n_1,n_2,\dots,n_m)$ deformation potential functions (i.e., such that solutions $\alpha_{ij}$ exist for all values of $\mu\geq0$ and $K>0$) are 
 \begin{equation}
     m=2: \quad (1,3), \ (2,3);\qquad m=3:\quad (1,2,2),\ (1,1,2),\ (2,2,2),
 \end{equation}
		and permutations thereof. 
\begin{remark}
    Perfect fluids  are of type $(1,1,...)$ with $\mu=0$. In this case, only for $m=2$ does the system admit a unique solution of type $(1,1)$.
\end{remark}

  \subsection{The wave speeds}
Besides satisfying the energy conditions (see Definition~\ref{EC}), further restrictions on the choice of elastic law come from requiring causal wave propagation within the material. As shown by Karlovini \& Samuelsson~\cite{Karlovini:2002fc}, in spherical symmetry there are five independent wave speeds, of which only three can be written from a given spherically symmetric stored energy function $\widehat{\epsilon}(\delta,\eta)$ (equivalently $\widehat{\rho}(\delta,\eta)$). In general, the other two independent velocities must be obtained through the general form of the stored energy function without symmetries. The wave speeds without any symmetry assumption can be found in Appendix~\ref{App:A}, and the degenerate case in subsection~\ref{Deg}, which we use in Appendix~\ref{WSetadelta} to write the wave speeds in spherical symmetry as functions of the variables $(\delta,\eta)$ given in Definition~\ref{DefElaBall}. We summarize these results in the following definition:
\begin{definition}[Spherically symmetric elastic wave speeds]\label{Speeds}
	The squared speed of elastic longitudinal waves in the radial direction is given by
	\begin{equation}\label{cLR}
	c^2_\mathrm{L}(\delta,\eta)= \frac{\delta\partial_\delta \widehat{p}_{\mathrm{rad}}(\delta,\eta)}{\widehat{\rho}(\delta,\eta)+\widehat{p}_{\mathrm{rad}}(\delta,\eta)},
	\end{equation}
and	the squared speed of transverse waves in the radial direction is given by
	\begin{equation}\label{cTR}
	c^2_\mathrm{T}(\delta,\eta)= \frac{\widehat{p}_{\mathrm{tan}}(\delta,\eta) - \widehat{p}_{\mathrm{rad}}(\delta,\eta)}{\left(\widehat{\rho}(\delta,\eta)+\widehat{p}_{\mathrm{tan}}(\delta,\eta)\right)\left(1 - \left(\frac{\delta}{\eta}\right)^2\right)},
	\end{equation}
 while the squared speed of transverse waves in the tangential direction oscillating in the radial direction is given by 
	\begin{equation}\label{cTT}
	\tilde{c}^2_\mathrm{T}(\delta,\eta)= \left(\frac{\delta}{\eta}\right)^2\frac{\widehat{p}_{\mathrm{tan}}(\delta,\eta)-\widehat{p}_{\mathrm{rad}}(\delta,\eta)}{\left(\widehat{\rho}(\delta,\eta)+\widehat{p}_{\mathrm{rad}}(\delta,\eta)\right)\left(1 - \left(\frac{\delta}{\eta}\right)^2\right)}.
	\end{equation}
	Moreover the squared speed of longitudinal waves along the tangential direction, $\tilde{c}^2_\mathrm{L}$, and the squared speed of transverse waves in the tangential direction oscillating in the tangential direction, $\tilde{c}^2_\mathrm{TT}$, satisfy the relation
	\begin{equation}\label{RelWS}
 \begin{split}
 \tilde{c}^2_\mathrm{L}(\delta,\eta)-\tilde{c}^2_\mathrm{TT}(\delta,\eta)
&=\frac{\frac{3}{4}\eta\partial_{\eta}\widehat{\rho}(\delta,\eta)+\delta^2\partial^2_\delta \widehat{\rho}(\delta,\eta)+3\delta\eta\partial^2_{\eta\delta}\widehat{\rho}(\delta,\eta)+\frac{9}{4}\eta^2\partial^2_\eta\widehat{\rho}(\delta,\eta)}{\widehat{\rho}(\delta,\eta)+\widehat{p}_\mathrm{tan}(\delta,\eta)}. 
 \end{split}
	\end{equation}
\end{definition}
\begin{remark}
	This means that to construct physically admissible elastic stars one needs to specify not only $\widehat{\rho}(\delta,\eta)$ but also either $\tilde{c}^2_{\mathrm{L}}(\delta,\eta)$ or $\tilde{c}^2_\mathrm{TT}(\delta,\eta)$.
\end{remark}
Overall, reality of the wave speeds and causality require 
\begin{equation} \label{realcausal}
0\leq {c}^2_\mathrm{L,T}(\delta,\eta)\leq1,\qquad  0\leq \tilde{c}^2_\mathrm{L,T,TT}(\delta,\eta)\leq1.
\end{equation}
%
In the isotropic state $(\delta,\eta)=(\delta,\delta)$, the squared wave speeds satisfy
\begin{subequations}\label{WSIso}
	\begin{align}
	&c^2_\mathrm{L}(\delta,\delta)=\tilde{c}^2_\mathrm{L}(\delta,\delta)=\frac{\delta\partial_\delta\widehat{p}_\mathrm{rad}(\delta,\delta)}{\widehat{\rho}(\delta,\delta)+\widehat{p}_\mathrm{iso}(\delta,\delta)} ,\label{cLCenter}\\ &c^2_\mathrm{T}(\delta,\delta)=\tilde{c}^2_\mathrm{T}(\delta,\delta)=\tilde{c}^2_\mathrm{TT}(\delta,\delta)=-\frac{3}{4}\frac{\delta\partial_\eta\widehat{p}_\mathrm{rad}(\delta,\delta)}{\widehat{\rho}(\delta,\delta)+\widehat{p}_\mathrm{iso}(\delta,\delta)} \label{cTCenter}
	\end{align}
\end{subequations}
and
\begin{equation}\label{Kcenter}
\frac{d\widehat{p}_\mathrm{iso}}{d\widehat{\rho}}(\delta,\delta)=c^2_\mathrm{L}(\delta,\delta)-\frac{4}{3}c^2_\mathrm{T}(\delta,\delta).
\end{equation}
Note that for a perfect fluid $c^2_\mathrm{T}(\delta,\delta)=\tilde{c}^2_\mathrm{T}(\delta,\delta)=\tilde{c}^2_\mathrm{TT}(\delta,\delta)=0$. Causality together with the reality conditions demand
\begin{equation}
0\leq\delta\partial_\delta\widehat{p}_\mathrm{rad}(\delta,\delta)\leq\widehat{\rho}(\delta,\delta)+\widehat{p}_\mathrm{iso}(\delta,\delta),
\end{equation}
\begin{equation}
0\leq-\frac{3}{4}\delta\partial_\eta\widehat{p}_\mathrm{rad}(\delta,\delta)\leq\widehat{\rho}(\delta,\delta)+\widehat{p}_\mathrm{iso}(\delta,\delta)
\end{equation}
and
\begin{equation}
\delta\partial_\delta\widehat{p}_\mathrm{rad}(\delta,\delta)>-\delta\partial_\eta\widehat{p}_\mathrm{rad}(\delta,\delta).
\end{equation}
Specializing to the reference state $(\delta,\eta)=(1,1)$, the wave speeds satisfy
\begin{equation}\label{WSRefSt}
\tilde{c}^2_\mathrm{L}(1,1)=c^2_\mathrm{L}(1,1)=\frac{\lambda+2\mu}{\rho_0+p_0},\qquad c^2_\mathrm{T}(1,1)=\tilde{c}^2_\mathrm{T}(1,1)=\tilde{c}^2_\mathrm{TT}(1,1)=\frac{\mu}{\rho_0+p_0},
\end{equation}
and we must have
\begin{equation}
0\leq\lambda + 2 \mu \leq \rho_0 +p_0, \qquad 0\leq \mu \leq \rho_0 +p_0 .
\end{equation}
Furthermore
\begin{equation}
\frac{d\widehat{p}_\mathrm{iso}}{d\widehat{\rho}}(1,1)=c^2_\mathrm{L}(1,1)-\frac{4}{3}c^2_\mathrm{T}(1,1)=\frac{\lambda+\frac{2}{3}\mu}{\rho_0 +p_0},
\end{equation}
is assumed to be positive, which is equivalent to the condition
\begin{equation}\label{Kpos}
K=\lambda+\frac{2}{3}\mu>0.
\end{equation}
\begin{remark}
    The physical admissibility conditions $L\geq 0$, $\mu\geq0$ and $K>0$ (see Remark~\ref{modulicond}) imply that $E\geq0$ and $\nu\in(-1,\frac{1}{2}]$. In the fluid limit $\mu=0$ one has $\nu=\frac{1}{2}$ and $E=0$, while the lower bound on $\nu$ arises from the condition $K>0$.
\end{remark}

\subsubsection{The natural choice of $\tilde{c}_{\mathrm{L}}^2$ and $\tilde{c}_{\mathrm{TT}}^2$}

Motivated by relation~\eqref{RelWS}, we will now introduce a special class of spherically symmetric elastic materials for which the unknown squared speeds $\tilde{c}^2_\mathrm{L}$  and $\tilde{c}^2_\mathrm{TT}$ are linear functions of the first and second partial derivatives of $\widehat{\rho}$ only. Starting with the general expression
\begin{equation}
(\widehat{\rho}+\widehat{p}_\mathrm{tan})\tilde{c}^2_\mathrm{L} = A\delta^2\partial^2_\delta\widehat{\rho}+B\delta\eta\partial^2_{\eta\delta}\widehat{\rho}+C\eta^2\partial^2_\eta \widehat{\rho}+D\delta\partial_\delta\widehat{\rho}+E\eta\partial_\eta\widehat{\rho},
\end{equation}
with constants $A, B, C, D, E \in \mathbb{R}$, we obtain $A=1$, $B=C$ and $D=0$ from the isotropy condition $c_\mathrm{L}(\delta,\delta) = \tilde{c}_\mathrm{L}(\delta,\delta)$, so that 
\begin{equation}
(\widehat{\rho}+\widehat{p}_\mathrm{tan})\tilde{c}^2_\mathrm{L}  = \delta^2 \partial^2_\delta \widehat{\rho} + B \delta \eta \partial_\delta \partial_\eta \widehat{\rho} + B \eta^2 \partial^2_\eta \widehat{\rho} + E \eta \partial_\eta \widehat{\rho}.
\end{equation}
Using~\eqref{RelWS}, we obtain
\begin{equation}\label{LinWScTT}
(\widehat\rho + \widehat{p}_{\rm tan})\tilde{c}_{\mathrm{TT}}^2 = (B-3) \delta \eta \partial_\delta \partial_\eta \widehat\rho + \left(B-\frac94\right) \eta^2 \partial^2_\eta \widehat\rho + \left(E-\frac34\right) \eta \partial_\eta \widehat\rho,
\end{equation}
which in the isotropic limit satisfies~\eqref{cTCenter}, and therefore $B$ and $E$ are free parameters. The choice of $B$ and $E$ amounts to choosing the material.
When written in terms of $\widehat{p}_\mathrm{rad}$ and $\widehat{p}_\mathrm{tan}$, the two squared velocities read
\begin{equation} \label{cTBE}
 (\widehat\rho + \widehat{p}_{\rm tan})\tilde{c}_{\mathrm{L}}^2 = \left(1 - \frac23 E \right) \delta \partial_\delta \widehat{p}_{\rm rad} + \left(\frac13 B - E \right) \eta \partial_\eta \widehat{p}_{\rm rad} + \frac23 E \delta \partial_\delta \widehat{p}_{\rm tan} + \frac23 B \eta \partial_\eta \widehat{p}_{\rm tan}
\end{equation}
and
\begin{align}
 (\widehat\rho + \widehat{p}_{\rm tan})\tilde{c}_{\mathrm{TT}}^2 = & \left(1 - \frac23 E \right) \delta \partial_\delta \widehat{p}_{\rm rad} + \left(\frac13 B - E \right) \eta \partial_\eta \widehat{p}_{\rm rad} \nonumber \\
  & \qquad \qquad \qquad + \left(\frac23 E - 1 \right) \delta \partial_\delta \widehat{p}_{\rm tan} + \left(\frac23 B - \frac32 \right) \eta \partial_\eta \widehat{p}_{\rm tan}. \label{cTTBE}
\end{align}
\begin{definition}[The natural choice of $\tilde{c}_\mathrm{L}^2$ and $\tilde{c}_\mathrm{TT}^2$]\label{NaturalChoice}
The obvious choice to simplify these expressions is to set $E=\frac32$ and $B=\frac92$, yielding
\begin{equation}\label{cNC}
 \tilde{c}_\mathrm{L}^2(\delta,\eta) = \frac{\delta \partial_\delta \widehat{p}_{\rm tan}(\delta,\eta) + 3 \eta \partial_\eta \widehat{p}_{\rm tan}(\delta,\eta) }{\widehat\rho(\delta,\eta) + \widehat{p}_{\rm tan}(\delta,\eta)},\qquad  \tilde{c}_\mathrm{TT}^2(\delta,\eta) = \frac{\frac32 \eta \partial_\eta \widehat{p}_{\rm tan}(\delta,\eta) }{\widehat\rho(\delta,\eta) + \widehat{p}_{\rm tan}(\delta,\eta)}.
\end{equation}
\end{definition}
%

\subsection{Eulerian Einstein-elastic equations}
Just as in the fluid case, the evolution equation for the variable $\delta(t,r)$ follows from the conservation of particle current, which can be written in the form $\nabla_\mu (\delta u^\mu)=0$. On the other hand, from the expression of $\eta(t,r)$ given in Definition~\ref{DefElaBall}, and using~\eqref{Divn}, it follows that
\begin{equation}\label{EtaEqs}
\partial_t\eta=-\frac{3}{r}\frac{e^{\phi}\delta v}{\left(1-\frac{2m}{r}\right)^{\frac{1}{2}}},\qquad \partial_r \eta = -\frac{3}{r}\left(\eta-\frac{\langle v\rangle \delta}{\left(1-\frac{2m}{r}\right)^{\frac{1}{2}}}\right).
\end{equation}
When written in terms of the variables $(\phi,m,\delta,\eta,v)$, the spherically symmetric Einstein-elastic equations in Schwarzschild coordinates yield the following first order system of partial differential equations:
\begin{subequations}\label{EinNEW}
\begin{align}
&e^{-\phi}\langle v\rangle \partial_t \delta+v\partial_r \delta +\frac{\delta v} {1-\frac{2m}{r}+v^2}e^{-\phi}\langle v\rangle\partial_t v +\delta\partial_r v=-\frac{2}{r}\delta v;\\
&\frac{1}{1-\frac{2m}{r}+v^2}\left(e^{-\phi}\langle v\rangle\partial_t v+v\partial_r v\right)+\frac{ve^{-\phi}\langle v\rangle}{1-\frac{2m}{r}+v^2}\left(c^2_\mathrm{L}(\delta,\eta)\delta^{-1}\partial_t\delta+\frac{\eta\partial_\eta \widehat{p}_\mathrm{rad}}{\widehat{\rho}(\delta,\eta) +\widehat{p}_{\mathrm{rad}}(\delta,\eta)}\eta^{-1}\partial_t\eta\right) \nonumber \\ 
&\qquad\qquad\qquad+c^2_\mathrm{L}(\delta,\eta)\delta^{-1} \partial_r\delta +s\left(\delta,\eta,v,\frac{m}{r}\right)\eta^{-1} \partial_r\eta=-\frac{\frac{m}{r^2}+4\pi r \widehat{p}_\mathrm{rad}}{1-\frac{2m}{r}+v^2} ; \\
& e^{-\phi}\langle v\rangle \partial_t\eta+v\partial_r \eta=-\frac{3}{r}\eta v;  \\
&e^{-\phi}\langle v\rangle\partial_t m +v\partial_r m=-4\pi r^2 \widehat{p}_\mathrm{rad}(\delta,\eta)v,  
\end{align}
\end{subequations}
together with
\begin{align}
& \partial_r \eta = -\frac{3}{r}\left(\eta-\frac{\langle v\rangle \delta}{\left(1-\frac{2m}{r}\right)^{\frac{1}{2}}}\right), \quad \partial_r m = 4\pi r^2 \left(\rho +\frac{\rho+p_\mathrm{rad}}{1-\frac{2m}{r}}v^2\right), \label{integro1} \\
& \partial_r\phi =\frac{\frac{m}{r^2}+4\pi r p_\mathrm{rad}}{1-\frac{2m}r} + \frac{4\pi r (\rho+p_\mathrm{rad})v^2}{\left(1-\frac{2m}r\right)^2}, \label{integro2}
\end{align}
where we introduced
\begin{equation}
\begin{split}
    s\left(\delta,\eta,v,\frac{m}{r}\right) = 
    &\frac{\eta\partial_{\eta}\widehat{p}_{\mathrm{rad}}(\delta,\eta)}{\left(\widehat{\rho}(\delta,\eta) +\widehat{p}_{\mathrm{rad}}(\delta,\eta)\right)}+\frac{2}{3}\left(\frac{\delta}{\eta}\right)^{-2}\frac{\left(1-\left(\frac{\delta}{\eta}\right)^2\right)}{\left(1-\frac{\langle v \rangle\left(\frac{\delta}{\eta}\right)}{\left(1-\frac{2m}{r}\right)^{\frac{1}{2}}}\right)}\tilde{c}^2_\mathrm{T}(\delta,\eta).
    \end{split}
\end{equation}
For solutions with a regular center of symmetry (Definition~\ref{Def1}), for which $\lim_{r\rightarrow 0^+} m(t,r)=0$ and $\lim_{r\rightarrow 0^+} v(t,r)=0$, we have
	\begin{equation}\label{RCcond}
	\lim_{r\rightarrow 0^+} \eta(t,r)=\lim_{r\rightarrow0^+}\delta(t,r)=\delta(t,0)=\delta_\mathrm{c}(t), \qquad t\in[0,T].
	\end{equation}
 	Hence, the center of symmetry corresponds to an isotropic state where all relations in Remark~\ref{Isotropicrelations} hold. In particular,
\begin{equation}
    s(\delta_\mathrm{c}(t),\delta_\mathrm{c}(t),0,0)=0, \qquad t\in[0,T].
\end{equation}
%
%
%
%

As seen in propositions~\ref{totalmass} and \ref{totalnumber}, two conserved quantities are the total number of particles $\mathcal{N}$ and the total mass $\mathcal{M}$:
\begin{equation}
	\mathcal{N} \equiv \frac{4\pi n_0}{3}\mathcal{R}(t)^3\eta(t,\mathcal{R}(t)), \qquad
	\mathcal{M} \equiv m(t,\mathcal{R}(t)).
\end{equation}

The system \eqref{EinNEW}-\eqref{integro2} can be seen as a system of integro-differential equations, where $\phi$ and $\eta$ can be obtained from \eqref{integro1} and \eqref{integro2} in the integral forms \eqref{integrophi} and \eqref{DefDeltaEta}. In the simpler case where gravity is neglected, i.e., in Minkowski spacetime, the system~\eqref{EinNEW}  can be written in the form $\bm{A}^{0}(\bm{u})\partial_t \bm{u}+\bm{A}^{r}(\bm{u})\partial_r \bm{u}=\bm{B}(\bm{u},r)$, with
\begin{equation}
\bm{u}=
\begin{pmatrix}
\delta \\
v \\
\eta 
\end{pmatrix}
,\quad
\bm{A}^{0}(\bm{u})=
\begin{pmatrix}
\langle v\rangle & \frac{\delta v}{\langle v\rangle} & 0 \\
\frac{v c^2_\mathrm{L}}{\langle v\rangle \delta} & \frac{1}{\langle v\rangle} & \frac{\partial_\eta \widehat{p}_\mathrm{rad}}{\widehat{\rho}(\delta,\eta) +\widehat{p}_{\mathrm{rad}}(\delta,\eta)} \\
0 &0 & \langle v\rangle 
\end{pmatrix}
,\quad 
\bm{A}^{r}(\bm{u})=
\begin{pmatrix}
v & \delta & 0 \\
\frac{c^2_\mathrm{L}}{\delta} & \frac{v}{\langle v\rangle^2} & \frac{s}{\eta} \\
0 & 0 & v 
\end{pmatrix}
.
\end{equation}
The eigenvalues of the matrix $\left( \bm{A}^{0} \right)^{-1} \bm{A}^{r}$ are
\begin{equation}
\lambda_0 = v_b, \qquad \lambda_+ = \frac{v_b+c_\mathrm{L}}{1+v_b c_\mathrm{L}}, \qquad \lambda_- = \frac{v_b-c_\mathrm{L}}{1-v_b c_\mathrm{L}},
\end{equation}
where $c_\mathrm{L}(\delta,\eta)=\sqrt{c^2_\mathrm{L}(\delta,\eta)}$ and
\begin{equation}
v_b = \frac{v}{\langle v\rangle}
\end{equation}
is the so-called \emph{boost velocity}, that is, the velocity of the material particles as measured by observers at rest. Notice that Einstein's velocity addition formula then implies that the eigenvalues $\lambda_+$ and $\lambda_-$ are the velocities of outgoing and ingoing longitudinal waves as measured by observers at rest.
Thus, the system is strictly hyperbolic if
\begin{equation}
    c^2_\mathrm{L}(\delta,\eta)>0 \quad \Leftrightarrow \quad \delta\partial_\delta\widehat{p}_\mathrm{rad}(\delta,\eta)>0,
\end{equation}
which is a stronger condition than simply requiring $c^2_\mathrm{L}(\delta,\eta)$ to be nonnegative (reality condition, see~\eqref{realcausal}). In particular, in the reference state the condition $\lambda+2\mu>0$ must hold.

\newpage

\section{Steady states}\label{SS}
%
In terms of the variables $\delta(r)$, $\eta(r)$ and $m(r)$, the TOV equations~\eqref{TOVeq} become a closed first order system,
\begin{subequations}\label{TOV2}
	\begin{align}
	c^2_{\mathrm{L}}(\delta,\eta) \frac{1}{\delta}\frac{d\delta}{dr} &=-s(\delta,\eta,\frac{m}{r})\frac{1}{\eta}\frac{d\eta}{dr}-\frac{\left(\frac{m}{r^2} +4\pi r \widehat{p}_{\mathrm{rad}}(\delta,\eta)\right)}{\left(1-\frac{2m}{r}\right)} , \label{TOV2a}\\
	\frac{d\eta}{dr} &=-\frac{3}{r}\left(\eta-\frac{\delta}{\left(1-\frac{2m}{r}\right)^{1/2}}\right),\label{Etastatic} \\
	\frac{dm}{dr}&=4\pi r^2 \widehat{\rho}(\delta,\eta) \label{massstatic}.
	\end{align}
\end{subequations}
Given an EoS~\eqref{EoS}, the system of equations~\eqref{TOV2} for the stellar structure should be solved subject to the regular center conditions
\begin{equation}\label{icRelSS}
    \eta(0)=\delta(0)=\delta_\mathrm{c}, \qquad m(0)=0.
\end{equation}
Moreover, the central density $\rho_c=\widehat\rho(\delta_c,\delta_c)$ and the central pressure $p_c=\widehat{p}_{\mathrm{iso}}(\delta_c,\delta_c)$ should satisfy some physical admissibility conditions:
\begin{definition}[Physically admissible initial data for static configurations]\label{IDstatic}
	Let $\rho_c=\widehat{\rho}(\delta_c,\delta_c)$, and  $p_c=\widehat{p}_\mathrm{iso}(\delta_c,\delta_c)$ denote the central density and pressure, respectively. An initial data set $(\rho_c,p_c)$ is said to be \emph{physically admissible} if
	\begin{subequations}
		\begin{align}
		&\rho_c>0, \label{rhoc}\\
		&p_c>0\quad \text{for radially compressed balls}, \\
		&p_c<0\quad \text{for radially stretched balls},
		\end{align}
	\end{subequations}
	and
	\begin{equation}\label{InitialData}
		0<c^2_\mathrm{L}(\delta_c,\delta_c)\leq1,\quad 0\leq c^2_\mathrm{T}(\delta_c,\delta_c)\leq1,\quad c^2_\mathrm{L}(\delta_c,\delta_c)>\frac{4}{3}c^2_\mathrm{T}(\delta_c,\delta_c).
	\end{equation}
	In addition, the standard relativistic energy conditions at the center of symmetry should be satisfied:
	\begin{subequations}
		\begin{align}
		\mathrm{SEC}&:\quad \rho_c+3p_c\geq 0\quad\text{and} \quad \rho_c+p_c\geq 0; \label{SEC}\\
		\mathrm{WEC}&:\quad \rho_c\geq 0;\quad \rho_c+p_c\geq 0; \\
		\mathrm{NEC}&:\quad \rho_c+p_c\geq 0; \\
		\mathrm{DEC}&:\quad \rho_c\geq |p_c|. \label{DEC}
		\end{align}
	\end{subequations}
\end{definition}
\begin{remark}
	In practice, it suffices to consider~\eqref{rhoc}, the conditions on the wave speeds~\eqref{InitialData}, and either the dominant energy condition $(\mathrm{DEC})$~\eqref{DEC} if $p_c>0$, or the first inequality on the strong energy condition $(\mathrm{SEC})$~\eqref{SEC} if $p_c<0$.
\end{remark}
\begin{theorem}[Strong regularity]\label{strongregthm}
Let $(\phi,m,\rho,p_\mathrm{rad},p_\mathrm{tan})$ be a regular static self-gravitating elastic ball with initial data $(\phi_\mathrm{c},0, \rho_\mathrm{c},p_\mathrm{c},p_\mathrm{c})$, where $\rho_\mathrm{c}=\widehat{\rho}(\delta_\mathrm{c},\delta_\mathrm{c})$, $p_c=\widehat{p}_\mathrm{rad}(\delta_\mathrm{c},\delta_\mathrm{c})=\widehat{p}_\mathrm{tan}(\delta_\mathrm{c},\delta_\mathrm{c})$. If $\partial_\delta\widehat{p}_\mathrm{rad}(\delta_\mathrm{c},\delta_\mathrm{c})\neq0$ then $(\phi,m,\rho,p_\mathrm{rad},p_\mathrm{tan})$ is strongly regular, and the following estimates hold:
\begin{equation}
|m^{\prime}(r)|\leq C r^2, \qquad |\phi^\prime(r)|\leq C r,
\end{equation}
\begin{equation}\label{estistronreg}
|\rho^\prime(r)|+|p^\prime_\mathrm{rad}(r)|+|p^\prime_\mathrm{tan}(r)|+|\delta^\prime(r)|+|\eta^\prime(r)|\leq C r, \qquad 0\leq r\leq r_*,
\end{equation}
for some positive constants $C$ and $r_*$.
\end{theorem}
\begin{proof}
    The proof is similar to the proof of the strong regularity theorem in the Newtonian setting given in~\cite{Alho:2019fup}. Since we are interested in the behaviour of regular ball solutions as $r\rightarrow0^+$, we may restrict to the domain $r\in[0,r_*)$, where $r_*$ can be chosen arbitrarily small. Since $\widehat{\rho}(\delta_\mathrm{c},\delta_\mathrm{c})>0$, then by continuity $\widehat{\rho}(\delta,\eta)>0$ in $r\in[0,r_*)$, and from~\eqref{massstatic}, $m^\prime(r)\leq C r^2$, which implies
    \begin{equation}\label{Boundmass}
    m(r)\leq C r^3, \qquad r< r_*.
    \end{equation}
    Moreover, since the constitutive functions are assumed to be $C^2$, by Taylor's theorem there exist functions $h_1(\delta,\eta)$ and $h_2(\delta,\eta)$, bounded in a small disk $D$ around $(\delta,\delta)$, such that
    \begin{equation}
    \partial_\eta\widehat{p}_\mathrm{rad}(\delta,\eta)=\partial_\eta\widehat{p}_\mathrm{rad}(\delta,\delta)+ h_1(\delta,\eta)(\eta-\delta)
    \end{equation}
    and
    \begin{equation}
    \widehat{q}(\delta,\eta)=\widehat{q}(\delta,\delta)+ \partial_\eta \widehat{q}(\delta,\delta)(\eta-\delta)+h_2(\delta,\eta)(\eta-\delta)^2.
    \end{equation}
    Plugging into~\eqref{TOV2a} and using the fact that at the isotropic state (see Definition~\ref{IsoState}) $\widehat{q}(\delta,\delta)=0$ and $3\partial_\eta\widehat{p}_\mathrm{rad}(\delta,\delta)+2\partial_\eta\widehat{q}(\delta,\delta)=0$, we obtain
    \begin{equation}\label{Expdeltaprime}
    \begin{split}
    \partial_\delta\widehat{p}_\mathrm{rad}(\delta,\eta)\frac{d\delta}{dr} =&\left[h_1(\delta,\eta)+h_2(\delta,\eta)\right]\frac{\left(\eta-\left(1-\frac{2m}{r}\right)^{-\frac{1}{2}}\delta\right)^2}{r}\\
    &+\left[\delta h_1(\delta,\eta)-2\delta\eta h_2(\delta,\eta)\right]\frac{1-\left(1-\frac{2m}{r}\right)^{\frac{1}{2}}}{r\left(1-\frac{2m}{r}\right)^{\frac{1}{2}}} -\delta^2 h_2(\delta,\eta)\frac{\frac{2m}{r^2}}{\left(1-\frac{2m}{r}\right)} \\
    &-(\widehat{\rho}(\delta,\eta)+\widehat{p}_\mathrm{rad}(\delta,\eta))\frac{\left(\frac{m}{r^2} +4\pi r \widehat{p}_{\mathrm{rad}}(\delta,\eta)\right)}{\left(1-\frac{2m}{r}\right)}.
    \end{split}
    \end{equation}
    Now let
    \begin{equation}
    u(r)=r\left(\frac{\delta}{\left(1-\frac{2m}{r}\right)^{1/2}}\right)^\prime = \frac{1}{\left(1-\frac{2m}{r}\right)^{1/2}}\left(r\delta^\prime-\delta\frac{\frac{m}{r}-4\pi r^2\widehat{\rho}(\delta,\eta)}{1-\frac{2m}{r}}\right).
    \end{equation}
   If $\partial_\delta\widehat{p}_\mathrm{rad}(\delta_\mathrm{c},\delta_\mathrm{c})\neq 0$ then, for $r_*$ sufficiently small, $\inf_{r\in(0,r_*)}\partial_\delta\widehat{p}_\mathrm{rad}(\delta,\eta)\neq0$, and hence for regular ball solutions $u(r)\in C^{0}((0,r_*))$, with $\lim_{r\rightarrow 0^+}u(r)= 0$. Furthermore, from the definition of $\eta(r)$ and integration by parts we have
    \begin{equation}
        \eta(r)-\frac{\delta(r)}{\left(1-\frac{2m(r)}{r}\right)^{1/2}}=-\frac{1}{r^3}\int^{r}_{0}u(s)s^2 ds.
    \end{equation}
    By the Cauchy-Schwarz's inequality,
    \begin{equation}
        \left|\eta(r)-\frac{\delta}{\left(1-\frac{2m(r)}{r}\right)^{1/2}}\right|\leq\frac{1}{\sqrt{5r}}\left(\int^{r}_{0} u(s)^2 ds \right)^{1/2},
    \end{equation}
    and using
    \begin{equation}
        \left|\frac{1-\left(1-\frac{2m}{r}\right)^{\frac{1}{2}}}{r\left(1-\frac{2m}{r}\right)^{\frac{1}{2}}}\right|\leq \frac{m}{r^2}
    \end{equation}
    together with~\eqref{Boundmass}, we get from~\eqref{Expdeltaprime}
    \begin{equation}
        u(r)\leq C \left(r^2+\frac{1}{r}\int^{r}_{0} u(s)^2 ds\right)
    \end{equation}
    with $C$ a positive constant. The rest of the proof is identical to the proof in~\cite{Alho:2019fup}, i.e., using the Stachurska lemma we can conclude that for sufficiently small $r$ one has $u(r)\leq C r^2$, and hence $\delta^\prime (r)\leq C r$. Moreover
    \begin{equation}
        \left|\eta(r)-\frac{\delta(r)}{\left(1-\frac{2m(r)}{r}\right)^{1/2}}\right|\leq C r^2,
    \end{equation}
    and, by~\eqref{Etastatic}, $|\eta^{\prime}(r)|\leq Cr$. Finally, from the relations
    \begin{subequations}
\begin{align}
    \rho^\prime (r) &= \bigl(\widehat{\rho}(\delta,\eta)+\widehat{p}_\mathrm{rad}(\delta,\eta)\bigr)\delta^{-1}\delta^\prime(r)+\frac{2}{3}\widehat{q}(\delta,\eta)\eta^{-1}\eta^\prime(r), \\
    p^\prime_\mathrm{rad} (r) &= \partial_\delta\widehat{p}_\mathrm{rad}(\delta,\eta)\delta^\prime(r)+\partial_\eta\widehat{p}_\mathrm{rad}(\delta,\eta)\eta^\prime(r), \\
    p^\prime_\mathrm{tan} (r) &= \partial_\delta\widehat{p}_\mathrm{tan}(\delta,\eta)\delta^\prime(r)+\partial_\eta\widehat{p}_\mathrm{tan}(\delta,\eta)\eta^\prime(r),
\end{align}
    \end{subequations}
    the estimates~\eqref{estistronreg} follow.
\end{proof}
By the above Theorem, the first condition in~\eqref{InitialData} implies that regular solutions of~\eqref{TOV2} are strongly regular, and $\delta^{\prime}(r)\rightarrow 0$, $\eta^{\prime}(r)\rightarrow 0$ as $r\rightarrow 0^+$. In fact they are even analytic having the following Taylor expansion  
\begin{subequations}
    \begin{align}
        \delta(r) &=  \delta_c - \frac{2\pi}{3}\frac{\delta_c}{c_{\rm L}^2(\delta_c,\delta_c)}\left( \rho_c + 3 p_c - 2 \rho_c \left(c_{\rm L}^2(\delta_c,\delta_c) + \frac{4}{3} c_{\rm T}^2(\delta_c,\delta_c)\right)\right)r^2+ {\cal O}(r^4);\\
        \eta(r) &= \delta_c -\frac{2}{5}\pi\frac{\delta _c}{c_{\rm L}^2(\delta_c,\delta_c)}\left(\rho _c+3 p_c-4  \rho _c \left(c_{\rm L}^2(\delta_c,\delta_c)+\frac{2 }{3 }c_{\rm T}^2(\delta_c,\delta_c)\right)\right)r^2 + {\cal O}(r^4);
        \\
        m(r) &= \frac{4 }{3} \pi \rho_c r^3 + \frac{8 \pi ^2}{15c_{\rm L}^2(\delta_c,\delta_c)}  \left(\rho _c-p_c\right) \left(  \rho_c + 3 p_c - 2\rho _c \left( c_{\rm L}^2(\delta_c,\delta_c)+\frac{4}{3} c_{\rm T}^2(\delta_c,\delta_c)\right)\right)r^5+ {\cal O}(r^7).
    \end{align}
\end{subequations}
as $r\rightarrow 0^{+}$.

\newpage

    \section{Radial perturbations}\label{RS}
Linearizing the Einstein equations~\eqref{EinNEW} around a solution of the static background equations~\eqref{TOV2} leads to equations for $(\phi_\mathrm{E},m_\mathrm{E},\delta_\mathrm{E},\eta_\mathrm{E},v_\mathrm{E})$.
In terms of the fundamental radial displacement field $\xi(t,r)$, the linearized velocity, mass and energy density are given by equations~\eqref{LinVel},~\eqref{LinMass}, and~\eqref{LinRho} respectively. Using Definition~\ref{DefElaBall}, it follows that for spherically symmetric elastic solid materials we have
\begin{subequations}
	\begin{align}
	\rho_\mathrm{E}&=(\widehat{\rho}+\widehat{p}_\mathrm{rad})\frac{\delta_\mathrm{E}}{\delta}+\frac{2}{3}\widehat{q}\frac{\eta_\mathrm{E}}{\eta}, \label{LinRhoEl}\\ 
 (p_\mathrm{rad})_\mathrm{E} &=\delta\partial_\delta\widehat{p}_\mathrm{rad}\frac{\delta_\mathrm{E}}{\delta}+\eta\partial_\eta\widehat{p}_\mathrm{rad} \frac{\eta_\mathrm{E}}{\eta}, \\
 (p_\mathrm{tan})_\mathrm{E} &=\delta\partial_\delta\widehat{p}_\mathrm{tan}\frac{\delta_\mathrm{E}}{\delta}+\eta\partial_\eta\widehat{p}_\mathrm{tan} \frac{\eta_\mathrm{E}}{\eta},
	\end{align}
\end{subequations}
which generalizes equation~\eqref{LinEoS} for fluids. The linearized quantity $\delta_\mathrm{E}$ is obtained from the particle current conservation equation, which yields the same equation as for fluids, i.e.
\begin{equation}\label{Lindelta}
\delta_\mathrm{E}(t,r) = -\delta \partial_r\xi(t,r)-\delta\left[\frac{2}{r}-\frac{d\phi}{dr}+\frac{1}{\delta}\frac{d\delta}{dr}\right]\xi(t,r),
\end{equation}
while the linearization of the equation~\eqref{EtaEqs} for $\eta$ yields
\begin{equation}
e^{-\phi}\partial_t\eta_\mathrm{E}=-\frac{3}{r}\frac{\delta}{\left(1-\frac{2m}{r}\right)^{1/2}}v_\mathrm{E}.
\end{equation}
Using~\eqref{LinRhoEl} and~\eqref{Lindelta} in~\eqref{LinRho}, leads to
\begin{equation}
\eta_\mathrm{E}(t,r)=-\frac{3}{r}\frac{\delta}{\left(1-\frac{2m}{r}\right)^{1/2}}\xi(t,r).
\end{equation}
Hence, \eqref{linearized2} becomes (using~\eqref{LinVel} and~\eqref{LinMass})
%
%
%
\begin{equation}
\begin{split}
  &\frac{e^{-2\phi}}{1-\frac{2m}{r}} \partial^2_t\xi(t,r)=c^2_\mathrm{L}\partial^2_r\xi(t,r) +\Big[\frac{d}{dr}\left(c^2_\mathrm{L}-\phi\right)+\frac{2}{r}\frac{\delta\partial_\delta \widehat{q}}{\rho+p_\mathrm{rad}}-\frac{3}{r}\frac{\eta\partial_\eta\widehat{p}_\mathrm{rad}}{\rho+p_\mathrm{rad}}\frac{\delta/\eta}{\left(1-\frac{2m}{r}\right)^{1/2}} \\
  &+c^2_\mathrm{L}\left(\frac{2}{r}-\frac{d\phi}{dr}+\frac{4\pi r (\rho+p_\mathrm{rad})}{1-\frac{2m}{r}}+\frac{2}{r}\frac{q}{\rho+p_\mathrm{rad}}\frac{\delta/\eta}{\left(1-\frac{2m}{r}\right)^{1/2}}\right)\Big]\partial_r\xi(t,r) \\
   &+\Big[2\left(\frac{d\phi}{dr}\right)^2+\frac{d}{dr}\left(c^2_\mathrm{L}\left(\frac{2}{r}-\frac{d\phi}{dr}\right)\right)+\left(\frac{2}{r}-\frac{d\phi}{dr}\right)c^2_\mathrm{L}\left(\frac{4\pi r(\rho+p_\mathrm{rad})}{1-\frac{2m}{r}}+\frac{2}{r}\frac{q}{\rho+p_\mathrm{rad}}\frac{\delta/\eta}{\left(1-\frac{2m}{r}\right)^{1/2}}\right) \\
   &+\frac{2m}{r^3\left(1-\frac{2m}{r}\right)}+\left(\frac{2}{r}-\frac{d\phi}{dr}\right)\frac{2}{r}\frac{q }{\left(p_{\rm rad}+\rho \right)}\frac{\delta/\eta}{\left(1-\frac{2m}{r}\right)^{1/2}}-\frac{8 \pi  q }{1-\frac{2 m}{r}}-\frac{6}{r^2 }\frac{\eta\partial_\eta q}{\left(p_{\rm rad}+\rho \right) }\frac{\delta/\eta}{\left(1-\frac{2m}{r}\right)^{1/2}}\\
   &-\frac{6}{r^2}\frac{\delta  \eta  \partial^2_{\eta \delta}\widehat{q}}{ \left(p_{\rm rad}+\rho \right)} \left(1-\frac{\delta/\eta }{\left(1-\frac{2m}{r}\right)^{1/2}}\right)-\frac{2}{r}\frac{2\delta \partial_\delta \widehat{q}}{ \left(p_{\rm rad}+\rho \right)} \left(\frac{1}{r}-\frac{4 \pi  r \left(p_{\rm rad}+\rho \right)}{1-\frac{2 m}{r}}+\frac{\delta /\eta }{\left(1-\frac{2m}{r}\right)^{1/2}}\left(\frac{2}{r}-\frac{d\phi}{dr}\right)-\frac{d\phi}{dr}\right)\\
   &-\frac{2}{r}\frac{\delta ^2 \partial^2_\delta q}{c_{\rm L}^2 \left(p_{\rm rad}+\rho \right)}\left( \frac{d\phi}{dr}-\frac{2}{r}\frac{\delta \partial_\delta q}{\left(p_{\rm rad}+\rho \right)}\left(1-\frac{ \delta/\eta}{\left(1-\frac{2m}{r}\right)^{1/2}}\right)-\frac{2}{r}\frac{\widehat{q}}{\left(p_{\rm rad}+\rho \right)}\frac{ \delta/\eta}{\left(1-\frac{2m}{r}\right)^{1/2}}\right)\Big]\xi(t,r),
\end{split}
   \end{equation}
which generalizes equation~\eqref{WEdisp} to the elastic solid setting. For strongly regular balls, it follows from  Theorem~\ref{strongregthm} and L'H\^opital's rule that
\begin{equation}
\lim_{r\rightarrow 0^+} \frac{\widehat{q}(\delta(r),\eta(r))}{r} = \lim_{r\rightarrow 0^+} \frac{\partial_{\delta}\widehat{q}(\delta(r),\eta(r))}{r} = \lim_{r\rightarrow 0^+} \frac{\partial_{\eta}\widehat{q}(\delta(r),\eta(r))}{r} = 0,
\end{equation}
while, as in the fluid case, the unbounded terms in this right-hand side cancel the term $c^2_\mathrm{s}\partial^2_r \xi(t,r)$ in this limit, so that
\begin{equation}
\lim_{r\rightarrow 0^+} \partial^2_t \xi(t,r)=0.
\end{equation}

At the surface of the ball conditions~\eqref{BoundCond} yield, in the elastic case,
\begin{equation}
    (p_\mathrm{rad})_\mathrm{L}(t,\mathcal{R}(t)) = -\frac{e^{\phi(\mathcal{R})}}{\mathcal{R}^2}\delta\partial_\delta\widehat{p}_\mathrm{rad}(\delta,\eta)\partial_r\left(r^2 e^{-\phi}\xi\right)(t,\mathcal{R}(t))-\frac{3}{\mathcal{R}}\eta\partial_\eta\widehat{p}_\mathrm{rad}(\delta,\eta)\xi(t,\mathcal{R}(t))=0.
\end{equation}
The perturbation variables are 
\begin{equation}
\zeta(t,r)\equiv r^2 e^{-\phi(r)}\xi(t,r),\qquad \chi(t,r)=-\frac{e^{2\phi(r)}}{\left(1-\frac{2m}{r}\right)^{1/2}}(p_\mathrm{rad})_\mathrm{L}(t,r),
\end{equation}
where $\zeta$, $\chi$ are the renormalized radial displacement and the  renormalized Lagrangian perturbation of the radial pressure $(p_\mathrm{rad})_\mathrm{L}$, respectively. 
%
%
%
%
%
Making the ansatz 
\begin{equation}
\zeta(t,r)=e^{i  \omega t}\zeta_0(r),\qquad \chi(t,r)=e^{i  \omega t}\chi_0(r),
\end{equation}
where $\zeta_0(r)=\zeta(0,r)$ and $\chi_0(r)=\chi(0,r)$, leads to an eigenvalue problem for the system of first order ordinary differential equations
\begin{subequations} \label{SturmLiouville}
	\begin{align}
	\delta\partial_\delta\widehat{p}_\mathrm{rad}(\delta,\eta)\frac{d\zeta_0}{dr} &=-\frac{3}{r}\eta\partial_\eta\widehat{p}_\mathrm{rad}(\delta,\eta)\zeta_0 + \left(1-\frac{2m}{r}\right)^{1/2}e^{-3\phi}r^2\chi_0, \\
	\delta\partial_\delta\widehat{p}_\mathrm{rad}(\delta,\eta) \frac{d\chi_0}{dr}&=\frac{3}{r}\eta\partial_\eta\widehat{p}_\mathrm{rad}(\delta,\eta)\chi_0-\left[Q+\omega^2W\right]\zeta_0,
	\end{align}
\end{subequations}
where
\begin{subequations}
	\begin{align} 
	Q &= \frac{e^{3\phi}}{r^2\left(1-\frac{2m}{r}\right)^{1/2}}\Big[\frac{4}{r^2}(\delta\partial_\delta\widehat{q}-\widehat{q})^2+\delta\partial_\delta\widehat{p}_\mathrm{rad}\Big(\frac{2}{r^2}\widehat{q}-\frac{6}{r^2}\eta\partial_\eta\widehat{q}-\frac{4}{r}\frac{dp_\mathrm{rad}}{dr} \nonumber\\
	&\qquad\qquad\qquad\qquad\qquad\qquad +\left(\frac{6}{r}\widehat{q}-\frac{dp_\mathrm{rad}}{dr}\right)\frac{d\phi}{dr}-8\pi\frac{(\widehat{\rho}+\widehat{p}_\mathrm{rad})}{\left(1-\frac{2m}{r}\right)}\widehat{p}_\mathrm{rad}\Big)\Big], \\
	W &=\frac{(\widehat{\rho}+\widehat{p}_\mathrm{rad})e^{\phi}}{r^2\left(1-\frac{2m}{r}\right)^{3/2}}\delta\partial_\delta\widehat{p}_\mathrm{rad} .
	\end{align}
\end{subequations}
This system of ordinary differential equations, coupled to the background equations~\eqref{TOV2}, must be solved subject to the boundary conditions~\eqref{BoundCond}, and~\eqref{BoundCond1}, i.e.,
\begin{equation}
\lim_{r\rightarrow 0^+}\frac{\zeta_0(r)}{r^2}=0,\qquad \lim_{r\rightarrow R}\chi_0(r)=0.
\end{equation}
The system~\eqref{SturmLiouville} is not in standard Sturm-Liouville form; however, it has been shown by Karlovini \& Samuelsson~\cite{Karlovini:2003xi} that the main Sturm-Liouville results still hold in the elastic case, assuming continuity of $\zeta$ and $\chi$. For strongly regular solutions the following expansions at the center of symmetry hold
 \begin{subequations}
     \begin{align}
         \zeta_0(r) &= e^{-\phi_\mathrm{c}} g_0 r^3 +\mathcal{O}(r^5), \\
         \chi_0(r) &= 3g_0 e^{2\phi_c} \delta_c \left(\partial_\eta p_{\rm rad}(\delta_c,\delta_c)+\partial_\delta p_{\rm rad}(\delta_c,\delta_c)\right)+{\cal O}(r^2)\,.
     \end{align}
 \end{subequations}
as $r\rightarrow 0^+$ with $g_0$ a constant that parameterizes the solutions, as follows from~\eqref{sRg}.

\newpage

    \section{The Newtonian limit}\label{Newtonian}
%
So far we have been using units such that the speed of light is $c=1$. If we reinstate $c$ into our equations and take the limit as $c\rightarrow \infty$ we obtain $\phi=\phi_N$, where $\phi_N$ is the Newtonian gravitational potential, $g_{rr}=1$, and\footnote{Note that ultra-relativistic materials do not have a Newtonian limit.}
	\begin{equation}\label{NewForm1}
	\delta_N(t,r) =\frac{\varrho(t,r)}{\varrho_0} ,\qquad \eta_N(t,r) =\frac{\bar{\varrho}(t,r)}{\varrho_0},
	\end{equation}
 where
\begin{equation}\label{NewForm2}
 \bar{\varrho}(t,r)=\frac{3m(t,r)}{4\pi r^3}=\frac{3}{r^3}\int^{r}_{0}\varrho(t,s)s^2 ds
\end{equation}
 is the averaged mass density of a ball of radius $r$, so that we recover the results in Alho \& Calogero~\cite{Alho:2018mro,Alho:2019fup} for the static case, and in Calogero~\cite{Cal21} for the time-dependent equations of motion.
Furthermore, if the constitutive functions can be deduced\footnote{This case corresponds to the so-called \emph{hyperelastic} materials.} from a potential deformation function $\widehat{w}(\delta,\eta)$ (called the Newtonian stored energy function), it has been further shown that\footnote{For clarity, we drop the subscripts from $\delta_N, \eta_N$.}
\begin{equation}\label{hyperDef}
\widehat{p}_{\mathrm{rad}}(\delta,\eta) = \delta^2 \partial_\delta \widehat{w}(\delta,\eta) \quad \, \quad \widehat{p}_{\mathrm{tan}}(\delta,\eta)=\widehat{p}_{\mathrm{rad}}(\delta,\eta)+\frac{3}{2}\delta\eta \partial_{\eta} \widehat{w}(\delta,\eta)\,.
\end{equation}
For perfect fluids, the stored energy function is a function of $\delta$ only, i.e.\ $w=\widehat{w}(\delta)$, and $p_\mathrm{iso}(t,r)=\widehat{p}_\mathrm{iso}(\delta(t,r))=\delta^2 \frac{d\widehat{w}}{d\delta}(\delta)$. 

In the Newtonian limit, the squared speed of elastic longitudinal waves in the radial direction~\eqref{cLR} reduces to
\begin{equation}\label{cLRN}
	c^2_\mathrm{L}(\delta,\eta)= \frac{\delta\partial_\delta \widehat{p}_{\mathrm{rad}}(\delta,\eta)}{\varrho_0\delta}.
\end{equation}
The squared speed of transverse waves in the radial direction~\eqref{cTR} is given by
	\begin{equation}\label{cTRN}
	c^2_\mathrm{T}(\delta,\eta)= \frac{\widehat{p}_{\mathrm{tan}}(\delta,\eta) - \widehat{p}_{\mathrm{rad}}(\delta,\eta)}{\varrho_0\delta\left(1 - \delta^2/\eta^2\right)},
	\end{equation}
 while the squared speed of transverse waves in the tangential direction and oscillating in the radial direction~\eqref{cTT} is given by 
	\begin{equation}\label{cTTN}
	\tilde{c}^2_\mathrm{T}(\delta,\eta)= \left(\frac{\delta}{\eta}\right)^2\frac{\widehat{p}_{\mathrm{tan}}(\delta,\eta) - \widehat{p}_{\mathrm{rad}}(\delta,\eta)}{\varrho_0\delta\left(1 - \left(\frac{\delta}{\eta}\right)^2\right)}.
	\end{equation}
Finally, equations~\eqref{cNC} for the natural choice of the squared speed of longitudinal waves along the tangential direction, $\tilde{c}^2_\mathrm{L}$, and of the squared speed of transverse waves in the tangential direction oscillating in the tangential direction, $\tilde{c}^2_\mathrm{TT}$, become
\begin{equation}\label{cNCN}
 \tilde{c}_\mathrm{L}^2(\delta,\eta) = \frac{\delta \partial_\delta \widehat{p}_{\rm tan}(\delta,\eta) + 3 \eta \partial_\eta \widehat{p}_{\rm tan}(\delta,\eta) }{\varrho_0\delta},\qquad  \tilde{c}_\mathrm{TT}^2(\delta,\eta) = \frac{\frac32 \eta \partial_\eta \widehat{p}_{\rm tan}(\delta,\eta) }{\varrho_0\delta}.
\end{equation}
Physical viability further requires that
\begin{subequations}\label{VelNew}
\begin{align}
    &c^{2}_\mathrm{L}(\delta,\eta)>0, \quad \tilde{c}^2_\mathrm{L}(\delta,\eta)>0, \quad c^{2}_\mathrm{T}(\delta,\eta)\geq0, \quad \tilde{c}^2_\mathrm{T}(\delta,\eta)\geq 0, \quad \tilde{c}^2_\mathrm{TT}(\delta,\eta)\geq 0, \\
    &c^{2}_\mathrm{L}(\delta,\delta)=\tilde{c}^2_\mathrm{L}(\delta,\delta)>0, \quad c^{2}_\mathrm{T}(\delta,\delta)=\tilde{c}^2_\mathrm{T}(\delta,\delta)=\tilde{c}^2_\mathrm{TT}(\delta,\delta)\geq0, \quad c^2_\mathrm{L}(\delta,\delta)-\frac{4}{3}c^2_\mathrm{T}(\delta,\delta)>0. 
\end{align}
\end{subequations}

\begin{remark}
In the non-relativistic theory of elasticity, the reality condition on the transverse wave speeds~\eqref{cTRN} is equivalent to the so called \emph{Baker-Erickson inequality}, see for example John~\cite{John63}. 
\end{remark}
The Newtonian limit of the system~\eqref{EinNEW} of spherically symmetric Einstein-elastic equations results in the local conservation of mass and momentum, and an equation for $\eta$:
\begin{subequations}\label{sssystem2}
	\begin{align}
	&\partial_t\delta+\frac{1}{r^2}\partial_r(r^2\delta v)=0,\\
	&\partial_t (\delta v)+\frac{1}{r^2}\partial_r(r^2\delta v^2)+c^2_\mathrm{L}(\delta,\eta)\partial_r\delta+s(\delta,\eta)\partial_r\eta =-\frac{4\pi\varrho_0}{3}\delta\eta r, \\
	&\partial_t\eta+v\partial_r\eta=-\frac{3}{r}\eta v, 
	\end{align}
\end{subequations}
where the function
\begin{equation}
    s(\delta,\eta)=\left(\frac{\partial_\eta\widehat{p}_\mathrm{rad}(\delta,\eta)}{\varrho_0}+\frac{2}{3\varrho_0}\frac{\widehat{q}(\delta,\eta)}{(\eta-\delta)}\right)
\end{equation}
vanishes on isotropic states:
\begin{equation}
    s(\delta,\delta)=0.
\end{equation}
Thus $\widehat{s}(\delta,\eta)$ is a measure of anisotropy/shear. 
%
The energy density $E$ of an hyperelastic ball is defined as
\begin{equation}
E(t,r)=\varrho_0\delta(t,r) \left(\frac{v(t,r)^2}{2}+\varrho_0^{-1}w(t,r)\right).
\end{equation}
Self-gravitating hyperelastic balls satisfy the local energy balance equation
\begin{equation}\label{localconsenergy}
\partial_t\left(E-\frac{1}{8\pi}\frac{m^2}{r^4}\right)+\frac{1}{r^2}\partial_r\left(r^2(p_\mathrm{rad}+E)v\right)=0.
\end{equation}
The Newtonian limit~\eqref{sssystem2} of the system~\eqref{EinNEW} has the form $\partial_t \bm{u} +\bm{A}(\bm{u})\partial_r \bm{u} =\bm{B}(\bm{u},r)$, with
\begin{equation}
\bm{u}=
\begin{pmatrix}
\delta \\
\delta v \\
\eta
\end{pmatrix}
,\quad
\bm{A}(\bm{u})=
\begin{pmatrix}
0 & 1 & 0 \\
-\frac{(\delta v)^2}{\delta^2}+c^2_\mathrm{L}(\delta,\eta) & 2\frac{(\delta v)}{\delta} & \widehat{s}(\delta,\eta)  \\
0 &0 & v 
\end{pmatrix}
,\quad
\bm{B}(\bm{u})=
\begin{pmatrix}
-\frac{2}{r}(\delta v) \\
-\frac{2}{r}\frac{(\delta v)^2}{\delta}-\frac{4\pi\varrho_0}{3}\eta\delta r \\
-\frac{3}{r}\eta v
\end{pmatrix}.
\end{equation}
The matrix $\bm{A}(\bm{u})$ has eigenvalues  $v+\sqrt{c^2_\mathrm{L}(\delta,\eta)}$, $v-\sqrt{c^2_\mathrm{L}(\delta,\eta)}$, and $v$. Thus, the system is strictly hyperbolic if $c^2_\mathrm{L}(\delta,\eta)>0$, i.e., if$\partial_\delta\widehat{p}_{\mathrm{rad}}(\delta,\eta)>0$ (see~\cite{Cal21}).
\subsection{Steady states}
Non-relativistic steady states are solutions of the system~\eqref{sssystem2} with $v=0$ and $\partial_t\delta =0$, i.e., solutions of
 	\begin{subequations}\label{CPSS_Rho}
 		\begin{align}
 	c^2_\mathrm{L}(\delta,\eta)	\frac{d\delta}{dr} &=\frac{3}{r}s(\delta,\eta)(\eta-\delta)-\frac{4\pi\varrho_0}{3} \eta \delta r    \label{SSdelta}\, ,\\
 	\frac{d\eta}{dr} &= -\frac{3}{r}\left(\eta-\delta\right) \label{SSeta},
   \end{align}
 	\end{subequations}
%
%
%
subject to the regular center conditions
\begin{equation}\label{icNewtSS}
    \eta(0)=\delta(0)=\delta_\mathrm{c}.
\end{equation}
In~\cite{Alho:2019fup} it was proved that solutions with a regular center are strongly regular if $c^2_\mathrm{L}\neq0$. In fact, such solutions are even analytic and have the following Taylor expansion as $r\rightarrow 0^{+}$:
\begin{equation}
    \delta(r)=\delta_\mathrm{c}+\frac{4\pi\varrho_0\delta^2_\mathrm{c}}{3 c^2_\mathrm{L}(\delta_c,\delta_c)} r^2 + \mathcal{O}(r^3) , \qquad \eta(r)=\delta_\mathrm{c}- \frac{4\pi\varrho_0\delta_\mathrm{c}}{5 c^2_\mathrm{L}(\delta_c,\delta_c )} r^2 + \mathcal{O}(r^3).
\end{equation}
The asymptotics for the mass density $\varrho(r)$ and the mass function $m(r)$ can be obtained straightforwardly by using relations~\eqref{NewForm1} and~\eqref{NewForm2}, while for the Newtonian gravitational potential one obtains the asymptotics
\begin{equation}
    \phi_N(r) =\phi_c+\frac{2\pi\varrho_0 \delta_c}{3} r^2 + \mathcal{O}(r^3), \quad r\rightarrow 0^{+}.
\end{equation}	
In~\cite{Alho:2018mro}, a more general form of the equations, allowing for studying shells, was given, and the particular example of the (non-hyperelastic) Seth model was studied in detail. The existence of single and multi-body configurations (consisting of a ball, or a vacuum core, surrounded by an arbitrary number of shells) was proved, along with sharp mass/radius inequalities. In the follow-up paper~\cite{Alho:2019fup}, a definition of spherically symmetric power-law materials, and a new dynamical systems formulation of the above equations, which made use of Milne-type scaling invariant variables, was introduced. These showed that the qualitative properties of power-law materials depend uniquely on two parameters $(\mathrm{n},\mathrm{s})$,\footnote{Related to the parameters $(a,b)$ in~\cite{Alho:2019fup} by $a = - 3(1 + \frac{1}{\mathrm{n}})$ and $b=-\frac{1}{\mathrm{s}}$.} which allowed for proofs of existence of balls for several elastic material laws (by making use of a quite general theorem on the asymptotic properties of regular center solutions). 
In particular, it was shown that if (strongly) regular balls exist then the following relation holds:
    \begin{equation}\label{MainNewRes}
        \varrho(r)<\bar{\varrho}(r),\quad r\in(0,\mathcal{R}].
    \end{equation}
Another important issue concerns the physical viability of such solutions, i.e., whether conditions~\eqref{VelNew} hold. Due to~\eqref{MainNewRes}, the velocities $c^2_\mathrm{T}$, and $\tilde{c}^2_\mathrm{T}$ are nonnegative if and only if 
\begin{equation}
 p_\mathrm{tan}(r)-p_\mathrm{rad}(r)\geq0, \quad    r\in(0,\mathcal{R}].
\end{equation}
Moreover, the existence of exact solutions for scale-invariant models with an irregular center of symmetry (discussed in detail by Chandrasekhar~\cite{Cha39} for  polytropic fluids), and the dynamical systems formulation and lower-dimensional boundaries for scale-invariant models, was briefly discussed. 
%
In~\cite{ACL21}, a subclass of power-law materials which just depend on the two elastic constants of linear elasticity (called \emph{Lam\'e materials}) was identified. Moreover, the existence of exact solutions with an irregular center of symmetry  was established. 
%
%
\subsubsection{Scale invariance, exact solutions, and Milne variables}
In the Newtonian theory of self-gravitating fluids, an important role is played by the polytropic EoS and the isothermal EoS. This is because for these models there exists the so-called \emph{homology theorems}~\cite{Cha39}. Roughly, a differential equation is said to admit a \emph{homology transformation} if, given a solution, a whole class of solutions can be obtained by a simple change of scale. 
In the case of fluids, the deformation potential is a function of a single variable, $\widehat{w}=\widehat{w}(\delta)$, and $\widehat{s}(\delta,\eta)=0$, so that the equation for steady states can be written as a second order ordinary differential equation for $\delta(r)$:
%
%
\begin{equation}\label{POissonSS}
\frac{d}{dr}\left(r^2 \frac{c^2_\mathrm{s}(\delta)}{\delta} \frac{d\delta}{dr} \right) = -4\pi\varrho_0\delta r^2.
\end{equation}
Let us first consider a general class of models which arise from the condition that the speed of sound $c^2_\mathrm{s}(\delta)$ is a homogeneous function of $\delta$ of some degree $\kappa$, i.e. 
\begin{equation}\label{fluidSI}
c^2_\mathrm{s}(\delta)=\frac{\lambda}{\varrho_0}\delta^{\kappa},
\end{equation}
with $\kappa\in\mathbb{R}$. The isotropic pressure is
\begin{equation}
    \widehat{p}_\mathrm{iso}(\delta) = 
    \begin{cases}
        \frac{\lambda}{ (\kappa+1)} \delta^{\kappa+1}+D, \qquad &\kappa\neq -1, \\
        \lambda\ln{(\delta)}+D, \qquad &\kappa= -1,
    \end{cases}
\end{equation}
and the associated deformation potential function is
\begin{equation}
\qquad \widehat{w}(\delta)=
\begin{cases}
    \frac{\lambda}{\kappa(\kappa+1)} \delta^{\kappa}-D\delta^{-1}+E, \qquad &\kappa\neq -1,0, \\
    \lambda\ln{(\delta)} -D\delta^{-1}+E, \qquad &\kappa=0, \\
    -\lambda\left(1+\ln{(\delta)}\right)\delta^{-1}-D\delta^{-1}+E, \qquad &\kappa=-1,
\end{cases}
\end{equation}
with $D,E\in\mathbb{R}$ integration constants. For a natural stress-free reference state, i.e. $\widehat{p}^{(\mathrm{sf})}_\mathrm{iso}(1)=0$ and $\widehat{w}^{(\mathrm{sf})}(1)=0$, these constants are given by $D=-\lambda/(\kappa+1)$ and $E=-\lambda/\kappa$ for $\kappa\neq-1,0$, $D=-\lambda$, $E=-\lambda$ for $\kappa=0$, and $D=0$, $E=\lambda$ for $\kappa=-1$. 

The usual polytropic  EoS~\footnote{Usually written in the standard form 
$p_\mathrm{iso}(\varrho)=\mathcal{K}\varrho^{1+\frac{1}{\mathrm{n}}}$, $\varrho^{-1}_{0}w(\varrho)=\mathrm{n}\mathcal{K}\varrho^{\frac{1}{\mathrm{n}}}$, where $\mathcal{K}=\frac{\mathrm{n}\lambda}{1+\mathrm{n}}/\varrho^{1+\frac{1}{\mathrm{n}}}_0$.} has $D=E=0$ and $\kappa=1/\mathrm{n}$, where $\mathrm{n}\neq-1,0$ is the polytropic index~\footnote{The subclass with $\kappa=0$, i.e. $\mathrm{n}\rightarrow+\infty$, is the isothermal EoS.}, i.e.
\begin{equation}
    \widehat{w}_{\mathrm{pol}}(\delta)=\frac{\mathrm{n}^2\lambda}{1+\mathrm{n}} \delta^{1/\mathrm{n}},
\end{equation}
and is characterized by a pre-stressed reference state with
\begin{equation}
    \widehat{p}^{(\mathrm{ps})}_\mathrm{iso}(1)=p_0=\frac{\mathrm{n}\lambda}{1+\mathrm{n}}, \qquad \widehat{w}^{(\mathrm{ps})}(1)=w_0=\frac{\mathrm{n}^2\lambda}{1+\mathrm{n}}=\mathrm{n}p_0.
\end{equation}
It can be obtained from the natural stress-free reference state stored energy function $\widehat{w}^{(\mathrm{sf})}(\delta)$ of power-law type $(1,1)$, with $\kappa\neq-1,0$, via the transformation~\eqref{SFPSTransfw}, with the natural choice of the reference state pressure given by $p_0=\mathrm{n}\lambda/(1+\mathrm{n})$. 
\begin{proposition}\label{scaling invariance_polytrope}
	For a polytropic EoS, if $\delta(r)$ is a solution of the second order differential equation~\eqref{POissonSS}, then so is $A^{\frac{2}{\mathrm{n}-1}}\delta(Ar)$, with $A$ an arbitrary positive real number.
\end{proposition}
\begin{remark}
    The classical homology theorem for polytropic fluids makes use of the variable $h(r)=\delta(r)^{1/\mathrm{n}}$, which leads to the so-called \emph{Lane-Emden equation of index $\mathrm{n}$} \cite{1870AmJS...50...57L}. Solutions of the Lane-Emden equation with a regular center and compact support exist in the range $0\leq \mathrm{n} <5$,  while for $\mathrm{n}=5$ solutions with a regular center have infinite radius for any mass, and for $\mathrm{n}>5$ the polytropes have infinite mass and radius. The case  $\mathrm{n}=3$ is special, since the mass of the polytrope is independent of its radius. The Lane-Emden equation is also explicitly solvable for $\mathrm{n}=0,1,5$.
\end{remark}
Due to the homology theorem, equation~\eqref{POissonSS} admits scale-invariant solutions, obtained by making the power-law ansatz  
\begin{equation}\label{PLSol}
\delta(r)=c r^{p}
\end{equation}
and substituting into equation~\eqref{POissonSS}, which yields the relations 
\begin{equation}
p=-2\mathrm{n}/(\mathrm{n}-1) \quad  \text{ and } \quad c^{\mathrm{n}-1}=\left(\frac{\mathrm{n}\lambda}{4\pi\varrho^2_0}\right)\left(2\frac{2(\mathrm{n}-3)}{(\mathrm{n}-1)^2}\right). 
\end{equation}
Hence, for $\mathrm{n}>3$ we have the explicit solutions
\begin{equation}
\delta(r)=\left(\frac{\mathrm{n}\lambda}{4\pi\varrho^2_0}\right)^{\frac{\mathrm{n}}{\mathrm{n}-1}}\left(\frac{2(\mathrm{n}-3)}{(\mathrm{n}-1)^2}\right)^{\frac{\mathrm{n}}{\mathrm{n}-1}}r^{-\frac{2\mathrm{n}}{\mathrm{n}-1}}\,,
\end{equation}
which are singular at the center $r\rightarrow0^{+}$ and have infinite radius.

The second order differential equation~\eqref{POissonSS} has a general solution depending on two integration constants. However, one of these constants is determined by the homology constant $A$, and hence the second order equation can be reduced to a first order differential equation by finding a suitable set of scale-invariant variables, i.e., functions that are invariant under homologous transformations. One such set of functions consists of the \emph{Milne variables} $(u,v)\in(0,+\infty)^2$, defined by
\begin{equation}
u(r) = \frac{r}{m}\frac{dm}{dr}= 3\frac{\delta}{\eta} , \quad v(r)=-\frac{r}{(1+\mathrm{n})p_\mathrm{iso}}\frac{dp_\mathrm{iso}}{dr}=\frac{4\pi\varrho^{2}_0}{3\mathrm{n}\lambda}r^2\eta^{1-\frac{1}{\mathrm{n}}} \left(\frac{\delta}{\eta}\right)^{-\frac{1}{\mathrm{n}}}.
\end{equation}
These variables obey the first order system of autonomous differential equations (in the independent variable $\xi=\ln{r}$)
\begin{equation}\label{uvMilne}
r\frac{du}{dr} =  (3-u-\mathrm{n}v)u \,,\quad r\frac{dv}{dr} = (u+v-1)v,
\end{equation}
from which the single first order differential equation on the homology invariant variables follows:
\begin{equation}
\frac{v}{u}\frac{du}{dv} = -\frac{u-3+\mathrm{n}v}{u+v-1}.
\end{equation}
It should be noted that the Milne variable $v$ is positive only for $\mathrm{n}>0$. When the coefficient $c^2_\mathrm{s}(\delta)$ is not a homogeneous function, like in the case of asymptotically polytropic equations of state, the state space is higher-dimensional, since in this case there exists a scale-dependent variable which does not decouple from the Milne $(u,v)$ system (see~\cite{HEINZLE200318} for more details).
The 2-dimensional reduced state space for the polytropic EoS can be found in Figure 3 of~\cite{HEINZLE200318}. The self-similar power-law solutions appear as an interior fixed point with coordinates $(u,v)=\left(\frac{\mathrm{n}-3}{\mathrm{n}-1},\frac{2}{\mathrm{n}-1}\right)$. 

We now turn to the anisotropic elastic setting. Due to the additional term $s(\delta,\eta)$, it is no longer possible to write equation~\eqref{SSdelta} as a second order ordinary differential equation for $\delta$ alone. Nevertheless, and similarly to the fluid case, a homology theorem for power-law deformation functions can be easily established. First, we must find a condition for power-law materials so that $c^2_\mathrm{L}(\delta,\eta)$ and $s(\delta,\eta)$ are homogeneous functions of some degree $\kappa$, and then prove the invariance under homologous transformations of the first order integro-differential equation~\eqref{SSdelta}.
%
%
%
%
%
Let
\begin{equation}
I_j :=\{i\in\{1,...,n_j\}\,:\,\beta_{ij}\neq-1,\beta_{ij}\neq0\}.
\end{equation}
By condition $(iii)$ on $\beta_{ij}$ in Definition~\ref{PLDP}, $I_j$ is not empty for at least one $j\in\{1,...,m\}$.
\begin{proposition}\label{Homcs}
Given a power-law deformation potential  function $\widehat{w}(\delta,\eta)$ of type $(n_1,n_2,...,n_m)$,  $j=1,..,m$,  the functions $c^2_\mathrm{L}(\delta,\eta)$ and $s(\delta,\eta)$ are homogeneous functions of degree $\kappa$ if and only if there exists a unique  $p\in\{1,2,\cdots,m\}$ such that $I_p \neq \varnothing$, and in this case $\kappa=\theta_p$.
\end{proposition}
\begin{proof}
For the power-law stored energy functions we have from~\eqref{PPL}
\begin{subequations}
\begin{align}
        c^2_\mathrm{L}(\delta,\eta) &= \varrho^{-1}_0\sum^{m}_{j=1}\eta^{\theta_j}\sum^{n_j}_{i=1}\alpha_{ij}\beta_{ij}(1+\beta_{ij})\left(\frac{\delta}{\eta}\right)^{\beta_{ij}},\label{c_PL} \\
        s(\delta,\eta) &= \varrho^{-1}_0\sum^{m}_{j=1}\eta^{\theta_j}\sum^{n_j}_{i=1}\alpha_{ij}(\theta_j-\beta_{ij})\left[-\frac{1-\left(\frac{\delta}{\eta}\right)^{1+\beta_{ij}}}{1-\left(\frac{\delta}{\eta}\right)}+\beta_{ij}\left(\frac{\delta}{\eta}\right)^{1+\beta_{ij}}\right].
\end{align}
\end{subequations}
Under the scaling transformation $(\delta,\eta)\rightarrow A (\delta,\eta)$, the quantity $\delta/\eta$ is invariant. Moreover the $i$-sums in the function $c^2_\mathrm{L}$ vanish for all $(\delta,\eta)$ if and only if $I_j$ are empty sets. Hence, $c^2_\mathrm{L}$ is a homogeneous function for all $(\delta,\eta)$ if and only if the $i$-sums vanish for all $I_j$ sets except one, i.e., if there is a unique $p\in\{1,2,\dots,m\}$ such that $I_p\neq\varnothing$. By definition, $\theta_p$ exists, is unique, and $\kappa=\theta_p$. Now, the $i$-sums on the empty sets $I_j$ in the function $s(\delta,\eta)$ give
	\begin{equation}
	\sum_{j\neq p}^m \eta^{\theta_j}\left( \alpha_{0j}\theta_j
	+\alpha_{-1j}(\theta_j+1)\right).
	\end{equation}
By condition~\eqref{PLSEF2nd}, each of these terms vanish, and $s(\delta,\eta)$ is a homogeneous function of degree $\kappa=\theta_p$. 
\end{proof}
The next proposition generalizes the polytropic homology theorem to hyperelastic materials with power-law deformation potentials.
\begin{proposition}\label{scalinginvariance}
	Let $\widehat{w}(\delta,\eta)$ be a power-law deformation potential function with a unique non-empty $I_j$ set. If $\delta(r)$ is a solution of the integro-differential equation~\eqref{SSdelta}, then so is $A^{\frac{2}{1-\theta_p}}\delta(Ar)$, with $A$ an arbitrary positive real number.
\end{proposition}
\begin{proof}
	Let $\tilde{r}=Ar$, $\delta(r)=A^{\frac{2}{1-\theta_p}}\tilde{\delta}(\tilde{r})$. By the definition of $\eta$~\eqref{NewForm1}, it follows that $\eta(r)=A^{\frac{2}{1-\theta_p}}\tilde{\eta}(\tilde{r})$. Since $c^2_\mathrm{L}$ and $s$ are homogeneous functions of degree $\kappa=\theta_p$, we have
	\begin{equation}
	c^2_\mathrm{L}(\delta,\eta)=A^{\frac{2\theta_p}{1-\theta_p}}c^2_\mathrm{L}(\tilde{\delta},\tilde{\eta})\,,\quad s(\delta,\eta)=A^{\frac{2\theta_p}{1-\theta_p}}s(\tilde{\delta},\tilde{\eta}),
	\end{equation}
	and equation~\eqref{SSdelta} reads
	\begin{equation}
	c^2_\mathrm{L}(\tilde{\delta},\tilde{\eta})\tilde{r}\frac{d\tilde{\delta}}{d\tilde{r}} = 3s(\tilde{\delta},\tilde{\eta})(\tilde{\delta}-\tilde{\eta})- \frac{4\pi\varrho_0}{3}\tilde{r}^2\tilde{\delta}\tilde{\eta},
	\end{equation}
	which shows its invariance under the homology transformation.
\end{proof}

\begin{examples}
Some examples of scale-invariant materials of power-law type are:
\begin{itemize}
    \item[(i)] Polytropic fluid:
\end{itemize}
The power-law deformation potential  of type $(1,1)$ with a stress-free natural reference state,
    \begin{equation}
        \widehat{w}^{(\mathrm{sf})}_{\mathrm{pol}}(\delta)=-\mathrm{n}\lambda+\frac{\mathrm{n}\lambda}{(1+\mathrm{n})}\eta^{-1}\left(\frac{\delta}{\eta}\right)^{-1}+\frac{\mathrm{n}^2\lambda}{1+\mathrm{n}}\eta^{\frac{1}{\mathrm{n}}}\left(\frac{\delta}{\eta}\right)^{\frac{1}{\mathrm{n}}},
    \end{equation}
   has a unique non-empty $I_j$ set, with $\theta_p=1/\mathrm{n}$.
   \begin{itemize}
       \item[(ii)] Quasi-linear John materials:
   \end{itemize}
   Quasi-linear John materials are an example of type $(1,3,2)$, with stored energy function  
\begin{equation}
\begin{split}
    w^{(\mathrm{sf})}_{\mathrm{John}}(\delta,\eta)=&\frac{1}{2}(9\lambda+10\mu)+2\mu \eta^{-1}\left(\frac{\delta}{\eta}\right)^{-1}+(\lambda+2\mu)\eta^{-\frac{2}{3}}\left(2+2\left(\frac{\delta}{\eta}\right)^{-1}+\frac{1}{2}\left(\frac{\delta}{\eta}\right)^{-2}\right) \\
    &-(3\lambda+4\mu)\eta^{-\frac{1}{3}}\left(2+\left(\frac{\delta}{\eta}\right)^{-1}\right),
\end{split}
\end{equation}
     which belongs to a more general class of harmonic materials introduced by Fritz John~\cite{John63}. This material has a unique non-empty $I_j$ set, with $\theta_p=-2/3$, and hence it is scale-invariant.
      \begin{itemize}
       \item[(iii)] Lamé type:
   \end{itemize}
   For the class of hyperelastic power-law materials, if $I_j \geq 3$ then $I_j$ must be non-empty. Therefore, scale invariance requires such a set $I_j$ to be unique, and all others sets $I_k$ ($k\neq j$) to be empty, so they are of type $1$ or $2$. For type $(2,3)$ or $(1,3)$ and their permutations, $I_1$ ($I_2$ in the permuted case) is necessarily empty. The simplest choice is given by\footnote{The non-empty set $I_p$ is chosen such that $\beta_{1p}=-1$, and $\beta_{2p}=0$, which simplifies considerably the form of $c^{2}_\mathrm{L}$ given in~\eqref{c_PL}.}
    \begin{equation}
    \begin{split}
        \widehat{w}^{(\mathrm{sf})}_{(1,3)}(\delta,\eta)= &-\mathrm{n} K +\frac{\mathrm{n}K}{1+\mathrm{n}}\delta^{-1}+\eta^{\frac{1}{\mathrm{n}}}\left[-\left((\mathrm{s}-\mathrm{n})K+\mathrm{s}\frac{4\mu}{3}\right)\right.\\
        &\left.+\frac{1}{1+\mathrm{s}}\left(\frac{(\mathrm{s}-\mathrm{n})K}{1+\mathrm{n}}+\mathrm{s}\frac{4\mu}{3}\right)\left(\frac{\delta}{\eta}\right)^{-1} +\frac{\mathrm{s}^2}{1+\mathrm{s}}\left(K+\frac{4\mu}{3}\right)\left(\frac{\delta}{\eta}\right)^{\frac{1}{\mathrm{s}}}\right],
        \end{split}
    \end{equation}
    which belongs to a more general class of scale-invariant elastic materials introduced by Calogero in~\cite{Cal21}, containing also generalizations of the isothermal fluid EoS, corresponding to $\mathrm{n}\rightarrow+\infty$.    
\end{examples}
\begin{remark}
    The elastic polytropic material with a pre-stressed reference state can be obtained from the $(1,3)$ power-law deformation potential by using the transformation~\eqref{SFPSTransfw} with the natural choice $p_0=\mathrm{n}K/(1+\mathrm{n})$, resulting in the stored energy function 
\begin{equation}\label{NSE}
\begin{split}
\widehat{w}^{(\mathrm{ps})}_{(1,3)}(\delta,\eta)=    & \,\, \eta^{\frac{1}{\mathrm{n}}}\left[-\left((\mathrm{s}-\mathrm{n})K+\mathrm{s}\frac{4\mu}{3}\right)\right.\\
&\left.+\frac{1}{1+\mathrm{s}}\left(\frac{(\mathrm{s}-\mathrm{n})K}{1+\mathrm{n}}+\mathrm{s}\frac{4\mu}{3}\right)\left(\frac{\delta}{\eta}\right)^{-1} +\frac{\mathrm{s}^2}{1+\mathrm{s}}\left(K+\frac{4\mu}{3}\right)\left(\frac{\delta}{\eta}\right)^{\frac{1}{\mathrm{s}}}\right].
\end{split}
\end{equation}
In Section~\ref{Polytropes}, when discussing relativistic polytropes, we deduce the above deformation potential by starting with the simplest quasi-Hookean material (the quadratic model) which, in view of propositions~\ref{Homcs}-\ref{scalinginvariance}, is Newtonian scale-invariant for $\mathrm{n}=1$, and then generalize the scale invariance property for all $\mathrm{n}$. Moreover, starting with the pre-stressed relativistic stored energy function, the generalization to the isothermal EoS can be easily obtained by taking the limit $\mathrm{n}\rightarrow+\infty$. Other materials can be obtained from the limit $\mathrm{s}\rightarrow+\infty$.
\end{remark}

Similarly to perfect fluids with polytropic EoS, for scale-invariant power-law elastic materials equation~\eqref{SSdelta} also admits power-law solutions of the type~\eqref{PLSol}. A result about the existence of such exact solutions was presented in~\cite{ACL21}. Here we will discuss them in a dynamical systems context.
%
%
%

In~\cite{Alho:2019fup}, a theorem on the existence and uniqueness of ball solutions for power-law type stored energy functions was given, exemplified by several material laws, such as the Saint-Venant Kirchhoff, John, and Hadamard materials. The analysis relied on the introduction of a new dynamical systems formulation of the static Newtonian equations, which made use of the dimensionless variables
\begin{equation}\label{NewMilne}
    x(r)=\eta(r), \quad y(r)=\frac{\delta(r)}{\eta(r)},\quad z(r)=\frac{4\pi\rho^2_0}{3(\lambda+2\mu)}\eta(r)^{1-\frac{1}{n}}\left(\frac{\delta(r)}{\eta(r)}\right)^{\tau-\frac{1}{\mathrm{s}}}r^2,
\end{equation}
the new independent variable $\xi\in(-\infty,+\infty)$ defined by
\begin{equation}
\frac{d}{d\xi} = y^{\tau}\Gamma(x,y)r\frac{d}{dr},  
\end{equation}
and the functions $\Gamma(x,y)$ and $\Upsilon(x,y)$ defined by%
\begin{subequations}
\begin{align}
    \Gamma(x,y) &=\frac{\varrho_0\delta c^2_\mathrm{L}(\delta,\eta)}{(\lambda+2\mu)\eta^{1+\frac{1}{\mathrm{n}}}\left(\frac{\delta}{\eta}\right)^{1+\frac{1}{\mathrm{s}}}} 
    =\frac{y\partial_y \mathcal{P}_\mathrm{rad}(x,y)}{(\lambda+2\mu)x^{1+\frac{1}{\mathrm{n}}}y^{1+\frac{1}{\mathrm{s}}}} , \\
    \Upsilon(x,y) &= \frac{\varrho_0\left(\delta c^2_\mathrm{L}(\delta,\eta)+\eta s(\delta,\eta)\right)}{(\lambda+2\mu)\eta^{1+\frac{1}{\mathrm{n}}}\left(\frac{\delta}{\eta}\right)^{1+\frac{1}{\mathrm{s}}-\tau}} 
    =\frac{x\partial_x \mathcal{P}_\mathrm{rad}(x,y) +\frac{2}{3}\frac{\mathcal{Q}(x,y)}{1-y}}{(\lambda+2\mu)x^{1+\frac{1}{\mathrm{n}}}y^{1+\frac{1}{\mathrm{s}}-\tau}},
\end{align}
	\end{subequations}
 where $\mathcal{P}_\mathrm{tan}(x,y)=\widehat{p}_\mathrm{rad}(\delta,\eta)$, $\mathcal{P}_\mathrm{tan}(x,y)=\widehat{p}_\mathrm{tan}(\delta,\eta)$, and $\mathcal{Q}(x,y)=\mathcal{P}_\mathrm{tan}(x,y)-\mathcal{P}_\mathrm{tan}(x,y)$. The values of the parameters $(\mathrm{n},\mathrm{s},\tau)$ are such that the functions $\Gamma$, and $\Upsilon$ extend continuously to $(0,0)$, with $\Gamma(0,0)\neq0$ and $\Upsilon(0,0)\neq0$. The parameters $\mathrm{n}$ and $\mathrm{s}$ are the \emph{polytropic index} and the \emph{shear index}, respectively, while the constant $\tau$ is a regularity parameter and has no physical significance. Now, let $\theta_*=\min\{\theta_j : j=1,...,m, I_j\neq\varnothing\}$, and $\beta_*=\min\{\beta_{ij}: i\in I_j, j=1,...,m, I_j\neq\varnothing\}$. Since $\theta_j$ are increasing, then $\theta_*=\theta_p$, where $p$ is the lowest value of $j$ such that $I_j$ is not empty. Denoting by $\alpha_{(-1)}$ and $\alpha_{(0)}$ the coefficients $\alpha_{ij}$ of the term with exponent $\beta_{ip}=-1$ and $\beta_{ip}=0$, respectively, let
\begin{equation}
    \sigma = \alpha_{(0)}\theta_p+\alpha_{(-1)}(1+\theta_p).
\end{equation}
The following proposition was proved in~\cite{Alho:2019fup}.
\begin{proposition}

For hyperelastic power-law materials, the functions $\Gamma(x,y)$ and  $\Upsilon(x,y)$ satisfy $\Gamma(0,0)\neq0$, $\Upsilon(0,0)\neq0$ if and only if $\beta_*=\beta_{pq}$ for some (necessarily unique) $q\in I_p$. Moreover $\mathrm{n}=\theta^{-1}_*$, $\mathrm{s}=\beta^{-1}_*$, and
\begin{equation}
    \tau=
    \begin{cases}
        0,\qquad &\text{if}\quad \beta_*<0\quad\text{or}\quad\beta_*>0\quad\text{and}\quad\sigma=0, \\
        \beta_*  &\text{if}\quad \beta_*>0\quad\text{and}\quad\sigma\neq0.
    \end{cases}
\end{equation}
\end{proposition}
\begin{examples}
    Some examples of the functions $\Gamma(x,y)$ and $\Upsilon(x,y)$ are:
    \begin{itemize}
        \item[(i)] Polytropic fluid: $\mathrm{s}=\mathrm{n}$, $\sigma=0$, $\tau=0$, and $\Upsilon=\Gamma=1$.

        \item[(ii)] Quasi-linear John's materials: $\mathrm{n}=-\frac{3}{2}$, $\mathrm{s}=-\frac{1}{2}$,   $\Gamma=1$, $\sigma=-\frac{2}{3}(\lambda+2\mu)$, and 
\begin{equation}
    \Upsilon(y) = \frac{1}{3}(1+2y)
\end{equation}
        \item[(iii)] Type $(1,3)$ or $(3,1)$ elastic polytrope: $\Gamma=1$, $ \sigma=-\frac{(\lambda+2\mu)\mathrm{s}}{\mathrm{n}(1+\mathrm{s})}(\mathrm{s}-\mathrm{n})$, and
        \begin{subequations}
        \begin{align}
         &\Upsilon(y) = 1,\quad  \mathrm{s}=\mathrm{n}>0, \quad \tau=0; \\
         &\Upsilon(y) =\frac{\mathrm{s}}{\mathrm{n}}\left[\frac{(1+\mathrm{n})}{(1+\mathrm{s})}+\frac{(\mathrm{s}-\mathrm{n})}{(1+\mathrm{s})}\left(\frac{1-y^{-\frac{1}{\mathrm{s}}}}{1-y}\right)\right] ,\quad \mathrm{s}<0,\quad \tau=0;    \label{s<0}        \\
         &\Upsilon(y) =\frac{1}{1-y}\frac{\mathrm{s}}{\mathrm{n}}\left[-\frac{(\mathrm{s}-\mathrm{n})}{(1+\mathrm{s})}+y^{\frac{1}{\mathrm{s}}}-\frac{(1+\mathrm{n})}{(1+\mathrm{s})}y^{1+\frac{1}{\mathrm{s}}}\right] ,\quad\mathrm{s}\neq\mathrm{n},\quad \mathrm{s}>0,\quad \tau=\frac{1}{\mathrm{s}}.
        \end{align}
        \end{subequations}
    \end{itemize}
\end{examples}
If $\Gamma,\Upsilon\in C^1([0,+\infty)^2)$, we obtain the following $C^1$ dynamical system:
\begin{subequations}\label{MasterDS}
	\begin{align}
	\frac{dx}{d\xi} &=-3\Gamma(x,y) y^{\tau} (1-y) x, \\
	\frac{dy}{d\xi} &= \left[3\Upsilon(x,y)(1-y)-z\right]y, \\
	\frac{dz}{d\xi} &= \left[\Gamma(x,y) y^{\tau} \left(2-\frac{3}{\mathrm{n}}\left(\mathrm{n}-1\right)(1-y)\right)+\left(\tau-\frac{1}{\mathrm{s}}\right)(3\Upsilon(x,y)(1-y)-z)\right]z.
	\end{align}
\end{subequations}
If $\tau=0$ or $\tau\geq1$, the state space $\mathbf{S}=\{(x,y,z)\in\mathbb{R}^3: x>0, y>0, z>0\}$ can be extended in a $C^1$ manner to its invariant boundaries contained in the planes $\{x=0\}$, $\{y=0\}$ and $\{z=0\}$. In these variables, the isotropic states $(\delta,\eta)=(\delta,\delta)$ are the points $(x,y)=(x,1)$, and the reference state is $(x,y)=(1,1)$. Due to~\eqref{ISOCond}--\eqref{ISOCond2} and~\eqref{DevPreference}, the functions $\Gamma(x,y)$, $\Upsilon(x,y)$ satisfy
\begin{equation}
    \Upsilon(x,1)=\Gamma(x,1), \qquad \Upsilon(1,1)=\Gamma(1,1)=1.
\end{equation}
For scale-invariant models, the functions $\Gamma$ and $\Upsilon$ just depend on the variable $y$, and the scale-dependent variable $x$ decouples, leaving a reduced $2$-dimensional dynamical system for the scale-invariant variables $(y,z)$, which are related to the Milne homology invariant variables $(u,v)$ by $y=u/3$ and $z=\mathrm{n}v$. In the polytropic fluid limit, the system reduces to~\eqref{uvMilne}. 
For scale-invariant materials, the irregular center solutions appear as straight orbits in the full state space, see~\cite{Alho:2019fup}. In this case, the proper scale-invariant solutions of power-law type appear as an interior fixed point located at
\begin{equation}
(y_{\star},z_{\star})=\left(\frac{\mathrm{n}-3}{3(\mathrm{n}-1)},\frac{2\mathrm{n}\Upsilon\left(\frac{\mathrm{n}-3}{3(\mathrm{n}-1)}\right)}{\mathrm{n}-1}\right),
\end{equation}
which exists if $\mathrm{n}>3$ or $\mathrm{n}<0$ and $\Upsilon\left(\frac{\mathrm{n}-3}{3(\mathrm{n}-1)}\right)>0$, or if $0<\mathrm{n}<1$ and $\Upsilon\left(\frac{\mathrm{n}-3}{3(\mathrm{n}-1)}\right)<0$. Moreover, for $\mathrm{n}>3$ or $\mathrm{n}<0$, it follows that $y_\star<1$.
\subsubsection{Existence, uniqueness, and asymptotics for power-law type materials}
Here, we briefly review results on existence and uniqueness of solutions with a regular center and compact support, see~\cite{Alho:2019fup} for more details. For regular center solutions, it follows from~\eqref{icNewtSS} and~\eqref{NewMilne} that
\begin{equation}
\lim_{r\rightarrow0^{+}}x(r)=x_\mathrm{c}=\delta_\mathrm{c},\qquad  \lim_{r\rightarrow0^{+}}y(r)=1, \qquad \lim_{r\rightarrow0^{+}}v(r)=0,   
\end{equation}
which in the extended state space consists of a normally hyperbolic line fixed points 
on the invariant boundary $\{v=0\}$. The following theorem, whose proof can be found in~\cite{Alho:2019fup}, provides a global existence and uniqueness result for regular center solutions: 
\begin{theorem}
Let $\Gamma,\Upsilon\in C^1([0,+\infty)^2)$, and let $X_\flat\in(0,+\infty]$ be such that $\Gamma(x,1)>0$ for all $x\in(0,X_\flat)$, and $\Gamma(X_\flat,1)=0$ if $X_\flat<+\infty$. Assume further that:
\begin{itemize}
    \item[(i)] $\Gamma(x,y)>0$ for all $(x,y)\in[0,X_\flat)\times[0,1]$;
    \item[(ii)] $\Upsilon_0(y)=\Upsilon(0,y)>0$, and 
    $$
    \left[\left(\tau-\frac{1}{\mathrm{s}}\right)\Upsilon_{0}(y)-y\Upsilon^{\prime}_0(y)\right](1-y)+y\Upsilon_{0}(y)=y^{\tau}\left(1-\left(1+\frac{1}{\mathrm{n}}\right)(1-y)\right)\Gamma(0,y)>0,
    $$ 
    for all $y\in[0,1]$;
\item[(iii)] $\mathrm{s}<0$, and $\mathrm{n}>3$ or $\mathrm{n}<0$.
\end{itemize}
Then for all $x_c\in(0,X_\flat)$ there exists a unique solution trajectory of~\eqref{MasterDS} satisfying 
\begin{equation}
    \lim_{\xi\rightarrow-\infty}(x(\xi),y(\xi),z(\xi))=(x_\mathrm{c},1,0).
\end{equation}
Moreover, $x(\xi)<x_\mathrm{c}$, $y(\xi)<1$, for all $\xi\in(-\infty,+\infty)$, and
\begin{equation}\label{GlobAs}
    \lim_{\xi\rightarrow+\infty}(x(\xi),y(\xi),z(\xi))=(0,y_\star,z_\star),
\end{equation}
with $x(\xi)\sim e^{-\frac{2\mathrm{n}}{\mathrm{n}-1}\xi}$ as $\xi\rightarrow+\infty$.
\end{theorem}
%
%
    %
%
%
	%
\begin{remark}
    Condition $(i)$ is equivalent to the strict hyperbolicity condition, while conditions $(ii)$ and $(iii)$ are more restrictive. For example, in the case of hyperelastic power-law materials with $\mathrm{s}>0$ and $\sigma\neq0$, the invariant boundary $\{y=0\}$ is a surface of non-hyperbolic fixed points $(x_0,0,z_0)$, and the right-hand side of~\eqref{MasterDS} is $C^1$ at $y=0$ only for $\mathrm{s}\leq1$. The above result can be easily extended to include $\mathrm{n}=3$, although the asympotics for large radius differ due to the existence of a center manifold in this case.
\end{remark}
The qualitative properties of the invariant boundary $\{x=0\}$ can be found in Figure~1 in~\cite{Alho:2019fup}, where the fixed point structure depends on the sign of $3\Upsilon_0(0)+\frac{\mathrm{s}}{\mathrm{n}}(\mathrm{n}-3)\Gamma_0(0)$.  The asymptotics are given by
 \begin{equation}
    \varrho(r) = \mathcal{O}(r^{-\frac{2\mathrm{n}}{\mathrm{n}-1}}) ,\quad 	  m(r)  = \mathcal{O}(r^{\frac{\mathrm{n}-3}{\mathrm{n}-1}}), \quad 	\phi_N(r) = \mathcal{O}(r^{-\frac{2}{(\mathrm{n}-1)}})\quad  \text{as}\quad r\rightarrow +\infty.
 \end{equation}
From the above global result one can construct elastic balls as follows: Assume that for balls under radial compression there exist $X_{\pm}\in[0,\infty]$, with $X_-<X_+$ such that $\mathcal{P}_{\mathrm{tan}}(x_c,1)=\mathcal{P}_{\mathrm{rad}}(x_c,1)>0$ for all $x_c\in(X_-,X_+)$, and $\mathcal{P}_{\mathrm{tan}}(X_\pm,1)=\mathcal{P}_{\mathrm{rad}}(X_\pm,1)=0$. For materials with a stress-free reference state $X_-=1$, while for pre-stressed reference state $X_-=0$. Let $X_\sharp=\min\{X_\flat,X_+\}$; then, under the assumptions of the above theorem, for any given $x_c<X_\sharp$ there exists a unique global solution and the limit~\eqref{GlobAs} holds. Hence, if the radial pressure becomes zero, $\mathcal{P}_\mathrm{rad}(x,y)=0$ for some values of $(x,y)\in(0,X_\sharp)\times(0,1)$, then a static self-gravitating elastic ball can be constructed by truncating the global solution.

As an example, let us take the scale-invariant elastic material given by the deformation potential~\eqref{NSE}, which was analyzed in detail by Calogero in~\cite{Cal21}. For balls under radial compression we must have $\gamma>0$, i.e. $\mathrm{n}>0$ or $\mathrm{n}<-1$. As shown in~\cite{Cal21} (see also Section~\ref{SBNSI} below), the radial pressures vanishes either at $x=0$ or at constant values of $y$, denoted by $y_\mathrm{b}$, and given in equation~\eqref{ybDef}. Now, for this model, $\Gamma=1>0$ for all $(x,y)$, so that condition $(i)$ of the theorem is satisfied, while, by~\eqref{s<0}, condition $(ii)$ is satisfied if and only if $\mathrm{n}<0$. Hence, in this case, regular compressed balls with compact support exist for $\mathrm{s}<0$, $\mathrm{n}<-1$ and $y_\star<y_\mathrm{b}$, i.e.
\begin{equation}
   \frac{\mathrm{n}-3}{3(\mathrm{n}-1)}<\left(1-\frac{\mathrm{n}(1+\mathrm{s})}{3(1+\mathrm{n})\mathrm{s}}\left(\frac{1+\nu}{1-\nu}\right)\right)^{\frac{\mathrm{s}}{1+\mathrm{s}}}.
\end{equation}
In~\cite{Cal21}, balls with compact support were also proved to exists for $0<\mathrm{n}<1$ and $\mathrm{s}\geq\mathrm{n}>0$. However, in the important parameter range $\mathrm{n}\geq1$ $(1<\gamma\leq2)$ this is still an open problem.

%
%

\subsection{Homogeneous (separable) solutions}
%
Newtonian self-gravitating self-similar fluid solutions were introduced by Goldreich \& Weber~\cite{GW80}, and by Makino~\cite{Mak92}, where it was shown that for polytropic fluid models they exist for $\gamma=\frac{4}{3}$, i.e., $\mathrm{n}=3$. In a recent paper~\cite{Cal21}, Calogero used the Newtonian scale-invariant stored energy function which generalizes polytropic fluid models, and showed that self-similar solutions also exist similarly to the fluid case. Following Goldreich \& Weber~\cite{GW80}, Makino~\cite{Mak92}, and Fu \& Lin~\cite{FL98}, the idea is to start by introducing the ansatz
\begin{equation}\label{ansatz}
r = a(t) z, \qquad \delta(t,r)=a(t)^{-3}\delta_0(z),\qquad  v(t,r)=\dot{a}(t)z.
\end{equation}
It follows that the mass is scale-invariant:
\begin{equation}\label{ansatz2}
m(t,r)=m_0(z).
\end{equation}
Moreover, the average mass density scales as $\delta$, i.e.
\begin{equation}\label{etascaling}
\eta(t,r)=a(t)^{-3}\eta_0(z).
\end{equation}
If in addition
\begin{equation}\label{CoefSI}
c^2_\mathrm{L}(\delta(t,r),\eta(t,r))= a(t)^{-1}c^2_\mathrm{L}(\delta_0(z),\eta_0(z)),\qquad s(\delta(t,r),\eta(t,r))= a(t)^{-1}s(\delta_0(z),\eta_0(z)),
\end{equation}
then the system~\eqref{sssystem2} reduces to
\begin{subequations}
	\begin{align}
	a(t)^2\ddot{a}(t) \delta_0(z) z+ c^2_\mathrm{L}(\delta_0(z),\eta_0(z)) \frac{d\delta_0}{dz} &=\frac{3}{z}s(\delta_0(z),\eta_0(z))(\eta_0(z)-\delta_0(z))-\frac{4\pi\varrho_0}{3}\eta_0(z) \delta_0(z) z, \\
	\frac{d\eta_0}{dz} &=-\frac{3}{z}(\eta_0(z)-\delta_0(z)),
	\end{align}
\end{subequations}
which can be solved by separation of variables by letting 
\begin{equation}\label{MasterEq1}
a(t)^2\ddot{a}(t)=\alpha=-\frac{\Lambda}{3},
\end{equation}
where $\Lambda\in\mathbb{R}$. The following proposition yields a description of solutions of the above equation. 
\begin{proposition}[Makino~\cite{Mak92}, and Fu \& Lin~\cite{FL98},  Hadzic \& Jang~\cite{HJ18}]
	Let $a(t)$ be a solution of~\eqref{MasterEq1} with initial conditions
	\begin{equation}
	a(0)=a_0>0,\qquad \dot{a}(0)=a_1.
	\end{equation}
        Then: 
	\begin{itemize}
		\item[1)] If $\Lambda<0$ then $a(t)>0$ for all $t>0$, and $a(t)\sim c_1(1+c_2 t)$ as $t\rightarrow+\infty$ for some positive constants $c_1$, $c_2$.
		\item[2)] If $\Lambda=0$ then $a(t)=a_0+a_1 t$.
		\item[3)] If $\Lambda>0$, let $a^{*}_1 = \sqrt{2\Lambda/3a_0}$.   Then:
		\begin{itemize}
			\item[i)] If $a_1>a^{*}_1$ then $a(t)>0$ for all $t>0$, and the asymptotics are given as in $(1)$; 
			\item[ii)] If $a_1=a^{*}_1$ then $a(t)>0$ for all $t>0$, and
			\begin{equation}
			a(t)=a_0 \left(1+\frac{3}{2}\frac{a_1}{a_0}t\right)^{\frac{2}{3}} , \qquad t\geq 0;
			\end{equation} 
			\item[iii)] If $a_1<a^{*}_1$ then there exists $T\in(0,+\infty)$ such that $a(t)>0$ in $(0,T)$, and $a(t)=k_1 (T-k_2t)^{\frac{2}{3}}$ as $t\rightarrow k^{-1}_2T$ for some positive constants $k_1,k_2$. In particular, $a(t)\rightarrow0$ as $t\rightarrow k^{-1}_2T$.
		\end{itemize}
	\end{itemize}
\end{proposition}
This leads to the system of differential equations
\begin{subequations}\label{MasterEq2}
	\begin{align}
	c^2_\mathrm{L}(\delta_0,\eta_0)  \frac{d\delta_0}{dz} &=  \frac{3}{z}s(\delta_0,\eta_0)(\eta_0-\delta_0) -\delta_0 \frac{4\pi\varrho_0}{3}\left( \eta_0 - \frac{\Lambda}{4\pi\varrho_0}\right) z, \\
	\frac{d\eta_0}{dz} &=-\frac{3}{z}(\eta_0-\delta_0).
	\end{align}
\end{subequations}
When $\Lambda=0$ we recover the steady state equations. The notational choice of the constant $\Lambda$ is due to the fact that this set of equations is effectively equivalent to the Newtonian problem of static self-gravitating matter with a cosmological constant $\Lambda$.\footnote{In spherical symmetry, the Poisson equation with a cosmological constant reads
\begin{equation}
    \frac{1}{r^2}\frac{d}{dr}\left(r^2\frac{d\phi_N}{dr}\right)=4\pi\rho-\Lambda,
\end{equation}
which can be written as a set of two first order equations, one for the Newtonian potential, $\frac{d\phi_N}{dr}=\frac{m}{r^2}-\frac{\Lambda}{3}r$, and one for the mass, $\frac{dm}{dr}=4\pi r^2\rho$, see~\cite{Ehlers}.}  If $\Lambda\leq0$ then a similar analysis as in the previous section applies, and strongly regular ball solution can be proved to exist. When $\Lambda>0$ the problem is more involved. For the initial data
\begin{equation}
   \delta_0(0)  = \frac{\Lambda}{4\pi\varrho_0}
\end{equation}
solutions have constant mass density,
\begin{equation}
 \delta^{\prime}_0(z)=0, \qquad m_0(z)= \frac{\Lambda}{3}z^3,
\end{equation}
for which the pressure is constant, and therefore no ball solutions exist. For polytropic fluids, Fu \& Lin have given a characterization of the initial data for which ball solutions exist (see also the numerical work in~\cite{refId0}, where, for $\mathrm{n}=3$, it is shown that ball solutions exist for $\delta_\mathrm{0}(0)\gtrsim\frac{\Lambda}{24\pi\varrho_0} 10^{3}$).
%
%
%
%
For the Newtonian stored energy function~\eqref{NSE}, the scaling~\eqref{CoefSI} is satisfied only for $\mathrm{n}=3$ ($\gamma=4/3$).
As in the fluid case, it can be shown that the ball solutions satisfying  $a_1=\sqrt{\frac{2\Lambda}{3a_0}}$ have zero total energy, corresponding to true self-similar solutions, see Figure~(1) in~\cite{HJ18}.%

\subsection{Linear elastic wave solutions}
Linearizing the equations around the natural reference state, that is, taking $\delta=1+\check{\delta}$, $\eta=1+\check{\eta}$, $v=0+\check{v}$ and ignoring nonlinear terms, yields (neglecting gravity)
\begin{subequations}\label{LinSS}
	\begin{align}
	&\partial_t\check{\delta}+\frac{1}{r^2}\partial_r(r^2\check{v})=0,\\
	&\partial_t\check{\eta}=-\frac{3}{r}\breve{v}, \\
	&\partial_t \check{v}+c^2_\mathrm{L}(1,1)\partial_r\check{\delta}+s(1,1)\partial_r\check{\eta}=0, 
	\end{align}
\end{subequations}
where 
\begin{equation}
c^2_\mathrm{L}(1,1)=\frac{\lambda+2\mu}{\varrho_0},\qquad 	s(1,1)=0.
\end{equation}
Taking the time derivative of the first equation and the divergence of the second equation leads to
\begin{equation}
\partial^2_{tt}(r\check{\delta})-c^2_{\mathrm{L}}(1,1)\partial^2_{rr}(r \check{\delta})=(\partial_t-c_{\mathrm{L}}\partial_r )(\partial_t+c_{\mathrm{L}}\partial_r )(r\breve{\delta})=0.
\end{equation}
The solutions to the above second order equation contain incoming and outgoing waves:
\begin{subequations}
	\begin{align}
	\check{\delta}(t,r)&=\frac{1}{r}f(r-c_\mathrm{L}t)+\frac{1}{r}g(r+c_\mathrm{L}t), \\
	\check{\eta}(t,r) &= \frac{3}{r^3}\int^{r}_{0}\bigl[f(s-c_\mathrm{L}t)+g(s+c_\mathrm{L}t)\bigr]s ds, \\
	\check{v}(t,r) &= -\frac{c_\mathrm{L}}{r^2}\int^{r}_{0}\bigl[g^{\prime}(s+c_\mathrm{L}t)-f^{\prime}(s-c_\mathrm{L}t)\bigr]sds.
	\end{align}
\end{subequations}

\newpage

\section{Elastic polytropes}\label{Polytropes}
By far, the most popular and extensively studied model of a perfect fluid consists of the \emph{relativistic (or Tooper) polytropic EoS}, introduced in~\cite{Top1965}, and mentioned in the introduction (see Example~\ref{Ex1} and Remark~\ref{Rem1}). Originally, the energy density and the isotropic pressure were written in the form
\begin{equation}\label{TooperPol}
\rho =\mathfrak{\varrho} +\mathrm{n}\mathcal{K}\varrho^{1+\frac{1}{\mathrm{n}}},
\end{equation}
\begin{equation}
p_\mathrm{iso}=\mathcal{K}\varrho^{1+\frac{1}{\mathrm{n}}},
\end{equation}
where $\varrho$ is the \emph{baryonic mass density}, $\mathrm{n}\neq\{0,-1\}$ is the \emph{polytropic index} (which is usually taken to be nonnegative), and $\mathcal{K}$ is a constant which is assumed to be positive. %

In this section\footnote{Some of the results in this section were announced in the companion paper~\cite{Alho:2021sli}.} we introduce material laws that continuously deform the polytropic EoS to the elastic solid setting. Other models are given in Example~\ref{QH} and Remark~\ref{ShearPol} in Appendix~\ref{ApEx}, which, in spherical symmetry, reduce to those in Example~\ref{QHSS}. They belong to the class of \emph{quasi-Hookean materials}. We start by introducing a material law that is the simplest possible generalization to the elastic setting, obtained by adding a quadratic correction~\footnote{Lower powers lead to ill-defined materials, as the transverse wave velocities $\tilde{c}_\mathrm{T}$ blow up in the isotropic state; higher powers lead to materials where these velocities are zero in the isotropic state.} $\sim \mu (\delta-\eta)^2$ to the energy density; this law contains one more parameter, the shear modulus $\mu>0$, that fully characterizes the deviation from the fluid limit $\mu\rightarrow 0$. Besides its simplicity, the quadratic model also serves the purpose of introducing the concept of invariants under renormalization of the reference state (see Remarks~\ref{RefStaGauge} and \ref{InvRenRefSt}).

An important property of polytropic fluids is that they lead to Newtonian equations of motion which are invariant under homologous transformations~\cite{Cha39}. While simple, the quadratic model happens to be scale-invariant in the Newtonian limit only when $\mathrm{n}=1$. To generalize the quadratic model to a new material law that has Newtonian scale-invariant properties, we make use of Proposition~\ref{Homcs}, which gives a characterization of spherically symmetric power-law deformation potential functions for which the Newtonian equations are scale-invariant. A more general procedure to construct spherically symmetric scale-invariant materials beyond the power-law case has been developed by Calogero in~\cite{Cal21} (see also Sideris~\cite{sideris2022expansion} for a more recent account of scale invariance). This leads to a material law which contains an additional parameter, the \emph{shear index} $\mathrm{s}$. Moreover, the scale invariance property implies that the boundary of the ball is characterized by constant values of the \emph{shear variable}, defined by (see e.g.~\cite{Karlovini:2002fc})~\footnote{In~\cite{Karlovini:2002fc} the shear variable is denoted by $z$. Here we use the notation introduced in~\cite{Alho:2018mro,Alho:2019fup} in the Newtonian case, where the variable $y$ is closely related to the homology invariant Milne $u$-variable~\cite{Cha39}.}
\begin{equation}
y=\frac{\delta}{\eta},
\end{equation}
a property first noticed by Karlovini \& Samuelsson in~\cite{Karlovini:2002fc}, and later emphasized and explored in the Newtonian setting by Calogero~\cite{Cal21} (see also the section on exact solutions in~\cite{Alho:2019fup}). 
\subsection{The quadratic model (QP)}\label{sec:QM}
The simplest possible generalization of polytropic fluid models to the elastic setting consists in adding a quadratic correction to the energy density, namely
\begin{equation}\label{QuadraticModel}
\widehat{\rho}(\delta,\eta) = \left(\rho_0-\frac{\mathrm{n}^2K}{1+\mathrm{n}}\right)  \delta + \frac{\mathrm{n}^2K}{1+\mathrm{n}} \delta^{1+\frac{1}{\mathrm{n}}} + \frac{2\mu}{3}\left(\eta-\delta\right)^2,
\end{equation}
which can be seen as belonging to the class of quasi-Hookean materials discussed in Appendix~\ref{ApEx}, see Example~\ref{QH} and Remark~\ref{ShearPol}(and also Example~\ref{QHSS} for their form when restricted to spherical symmetry). Indeed, it corresponds to the choices
\begin{equation}\label{Polcheckrho} 
\check{\rho}(\delta) = \left(\rho_0-\frac{\mathrm{n}^2K}{1+\mathrm{n}}\right)  \delta + \frac{\mathrm{n}^2K}{1+\mathrm{n}} \delta^{1+\frac{1}{\mathrm{n}}}, \quad 
\check{\mu}(\delta)=\frac{2\mu}{3}\delta^2,\quad S\left(\frac{\delta}{\eta}\right) = \left(\frac{\delta}{\eta}\right)^{-2}\left(1-\frac{\delta}{\eta}\right)^2.
\end{equation}
The associated relativistic stored energy function is given by
\begin{equation}\label{SEF_quadratic}
\widehat{\epsilon}(\delta,\eta)=\rho_0-\frac{\mathrm{n}^2 K}{1+\mathrm{n}}+\frac{\mathrm{n}^2 K}{1+\mathrm{n}} \delta^{\frac{1}{\mathrm{n}}}+\frac{2\mu}{3} \eta \left(-2+\left(\frac{\delta}{\eta}\right)^{-1}+\frac{\delta}{\eta}\right),
\end{equation}
which contains three parameters: the bulk modulus $K>0$, the shear modulus $\mu>0$, and the polytropic index $\mathrm{n}\neq0,-1$. The radial and tangential pressures are
\begin{subequations}
\begin{align}
&\widehat{p}_{\text{rad}}(\delta,\eta) = \frac{\mathrm{n} K}{1+\mathrm{n}}\delta^{1+\frac{1}{\mathrm{n}}}-\frac{2\mu}{3}\eta^2\left(1-\left(\frac{\delta}{\eta}\right)^2\right), \\
&\widehat{p}_{\text{tan}}(\delta,\eta) = \frac{\mathrm{n} K}{1+\mathrm{n}} \delta^{1+\frac{1}{\mathrm{n}}} +\frac{2\mu}{3}\eta^2\left(1-\left(\frac{\delta}{\eta}\right)\right)\left(2-\left(\frac{\delta}{\eta}\right)\right).
\end{align}
\end{subequations}
Therefore, in the limit $\mu\rightarrow0$ we recover a perfect fluid with a polytropic EoS~\eqref{TooperPol},
%
%
with the identifications
 \begin{equation}\label{fluidQ}
\varrho(\delta) =\varrho_0 \delta,\qquad \mathcal{K} = \frac{\frac{\mathrm{n} K}{1+\mathrm{n}}}{\varrho_0^{1+\frac{1}{\mathrm{n}}}},
\end{equation}
where the reference state baryonic mass density $\varrho_0$ is given by
\begin{equation}\label{varrho0}
\varrho_0=\rho_0-\mathrm{n}p_0=\rho_0-\frac{\mathrm{n}^2 K}{1+\mathrm{n}}.
\end{equation}
%
The constant $\mathcal{K}$ is well defined for all $\mathrm{n}$ only when the baryon mass density is positive, which in turn implies the upper bound for the bulk modulus:
\begin{equation}
K<\frac{(1+\mathrm{n})\rho_0}{\mathrm{n}^2},\quad \text{for}\quad \mathrm{n}\in(-1,+\infty)\backslash\{0\}.
\end{equation}
Moreover, positivity of the bulk modulus, $K>0$, implies that
\begin{subequations}
\begin{align}
&\mathcal{K}>0\quad \text{for}\quad \mathrm{n}\in(-\infty,-1)\cup(0,+\infty) ,\\
&\mathcal{K}<0\quad \text{for}\quad \mathrm{n}\in(-1,0).
\end{align}	
\end{subequations} 
The cases with $\mathcal{K}<0$ have no physical meaning in the fluid case, since this will lead to a negative radial pressure. However, there is nothing a priori preventing the existence of streched balls in the elastic setting.
%
\subsubsection{Invariants under renormalization of the reference state}
%
If we take a reference state that is compressed by a factor of $f$ (in volume) with respect to the original reference state, then the new variables $(\tilde\delta,\tilde\eta)$ are related to the original variables $(\delta,\eta)$ by
\begin{equation}
\delta = f \tilde\delta, \qquad  \eta = f \tilde\eta.
\end{equation}
Consequently, the energy density is
\begin{equation}
\widehat{\rho}(\tilde\delta,\tilde\eta) = \left(\rho_0- \frac{\mathrm{n}^2 K}{1+\mathrm{n}}\right) f \tilde\delta + \frac{\mathrm{n}^2 K}{1+\mathrm{n}} f^{1+\frac{1}{\mathrm{n}}} {\tilde\delta}^{1+\frac{1}{\mathrm{n}}} + \frac{2\mu}{3} f^2(\tilde\delta-\tilde\eta)^2.
\end{equation}
If, on the other hand, we write it in the same form as Eq.~\eqref{QuadraticModel},
\begin{equation}
\widehat\rho(\tilde\delta,\tilde\eta) = \left(\tilde\rho_0- \frac{\mathrm{n}^2 \tilde{K}}{1+\mathrm{n}}\right) \tilde\delta + \frac{\mathrm{n}^2 \tilde{K}}{1+\mathrm{n}} {\tilde\delta}^{1+\frac{1}{\mathrm{n}}} + \frac{2\tilde\mu}{3} (\tilde\delta-\tilde\eta)^2,
\end{equation}
we obtain
\begin{align}
&\tilde\rho_0 = \left(\rho_0 - \frac{\mathrm{n}^2 K}{1+\mathrm{n}}\left(1-f^{\frac{1}{\mathrm{n}}}\right)\right) f  ,\qquad \tilde{K} = K f^{1+\frac{1}{\mathrm{n}}}, \qquad \tilde\mu = \mu f^2 .
\end{align}
So, by changing the reference state, we obtain an equivalent description of the material, moving from the parameters $(\rho_0,K,\mu)$ to $(\tilde\rho_0,\tilde{K},\tilde\mu)$; the choice of reference state is akin to a choice of gauge. The fluid quantities~\eqref{fluidQ} are in fact invariant under renormalization of the reference state, since by the above transformation, and Eq.~\eqref{varrho0}, it follows that $\tilde{\varrho}_0=\varrho_0 f$.
Other invariant quantities under renormalization are 
\begin{equation}\label{ElQ}
\varsigma(\eta)=\varrho_0 \eta,\qquad \mathcal{E} = \frac{2}{3}\frac{\mu \varrho^{\mathrm{n}-1}_0}{\left(\frac{\mathrm{n}K}{1+\mathrm{n}}\right)^\mathrm{n}}\,,
\end{equation}
with $\mathcal{E}$ dimensionless. Recall that for static solutions with a regular center we must have $\eta(0)=\delta(0)=\delta_\mathrm{c}>0$, and so $\varsigma(\delta_\mathrm{c})=\varrho(\delta_\mathrm{c})=\varrho_\mathrm{c}$, which we assume to be positive. 
So $(\mathcal{K},\mathcal{E},\mathrm{n})$ completely characterize the elastic material; together with $\varrho_\mathrm{c}$, they should give a three-parameter family of elastic stars. In practical terms, once $\mathrm{n}$, $\mathcal{K}$ and $\mathcal{E}$ are fixed, one can pick any $\rho_0$ (since the choice of reference state is arbitrary) and then solve~\eqref{fluidQ} for $K$ and~\eqref{ElQ} for $\mu$, while  $\delta_\mathrm{c}$ is obtained from $\varrho_\mathrm{c}$. When $\mathrm{n}=1$, for instance, we obtain
\begin{equation}
K = \frac{1 + 2\mathcal{K} \rho_0 - \sqrt{1+4\mathcal{K} \rho_0}}{\mathcal{K}},\qquad \mu = \frac{3}{4}K\mathcal{E},\qquad \delta_\mathrm{c}=\frac{\varrho_\mathrm{c}}{\rho_0-\frac{K}{2}}.
\end{equation}
Moreover, changing $\mathcal{K}$ only affects the equilibrium configuration by an overall scale, and thus does not change the value of dimensionless ratios such as the compactness,  while $\mathcal{E}$ is already dimensionless. 
Finally, when written in terms of the invariant quantities $(\varrho,\varsigma,\mathcal{K},\mathcal{E})$, the EoS becomes independent of $\rho_0$:
\begin{equation}
\rho= \varrho + \mathrm{n} {\cal K} \varrho^{1+\frac{1}{\mathrm{n}}}+ {\cal E} {\cal K}^\mathrm{n}(\varsigma-\varrho)^2,
\end{equation}
\begin{equation}
p_{\rm rad}\equiv\varrho\frac{\partial \rho}{\partial \varrho}-\rho = {\cal K} \varrho^{1+\frac{1}{\mathrm{n}}} - {\cal E}{\cal K}^\mathrm{n}(\varsigma^2-\varrho^2),
\end{equation}
\begin{equation}
q\equiv\frac{3}{2}\varsigma\frac{\partial \rho}{\partial \varsigma} = 3{\cal E}{\cal K}^\mathrm{n} \varsigma(\varsigma-\varrho).
\end{equation}
%
%
Note also that the differential equation for $\eta$ in~\eqref{TOV2} can be written in the equivalent form for $\varsigma$,
\begin{equation}
    \partial_r \varsigma = -\frac{3}{r}\left(\varsigma - \frac{\varrho}{\left(1-2m/r\right)^{1/2}}\right)\,,
\end{equation}
and so all field equations can be written only in terms of the invariant quantities. 
\subsubsection{Numerical results}
In this section we will develop our first numerical analysis of elastic self-gravitating configurations. Due to the simplicity of the QP model, the study developed in this section will be the basis for the analysis of more complex models that will be performed in the remainder of this work. 

We follow the procedure outlined in Sec.~\ref{SS}, and construct elastic stars by solving the system of the TOV equations~\eqref{TOV2} with the quadratic EoS~\eqref{QuadraticModel}. Due to the reference state invariance mentioned above, we choose to work with the invariant parameters $({\cal K},{\cal E})$ and the matter ``densities'' $(\varrho, \varsigma)$.  Alternatively, one could choose to work directly with $(K, \mu, \rho_0)$ and with the functions $\delta$ and $\eta$.

For a fixed set of parameters, there exists a one-parameter family of solutions, characterized by $\varrho_\mathrm{c}$ (or $\delta_\mathrm{c}$). In Fig.~\ref{fig:MR} we show the mass-radius and compactness diagrams for the QP model for $\mathrm{n}=1/2$ and $\mathrm{n}=1$, on the left and right panels, respectively, and for some chosen values of the elastic parameter ${\cal E}$ (namely $0$, $10^{-2}$, $10^{-1}$ and $1/2$). Recall that ${\cal K}$ only changes the scale and does not affect dimensionless ratios such as the compactness.
\begin{figure}[ht!]
	\centering
	\includegraphics[width=0.95\textwidth]{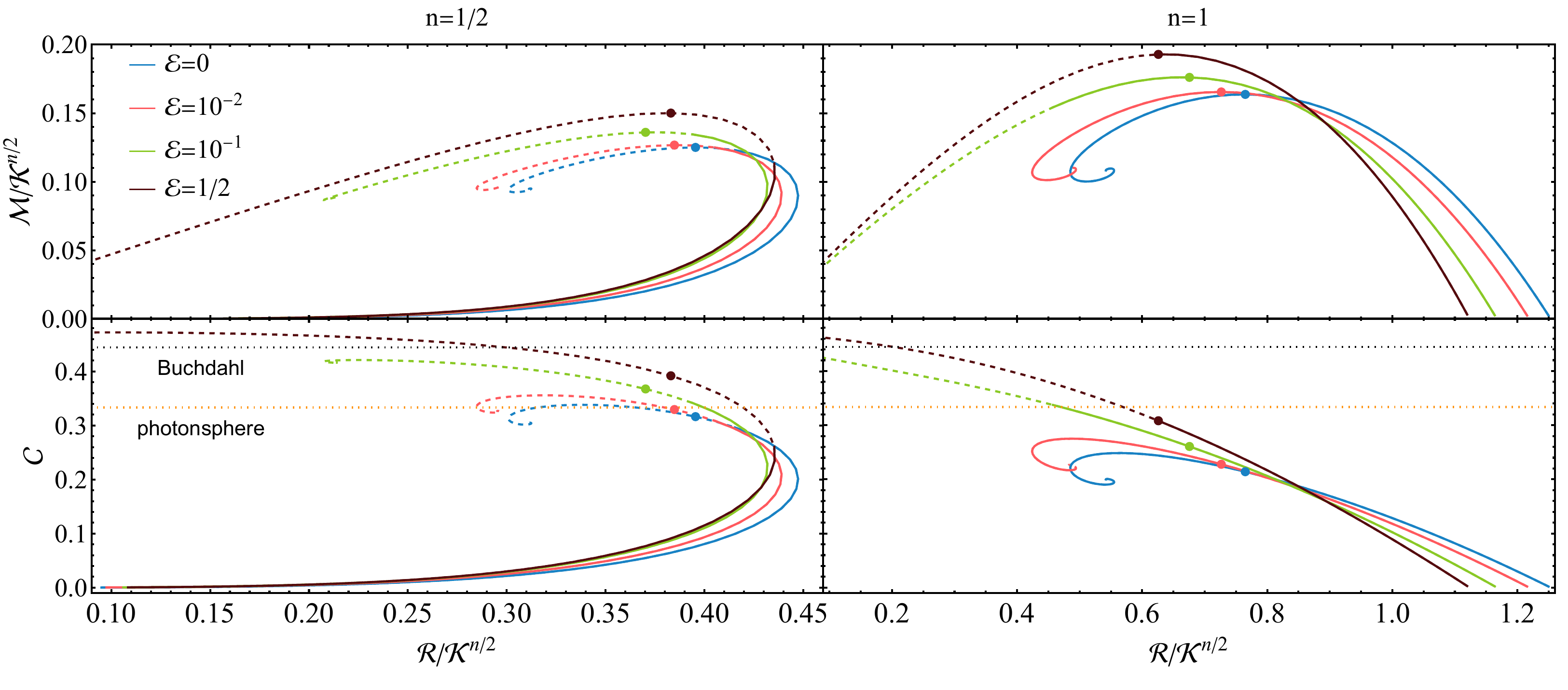} \\
	\caption{Mass-radius (top) and compactness-radius (bottom) diagrams for the quadratic polytropic model (QP) with $\mathrm{n}=1/2$ (left) and $\mathrm{n}=1$ (right). The dashed style denotes superluminal configurations, whereas the circular markers denote the maximum mass configuration, marking the onset of the radial instability.\label{fig:MR}
	}
\end{figure}

Notice that for low values of elasticity the mass-radius profile is qualitatively similar to the fluid case (compare the blue and red curves): for low densities there is an initial branch where the mass increases; once the mass reaches its maximum, there is a second branch where both the mass and the radius decrease for increasing central densities; finally, for sufficiently high densities there is a characteristic spiral-like behaviour. In contrast, for high enough values of elastic parameter, the second branch is much longer, and the spiral-like behaviour becomes smaller and eventually appears to vanish. A general property that can be observed in the diagrams is that elasticity contributes to increase the maximum mass and compactness of the solutions, potentially allowing for ultracompact ($\mathcal{M}/\mathcal{R}>1/3$) and beyond Buchdahl $(\mathcal{M}/\mathcal{R}>4/9)$ solutions (see bottom panels of Fig.~\ref{fig:MR}).
\begin{figure}[ht!]
	\centering
	
	\includegraphics[width=0.95\textwidth]{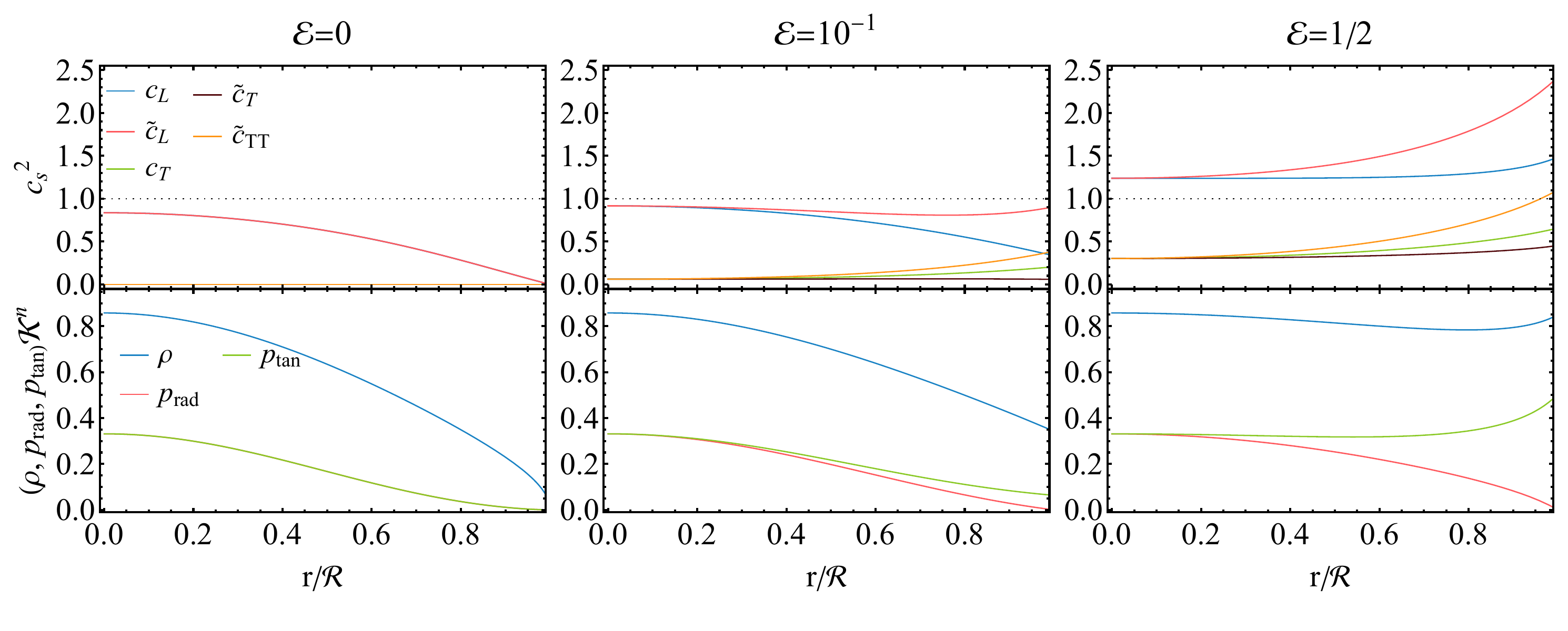} 
	\caption{
		Sound speeds (upper row) and density and pressure profiles (bottom row) for the quadratic polytropic model (QP) with $\mathrm{n}=1/2$. We compare the perfect fluid case (left panel) with two elastic configurations (middle and right panels). The solution represented in the middle panel satisfies all physical viability conditions and features a light ring ($\mathcal{M}/\mathcal{R}\approx0.35$). \label{fig:profiles_quadratic}
	}
\end{figure} 
Another important issue to address is whether these configurations satisfy the energy conditions~\eqref{EC} and have well-defined and subluminal wave propagation speeds~\eqref{realcausal}, that is, whether these ultracompact configurations are physically viable. Since higher values of elasticity contribute to higher sound speeds in the medium, it is crucial to monitor them to understand which regions of the parameter space are allowed. In Fig.~\ref{fig:MR}, we represent configurations for which one or more sound speeds are superluminal by a dashed style (using the natural choice given in Definition~\ref{NaturalChoice} for $\tilde{c}_{\mathrm{L}}$ and $\tilde{c}_{\mathrm{TT}}$). We see that in general the presence of elasticity can produce unphysical configurations. These typically appear in the high-density region of the solutions, but can extend to the lower density part of the curve if the elastic parameter is sufficiently high. While this behaviour is expected, especially in the case ${\rm n}< 1$, which can have superluminal waves in the fluid limit, it is interesting to note that these superluminal configurations can be found in cases where ${\rm n}\geq 1$, which always satisfy physical viability conditions in the fluid limit. In general, one also needs to check if the energy conditions are satisfied; however, in this particular model, we find that the energy conditions are always satisfied for materials with subluminal wave propagation. 
\begin{figure}[ht!]
	\centering
	\includegraphics[width=0.55\textwidth]{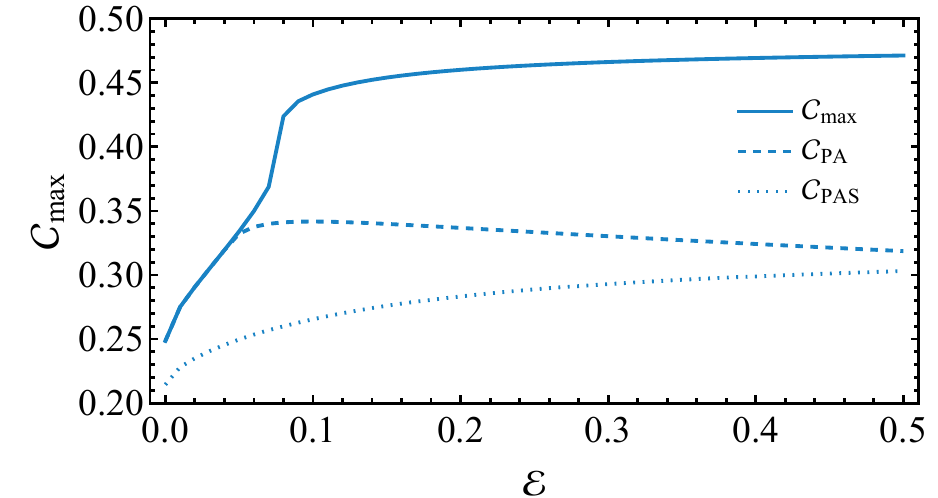} \\
	\caption{ Maximum compactness of elastic stars in the quadratic polytropic model (QP) with $\mathrm{n}=1$. The solid curve represents the absolute maximum compactness, while the dashed and dotted lines represents the maximum compactness of physically admissible configurations and radially stable configurations, respectively.\label{fig:maxcompquad}
	}
\end{figure}
To illustrate this point, we show the five independent sound speeds in Fig.~\ref{fig:profiles_quadratic} (top row), and also the profiles of the density, radial and tangential pressures (bottom row), for three representative configurations of a ${\rm n}=1/2$ polytrope with different values of ${\cal E}$ and the same central baryonic density. Not only do the sound speeds become larger at the center when the elasticity parameter increases, but also their profiles can grow within the star. We find that it is the sound speed of longitudinal waves along the tangential direction (red curve in top panel) that grows faster and breaks the causality limit (horizontal dotted line) first. Figure~\ref{fig:profiles_quadratic} also shows that the energy conditions are still satisfied even in the almost luminal and superluminal cases (middle and right panel, respectively), suggesting that causality is a stronger condition for physically admissible configurations. Interestingly, when we consider the physical viability conditions, all beyond-Buchdahl configurations are excluded, but for some parameters we still find viable configurations that are ultracompact and feature a light ring. In Fig.~\ref{fig:maxcompquad} we show a comparison of the  maximum compactness ${\cal C}_{\rm max}$ as function of ${\cal E}$  with and without assuming the viability conditions (dashed and solid lines, respectively). For simplicity, we chose ${\rm n}=1$, but qualitatively similar results hold also for other values of the index ${\rm n}$. One can qualitatively check that the transition around ${\cal E}\sim 0.08$ between the fast growing behaviour of ${\cal C}_{\rm max}$  and the saturation is related to the transition between a clear spiral-like profile of the unstable branch (e.g. red curve in right panel of Fig.~\ref{fig:MR}) and the long non-spiraling unstable branch (e.g. green curve in the right panel).
\begin{figure}[t]
	\centering
	\includegraphics[width=0.95\textwidth]{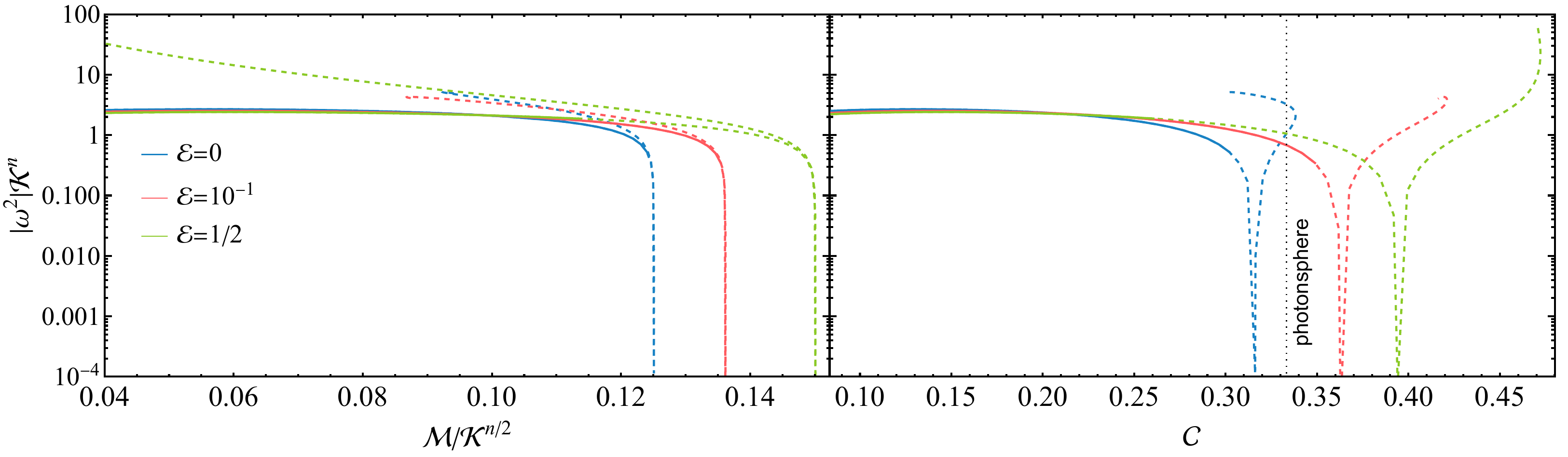} 
	\caption{Absolute value of the fundamental mode's squared frequency of radial perturbations in quadratic polytropic stars with a polytropic index of $\mathrm{n}=1/2$, displayed as a function of the mass (left) and compactness (right). Stable configurations ($\omega^2>0$) are depicted on the left side of the zero crossing point, while unstable configurations ($\omega^2<0$) are shown on the right side. Solid lines represent physically admissible configurations, while dashed lines indicate unphysical solutions. The zero crossing point corresponds to the maximum mass within numerical accuracy.
	\label{fig:stability}}
\end{figure} 

In addition to the physical admissibility, we also studied the radial stability of the self-gravitating configurations, following the procedure detailed in Sec.~\ref{RS}. In Fig.~\ref{fig:stability} we show the absolute value of the fundamental mode's squared frequency for the radial perturbations, as a function of the mass (left) and compactness (right), for the three values of ${\cal E}$ corresponding to those of Fig.~\ref{fig:profiles_quadratic}. We find that, similarly to perfect fluids, the threshold for radial stability is given by the turning point corresponding to the maximum mass, and thus we can classify the parts of the mass-radius diagram to the left and to the right of the maximum mass as the unstable and stable branch, respectively. Interestingly, we find that it is possible to obtain configurations of ultracompact objects with a light ring that are radially stable, as can be seen from the red curve in the right panel of Fig.~\ref{fig:stability}. We note that, however, at least for the QP model, these physically admissible and (radially) stable solutions appear to be restricted to a very narrow region of the parameter space.

\subsection{Model with a scale-invariant Newtonian limit (NSI)}\label{NSI}
%
Using Propositions~~\ref{Homcs}-\ref{scalinginvariance}, it is easy to check that the previous quadratic model leads to equations of motion which are invariant under homologous transformations in the Newtonian limit if and only if $\mathrm{n}=1$, since only in this case does the function $\widehat{\epsilon}(\delta,\eta)$ in~\eqref{SEF_quadratic} (and hence $\widehat{w}(\delta,\eta)$, since they differ by a constant for relativistic materials, see Definition~\ref{defrelstoredenergy}) have the required form
\begin{equation}\label{wPolEl}
 \widehat{\epsilon}(\delta,\eta)=\rho_0-\frac{K}{2}+\eta \left( -\frac{4\mu}{3}+\frac{2\mu}{3}\left(\frac{\delta}{\eta}\right)^{-1}+\left(\frac{2\mu}{3}+\frac{K}{2}\right)\left(\frac{\delta}{\eta}\right) \right).
\end{equation}
%
Given this characterization, it is straightforward to generalize the scale-invariant property for all $\mathrm{n}$: 
\begin{equation}\label{NEwtwn}
\widehat{\epsilon}(\delta,\eta)=\rho_0-w_0+\eta^{\frac{1}{\mathrm{n}}} \left( w_1+w_2\left(\frac{\delta}{\eta}\right)^{-1}+w_3\left(\frac{\delta}{\eta}\right)^{\frac{1}{\mathrm{n}}} \right), 
\end{equation}
which can be written in quasi-Hookean form (see Eq.~\eqref{QHEq} in Appendix~\eqref{ApEx} for details), with $\check{\rho}(\delta)$ as in~\eqref{Polcheckrho}, $w_0=\frac{\mathrm{n}^2K}{1+\mathrm{n}}$ and, for example,
\begin{equation}\label{QHPol}
\check{\mu}(\delta)=\frac{2\mu}{3}\delta^{1+\frac{1}{\mathrm{n}}},\qquad S\left(\frac{\delta}{\eta}\right)= \left(\frac{\delta}{\eta}\right)^{-\frac{1}{\mathrm{n}}} \left( \bar{w}_1+\bar{w}_2\left(\frac{\delta}{\eta}\right)^{-1}+\bar{w}_3\left(\frac{\delta}{\eta}\right)^{\frac{1}{\mathrm{n}}} \right),
\end{equation}
where $\bar{w}_1=\frac{3w_1}{2\mu}$, $\bar{w}_2=\frac{3w_2}{2\mu}$, and $\bar{w}_3=\frac{3}{2\mu}\left(w_3-\frac{\mathrm{n}^2 K}{1+\mathrm{n}}\right)$. Moreover, there is still freedom in changing the power of the shear variable $y=\delta/\eta$ in~\eqref{NEwtwn} from $1/\mathrm{n}$ to, say, $1/\mathrm{s}$. Doing so, using the reference state conditions given in Definition~\ref{RefState} and the isotropic conditions
of Definition~\ref{IsoState}, and solving for the coefficients, leads to the most general simple stored energy function generalizing polytropic fluids to the elastic setting:
\begin{equation}\label{PolyStore2}
\begin{split}
\widehat{\epsilon}(\delta,\eta)=& \rho_0-\frac{\mathrm{n}^2 K}{1+\mathrm{n}}+\eta^{\frac{1}{\mathrm{n}}}\Bigg[-\left((\mathrm{s}-\mathrm{n})K+\mathrm{s}\frac{4\mu}{3}\right) \\
&+\frac{1}{1+\mathrm{s}}\left(\frac{(\mathrm{s}-\mathrm{n})K}{1+\mathrm{n}}+\mathrm{s}\frac{4\mu}{3}\right)\left(\frac{\delta}{\eta}\right)^{-1}
+\frac{\mathrm{s}^2}{1+\mathrm{s}}\left(K+\frac{4\mu}{3}\right)\left(\frac{\delta}{\eta}\right)^{\frac{1}{\mathrm{s}}}\Bigg],
\end{split}
\end{equation}
where $\mathrm{n}$ is the polytropic index and $\mathrm{s}$ can be intrepreted as the shear index. When $\mathrm{s}=\mathrm{n}$, the model has the quasi-Hookean form~\eqref{QHPol} with $\bar{w}_1=-2\mathrm{n}$, $\bar{w}_2=2\mathrm{n}/(1+\mathrm{n})$, and $\bar{w}_3=2\mathrm{n}^2/(1+\mathrm{n})$, and when $\mathrm{s}=\mathrm{n}=1$ it is the same as the simpler quadratic model. We recover the usual relativistic polytropes in the limit $\mu\rightarrow0$ when $\mathrm{s}=\mathrm{n}$. From Eq.~\eqref{RSEF} in Definition~\ref{defrelstoredenergy}, the associated Newtonian stored energy function $\widehat{w}(\delta,\eta)$ has the form~\eqref{NSE}, satisfying $\widehat{w}(1,1)=w_0=\mathrm{n}p_0$.
%
%
Finally, notice that from~\eqref{PolyStore2} it is possible to obtain the materials laws for the limiting cases where $\mathrm{n}\rightarrow+\infty$, and/or $\mathrm{s}\rightarrow+\infty$. This leads to
\begin{equation}\label{eq:NSIstoresnlim}
\widehat{\epsilon}(\delta,\eta)=
    \begin{cases}
        &\rho_0+K \ln{(\eta )}+\left[-\left((\mathrm{s}-1)K+\mathrm{s}\frac{4\mu}{3}\right)+\frac{1}{1+\mathrm{s}}\left(-K+\mathrm{s}\frac{4\mu}{3}\right)\left(\frac{\delta}{\eta}\right)^{-1}\right.\\
        &\qquad\qquad\qquad\qquad\qquad\qquad\qquad\qquad\qquad\left.+\frac{\mathrm{s}^2}{1+\mathrm{s}}\left(K+\frac{4\mu}{3}\right)\left(\frac{\delta}{\eta}\right)^{\frac{1}{\mathrm{s}}}\right]  , \quad {\rm n} \to \infty; \\
        &\rho_0-\frac{\mathrm{n}^2 K}{1+\mathrm{n}}+\eta^{\frac{1}{\mathrm{n}}}\left[-\left((1-\mathrm{n})K+\frac{4\mu}{3}\right)\right.\\
        &\qquad\qquad\qquad\qquad\qquad\left.+\left(\frac{K}{1+\mathrm{n}}+\frac{4\mu}{3}\right)\left(\frac{\delta}{\eta}\right)^{-1} +\left(K+\frac{4\mu}{3}\right)\ln{\left(\frac{\delta}{\eta}\right)}\right], \quad {\rm s} \to \infty; \\
        &\rho_0+K \ln{(\delta)}+\frac{4 \mu }{3}\left[ -1+\left(\frac{\delta}{\eta}\right)^{-1}+\ln{ \left(\frac{\delta }{\eta }\right)}\right], \quad {\rm s, n} \to \infty.
    \end{cases}
\end{equation}

Due to Newtonian scale invariance, both the bulk modulus and the shear modulus transform under renormalization of the reference state as $K=\tilde{K}f^{1+\frac{1}{\mathrm{n}}}$ and $\mu=\tilde{\mu}f^{1+\frac{1}{\mathrm{n}}}$. Hence, in this case, the invariant $\mathcal{E}$ can be written in terms of a single dimensionless elastic parameter, the Poisson ratio $\nu\in(-1,\frac{1}{2}]$, as follows:
\begin{equation}
\mathcal{E}=\frac{\mathrm{n}+1}{\mathrm{n}}\frac{2\mu}{3K}=\frac{\mathrm{n}+1}{\mathrm{n}}\left(\frac{1-2\nu}{1+\nu}\right).
\end{equation}
%
%
In terms of the invariant parameters $(\mathcal{K},\nu)$, the energy density, is given by 
\begin{equation}
\begin{split}
\widehat{\rho}(\varrho,\varsigma) = \varrho +\mathrm{n}\mathcal{K} \varsigma^{1+\frac{1}{\mathrm{n}}}\Bigg[1 
	-\frac{1+\mathrm{n}}{\mathrm{n}}&\left(1-3\frac{\mathrm{s}}{\mathrm{n}}\left(\frac{1-\nu}{1+\nu}\right)\right)\left(1-\left(\frac{\varrho}{\varsigma}\right)\right)\\
 &-3\frac{(1+\mathrm{n})}{(1+\mathrm{s})}\left(\frac{\mathrm{s}}{\mathrm{n}}\right)^2\left(\frac{1-\nu}{1+\nu}\right)\left(1-\left(\frac{\varrho}{\varsigma}\right)^{1+\frac{1}{\mathrm{s}}}\right)\Bigg],
\end{split}
\end{equation}
the radial and tangential pressures are given by
\begin{equation}\label{PradNSI}
\widehat{p}_\mathrm{rad}(\varrho,\varsigma) = \mathcal{K} \varsigma^{1+\frac{1}{\mathrm{n}}}\left[1-3\frac{(1+\mathrm{n})}{(1+\mathrm{s})}\frac{\mathrm{s}}{\mathrm{n}}\left(\frac{1-\nu}{1+\nu}\right)\left(1-\left(\frac{\varrho}{\varsigma}\right)^{1+\frac{1}{\mathrm{s}}}\right)\right],
\end{equation}
\begin{equation}\label{PtanNSI}
\begin{split}
\widehat{p}_\mathrm{tan}(\varrho,\varsigma) = \mathcal{K} \varsigma^{1+\frac{1}{\mathrm{n}}}\Bigg[1&-\frac{3(1+\mathrm{n})}{2\mathrm{n}}\left(1-3\frac{\mathrm{s}}{\mathrm{n}}\left(\frac{1-\nu}{1+\nu}\right)\right)\left(1-\left(\frac{\varrho}{\varsigma}\right)\right)\\
&+\frac{3(1+\mathrm{n})}{2(1+\mathrm{s})}\frac{\mathrm{s}}{\mathrm{n}}\left(1-3\frac{\mathrm{s}}{\mathrm{n}}\right)\left(\frac{1-\nu}{1+\nu}\right)\left(1-\left(\frac{\varrho}{\varsigma}\right)^{1+\frac{1}{\mathrm{s}}}\right)	\Bigg],
 \end{split}
\end{equation}
and the isotropic and anisotropic pressures are
\begin{equation}
\begin{split}
\widehat{p}_\mathrm{iso}(\varrho,\varsigma)= \mathcal{K}\varsigma^{1+\frac{1}{\mathrm{n}}}\Bigg[1&-\frac{1+\mathrm{n}}{\mathrm{n}}\left(1-3\frac{\mathrm{s}}{\mathrm{n}}\left(\frac{1-\nu}{1+\nu}\right)\right)\left(1-\left(\frac{\varrho}{\varsigma}\right)\right) \\
&-3\frac{1+\mathrm{n}}{1+\mathrm{s}}\left(\frac{\mathrm{s}}{\mathrm{n}}\right)^2\left(\frac{1-\nu}{1+\nu}\right)\left(1-\left(\frac{\varrho}{\varsigma}\right)^{1+\frac{1}{\mathrm{s}}}\right)\Bigg],
\end{split}
\end{equation}
\begin{equation}
\begin{split}
\widehat{q}(\varrho,\varsigma)= \frac{3}{2}\mathcal{K} \varsigma^{1+\frac{1}{\mathrm{n}}}\Bigg[&-\frac{1+\mathrm{n}}{\mathrm{n}}\left(1-3\frac{\mathrm{s}}{\mathrm{n}}\left(\frac{1-\nu}{1+\nu}\right)\right)\left(1-\left(\frac{\varrho}{\varsigma}\right)\right)\\
&+3\frac{1+\mathrm{n}}{1+\mathrm{s}}\frac{\mathrm{s}}{\mathrm{n}}\left(1-\frac{\mathrm{s}}{\mathrm{n}}\right)\left(\frac{1-\nu}{1+\nu}\right)\left(1-\left(\frac{\varrho}{\varsigma}\right)^{1+\frac{1}{\mathrm{s}}}\right)\Bigg].
\end{split}
\end{equation}
\begin{remark}\label{rempiso}
Just as for the relativistic polytropes, the energy density can be written in the form
\begin{equation}
    \widehat{\rho}(\varrho,\varsigma)=\varrho +\mathrm{n}\widehat{p}_\mathrm{iso}(\varrho,\varsigma).
\end{equation}
\end{remark}
Given the spherically symmetric energy density function $\widehat{\rho}(\varrho,\varsigma)$, the three independent wave speeds $c_\mathrm{L}$, $c_\mathrm{T}$, and $\tilde{c}_\mathrm{T}$ given in Definition~\ref{Speeds} read
\begin{equation}\label{cLNSI}
c^2_\mathrm{L}(\varrho,\varsigma) = \frac{3\mathcal{K}\frac{1+\mathrm{n}}{\mathrm{n}}\left(\frac{1-\nu}{1+\nu}\right)\left(\frac{\varrho}{\varsigma}\right)^{1+\frac{1}{\mathrm{s}}}\varsigma^{1+\frac{1}{\mathrm{n}}}}{	 \widehat{\rho}(\varrho,\varsigma)+\widehat{p}_\mathrm{rad}(\varrho,\varsigma)},
\end{equation}
	\begin{align}
	c^2_\mathrm{T}(\varrho,\varsigma)=\frac{\left(1-\left(\frac{\varrho}{\varsigma}\right)^2\right)^{-1}\frac{3}{2}\left(\frac{1+\mathrm {n}}{\mathrm{n}}\right)\mathcal{K}\varsigma^{1+\frac{1}{\mathrm{n}}}\Bigg[\left(3\frac{\mathrm{s}}{\mathrm{n}}\left(\frac{1-\nu}{1+\nu}\right)-1\right)\left(1-\left(\frac{\varrho}{\varsigma}\right)\right)+3\frac{\mathrm{s}}{1+\mathrm{s}}\left(1-\frac{\mathrm{s}}{\mathrm{n}}\right)\left(\frac{1-\nu}{1+\nu}\right)\left(1-\left(\frac{\varrho}{\varsigma}\right)^{1+\frac{1}{\mathrm{s}}}\right)\Bigg]}{\widehat{\rho}(\varrho,\varsigma)+\widehat{p}_\mathrm{tan}(\varrho,\varsigma)},
 \end{align}
 \begin{align}
  \tilde{c}^2_\mathrm{T}(\varrho,\varsigma)=\frac{\left(\frac{\varrho}{\varsigma}\right)^2}{\left(1-\left(\frac{\varrho}{\varsigma}\right)^2\right)}\frac{\frac{3}{2}\left(\frac{1+\mathrm {n}}{\mathrm{n}}\right)\mathcal{K} \varsigma^{1+\frac{1}{\mathrm{n}}}\Bigg[-\left(1-3\frac{\mathrm{s}}{\mathrm{n}}\left(\frac{1-\nu}{1+\nu}\right)\right)\left(1-\left(\frac{\varrho}{\varsigma}\right)\right)+3\frac{\mathrm{s}}{1+\mathrm{s}}\left(1-\frac{\mathrm{s}}{\mathrm{n}}\right)\left(\frac{1-\nu}{1+\nu}\right)\left(1-\left(\frac{\varrho}{\varsigma}\right)^{1+\frac{1}{\mathrm{s}}}\right)\Bigg]}{\widehat{\rho}(\varrho,\varsigma)+\widehat{p}_\mathrm{rad}(\varrho,\varsigma)},
	\end{align} 
where 
\begin{equation}
\begin{split}
    \widehat{\rho}(\varrho,\varsigma)+\widehat{p}_\mathrm{rad}(\varrho,\varsigma)=&	\varrho+(1+\mathrm{n})\mathcal{K}\varsigma^{1+\frac{1}{\mathrm{n}}}\Big[1-\left(1-3\frac{\mathrm{s}}{\mathrm{n}}\left(\frac{1-\nu}{1+\nu}\right)\right)\left(1-\left(\frac{\varrho}{\varsigma}\right)\right) \\
    &\qquad\qquad\qquad-3\frac{\mathrm{s}}{\mathrm{n}}\left(\frac{1-\nu}{1+\nu}\right)\left(1-\left(\frac{\varrho}{\varsigma}\right)^{1+\frac{1}{\mathrm{s}}}\right)\Big]\,,
\end{split}
\end{equation}
\begin{equation}
\begin{split}
    \widehat{\rho}(\varrho,\varsigma)+\widehat{p}_\mathrm{tan}(\varrho,\varsigma)=&\varrho+(1+\mathrm{n})\mathcal{K}\varsigma^{1+\frac{1}{\mathrm{n}}}\Bigg[1+\left(1-\frac{3}{2\mathrm{n}}\right)\left(1-3\frac{\mathrm{s}}{\mathrm{n}}\left(\frac{1-\nu}{1+\nu}\right)\right)\left(1-\left(\frac{\varrho}{\varsigma}\right)\right) \\
    &\qquad\qquad-3\left(1-\frac{3}{2}\frac{1-\frac{\mathrm{s}}{\mathrm{n}}}{1+\mathrm{s}}\right)\frac{\mathrm{s}}{\mathrm{n}}\left(\frac{1-\nu}{1+\nu}\right)\left(1-\left(\frac{\varrho}{\varsigma}\right)^{1+\frac{1}{\mathrm{s}}}\right)\Bigg]\,,
\end{split}
\end{equation}
while the remaining wave speeds $\tilde{c}_\mathrm{L}$, and $\tilde{c}_\mathrm{TT}$ satisfy the relation~\eqref{RelWS}. 
%
%
Below, when studying numerically the viability of such material models, we shall make use of the natural choice given in Definition~\ref{NaturalChoice}. It follows that 
	\begin{align}
	&\tilde{c}^2_\mathrm{L}(\varrho,\varsigma)= \nonumber \\
 &\frac{\frac{3}{2}\left(\frac{1+\mathrm {n}}{\mathrm{n}}\right)\mathcal{K}\varsigma^{1+\frac{1}{\mathrm{n}}}\left[2\frac{1-\nu}{1+\nu}-\left(1+\frac{3}{\mathrm{n}}\right)\left(1-3\frac{\mathrm{s}}{\mathrm{n}}\left(\frac{1-\nu}{1+\nu}\right)\right)\left(1-\frac{\varrho}{\varsigma}\right)+\left(1+\frac{3}{\mathrm{n}}\left(\frac{\mathrm{s}-\mathrm{n}}{1+\mathrm{s}}\right)\right)\left(1-3\frac{\mathrm{s}}{\mathrm{n}}\right)\left(\frac{1-\nu}{1+\nu}\right)\left(1-\left(\frac{\varrho}{\varsigma}\right)^{1+\frac{1}{\mathrm{s}}}\right)\right]}{\widehat{\rho}(\varrho,\varsigma)+\widehat{p}_\mathrm{tan}(\varrho,\varsigma)},
 \end{align}
 \begin{align}\label{NCNSI}
	\tilde{c}^2_\mathrm{TT}(\varrho,\varsigma)= \frac{\frac{3}{2}\left(\frac{1+\mathrm {n}}{\mathrm{n}}\right)\mathcal{K}\varsigma^{1+\frac{1}{\mathrm{n}}}\left[\frac{1-2\nu}{1+\nu}-\frac{3}{2\mathrm{n}}\left(1-3\frac{\mathrm{s}}{\mathrm{n}}\left(\frac{1-\nu}{1+\nu}\right)\right)\left(1-\frac{\varrho}{\varsigma}\right)+\frac{3}{2\mathrm{n}}\left(\frac{\mathrm{s}-\mathrm{n}}{1+\mathrm{s}}\right)\left(1-3\frac{\mathrm{s}}{\mathrm{n}}\right)\left(\frac{1-\nu}{1+\nu}\right)\left(1-\left(\frac{\varrho}{\varsigma}\right)^{1+\frac{1}{\mathrm{s}}}\right)\right]}{\widehat{\rho}(\varrho,\varsigma)+\widehat{p}_\mathrm{tan}(\varrho,\varsigma)}.
	\end{align}
%

%
\subsubsection{Center of symmetry}%
For static solutions with a regular center of symmetry, we have $\eta(0)=\delta(0)=\delta_\mathrm{c}>0$. The central density $\rho_\mathrm{c}=\widehat{\rho}(\delta_\mathrm{c},\delta_\mathrm{c})$ and the central pressure $p_\mathrm{c} =\widehat{p}_\mathrm{rad}(\delta_\mathrm{c},\delta_\mathrm{c})=\widehat{p}_\mathrm{tan}(\delta_\mathrm{c},\delta_\mathrm{c})$ are given by
\begin{equation}
\rho_\mathrm{c} = \varrho_\mathrm{c}+\mathrm{n}\mathcal{K}\varrho^{1+\frac{1}{\mathrm{n}}}_\mathrm{c}, \qquad
p_\mathrm{c} = \mathcal{K}\varrho^{1+\frac{1}{\mathrm{n}}}_\mathrm{c},
\end{equation}
while the wave speeds at the center of symmetry are
\begin{subequations}
	\begin{align}
	c^2_\mathrm{L}(\varrho_\mathrm{c},\varrho_\mathrm{c}) = \frac{3\left(\frac{1-\nu}{1+\nu}\right)\left(\frac{1+\mathrm{n}}{\mathrm{n}}\right)\mathcal{K}\varrho^{1+\frac{1}{\mathrm{n}}}_\mathrm{c}}{\varrho_\mathrm{c}+(1+\mathrm{n})\mathcal{K}\varrho^{1+\frac{1}{\mathrm{n}}}_\mathrm{c}} 
	,\qquad c^2_\mathrm{T}(\varrho_\mathrm{c},\varrho_\mathrm{c}) =\frac{\frac{3}{2}\left(\frac{1-2\nu}{1+\nu}\right)\left(\frac{1+\mathrm{n}}{\mathrm{n}}\right)\mathcal{K}\varrho^{1+\frac{1}{\mathrm{n}}}_\mathrm{c}}{\varrho_\mathrm{c}+(1+\mathrm{n})\mathcal{K}\varrho^{1+\frac{1}{\mathrm{n}}}_\mathrm{c}},
	\end{align}
\end{subequations}
so that
\begin{equation}
c^2_\mathrm{L}(\varrho_\mathrm{c},\varrho_\mathrm{c})-\frac{4}{3}c^2_\mathrm{T}(\varrho_\mathrm{c},\varrho_\mathrm{c}) = \frac{\left(\frac{1+\mathrm{n}}{\mathrm{n}}\right)\mathcal{K}\varrho^{1+\frac{1}{\mathrm{n}}}_\mathrm{c}}{\varrho_\mathrm{c}+(1+\mathrm{n})\mathcal{K}\varrho^{1+\frac{1}{\mathrm{n}}}_\mathrm{c}}.
\end{equation}
Given Definition~\ref{IDstatic} of physically admissible initial data for equilibrium configurations, we deduce the following necessary conditions for the existence of physically admissible regular ball solutions:

\begin{proposition}
	Necessary conditions for the existence of physically admissible radially compressed balls solutions in the NSI model are given by
	\begin{equation}
	\mathrm{n}<-1, \quad\mathcal{K}>0,\quad 
	\begin{cases}
	&\varrho_\mathrm{c} \geq \left(-(1+\mathrm{n})\mathcal{K}\left(1-\frac{3}{\mathrm{n}}\left(\frac{1-\nu}{1+\nu}\right)\right)\right)^{-\mathrm{n}},\quad -1< \nu \leq \frac{3+\mathrm{n}}{3+5\mathrm{n}}, \\
	&\varrho_\mathrm{c} \geq \left(-(\mathrm{n}-1)\mathcal{K}\right)^{-\mathrm{n}},\qquad\qquad\qquad\quad \frac{3+\mathrm{n}}{3+5\mathrm{n}}\leq \nu\leq\frac{1}{2},
	\end{cases}
	\end{equation}
 or
	\begin{equation}
	\mathrm{n}>0,  \quad \mathcal{K}>0,\quad
	\begin{cases}
	&0<\varrho_\mathrm{c}\leq \left((1+\mathrm{n})\mathcal{K}\left(\frac{3}{\mathrm{n}}\left(\frac{1-\nu}{1+\nu}\right)-1\right)\right)^{-\mathrm{n}},\quad -1<\nu<\frac{3-\mathrm{n}}{3+\mathrm{n}}, \\
	&0<\varrho_\mathrm{c}<\infty ,\qquad\qquad\qquad\qquad\qquad\quad\quad \frac{3-\mathrm{n}}{3+\mathrm{n}}\leq \nu\leq\frac{1}{2},
	\end{cases} 
	\end{equation}
	while for radially stretched ball solutions the necessary conditions are given by
	\begin{equation}
	-1<\mathrm{n}<0, \quad 
\mathcal{K}<0,\quad  
\begin{cases}
\varrho_\mathrm{c}\geq (-(3+\mathrm{n})\mathcal{K})^{-\mathrm{n}},\qquad\qquad\qquad\quad -1<\nu\leq\frac{3+7\mathrm{n}}{2(3+\mathrm{n})}, \\
\varrho_\mathrm{c}\geq (-(1+\mathrm{n})\left(1-\frac{3}{2\mathrm{n}}\left(\frac{1-2\nu}{1+\nu}\right)\right)\mathcal{K})^{-\mathrm{n}},\quad \frac{3+7\mathrm{n}}{2(3+\mathrm{n})}\leq \nu< \frac{1}{2}.
\end{cases}  
   \end{equation}
\end{proposition}

\begin{remark}
	The above physically invariant characterization of initial data for steady state solutions can be easily translated into conditions on $\delta_\mathrm{c}$ in terms of $K$, $\nu$ and $\mathrm{n}$, with $0<K<\frac{(1+\mathrm{n})\rho_0}{\mathrm{n}^2}$ for $\mathrm{n} \in (-1,+\infty)\setminus \{0\}$ .
\end{remark}

\subsubsection{Boundary of the ball}\label{SBNSI}%
The boundary of the ball is defined by the condition $p_\mathrm{rad}(\mathcal{R})=\widehat{p}_\mathrm{rad}(\varrho(\mathcal{R}),\varsigma(\mathcal{R}))=0$. For the present model this can be easily solved to yield
\begin{equation}\label{yboundaryeta}
\varsigma(\mathcal{R})=0\qquad\text{or}\qquad  \varrho(\mathcal{R})=y_{\mathrm{b}}\varsigma(\mathcal{R}),
\end{equation}
where $y_\mathrm{b}\geq0$ is a constant given by~\footnote{In terms of the invariant $\mathcal{E}$, this condition reads $y_\mathrm{b}\equiv \left(\frac{{\rm s} (2 {\rm n} \mathcal{E}+{\rm n}+1)}{2 {\rm n} {\rm s} \mathcal{E}-{\rm n}+{\rm s}}\right)^{\frac{\mathrm{s}}{1+\mathrm{s}}}$.}
\begin{equation}\label{ybDef}
y_\mathrm{b}\equiv \left(1-\frac{(1+\mathrm{s})\mathrm{n}}{3(1+\mathrm{n})\mathrm{s}}\left(\frac{1+\nu}{1-\nu}\right)\right)^{\frac{\mathrm{s}}{1+\mathrm{s}}}<1.
\end{equation} 
This means that the boundary is characterized by constant values of the shear scalar, a condition called \emph{constant shear boundary condition} in~\cite{Cal21}. The condition $y_\mathrm{b}\geq0$ restricts the allowed values of the shear index:\footnote{In terms of the invariant $\mathcal{E}$, this condition reads ${\rm s}\geq\frac{\rm n}{1 + 2 {\rm n} \mathcal{E}}$.} 
\begin{equation}\label{ShearBound}
\mathrm{s}\geq\left(\frac{{\rm n}}{3\frac{1-\nu}{1+\nu}+2\mathrm{n}\frac{1-2\nu}{1+\nu}}\right).
\end{equation}
At the boundary of the ball, it follows that tangential pressure is given by
\begin{equation}
p_\mathrm{tan}(\mathcal{R})=\widehat{p}_\mathrm{tan}(\varrho(\mathcal{R}),\varsigma(\mathcal{R}))= \frac{3}{2}\mathcal{K}\varsigma(\mathcal{R})^{1+\frac{1}{\mathrm{n}}}\left[1-\frac{\mathrm{s}}{\mathrm{n}}-\frac{1+\mathrm{n}}{\mathrm{n}}\left(1-\frac{3\mathrm{s}}{\mathrm{n}}\left(\frac{1-\nu}{1+\nu}\right)\right)\left(1-y_\mathrm{b}\right) \right],
\end{equation}
and the energy density is 
\begin{equation}
	\rho(\mathcal{R})=\widehat{\rho}(\varrho(\mathcal{R}),\varsigma(\mathcal{R}))=\varsigma(\mathcal{R})\left\{y_\mathrm{b}+\mathcal{K} \varsigma(\mathcal{R})^{\frac{1}{\mathrm{n}}}\left[(\mathrm{n}-\mathrm{s})
	-\frac{\mathrm{n}+1}{\mathrm{n}}\left(\mathrm{n}-3\mathrm{s}\left(\frac{1-\nu}{1+\nu}\right)\right)\left(1-y_\mathrm{b}\right)\right]\right\}.
\end{equation}
\begin{remark}\label{SBoundNSI}
At the boundary of the ball we have $p_\mathrm{rad}(\mathcal{R})=0$, and so
$$
p_\mathrm{iso}\left(\mathcal{R}\right)=\frac{2}{3}p_\mathrm{tan}\left(\mathcal{R}\right)=\frac{2}{3}q\left(\mathcal{R}\right).
$$
Therefore, it follows from Remark~\ref{rempiso} that, at the boundary of the ball, the energy density can be written as
\begin{equation}
\rho(\mathcal{R})=\varrho(\mathcal{R})+\frac{2\mathrm{n}}{3}p_\mathrm{tan}(\mathcal{R}).
\end{equation}
\end{remark}
Besides condition~\eqref{ShearBound}, other necessary conditions for the existence of physically admissible ball solutions in the NSI model can be obtained from the strict hyperbolicity and causality conditions of the wave speeds at the boundary. In particular, notice that if $y_\mathrm{b}=0$ then $c_{\mathrm L}^2(\mathcal{R})=0$, and the strict hyperbolicity condition is violated at the boundary.

\subsubsection{Numerical results}
%
	%
%
%
%
%
%

%
\begin{figure}[ht!]
	\centering
 \includegraphics[width=0.95\textwidth]{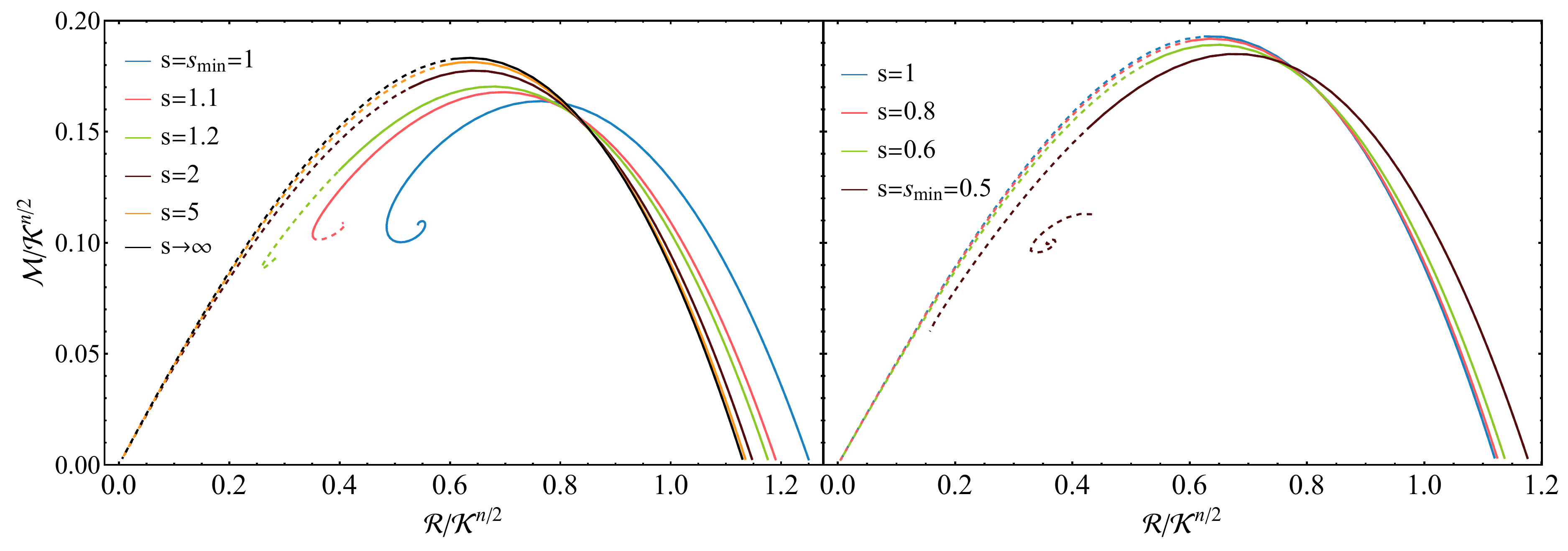}
	\caption{Mass-radius and compactness-radius diagrams for an Newtonian scale-invariant (NSI) model polytrope model with ${\rm n}=1$. Left: we set $\nu=\frac{1}{2}$ such that elastic effects enter in the model only via the shear index $s\geq 1$ (with ${\rm s=n}=1$ being the fluid case). Right: we set ${\cal \nu}=\frac13$ to explore shear index in the range ${\rm s}_{\rm min}<{\rm s}<{\rm n}$. 
 \label{fig:MRNSI_s_study}
	}
\end{figure}

In contrast with the quadratic model~\eqref{QuadraticModel} studied in Section~\ref{sec:QM}, the NSI model is characterized by two elastic parameters: the Poisson ratio $\nu$ (or equivalently, the dimensionless shear parameter ${\cal E}$), and a shear index ${\rm s}$. The effects of the Poisson ratio on the properties of the body are qualitatively similar to those studied in Section~\ref{sec:QM}, and thus we will focus our analysis on the effects of the shear index ${\rm s}$.
To isolate its effect, we will consider two cases with fixed Poisson ratio: $\nu=1/2$ and $\nu=1/3$. For simplicity of the analysis, we will consider only ${\rm n}=1$ polytropes. The two cases are represented in the left and right panels of Fig.~\ref{fig:MRNSI_s_study}, respectively.  In the former case, when ${\rm s}={\rm n}=1$, there are no elastic effects (blue line), and ${\rm s_{min}=n}=1$. In the latter case, although elasticity is already present through the Poisson ratio, we can analyze the effect of ${\rm s<n}$. From both panels of Fig.~\ref{fig:MRNSI_s_study} we are able to infer that the shear index ${\rm s}$ acts qualitatively in a similar fashion to the Poisson ratio, i.e., increasing the shear index leads to a higher maximum mass and a decrease in the radius of the star. For configurations with $\nu<1/2$ (see right panel), we also confirm these results by noting that decreasing the shear index decreases the effect of elasticity. 

Unlike the case of fluid stars with ${\rm n}=1$, which always satisfy the physical admissibility conditions for any central density, here the configurations are always superluminal for sufficiently high densities when elasticity is turned on. This effect can be seen by expanding the expressions of the longitudinal wave velocities close to the center of the star and for high densities\footnote{For simplicity, we show the expression for $c_{\rm L}$, but an analogous expression can be computed for $\tilde{c}_{\rm L}$, leading to the same conclusions.}:
\begin{equation}\label{eq:cLNSIcenter}
    {c}^2_\mathrm{L} = \left(\frac{1}{\rm n}+\frac{2 (1-2 \nu )}{{\rm  n}(1+\nu)}+{\cal O}(\varrho_c^{-\frac{1}{\rm n}})\right) + \frac{4 \pi  \mathcal{K} ({\rm n}+6)}{15 {\rm n}} {\varrho_c}^{1+\frac{1}{\rm n}} \left(\frac{2 (1-2 \nu )}{\nu +1}-\frac{\rm n-s}{\rm s}+{\cal O}(\varrho_c^{-\frac{1}{\rm n}})\right)r^2 +{\cal O}(r^4)\,.
\end{equation}
This equation allows us to infer that for ${\rm n}=1$ and very high central densities the material is superluminal at the center if the Poisson ratio is different from its corresponding fluid value (${\nu \neq 1/2}$). However, even when the ${\nu=1/2}$ and ${\rm n=1}$, the material can still be superluminal at high densities if ${\rm s \neq n}$. This can be seen by looking at the subleading coefficient in expression~\eqref{eq:cLNSIcenter}, which is positive when ${\rm s} >{\rm n}$. In the following section, when we discuss the LIS model, we will see that the longitudinal sound speed at the center, given by Eqs.~\eqref{eq:Linearspeedscenter}, corresponds exactly to the first term in of expansion~\eqref{eq:cLNSIcenter}. This is an example of the connection between the NSI model in the large density limit and the LIS model of the next section.

As elasticity increases, the configurations can become unphysical for lower densities, and superluminality can occur closer to the radius of the star. As previously done for the QP model, configurations that are unphysical are represented by dashed lines in Fig.~\ref{fig:MRNSI_s_study}. This unphysical branch becomes larger for increasing shear index, indicating that the maximum allowed central density for physically admissible configurations decreases with ${\rm s}$. Comparing the elastic curves in the left panel of Fig.~\ref{fig:MRNSI_s_study}, we note that the solutions saturate for large $s$, approaching the limit defined in Eq.~\eqref{eq:NSIstoresnlim} when ${\rm s}\to \infty$. The solutions corresponding to this limit are represented as a black solid line in Fig.~\ref{fig:MRNSI_s_study}.
\begin{figure}[ht!]
	\centering
 \includegraphics[width=0.49\textwidth]{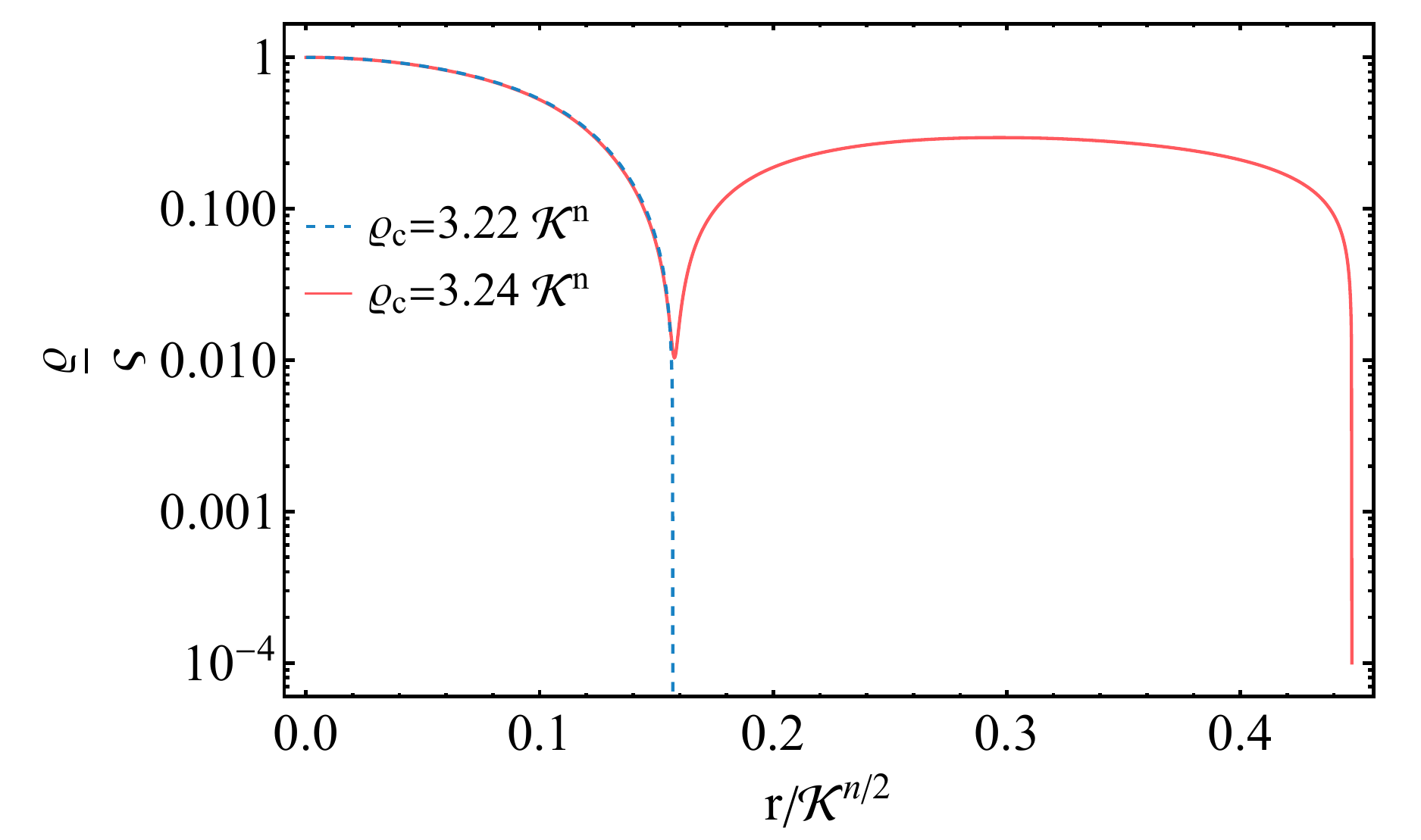}
 \includegraphics[width=0.49\textwidth]{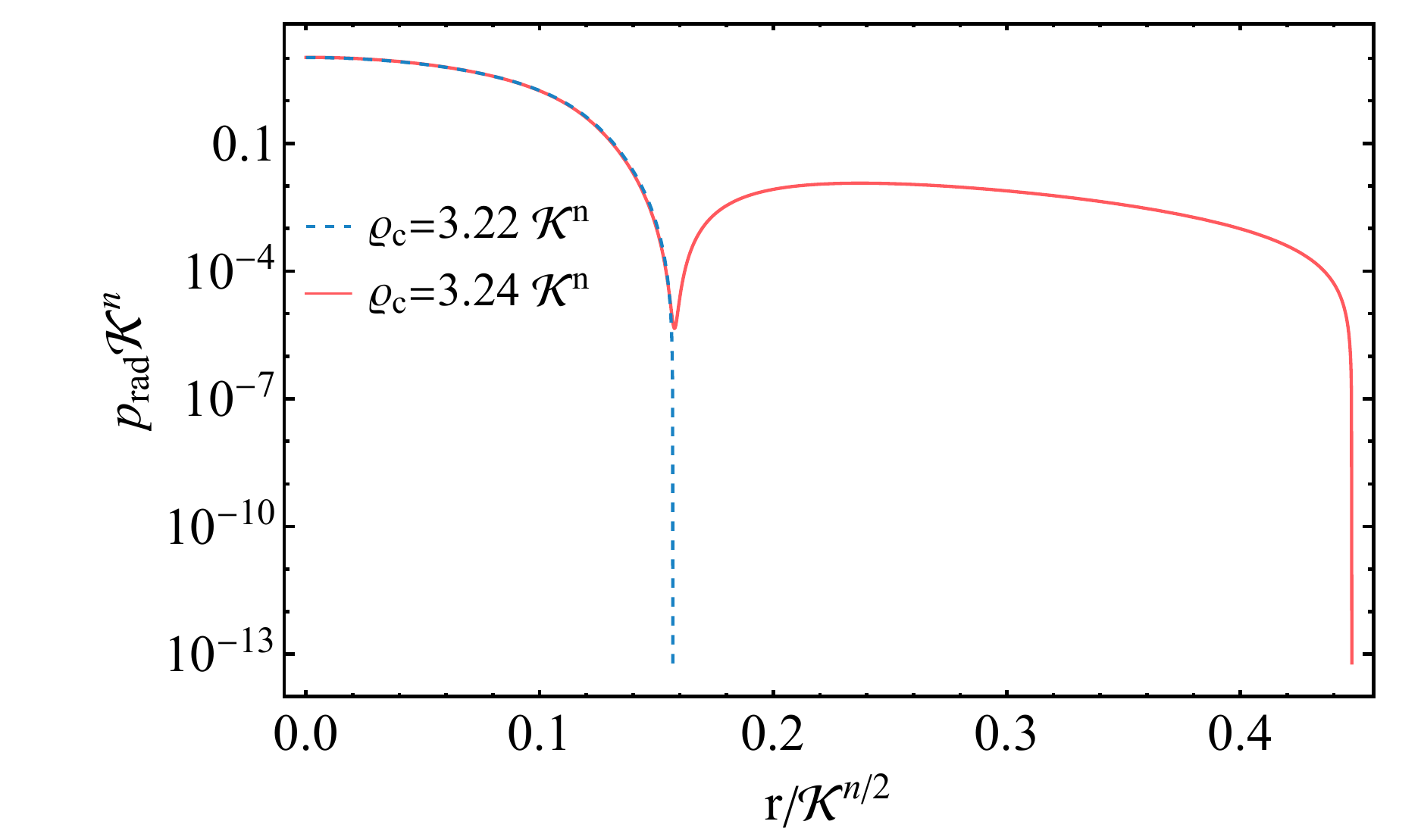}
	\caption{Shear ratio (left) and radial pressure (right) for a NSI polytrope with $\nu=1/3$, ${\rm n=1}$ and ${\rm s}={\rm s}_{\rm min}=1/2$. The blue curve corresponds to a configuration in the final part of the main branch of Fig.~\ref{fig:MRNSI_s_study}, whereas the red line corresponds to a configuration in the beginning of the secondary branch. \label{fig:shear_ratio}
	}
\end{figure}
Let us now turn our attention to the brown curve in the right panel of Fig.~\ref{fig:MRNSI_s_study}. We note that the mass-radius diagram for these parameters has two  disconnected branches, a main one, similar to the curves for the other choices of the parameters, and a smaller, spiral branch that appears in the high density regime. To better understand this behaviour, we look at the matter profiles in these two branches. Since the appearance of two disconnected curves is associated with a discontinuity of the radius of the star, it is particularly relevant to look at the conditions for which the radial pressure vanishes. As discussed in Eq.~\ref{ybDef} and above, the condition for the radial pressure to vanish is determined by the shear ratio $\varrho/\varsigma=y_\mathrm{b}$. In Fig.~\ref{fig:shear_ratio} we plot this quantity for two configurations, one in the end section of the main branch (blue curve), and a second one in the first section of the secondary branch (red curve). For the model parameters considered here we have $y_{\rm b}=0$. We see that, in the first and more standard case, the shear ratio decreases monotonically until it reaches $\varrho/\varsigma=y_\mathrm{b}$. The solution in the secondary branch initially overlaps with that of the first branch, but instead of decreasing monotonically, the shear index has a minimum near the point where the first solution reaches the boundary, departing from it. For larger radial coordinate, the shear index reaches a maximum which is followed by a decreasing behaviour until it reaches $y_{\rm b}$. This distinct behaviour profile is translated directly in the two different branches that we observe in the mass-radius diagram, and most likely signals a bifurcation where the larger ball splits into a smaller ball and an enveloping elastic shell (not considered in this work).

\begin{figure}[ht!]
	\centering
 \includegraphics[width=0.95\textwidth]{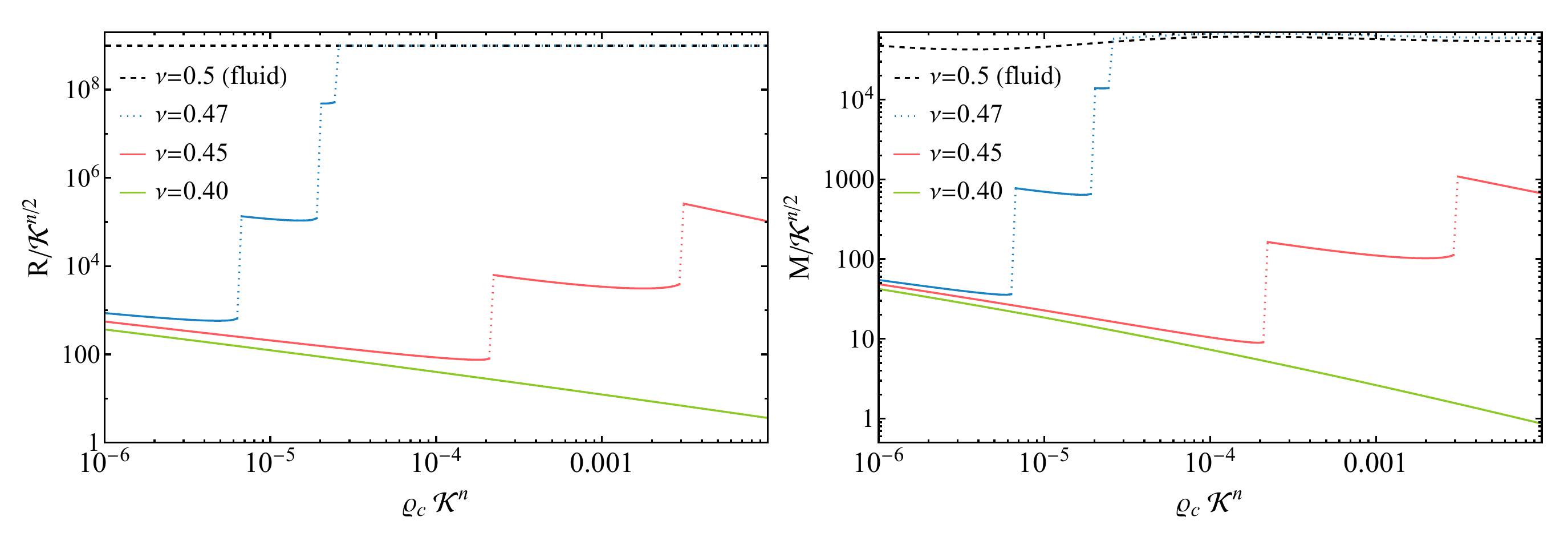} 
	\caption{Radius (left) and mass (right) diagrams of the NSI polytropic stars with ${\rm n=s=5}$, for different values of Poisson ratio, as a function of the corresponding values of the central baryonic density. Solid lines correspond to configurations which have bounded solutions, while the vertical (dotted) lines represent the sharp transition between two branches. The dashed black line corresponds to the numerical data for a fluid EoS ($\nu =1/2$), which coincides with the upper limits of the numerical integration domain (see text for details). \label{fig:MRNSI_s_studyRADIUS}
	}
\end{figure}
In addition to the above results for the polytropic index ${\rm n}=1$, our study also reveals that elasticity allows for the existence of stars with a polytropic index of ${\rm n}=5$, which are impossible for fluids. This remarkable result highlights the unique properties of elastic stars, which are not observable in fluid objects.
To illustrate this point, we present in Fig.~\ref{fig:MRNSI_s_studyRADIUS} the radius-central density and mass-central density diagrams for stars with a polytropic index ${\rm n}=5$ and different values of the Poisson ratio. The solutions that correspond to stars are represented by the colored solid lines. Our analysis shows that these solutions can have multiple branches (see the blue and red curves), with a sharp transitions between them (represented by dotted lines). We find that, for the same value of central density, the radius and mass of the solution decrease  when the elasticity increases (decreasing Poisson ratio). We also observe that, as  elasticity increases, the transition between different branches occurs for successively higher values of central density. However, for sufficiently high elasticity  (see green curve), the secondary branches vanish, and the solution is described by a single connected branch\footnote{We have checked that the solutions always lie in the primary branch when we consistently increase the central density up until a sufficiently large value ($\varrho\sim 10^6 {\cal K}^{- \rm n}$).}. The dashed black line corresponds to the fluid case and the values for the radius and mass correspond to the upper limit of the numerical domain, $r_{\rm max}=10^9 {\cal K}^{\rm n/2}$ (and the value mass function at this point). Since this behaviour is constant regardless of the value of $r_{\rm max}$, our analysis confirms that fluid solutions do not have a well-defined radius, and consequently no bounded solutions, for ${\rm n>5}$.

\begin{figure}[ht!]
	\centering
 \includegraphics[width=0.55\textwidth]{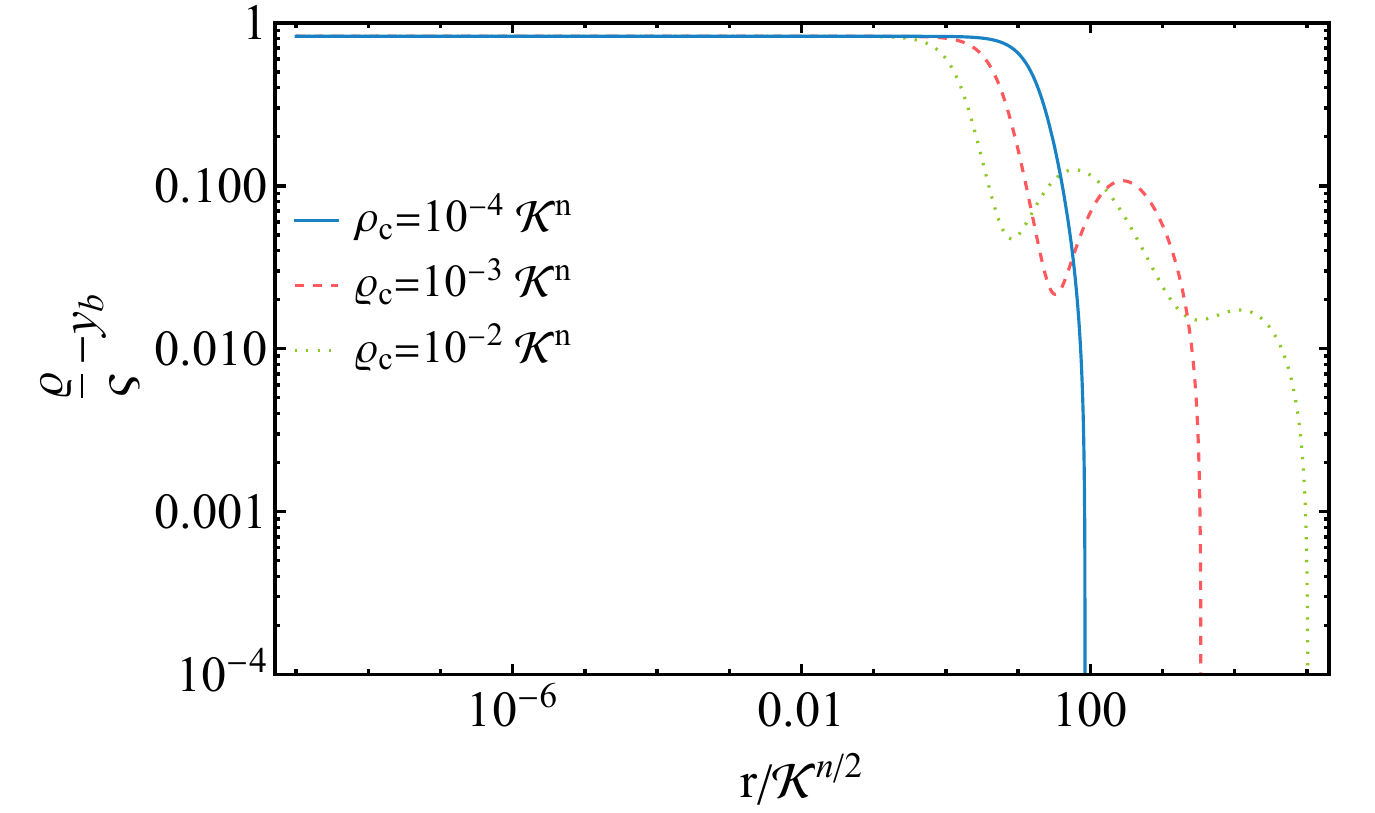} 
	\caption{Shear ratio profile for three configurations with ${\rm n=5}$ and $\nu=0.45$. The value of shear ratio at the boundary, given by Eq.~\eqref{ybDef}, is $y_{\rm b}\approx 0.172$. The values of central density are chosen such that each configuration is in an independent branch of the solution (see Fig..~\ref{fig:MRNSI_s_studyRADIUS}). The inflection points of the profile are directly related with the transition between branches of the solution.  \label{fig:shear_ratio_n=5}
	}
\end{figure}

We now turn our attention again to the different branches of the solutions. Similarly to the case of the brown curve in the right panel of Fig.~\ref{fig:MRNSI_s_study}, the appearance of new branches is directly related to the number of inflection points of the shear ratio profile, as can be seen in Fig.~\ref{fig:shear_ratio_n=5}. Solutions in the primary (lower density) branch are characterized by a shear ratio that monotonically decreases until it reaches $y_{\rm b}$, whereas solutions in the successive branches have a sharp inflection point in the shear radius profile where the ratio increases, followed by a subsequent decreasing behaviour.  As it can be seen from Fig.~\ref{fig:shear_ratio_n=5}, these inflection points can occur multiple times in the solutions, where each pair of inflections will correspond to a new transition between branches, probably associated to a ball-shell bifurcation. A more detailed investigation on the properties of these solutions is beyond the scope of this paper, and is left for future work.

Our main point here is to highlight the qualitative differences that occur due to elasticity, and to show that polytropes with ${\rm n}\geq5$ are possible only when elasticity is taken into account. The existence of such stars represents a departure from the standard fluid models, and highlights the unique properties of elastic stars. We remark that in this part of the analysis our main goal was to study the existence of stars under these conditions, and thus we did not consider the physical viability and the radial stability of such objects. Overall, our results provide valuable insights into the impact of elasticity on the properties of stars, with important implications for the study of astrophysical objects. 

\newpage

\section{Elastic materials with constant wave speeds in the isotropic state }\label{LinEOS}
In this section\footnote{Some of the results in this section were announced in the companion paper~\cite{Alho:2023mfc}.} we introduce a class of elastic materials with a pre-stressed reference state, $p_0\neq0$ (thus with no preferred reference state), and also a stress-free reference state material, $p_0=0$ (possessing a well defined reference state), which are elastic generalizations of perfect fluids with a \emph{linear} and an \emph{affine EoS}, respectively. In the fluid limit, these read:
\begin{subequations}\label{LinAfFluid}
	\begin{align}
    &\text{pre-stressed:}\quad p_\mathrm{iso}=c^2_\mathrm{s}\rho, \\
	&\text{stress-free:} \quad \quad p_\mathrm{iso}=c^2_\mathrm{s}(\rho-\rho_0).
	\end{align}
\end{subequations}	
Here the constant $c^2_\mathrm{s}$ is the fluid's sound speed,
\begin{equation}\label{Constsoundspeed}
c^2_\mathrm{s} =\frac{dp_\mathrm{iso}}{d\rho}=\gamma-1,
\end{equation}
where $\gamma$ is the \emph{adiabatic index}; we choose $\gamma\in(1,2]$ in order to satisfy the strict hyperbolicity and the causality conditions. These new elastic materials are characterized by having constant wave speeds in the isotropic state, so that
\begin{equation}
\frac{dp_\mathrm{iso}}{d\rho}=c^2_\mathrm{L}(\delta,\delta)-\frac{4}{3}c^2_\mathrm{T}(\delta,\delta)=\gamma-1.    
\end{equation}
Using the reference state identity
\begin{equation}\label{refStaId}
    \rho_0+p_0 = \frac{K}{\gamma-1},
\end{equation}
the longitudinal and transverse wave speeds in the isotropic state are constant and satisfy
\begin{equation}
    c^2_\mathrm{L}(\delta,\delta)=\frac{L}{K}(\gamma-1),\qquad c^2_\mathrm{T}(\delta,\delta)=\frac{\mu}{K}(\gamma-1).
\end{equation}

We discuss the pre-stressed and stress-free materials in the next two subsections.
\subsection{Linear isotropic state materials (LIS) and relativistic scale invariance} \label{LIS}
We start by considering the pre-stressed reference state materials, which generalise fluids with linear EoS. A remarkable property of these fluid models is that the resulting TOV equations are invariant under scaling transformations which follow from the positive homogeneity of the energy-density functional $\widehat{\rho}(\delta)$. For spherically symmetric power-law elastic materials we have the following result:
%
%
%
%
%
%
%
\begin{proposition}\label{scaling invariance}
	Let $\widehat{\rho}(\delta,\eta)$ be a positively homogeneous function of degree $\kappa$. If $\delta(r)$ is a solution to the integro-differential equation~\eqref{TOV2}, then so is $A^{\frac{2}{\kappa}}\delta(Ar)$, with $A$ an arbitrary positive real number.
\end{proposition}
\begin{proof}
	Let $\tilde{r}=Ar$ and $\delta(r)=A^{\frac{2}{\kappa}}\tilde{\delta}(\tilde{r})$. By the definition of $\eta$ in~\eqref{DefDeltaEta}, it follows that $\eta(r)=A^{\frac{2}{\kappa}}\tilde{\eta}(\tilde{r})$. If $\widehat{\rho}(\delta,\eta)$ is homogeneous of degree $\kappa$, then 
 \begin{equation}\label{rhorescaleA}
   \widehat{\rho} (\delta,\eta)= A^2 \widehat{\tilde{\rho}}(\tilde{\delta},\tilde{\eta}),
 \end{equation}
and from the definition of mass in~\eqref{TOVmassdefdelta} it follows that 
 \begin{equation}
     m(r)=A^{-1} \tilde{m}(\tilde{r}),
 \end{equation}
so that the ratio $m(r)/r$ is invariant under the scaling transformation. Moreover, if $\widehat{\rho}$ is homogeneous function of degree $\kappa$, then $c^2_\mathrm{L}(\delta,\eta)$ and $s(\delta,\eta,\frac{m}{r})$ are invariant functions, i.e.,
	\begin{equation}
	c^2_\mathrm{L}(\delta,\eta)=\tilde{c}^2_\mathrm{L}(\tilde{\delta},\tilde{\eta})\,,\quad s(\delta,\eta,\frac{m}{r})=\tilde{s}(\tilde{\delta},\tilde{\eta},\frac{\tilde{m}}{\tilde{r}}).
	\end{equation}
	Therefore, Eq.~\eqref{TOV2a} reads
	\begin{equation}
	\tilde{c}^2_{\mathrm{L}}(\tilde{\delta},\tilde{\eta}) \frac{\tilde{r}}{\tilde{\delta}}\frac{d\tilde{\delta}}{d\tilde{r}} =-\tilde{s}(\tilde{\delta},\tilde{\eta},\frac{\tilde{m}}{\tilde{r}})\frac{\tilde{r}}{\tilde{\eta}}\frac{d\tilde{\eta}}{d\tilde{r}}-\frac{\left(\frac{\tilde{m}}{\tilde{r}} +4\pi \tilde{r}^2 \widehat{\tilde{p}}_{\mathrm{rad}}(\tilde{\delta},\tilde{\eta})\right)}{\left(1-\frac{2\tilde{m}}{\tilde{r}}\right)} ,
	\end{equation}
	which shows its invariance under the scaling transformation.
\end{proof}
In analogy with the linear EoS with constant adiabatic index of Example~\ref{Ex2}, we consider the elastic pre-stressed polytropes given by the stored energy function~\eqref{PolyStore2}, and assume that the whole contribution to the energy density is due to the deformation potential, by removing the contribution of the baryonic rest mass density. Hence, these models belong to the class of ultra-relativistic materials, with a reference state characterized by
\begin{equation}\label{p0Lin}
    p_0=\frac{K}{\gamma}, \qquad \rho_0 = \frac{K}{\gamma(\gamma-1)},
\end{equation}
and
\begin{equation}\label{w0Lin}
    \rho_0=w_0=\frac{p_0}{\gamma-1}.
\end{equation}
The resulting stored energy function is written in terms of the polytropic index $1/\mathrm{n}=\gamma-1$ as
\begin{equation}\label{PolyStore3}
\widehat{\epsilon}(\delta,\eta)=\eta^{\frac{1}{\mathrm{n}}}\left[-\left((\mathrm{s}-\mathrm{n})K+\mathrm{s}\frac{4\mu}{3}\right)+\frac{1}{1+\mathrm{s}}\left(\frac{(\mathrm{s}-\mathrm{n})K}{1+\mathrm{n}}+\mathrm{s}\frac{4\mu}{3}\right)\left(\frac{\delta}{\eta}\right)^{-1}+\frac{\mathrm{s}^2}{1+\mathrm{s}}\left(K+\frac{4\mu}{3}\right)\left(\frac{\delta}{\eta}\right)^{\frac{1}{\mathrm{s}}}\right],
\end{equation}
and the energy density is
\begin{equation}\label{eq:prestreconstiso}
\widehat{\rho}(\delta,\eta) =  K\eta^{1+\frac{1}{\mathrm{n}}}\left[\frac{\mathrm{n}^2 }{1+\mathrm{n}} 
	-\left(\mathrm{n}-3\mathrm{s}\left(\frac{1-\nu}{1+\nu}\right)\right)\left(1-\left(\frac{\delta}{\eta}\right)\right)-\frac{3\mathrm{s}^2}{(1+\mathrm{s})}\left(\frac{1-\nu}{1+\nu}\right)\left(1-\left(\frac{\delta}{\eta}\right)^{1+\frac{1}{\mathrm{s}}}\right)\right],
\end{equation}
which under the scaling transformation of $\delta$ and $\eta$ transforms as in~\eqref{rhorescaleA} with $\kappa=1+\frac{1}{\mathrm{n}}$. 

%
\begin{remark}\label{LinRhoPiso}
From Remark~\ref{rempiso}, by taking $\varrho_0 \to 0$, it follows that
\begin{equation}
\rho= {\rm n}p_{\rm{iso}}=\frac{{\rm n}}3(p_\mathrm{rad} + 2 p_\mathrm{tan}),
\end{equation}
with $p_\mathrm{rad}$ and $p_\mathrm{tan}$ as given in~\eqref{PradNSI}, and~\eqref{PtanNSI} respectively. The resulting wave speeds are also easily obtained from Eqs.~\eqref{cLNSI}--\eqref{NCNSI}.
\end{remark}
\subsubsection{Center of symmetry}%
For static solutions with a regular center of symmetry, we have $\eta(0)=\delta(0)=\delta_c>0$. The central density $\rho_c=\widehat{\rho}(\delta_c,\delta_c)$ and central pressure $p_c =\widehat{p}_\mathrm{rad}(\delta_c,\delta_c)=\widehat{p}_\mathrm{tan}(\delta_c,\delta_c)$ are given by
\begin{equation}
    \rho_c = \frac{\mathrm{n}^2K}{1+\mathrm{n}}\delta^{1+\frac{1}{\mathrm{n}}}_c, \qquad p_c =\frac{\mathrm{n}K}{1+\mathrm{n}}\delta^{1+\frac{1}{\mathrm{n}}}_c,
\end{equation}
while the wave speeds at the center of symmetry are simply given by
	\begin{equation} \label{eq:Linearspeedscenter}
	c^2_\mathrm{L}(\delta_\mathrm{c},\delta_\mathrm{c}) = \frac{3}{\mathrm{n}}\left(\frac{1-\nu}{1+\nu}\right)
	,\qquad c^2_\mathrm{T}(\delta_\mathrm{c},\delta_\mathrm{c}) =\frac{3}{2\mathrm{n}}\left(\frac{1-2\nu}{1+\nu}\right), 
	\end{equation}
so that
\begin{equation}
c^2_\mathrm{L}(\delta_\mathrm{c},\delta_\mathrm{c})-\frac{4}{3}c^2_\mathrm{T}(\delta_\mathrm{c},\delta_\mathrm{c}) = \frac{1}{\mathrm{n}}.
\end{equation}
Given the Definition~\ref{IDstatic} of physically admissible initial data for equilibrium configurations, we deduce the following necessary conditions for the existence of physically admissible (strongly) regular ball solutions:

\begin{proposition}
	Necessary conditions for the existence of physically admissible radially compressed balls solutions are given by
	\begin{equation}
	K>0,  \qquad \delta_\mathrm{c}>0,\qquad\mathrm{n}\geq 3\left(\frac{1-\nu}{1+\nu}\right),\qquad -1<\nu\leq\frac{1}{2}.
	  \end{equation}
\end{proposition}
%

\subsubsection{Boundary of the ball}%
The boundary of the ball is defined by the condition $p_\mathrm{rad}(\mathcal{R})=0$. Just as for the NSI polytropes, the boundary condition yields Eqs.~\eqref{yboundaryeta}. However, in this case the energy density on the surface boundary is
\begin{equation}
\begin{split}
      \widehat{\rho}(\delta(\mathcal{R}),\eta(\mathcal{R}))=\mathrm{n}K\eta^{1+\frac{1}{\mathrm{n}}}(\mathcal{R})\left[\frac{(\mathrm{n}-\mathrm{s}) }{1+\mathrm{n}} 
	-\left(1-3\frac{\mathrm{s}}{\mathrm{n}}\left(\frac{1-\nu}{1+\nu}\right)\right)\left(1-y_\mathrm{b}\right)\right]\,,
\end{split}
\end{equation}
where $y_\mathrm{b}$ is given in~\eqref{ybDef}. 
\begin{remark}
From Remark~\ref{LinRhoPiso}, it follows that at the surface boundary of ball, the energy density can be written as
\begin{equation}
\rho(\mathcal{R})=\frac{2\mathrm{n}}{3}p_\mathrm{tan}(\mathcal{R}).
\end{equation}
\end{remark}

Besides condition~\eqref{ShearBound}, other necessary conditions for the existence of physically admissible ball solutions in the NSI model can be obtained from the strict hyperbolicity and causality conditions of the wave speeds at the boundary.

%

\subsubsection{Numerical results}
    %
      %

         %
    \begin{figure}[ht!]
    	\centering
         \includegraphics[width=0.65\textwidth]{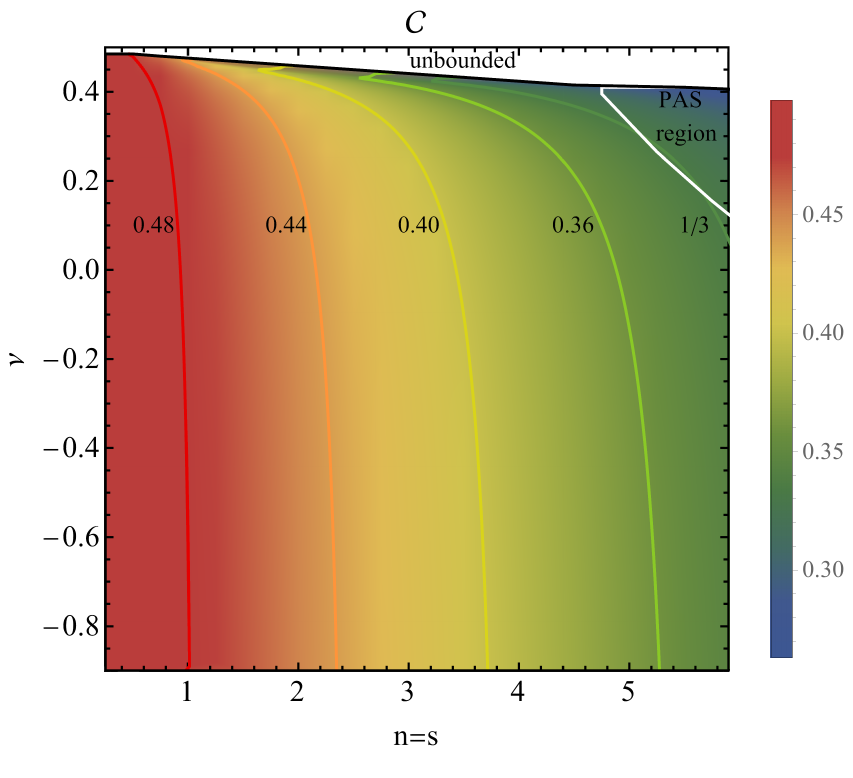} 
    	\caption{ Compactness of different solutions as a function of the Poisson ratio  $\nu$ and different ${\rm n}={\rm s}$.
    	\label{fig:CLIS}}
    \end{figure}
    Motivated by the fact that self-gravitating fluid solutions with a linear EoS do not exist, we start our analysis by considering whether there are bounded configurations for this model when elasticity is turned on. While for values of elastic parameters sufficiently close to the fluid limit we do not find bounded configurations, remarkably, for sufficiently high values of elasticity, we can find self-gravitating configurations. The parameter space of the solution is represented in Fig.~\ref{fig:CLIS}, where we show the compactness of the solutions as a function of the Poisson ratio $\nu$ and the polytropic index ${\rm n}$ (for ${\rm s}={\rm n}$). The colored area represents the region where we find bounded solutions, whereas the white region denotes where we could not find self-gravitating solutions.
    \begin{figure}[ht!]
    	\centering
    	\includegraphics[width=0.75\textwidth]{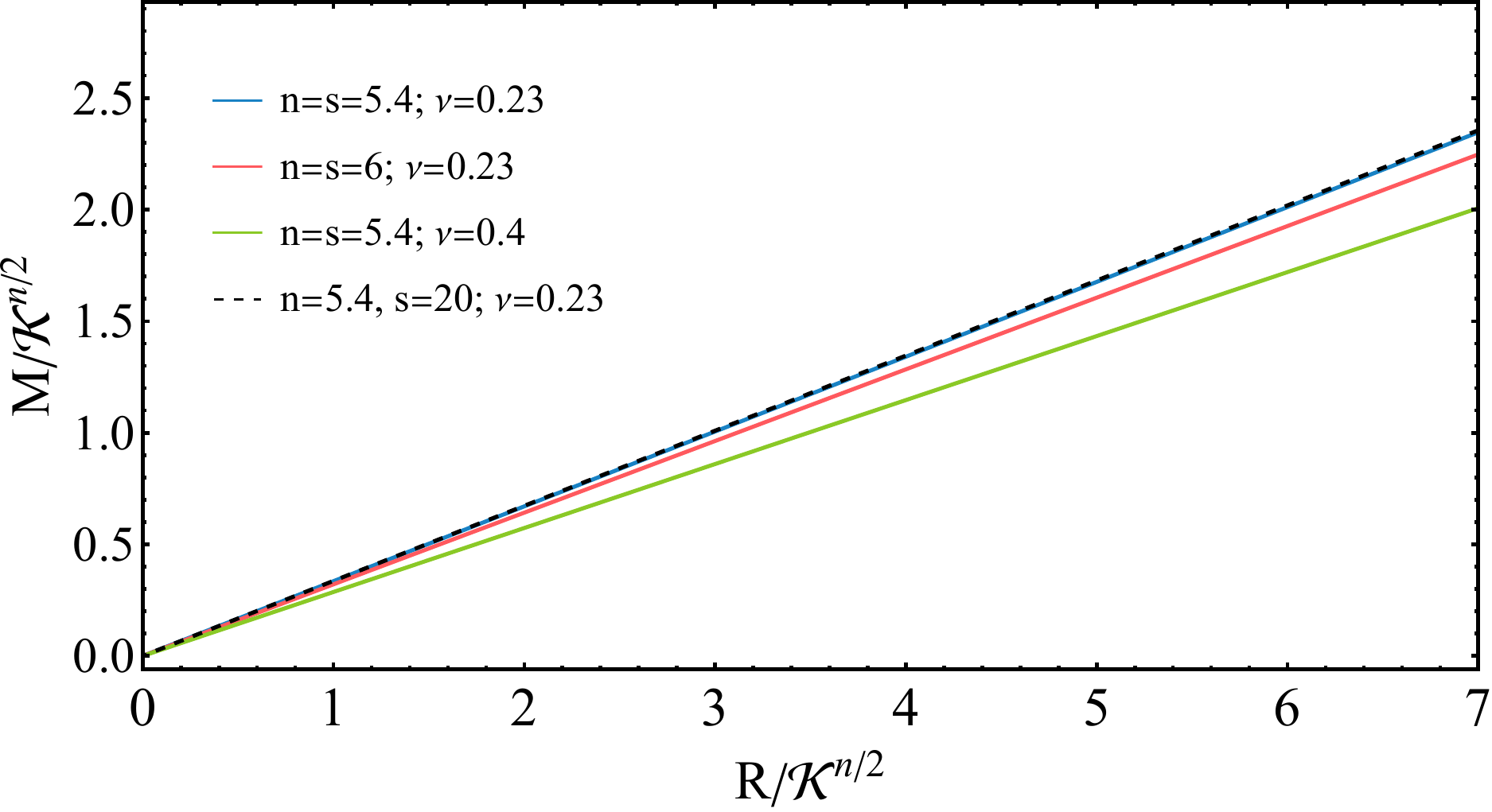} 
    	\caption{ Mass-radius diagram for some representative configurations described by the LIS model, which allows for scale-invariant solutions.
    	\label{fig:MRlinear}}
    \end{figure}
    Some representative examples of these configurations are shown in the mass-radius diagram on Fig.~\ref{fig:MRlinear}. As anticipated above, due to the scale invariance of the model, the mass-to-radius ratio is independent of the central value of the baryonic density, and therefore the mass-radius diagrams for these configurations are lines with constant slope (i.e., constant compactness for each set of model parameters). 
    By observing Fig.~\ref{fig:CLIS} we see that increasing the index ${\rm n}$ lowers the compactness of the solutions, while decreasing the Poisson ratio (increasing elasticity) increases the compactness of the solutions. Although the compactness of the solution can reach very high values, typically this happens for unphysical configurations, where some of the sound speeds in the material exceed the speed of light. The boundary of the region where we find physically admissible solutions is represented by a white curve on the top-right corner in Fig.~\ref{fig:CLIS}, corresponding both to  sufficiently high polytropic index ${\rm n}$ and Poisson ratio $\nu$. 

    \begin{figure}[ht!]
    	\centering
    	\includegraphics[width=0.75\textwidth]{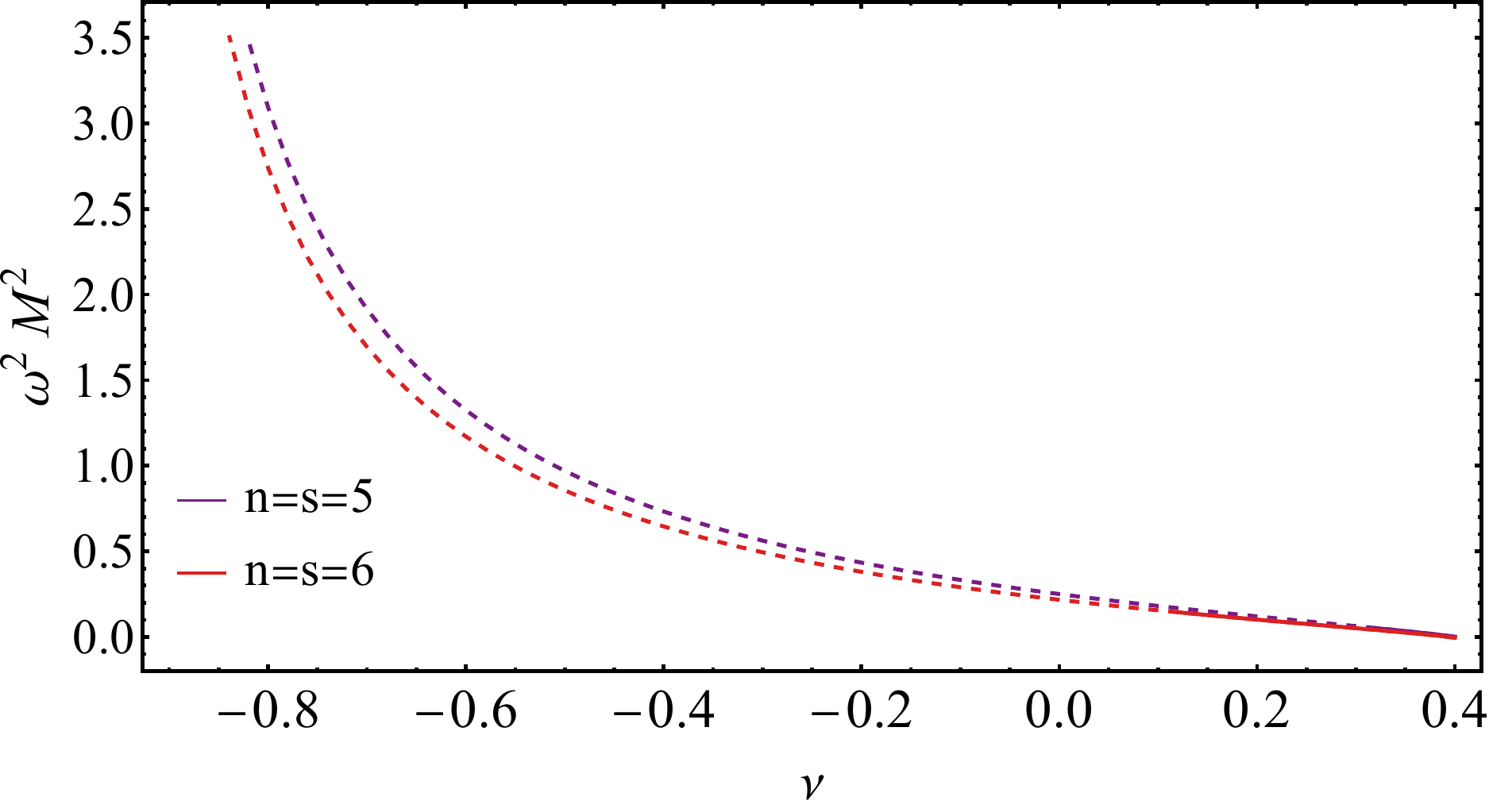} 
    	\caption{Squared frequency of the fundamental mode for linear radial perturbations of stars in the AIS model, as a function of the Poisson ratio  $\nu$, for some representative values of ${\rm n}={\rm s}$. The dashed and solid lines correspond to unphysical and physically admissible configurations, respectively. In all cases, we found no evidence for unstable modes. \label{fig:stabilitylinear}
    	}
    \end{figure}

    The linear mass-radius diagram is a truly remarkable feature which, to the best of our knowledge, has never been reported for viable matter within GR. Scale invariance implies that there is no maximum mass in the model, and so these solutions can exist with any mass. These characteristics make them akin to ordinary BHs, for which the mass is indeed a free parameter, and $\mathcal{M}/\mathcal{R}=1/2$ in the nonspinning case.
    Prompted by these quite unique properties, we analyze the radial stability of these solutions using the radial perturbation formalism presented in Sec.\ref{RS}. In Fig.~\ref{fig:stabilitylinear} we show the squared frequency of the fundamental mode as a function of the Poisson ratio of the material for some characteristic model parameters (${\rm n},{\rm s}$). In all cases that we have numerically explored, we always found $\omega^2>0$, indicating that the solutions are radially stable. 
    
    Finally, we analyzed the maximum compactness of physically admissible configurations for this model, and obtained the bound
    \begin{equation}\label{eq:boundlinearEOS}
        {\cal C}_{\rm max}^{\rm PAS}\lesssim 0.335\,,
    \end{equation}
    which is saturated for $\rm n=s\approx 5.4$ and $\nu \approx 0.23$ (solid blue curve in the left panel in Fig.~\ref{fig:MRlinear}).


    Another question that can be addressed is the effect of the shear index ${\rm s}$ in the solutions. Our preliminary analysis suggests that the effects of ${\rm s}$ are similar to those in the NSI model, i.e., increasing the mass and the compactness of the solution. Unfortunately, it does not seem to play a relevant role in the compactness of physically admissible solutions. This is  due to the fact that the increase of the compactness is counterbalanced by an increase in the sound speed of the solution within the star, which can turn physically admissible solutions into unphysical ones. 
    In the mass-radius diagram of Fig.~\ref{fig:MRlinear}, we also compare the curve corresponding to the solution that approaches our bound~\eqref{eq:boundlinearEOS} (blue) with a solution with the same parameters but much larger shear index ${\rm s}$ (black, dashed). We see that the profile is very similar to the previous case, with the exception that the increase in shear index was enough to make the speed of sound superluminal near the radius of the star.  We leave a more detailed analysis of the effects of the shear index for future work.

\newpage

\subsection{Affine isotropic state materials (AIS)} \label{AIS}
To generalize the fluid affine EoS of Example~\ref{Ex3}, we consider the stress-free reference state material with $p_0=0$ obtained from the LIS pre-stressed reference state material using the transformation~\eqref{TransfRho} with $\alpha_0=-(w_0+p_0)$ and $p_0$ given by Equations~\ref{p0Lin},~\ref{w0Lin}, together with the identity~\eqref{refStaId}, which gives $\varrho^{(\mathrm{sf})}_0=\rho^{(\mathrm{sf})}_0=\frac{K}{\gamma-1}$, and the fact that for ultra-relativistic materials $\varrho^{(\mathrm{ps})}_0=0$. This yields, similarly to the fluid case in Example~\ref{AffineTransLin}, the energy-density 
\begin{equation}\label{eq:strefreeconstiso}
\begin{split}
\widehat{\rho}(\delta,\eta) = \frac{\mathrm{n}K}{1+\mathrm{n}}+K\eta^{1+\frac{1}{\mathrm{n}}}\Bigg[\frac{\mathrm{n}^2 }{1+\mathrm{n}} 
	&-\left(\mathrm{n}- 3\mathrm{s}\left(\frac{1-\nu}{1+\nu}\right)\right)\left(1-\left(\frac{\delta}{\eta}\right)\right) \\
 & -\frac{3\mathrm{s}^2}{(1+\mathrm{s})}\left(\frac{1-\nu}{1+\nu}\right)\left(1-\left(\frac{\delta}{\eta}\right)^{1+\frac{1}{\mathrm{s}}}\right)\Bigg],
 \end{split}
\end{equation}
and the radial and tangential pressures
\begin{subequations}
    \begin{align}
        \widehat{p}_\mathrm{rad}(\delta,\eta) &=-\frac{K {\rm n}}{1 +{\rm n}} + \frac{K {\rm n}}{1 +{\rm n}} \eta^{1+\frac{1}{\mathrm{n}}}\left[1-3\frac{(1+\mathrm{n})}{(1+\mathrm{s})}\frac{\mathrm{s}}{\mathrm{n}}\left(\frac{1-\nu}{1+\nu}\right)\left(1-\left(\frac{\delta}{\eta}\right)^{1+\frac{1}{\mathrm{s}}}\right)\right],\\
        \widehat{p}_\mathrm{tan}(\delta,\eta) &=-\frac{K {\rm n}}{1 +{\rm n}} + \frac{K {\rm n}}{1 +{\rm n}} \eta^{1+\frac{1}{\mathrm{n}}}\Bigg[1-\frac{3(1+\mathrm{n})}{2\mathrm{n}}\left(1-3\frac{\mathrm{s}}{\mathrm{n}}\left(\frac{1-\nu}{1+\nu}\right)\right)\left(1-\left(\frac{\delta}{\eta}\right)\right)\\
&\qquad \qquad \qquad \qquad \quad +\frac{3(1+\mathrm{n})}{2(1+\mathrm{s})}\frac{\mathrm{s}}{\mathrm{n}}\left(1-3\frac{\mathrm{s}}{\mathrm{n}}\right)\left(\frac{1-\nu}{1+\nu}\right)\left(1-\left(\frac{\delta}{\eta}\right)^{1+\frac{1}{\mathrm{s}}}\right)\Bigg].
    \end{align}
\end{subequations}
\subsubsection{Center of symmetry}%
%
The central density $\rho_c=\widehat{\rho}(\delta_c,\delta_c)$ and central pressure $p_c =\widehat{p}_\mathrm{rad}(\delta_c,\delta_c)=\widehat{p}_\mathrm{tan}(\delta_c,\delta_c)$ are given by
\begin{equation}
\rho_c = \frac{K {\rm n}}{1+\mathrm{n}}\left(1+ \mathrm{n}\delta_c^{1+\frac{1}{\rm n}}\right)\,, \qquad p_c = \frac{K {\rm n}}{1+\mathrm{n}}\left(\delta_c^{1+\frac{1}{\rm n}}-1\right),
\end{equation}
while the wave speeds at the center of symmetry are
\begin{subequations}
	\begin{align}
	c^2_\mathrm{L}(\delta_\mathrm{c},\delta_\mathrm{c}) = \frac{3}{\mathrm{n}}\left(\frac{1-\nu}{1+\nu}\right)
	,\qquad c^2_\mathrm{T}(\delta_\mathrm{c},\delta_\mathrm{c}) =\frac{3}{2\mathrm{n}}\left(\frac{1-2\nu}{1+\nu}\right), 
	\end{align}
\end{subequations}
so that
\begin{equation}
c^2_\mathrm{L}(\delta_\mathrm{c},\delta_\mathrm{c})-\frac{4}{3}c^2_\mathrm{T}(\delta_\mathrm{c},\delta_\mathrm{c}) = \frac{1}{\mathrm{n}}.
\end{equation}
Given the Definition~\ref{IDstatic} of physically admissible initial data for equilibrium configurations, we deduce the following necessary conditions for the existence of physically admissible regular ball solutions:

\begin{proposition}
	Necessary conditions for the existence of physically admissible radially compressed balls solutions are given by
	\begin{equation}
	K>0, \qquad \delta_\mathrm{c}>1,\qquad \mathrm{n}\geq 3\left(\frac{1-\nu}{1+\nu}\right),\qquad -1<\nu\leq\frac{1}{2}.
   \end{equation}
\end{proposition}
%
%
%
%
\subsubsection{Numerical results}
The AIS material does not possess the same scaling invariance properties as the LIS material, and thus the properties of the solution for different values of central density cannot be obtained from a single one by a scaling factor. To study these solutions, we follow the same methodology as in the previous section. We investigate the properties of the material for configurations with different central densities and for different model parameters, to infer what is the maximum compactness and compare it with the corresponding fluid model. The main results of this analysis are presented in Fig.~\ref{fig:Cnuaffine}, where we display the maximum compactness of physically admissible configurations with and without radial stability conditions (solid and dotted lines, respectively). By analyzing the plot, we see that there is a maximum in the profile of each curve, corresponding to a given value of ${\nu}$. Interestingly, we find that this maximum increases as ${\rm n}={\rm s}\to 1$, but also the corresponding value of ${\nu}$ approaches ${\nu}=0.5$ in this limit. As a conclusion, we find that the bounds for maximum compactness in the AIS models are given precisely by those in the fluid limit,
\begin{equation}
    {\cal C}_{\rm max}^{\rm PA}\lesssim 0.364\,, \qquad {\cal C}_{\rm max}^{\rm PAS}\lesssim 0.354\,,
\end{equation}
for $\rm n= 1$ and $\nu = 1/2$,
which allow for slightly higher compactness than the pre-stressed LIS model.

\begin{figure}[H]
	\centering
	\includegraphics[width=0.65\textwidth]{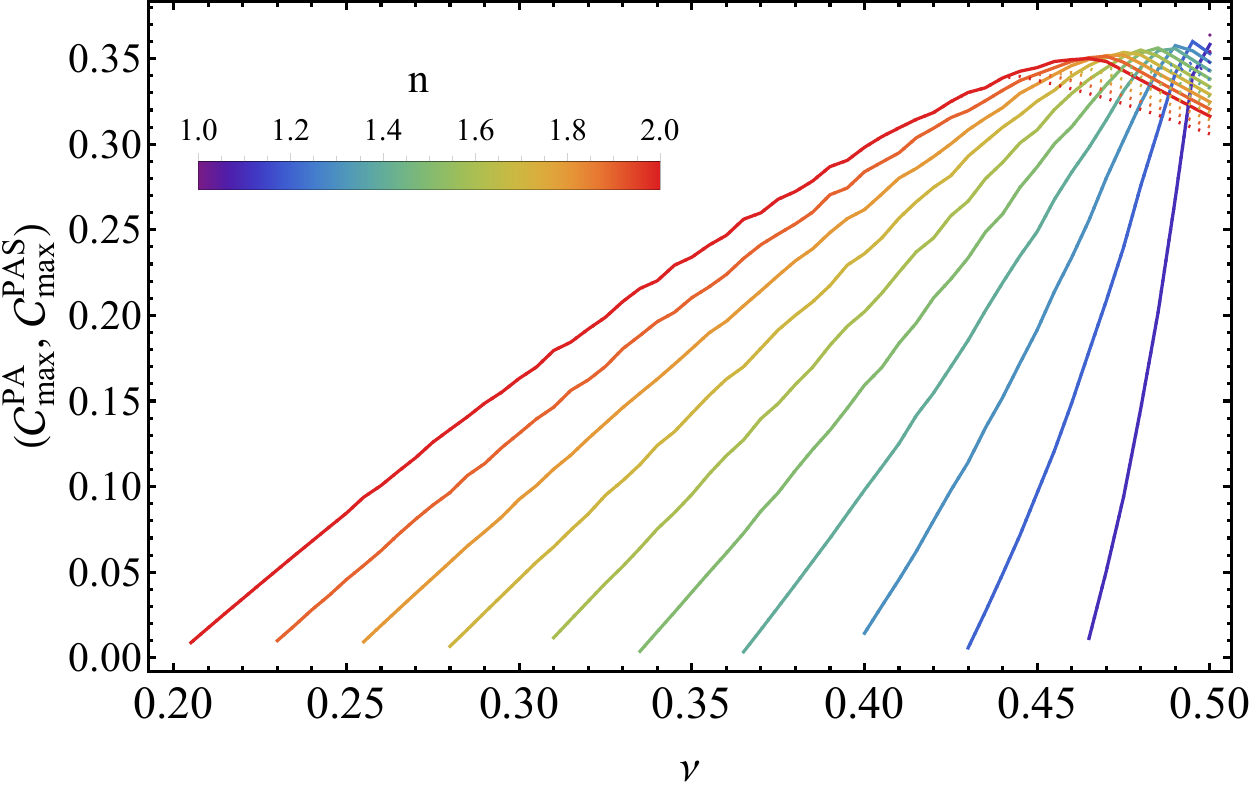} 
	\caption{Compactness of different solutions in the AIS model as a function of the Poisson ratio  $\nu$, for different values of ${\rm n}={\rm s}$. The solid  lines correspond to the maximum physically admissible compactness (${\cal C}_{\rm PA}$), while the dotted lines describe the maximum admissible and stable compactness (${\cal C}_{\rm PAS}$). The maximum compactness in both cases corresponds to the fluid model (${\nu =0.5}$) and ${\rm n=s=1}$. \label{fig:Cnuaffine}
	}
\end{figure}
\newpage

\section{Elastic materials with constant longitudinal wave speeds}\label{LinearES}
In this section\footnote{Some of the results in this section were announced in the companion paper~\cite{Alho:2022bki}.} we turn to material laws that are characterized by having constant longitudinal wave speeds,
\begin{equation}\label{CLConst}
c^2_\mathrm{L}(\delta,\eta)=\tilde{c}^2_\mathrm{L}(\delta,\eta)=\gamma-1,
\end{equation}
with $\gamma\in(1,2]$ in order to satisfy the strict hyperbolicity and the causality conditions. These materials, which can be found in Example~\ref{KarlSam} of Appendix~\ref{ApEx}, were introduced by Karlovini \& Samuelsson~\cite{Karlovini:2004ix}, and contain a \emph{third elastic constant} $C_{\mathrm{KS}}\in\mathbb{R}$. We consider both a stress-free reference state material, $p_0=0$ (possessing a well defined reference state), and a pre-stressed  reference state material, $p_0\neq0$ (thus with no preferred reference state), obtained from the previous one by the natural choice~\eqref{NatChoiceP0}. As in the case for the AIS and LIS materials discussed in Sec.~\ref{LinEOS}, they are elastic generalizations of perfect fluids with an \emph{affine} and \emph{linear EoS}, respectively, which have constant sound speed $c^2_\mathrm{s}=\gamma-1$ (see Table~\ref{table}).
From~\eqref{WSRefSt} and~\eqref{CLConst}, it follows that for this class of materials the following relation holds:
\begin{equation}
    \frac{L}{\rho_0+p_0}=\gamma-1,
\end{equation}
where recall that $L=\lambda+2\mu$.
%
\subsection{Stress-free materials (ACS)}\label{ACS}
%
The subfamily of stress-free reference state materials with constant longitudinal wave speeds has a reference state characterized by
\begin{equation}
p_0=0\qquad\text{and}\qquad \rho_0 =\frac{L}{\gamma-1}.
\end{equation}
The elastic law for these materials in spherical symmetry is given in Example~\ref{KarlSamSS} of Appendix~\ref{ApExSS}. It can be written in terms of the elastic constants $(L,\nu,\theta)$, where~\footnote{The definition of $\theta$ differs slightly from that in~\cite{Alho:2022bki}. The two definitions are related by $\theta_{\mathrm{new}}=\frac{\gamma}{\gamma-1}\theta_\mathrm{old}$.}
\begin{equation}
 \theta=\frac{\gamma C_{\mathrm{KS}}}{L} ,  
\end{equation}
is a third  dimensionless elastic constant. The energy density function is then given by 
%
%
%
\begin{equation}
\begin{split}
\frac{\gamma \widehat{\rho}(\delta,\eta)}{L}&=\frac{\gamma}{\gamma-1}+\left(\frac{1-(2-\gamma)\nu}{\gamma(\gamma-1)(1-\nu)}-\theta\right)\left(\delta^{\gamma}-1\right)+3\left(\frac{1}{\gamma}\left(\frac{1-2\nu}{1-\nu}\right)+\theta\eta^{\frac{\gamma}{3}}\right)\left(\eta^{\frac{\gamma}{3}}-1\right) \\
&\qquad\qquad\qquad +\eta^{\frac{\gamma}{3}}\left(\frac{1}{\gamma}\left(\frac{1-2\nu}{1-\nu}\right)+\theta \left(2\eta^{\frac{\gamma}{3}}-1\right)\right)\left(\left(\frac{\delta}{\eta}\right)^{\gamma}-1\right). 
\end{split}
\end{equation}
The radial and tangential pressures are 
\begin{equation}
\begin{split}
\frac{\gamma\widehat{p}_\mathrm{rad}(\delta,\eta)}{L}&=-\left(1-\frac{2}{\gamma}\frac{1-2\nu}{1-\nu}+\theta\right)+(\gamma-1)\left(\frac{1-(2-\gamma)\nu}{(\gamma-1)\gamma(1-\nu)}-\theta\right)\delta^{\gamma} \\
&+\left(-\frac{1}{\gamma}\frac{1-2\nu}{1-\nu}+\theta\right)\left(2-(\gamma-1)\left(\frac{\delta}{\eta}\right)^{\gamma}\right)\eta^{\frac{\gamma}{3}}+\theta\left(2(\gamma-1)\left(\frac{\delta}{\eta}\right)^{\gamma}-1\right)\eta^{\frac{2\gamma}{3}} ,
\end{split}
\end{equation}
\begin{equation}
\begin{split}
\frac{\gamma}{L}\widehat{p}_\mathrm{tan}(\delta,\eta)&=-\left(1+\theta - \frac{2\left(1-2\nu\right)}{\gamma \left(1-\nu\right)}\right)+\left(\frac{(\gamma -2) \nu +1}{\gamma  (1-\nu )}-(\gamma -1) \theta\right)\delta^\gamma\\
&+\left(\theta -\frac{1-2 \nu }{\gamma  (1-\nu )}\right) \left(2-\gamma +\left(\frac{\delta}{\eta}\right)^{\gamma }\right) \eta^{\gamma /3}+\theta   \left((\gamma -1)+(\gamma -2) \left(\frac{\delta}{\eta}\right)^{\gamma}\right)\eta^{\frac{2 \gamma }{3}},
\end{split}
\end{equation}
and from these we can compute the isotropic and anisotropic pressures:
\begin{subequations}
\begin{align}
    \frac{\gamma \widehat{p}_\mathrm{iso}(\delta,\eta) }{L}&= \frac{2}{\gamma}\frac{(1-2 \nu )}{(1-\nu )}-(\theta +1) + \left(\frac{1-(2-\gamma ) \nu}{\gamma  (1-\nu )}+(1-\gamma) \theta \right)\delta ^{\gamma } \nonumber \\ 
    +&\frac{1}{\gamma}\frac{\gamma -3}{3}\left(\frac{1-2 \nu }{1-\nu }-\gamma  \theta\right) \left(2+\left(\frac{\delta}{\eta}\right)^{\gamma }\right)\eta^{\gamma /3} +\frac{\theta }{3} (2 \gamma -3) \left(1+2 \left(\frac{\delta}{\eta}\right)^{\gamma }\right)\eta^{\frac{2 \gamma }{3}} ,\\
    \frac{\gamma\widehat{q}(\delta,\eta)}{L} &=\gamma\left[\left(\frac{1}{\gamma}\frac{1-2\nu}{1-\nu}-\theta\right)\eta^{\gamma/3}+\theta\eta^{2\gamma/3}\right]\left(1-\left(\frac{\delta}{\eta}\right)^{\gamma}\right).
\end{align}
\end{subequations}
Finally, the spherically symmetric transverse wave speeds~\eqref{cT1},~\eqref{cT2} are
\begin{equation}
c^{2}_\mathrm{T}(\delta,\eta)=\frac{\left[\left(\frac{1}{\gamma}\frac{1-2\nu}{1-\nu}-\theta\right)\eta^{\gamma/3}+\theta\eta^{2\gamma/3}\right]\left(1-\left(\frac{\delta}{\eta}\right)^{\gamma}\right)}{\left[\left(\frac{1}{\gamma-1}-\frac{1}{\gamma}\frac{1-2\nu}{1-\nu}-\theta\right)\delta^{\gamma}+\left(\frac{1}{\gamma}\frac{1-2\nu}{1-\nu}-\theta\right)\eta^{\gamma/3}+\theta\eta^{2\gamma/3}\left(1+\left(\frac{\delta}{\eta}\right)^{\gamma}\right)\right]\left(1-\left(\frac{\delta}{\eta}\right)^{2}\right)},
\end{equation}
\begin{equation}
\tilde{c}^{2}_\mathrm{T}(\delta,\eta)=\left(\frac{\delta}{\eta}\right)^{2-\gamma}\frac{\left[\left(\frac{1}{\gamma}\frac{1-2\nu}{1-\nu}-\theta\right)\eta^{\gamma/3}+\theta\eta^{2\gamma/3}\right]\left(1-\left(\frac{\delta}{\eta}\right)^{\gamma}\right)}{\left[\left(\frac{1}{\gamma-1}-\frac{1}{\gamma}\frac{1-2\nu}{1-\nu}-\theta\right)\eta^{\gamma}+\left(\frac{1}{\gamma}\frac{1-2\nu}{1-\nu}-\theta\right)\eta^{\gamma/3}+2\theta\eta^{2\gamma/3}\right]\left(1-\left(\frac{\delta}{\eta}\right)^{2}\right)},
\end{equation}
while~\eqref{cT3} reads
\begin{equation}
\tilde{c}^2_\mathrm{TT}(\delta,\eta)=\frac{\gamma}{2}\frac{\left(\frac{1}{\gamma}\frac{1-2\nu}{1-\nu}-\theta\right)\eta^{\gamma/3}+\theta\eta^{2\gamma/3}\left(\frac{\delta}{\eta}\right)^{\gamma}}{\left(\frac{1}{\gamma-1}-\frac{1}{\gamma}\frac{1-2\nu}{1-\nu}-\theta\right)\delta^{\gamma}+\left(\frac{1}{\gamma}\frac{1-2\nu}{1-\nu}-\theta\right)\eta^{\gamma/3}+\theta\eta^{2\gamma/3}\left(1+\left(\frac{\delta}{\eta}\right)^{\gamma}\right)}.
\end{equation}
\begin{remark}
    Notice that since for these models $\tilde{c}^2_\mathrm{L}(\delta,\eta)$ is known, by relation~\eqref{CLConst} the last wave speed could be obtained using Eq.~\eqref{RelWS}, to yield  
    \begin{equation}
        (\widehat\rho + \widehat{p}_{\rm tan})\tilde{c}_{\mathrm{TT}}^2 = -3\delta \eta \partial_\delta \partial_\eta \widehat\rho -\frac94 \eta^2 \partial^2_\eta \widehat\rho + \frac34 (2\gamma-3) \eta \partial_\eta \widehat\rho.
    \end{equation}
In particular, these materials belong to the class of materials satisfying~\eqref{LinWScTT} for the choice $B=0$ and $E=\frac32(\gamma-1)$.
\end{remark}
\subsubsection{Center of symmetry}%
For static solutions with a regular center of symmetry, we have $\eta(0)=\delta(0)=\delta_c>0$. The central density $\rho_c=\widehat{\rho}(\delta_c,\delta_c)$ and central pressure $p_c =\widehat{p}_\mathrm{rad}(\delta_c,\delta_c)=\widehat{p}_\mathrm{tan}(\delta_c,\delta_c)$ are given by
\begin{equation}
    \begin{split}
        \frac{\gamma\rho_\mathrm{c}}{L} 
 =\frac{\gamma}{\gamma-1} +(\delta^{\gamma/3}_\mathrm{c}-1)\Bigg[&\left(\frac{1-(2-\gamma)\nu}{\gamma(\gamma-1)(1-\nu)}-\theta\right)(\delta^{\gamma/3}_c-1)^2\\
 &+3\left(\frac{1-(2-\gamma)\nu}{\gamma(\gamma-1)(1-\nu)}\right)(\delta^{\gamma/3}_c-1)+\frac{3}{\gamma-1}\Bigg],
    \end{split}
\end{equation}

\begin{equation}
    \begin{split}
    \frac{\gamma p_c}{L} 
   =(\delta^{\gamma/3}_\mathrm{c}-1)\Bigg[&(\gamma-1)\left(\frac{1-(2-\gamma)\nu}{\gamma(\gamma-1)(1-\nu)}-\theta\right)(\delta^{\gamma/3}_\mathrm{c}-1)^2 \\
   &+3(\gamma-1)\left(\frac{1-(2-\gamma)\nu}{\gamma(\gamma-1)(1-\nu)}-\frac{\gamma\theta}{3(\gamma-1)}\right)(\delta^{\gamma/3}_\mathrm{c}-1)+\left(\frac{1+\nu}{1-\nu}\right)\Bigg].      
    \end{split}
\end{equation}
The dominant energy condition reads
\begin{equation}
\begin{split}
    \rho_c-p_c =& \frac{\gamma\rho_0}{L}+ (\delta^{\gamma/3}_\mathrm{c}-1)\Big[(2-\gamma)\left(\frac{1-(2-\gamma)\nu}{\gamma(\gamma-1)(1-\nu)}-\theta\right)(\delta^{\gamma/3}_\mathrm{c}-1)^2 \\
   &+3(2-\gamma)\left(\frac{1-(2-\gamma)\nu}{\gamma(\gamma-1)(1-\nu)}+\frac{\gamma\theta}{3(2-\gamma)}\right)(\delta^{\gamma/3}_\mathrm{c}-1)+4-\gamma-(2+\gamma)\nu\Big] \geq 0.
\end{split}
\end{equation}
The transverse wave speed at the center of symmetry, Eq.~\eqref{cTCenter}, is
\begin{subequations}
	\begin{align}
	c^2_\mathrm{T}(\delta_c,\delta_c) &=\frac{\gamma}{2}\frac{\theta(\delta^{\gamma/3}_\mathrm{c}-1)+\frac{1}{\gamma}\left(\frac{1-2\nu}{1-\nu}\right)}{\left(\frac{1-(2-\gamma)\nu}{\gamma(\gamma-1)(1-\nu)}-\theta\right)(\delta^{\gamma/3}_\mathrm{c}-1)^2+2\left(\frac{1-(2-\gamma)\nu}{\gamma(\gamma-1)(1-\nu)}\right)(\delta^{\gamma/3}_\mathrm{c}-1)+\frac{1}{\gamma-1}},
	\end{align}
\end{subequations}
and Eq.~\eqref{Kcenter} reads
\begin{equation}
\begin{split}
& c^2_\mathrm{L}(\delta_c,\delta_c)-\frac{4}{3}c^2_\mathrm{T}(\delta_c,\delta_c) = \\
& \frac{(\gamma-1)\left(\frac{1-(2-\gamma)\nu}{\gamma(\gamma-1)(1-\nu)}-\theta\right)(\delta^{\gamma/3}_\mathrm{c}-1)^2+2(\gamma-1)\left[\left(\frac{1-(2-\gamma)\nu}{\gamma(\gamma-1)(1-\nu)}\right)-\frac{\gamma\theta}{3(\gamma-1)}\right](\delta^{\gamma/3}_\mathrm{c}-1)+\frac{1}{3}\left(\frac{1+\nu}{1-\nu}\right)}{\left(\frac{1-(2-\gamma)\nu}{\gamma(\gamma-1)(1-\nu)}-\theta\right)(\delta^{\gamma/3}_\mathrm{c}-1)^2+2\left(\frac{1-(2-\gamma)\nu}{\gamma(\gamma-1)(1-\nu)}\right)(\delta^{\gamma/3}_\mathrm{c}-1)+\frac{1}{\gamma-1}}.
\end{split}
\end{equation}
Given the Definition~\ref{IDstatic} of physically admissible initial data for equilibrium configurations, we deduce the following necessary conditions for the existence of physically admissible regular ball solutions:

\begin{proposition}
	A set of necessary conditions for the existence of physically admissible radially compressed balls solutions is given by 
	\begin{equation}
		  \begin{cases}
		       &1<\delta_\mathrm{c}<+\infty,\quad \text{for} \quad 1<\gamma\leq\frac{3}{2},\quad -1<\nu\leq\frac{1}{2},\quad  0\leq\theta\leq\frac{3(\gamma-1)}{\gamma}\frac{1-(2-\gamma)\nu}{\gamma(\gamma-1) (1-\nu)}, \\
             &1<\delta_\mathrm{c}\leq 1-\frac{1}{\gamma\theta}\left(\frac{1-2\nu}{1-\nu}\right) ,\quad \text{for} \quad 1<\gamma\leq\frac{3}{2},\quad -1<\nu<\frac{1}{2},\quad  -\frac{3(2-\gamma)}{\gamma}\frac{1-(2-\gamma)\nu}{\gamma(\gamma-1) (1-\nu)}\leq\theta<0, \\
              &1<\delta_\mathrm{c}<+\infty,\quad \text{for} \quad \frac{3}{2}\leq\gamma\leq2,\quad -1<\nu\leq\frac{1}{2},\quad  -\frac{3(2-\gamma)}{\gamma}\frac{1-(2-\gamma)\nu}{\gamma(\gamma-1) (1-\nu)}\leq\theta\leq\frac{1-(2-\gamma)\nu}{\gamma(\gamma-1) (1-\nu)}. \\
		   \end{cases} 
	\end{equation}
	%
	%
	%
\end{proposition}
\begin{figure}[ht!]
	\centering
	
	\includegraphics[width=0.80\textwidth]{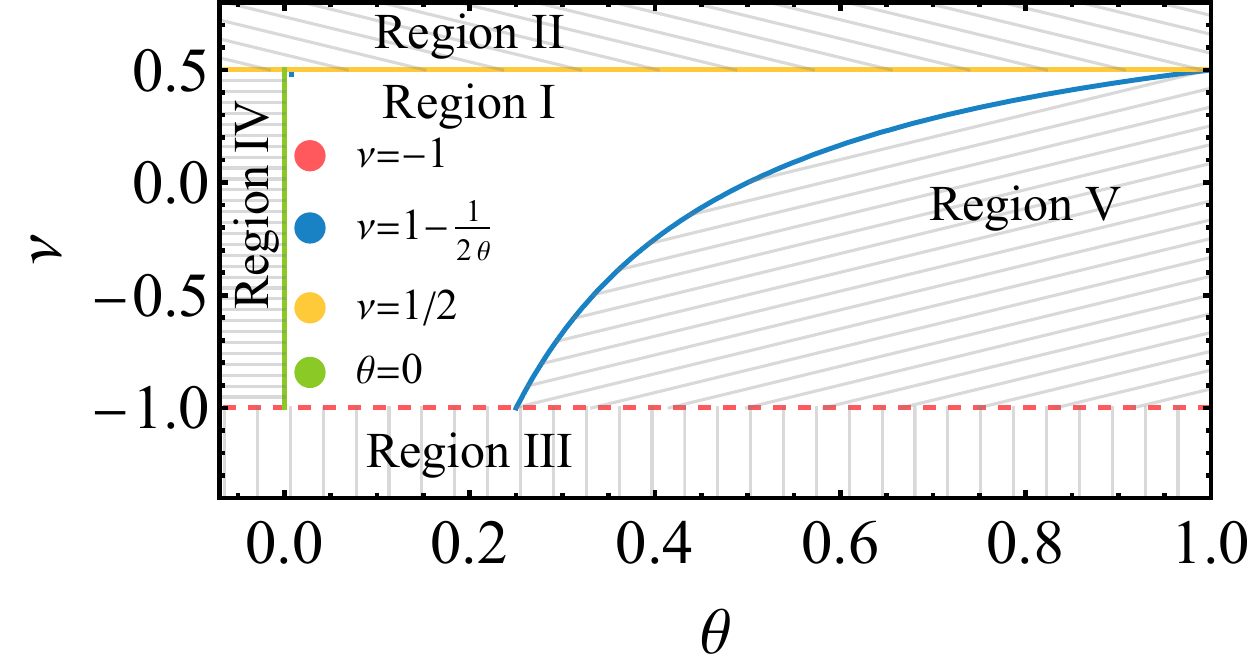} \\
	\caption{Parameter space of the ACS model. The necessary conditions for physical admissibility of the solutions split the parameter space into 5 independent regions. See main text for more details.\label{fig:ultrarigid_region}
	}
\end{figure}

Based on the above conditions, we can identify five regions in the two dimensional $(\theta-\nu)$ parameter space (see Fig.~\ref{fig:ultrarigid_region}): 

\begin{itemize}
    \item \textbf{Region I}: The necessary conditions for physical admissibility are satisfied for any $\delta_c>1$;
    \item \textbf{Region II}: Imaginary sound speeds of transverse waves when $\delta_c \to 1$;
    \item \textbf{Region III}: $c_T^2>3/4 c_L^2$ at the center for $\delta_c \to 1$ and central pressure is negative for large $\delta_c$; 
    \item \textbf{Region IV}: Breaks dominant energy condition and imaginary sound speeds of transverse waves when $\delta_c \to 1$;
    \item \textbf{Region V}: Imaginary sound speeds of transverse waves, negative energy density and pressure at the center for large $\delta_c$.
\end{itemize}

It is worthwhile to make two remarks. First, while Region I satisfies all necessary conditions for physically admissible solutions when $\delta_c>1$, nothing can be said about whether these conditions are sufficient. It is crucial to perform a numerical analysis of the model in order to explore if some of these conditions are broken somewhere within the star. We will do this analysis in the next section. Secondly, while the remaining regions do not satisfy the necessary conditions for \emph{every} central compression, they can still satisfy the necessary conditions for \emph{some} range of $\delta_c$. As an example, while in Region~II the solutions have imaginary sound speeds for configurations close to the reference state, for high values of $\delta_c$ this is no longer an issue, and we can find solutions which are free of pathologies. In our analysis, we will start by focusing first on solutions that have a physically admissible reference state, which restricts the parameter space to that of Regions~I and V. 

\newpage

\subsubsection{Numerical results}

We follow the same procedure of the previous sections and solve numerically the system of Einstein-elastic equations to compute the macroscopic properties of the configurations in the ACS model. A given self-gravitating configuration is described by five parameters: 4 elastic model parameters (the adiabatic index $\gamma$, the Poisson ratio $\nu$, a new third elastic constant $\theta$ and the reference state density $\rho_0$) and the central value for the compression $\delta_c$. Since it would be extremely time-consuming and inefficient to explore the full parameter space, we instead focus on configurations with $\gamma=2$. It is easy to see that this is the most interesting choice, as it will yield the physically admissible configurations with larger longitudinal speeds of sound $c_{\rm L}^2=\tilde{c}_{\rm L}^2=\gamma -1$. Additionally, we will also remove $\rho_0$ from the parameter space by setting $\rho_0=1$, as this amounts to fixing the scale for this analysis. Therefore, dimensionless quantities such as the compactness will not depend on this choice. 
\begin{figure}[ht!]
	\centering
 \includegraphics[width=0.475\textwidth]{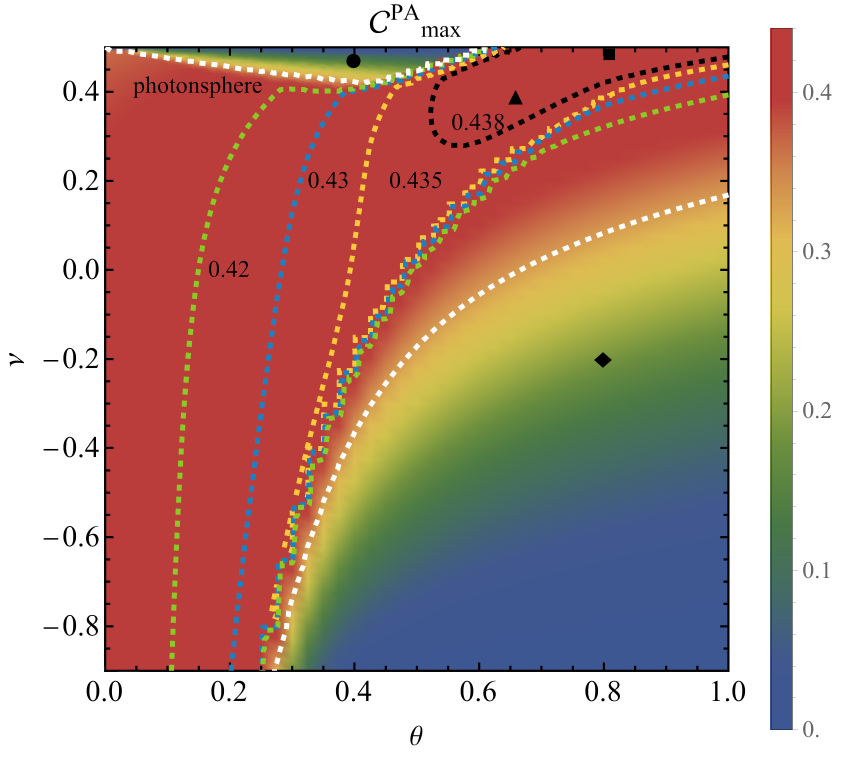}
 \includegraphics[width=0.475\textwidth]{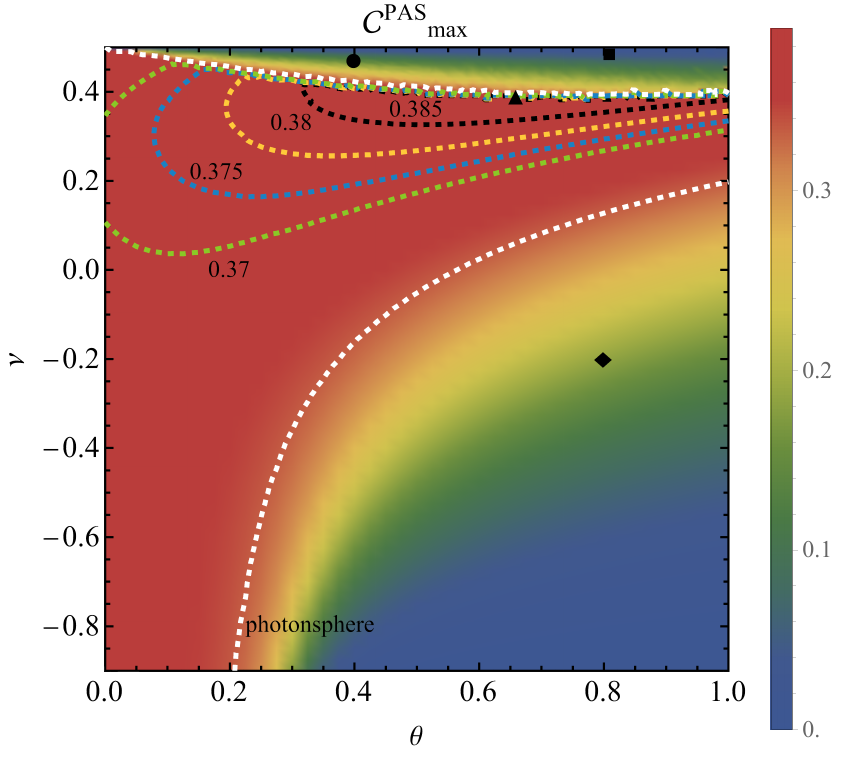} 
	\caption{Density plots of the maximum physically admissible (left) and maximum physically admissible and stable (right) compactness configurations for the ACS model. As expected, we see that the ACS model yields the highest compactness when comparing with other EoS. \label{fig:ACS_contours}
	}
\end{figure}

The main results of this analysis are presented in Fig.~\ref{fig:ACS_contours}, where we show in a density plot the maximum physically admissible (left) and physically admissible and radially stable (right) compactness possible for a configuration in the ACS model as a function of the two elastic parameters $(\nu,\theta)$. In Fig.~\ref{fig:ACS_contours} we restrict ourselves to the main region of physical interest $(0\leq \theta \leq 1, -1<\nu\leq 1/2)$. In both cases represented in the figure, we note that the contours are nested and, apart from numerical errors, there do not appear to be isolated regions with higher compactness. We note also here that the numerical noise associated with the jumps and sharp lines in the profile are related to the interpolation between different points in the three dimensional parameter space analysis. In contrast with the AIS model of Sec.~\ref{AIS}, here the maximum compactness is not reached in the fluid limit, but instead for some critical values of $\nu$ and $\theta$. Our analysis suggests that the maximum compactness that can be obtained for physically admissible configurations in this model is much higher than the maximum fluid bounds, and also than any of the previously studied elastic equations of state. For the ACS model and regular elastic parameters, we obtain the following bounds:
\begin{equation}\label{eq:ACS_bounds_region1}
    {\cal C}_{\rm max}^{\rm PA}\lesssim 0.443\,, \qquad {\cal C}_{\rm max}^{\rm PAS}\lesssim 0.384\,,
\end{equation}
for physically admissible and physically admissible and stable configurations, respectively. The points in the 2-dimensional elastic parameter space for which these bounds were obtained are marked by the black square and triangle, respectively, in Fig.~\ref{fig:ACS_contours}. It should be pointed out that the bound for physically admissible configurations is different from the value ${\cal C}^{\rm PA}_{\rm max}\lesssim 0.463$ presented in~\cite{Alho:2022bki}. The reason for this difference is related to the fact that the aforementioned bound was obtained considering only physical viability conditions, and not imposing any range for the elastic parameters. Removing the restrictions on the range of $\nu$ and exploring configurations with $\nu>1/2$ leads to the bound presented in~\cite{Alho:2022bki}. Although the configurations with $\nu>1/2$ can satisfy all physical viability conditions, as we have shown in~\cite{Alho:2022bki}, it is important to bear in mind that they are still exotic, in the sense that the matter composing the star has unphysical properties in its reference state.
\begin{figure}[ht!]
	\centering
 \includegraphics[width=0.65\textwidth]{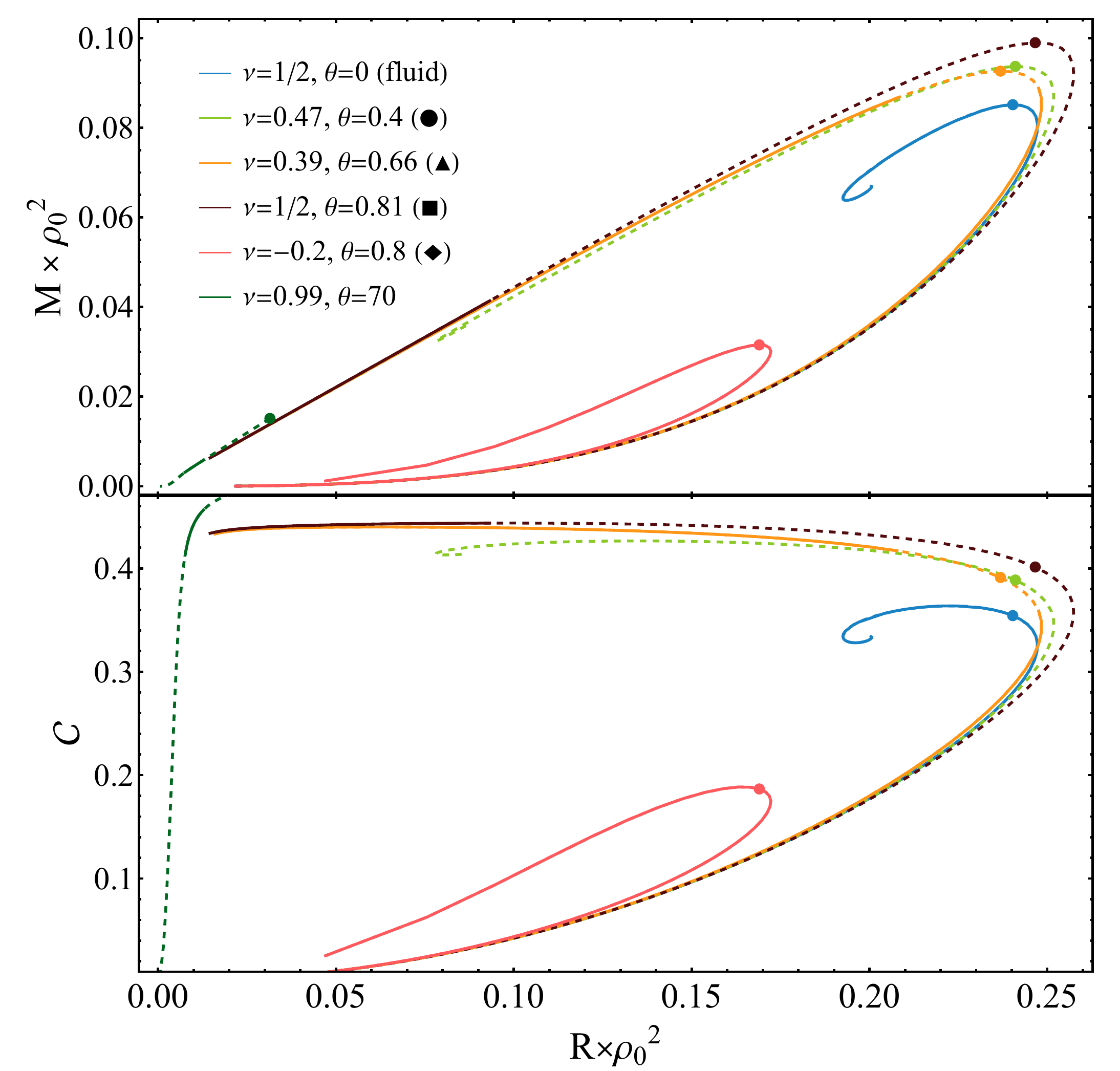}
	\caption{Mass-radius (top) and compactness-radius diagrams for the same representative configurations in the ACS model.\label{fig:MR_ACS}
	}
\end{figure}

For completeness, and to better understand Fig.~\ref{fig:ACS_contours}, the mass-radius and compactness-radius diagrams for the elastic solutions mentioned in this section are shown in Fig.~\ref{fig:MR_ACS}. The curves that contain the configurations that approximate the bounds~\eqref{eq:ACS_bounds_region1} are represented in orange and brown. With the exception of the exotic case with $\nu>1/2$, all configurations have typical mass-radius profiles when compared with the results obtained in the previous sections. 

The analysis of the mass-radius diagrams can  help us understand the low-compactness region presented towards the top of each panel in Fig.~\ref{fig:ACS_contours}, and also why the two panels are qualitatively different in this region of the parameter space. If we look at the configurations that lie close to $\nu = 1/2 $ and $\theta =0$ (e.g. black circle in Fig.~\ref{fig:ACS_contours} and light green curve in Fig.~\ref{fig:MR_ACS}), we see that they are unphysical (dashed curves) for high central compression, as $\tilde{c}_{\rm TT}^2(r)$  becomes negative within the star. As the elastic parameters increase (e.g. black square), the value of $\tilde{c}_{\rm TT}^2$ at the center also increases, and after some critical value of the elastic parameters it compensates the decreasing behaviour of $\tilde{c}_{\rm TT}^2$, allowing us to find physically admissible configurations at high density (see brown curve in Fig.~\ref{fig:MR_ACS}). A consequence of this behaviour is that some families of solutions can have unphysical configurations for intermediate values of central compression, as is represented in the orange curve in the mass-radius diagram. 

We now turn our attention to the lower compactness region in the bottom-right part of each panel. In this region, the lower compactness is a consequence of a decreasing behaviour of the central pressure with increasing central compression. This can be checked by taking the derivative of the central pressure with respect to the central compression,
\begin{equation}
    \frac{d p_{\rm rad}}{d{\delta_c}} = \frac{\delta_c (-2 \theta  (\nu -1)-1)}{2 (\nu -1)} + {\cal O} (1)\,,
\end{equation}
which  was represented in the $\delta_c\gg 1$ limit for clarity purposes. It is easy to see that in this limit the central pressure decreases when $\nu<1-1/(2\theta)$.  The red curve in Fig.~\ref{fig:MR_ACS} represents a particular configuration in this region (black diamond shape in Fig.~\ref{fig:ACS_contours}). 

Finally, the dark green curve represents the family of solutions that (approximately) yield the maximum physically admissible compactness when there is no restriction on the elastic parameters and that leads to the bound presented in~\cite{Alho:2022bki}. This solution is particularly exotic, in the sense that it is described by two branches, an initial branch at lower densities (not represented in the diagram), that has low compactness and yields typically unphysical configurations, and a second branch that appears at higher densities (represented in the figure), which is unstable and mostly unphysical, with the exception of a small section which satisfy all physical viability conditions.

	

\newpage

\subsection{Pre-stressed materials (LCS)} \label{LCS}
The family of natural pre-stressed materials with constant longitudinal wave speeds has a reference state characterized by 
\begin{equation}
p_0= \frac{\lambda+2\mu}{\gamma}-\frac{4\mu}{\gamma^2}+C_{\mathrm{KS}}\qquad\text{and}\qquad \rho_0 = \frac{\lambda+2\mu}{\gamma(\gamma-1)}+\frac{4\mu}{\gamma^2}-C_\mathrm{KS},
\end{equation}
where again $C_{\mathrm{KS}}\in\mathbb{R}$ is a third elastic constant. In Appendix~\ref{ApEx} we deduce the natural pre-stressed reference state material without symmetry assumptions in Example~\ref{KarlSam}, and its form in spherical symmetry is obtained in Appendix~\ref{ApExSS}, see~\eqref{PSCLSSS}.
The energy density is given by 
\begin{equation}\label{KSPSSS}
\begin{split}
\widehat{\rho}^{(\mathrm{ps})}(\delta,\eta) = &  \,\, \left(\frac{\lambda+2\mu}{\gamma(\gamma-1)} - \frac{2\mu}{\gamma^2}-C_\mathrm{KS}\right) \delta^\gamma  + \left(\frac{2\mu}{\gamma^2}-C_\mathrm{KS}\right)\eta^{\gamma/3}\left(2+\left(\frac{\delta}{\eta}\right)^\gamma\right) \\
&+ C_\mathrm{KS}\eta^{2\gamma/3}\left(1+2\left(\frac{\delta}{\eta}\right)^\gamma\right).
\end{split}
\end{equation}
%
 \subsubsection{Invariants under renormalization of the reference state}
%
If we take a reference state that is compressed by a factor of $f$ (in volume) with respect to the original reference state, then the new variables $(\tilde\delta,\tilde\eta)$ are related to the original variables $(\delta,\eta)$ by
\begin{equation}
\delta = f \tilde\delta, \qquad  \eta = f \tilde\eta.
\end{equation}
If we write the energy density as
\begin{equation}
\widehat{\rho}^{(\mathrm{ps})}(\delta,\eta) = A \delta^\gamma  + B \eta^{\gamma/3}\left(2+\left(\frac{\delta}{\eta}\right)^\gamma\right) + C_\mathrm{KS} \eta^{2\gamma/3}\left(1+2\left(\frac{\delta}{\eta}\right)^\gamma\right),
\end{equation}
then it can be rewritten as
\begin{equation}
\widehat{\rho}^{(\mathrm{ps})}(\tilde\delta,\tilde\eta) = A f^\gamma {\tilde\delta}^\gamma  + B f^{\gamma/3}{\tilde\eta}^{\gamma/3}\left(2+\left(\frac{\tilde\delta}{\tilde\eta}\right)^\gamma\right) + C_\mathrm{KS} f^{2\gamma/3}{\tilde\eta}^{2\gamma/3}\left(1+2\left(\frac{\tilde\delta}{\tilde\eta}\right)^\gamma\right).
\end{equation}
If on the other hand we write it as
\begin{equation}
\widehat{\rho}^{(\mathrm{ps})}(\tilde\delta,\tilde\eta) = \tilde{A} {\tilde\delta}^\gamma + \tilde{B}{\tilde\eta}^{\gamma/3}\left(2+\left(\frac{\tilde\delta}{\tilde\eta}\right)^\gamma\right) + \tilde{C}_\mathrm{KS}{\tilde\eta}^{2\gamma/3}\left(1+2\left(\frac{\tilde\delta}{\tilde\eta}\right)^\gamma\right),
\end{equation}
we obtain, setting $L=\lambda+2\mu$,
\begin{subequations}
\begin{align}
\tilde{L} &= \left[L-\gamma(\gamma-1)\left(\frac{2\mu}{\gamma^2}+C_\mathrm{KS}+\left(\frac{2\mu}{\gamma^2}-C_\mathrm{KS}\right)f^{-\gamma/3}\right)\left(1-f^{-\gamma/3}\right)\right] f^\gamma, \\
\tilde{\mu} &= \left(\mu-\frac{\gamma^2 C_\mathrm{KS}}{2}\left(1-f^{\gamma/3}\right) \right)f^{\gamma/3}, \\ \tilde{C}_\mathrm{KS} &= f^{2\gamma/3} C_\mathrm{KS}.
\end{align}
\end{subequations}
Therefore, by changing the reference state we obtain an equivalent description of the material, moving from the parameters $(L,\mu,C_\mathrm{KS})$ to $(\tilde{L},\tilde{\mu},\tilde{C}_\mathrm{KS})$; the choice a of reference state is akin to a choice of gauge. A convenient choice of gauge-invariant parameters are given by
\begin{equation}
     \mathcal{K}=\frac{(\gamma-1)A}{\varrho^{\gamma}_0}, \qquad \mathcal{E}_1=\frac{((\gamma-1)A)^{\frac{1}{3}\left(\frac{3-\gamma}{\gamma-1}\right)}B}{\varrho^{\frac{2\gamma}{3(\gamma-1)}}_0}, \qquad \mathcal{E}_2 = \frac{((\gamma-1)A)^{\frac{1}{3}\left(\frac{3-2\gamma}{\gamma-1}\right)}C_\mathrm{KS}}{\varrho^{\frac{\gamma}{3(\gamma-1)}}_0},
\end{equation}
where $\varrho_0$ is the baryonic mass density of the reference state. Note that $\mathcal{E}_1$ and $\mathcal{E}_2$ are dimensionless. Other (dimensionful) invariant quantities related to elasticity are
\begin{equation}
\varrho(\delta)=\varrho_0\delta, \qquad
 \varsigma(\eta)=\varrho_0\eta. 
\end{equation}
Written in terms of the above invariants, the energy density reads
\begin{equation}
    \widehat{\rho}^{(\mathrm{ps})}(\varrho,\varsigma)=\frac{\mathcal{K}}{\gamma-1}\varrho^{\gamma}+\mathcal{E}_1 \mathcal{K}^{-\frac{1}{3}\left(\frac{3-\gamma}{\gamma-1}\right)}\varsigma^{\gamma/3}\left(2+\left(\frac{\varrho}{\varsigma}\right)^{\gamma}\right)+\mathcal{E}_2\mathcal{K}^{-\frac{1}{3}\left(\frac{3-2\gamma}{\gamma-1}\right)}\varsigma^{2\gamma/3}\left(1+2\left(\frac{\varrho}{\varsigma}\right)^{\gamma}\right).
\end{equation}
%
%
\subsubsection{Invariants under scaling of the baryonic mass density}
%
Unlike in the case of the polytropes, the EoS for the LCS material is insensitive to the baryonic mass density (i.e., this material is ultra-relativistic, see Remark~\ref{UltraRelExtraSym}). Therefore, there is an extra  invariance, corresponding to scaling the reference state baryonic mass density (while keeping $\delta$ and $\eta$ fixed):
\begin{equation}
    \tilde\varrho_0 = f \varrho_0.
\end{equation}
Notice that this is quite different from renormalizing the reference state, which is simply a change in the description of a \emph{fixed} material: here we are modifying the material itself by adding more baryons per unit volume, so that its (gauge-invariant) elastic parameters will change. Under this scaling of $\varrho_0$, we have the new gauge-invariant quantities
\begin{equation}
    \tilde\varrho = f \varrho, \qquad \tilde\varsigma = f \varsigma,
\end{equation}
and so
\begin{equation}
    \widehat{\rho}^{(\mathrm{ps})}(\tilde\varrho,\tilde\varsigma)=\frac{\mathcal{K}}{\gamma-1}f^{\gamma}{\tilde\varrho}^{\gamma}+\mathcal{E}_1 \mathcal{K}^{-\frac{1}{3}\left(\frac{3-\gamma}{\gamma-1}\right)}f^{\gamma/3}{\tilde\varsigma}^{\gamma/3}\left(2+\left(\frac{\tilde\varrho}{\tilde\varsigma}\right)^{\gamma}\right)+\mathcal{E}_2\mathcal{K}^{-\frac{1}{3}\left(\frac{3-2\gamma}{\gamma-1}\right)}f^{2\gamma/3}{\tilde\varsigma}^{2\gamma/3}\left(1+2\left(\frac{\tilde\varrho}{\tilde\varsigma}\right)^{\gamma}\right).
\end{equation}
This can be rewritten as
\begin{equation}
    \widehat{\rho}^{(\mathrm{ps})}(\tilde\varrho,\tilde\varsigma)=\frac{\mathcal{\tilde K}}{\gamma-1}{\tilde\varrho}^{\gamma}+\mathcal{\tilde E}_1 \mathcal{\tilde K}^{-\frac{1}{3}\left(\frac{3-\gamma}{\gamma-1}\right)}{\tilde\varsigma}^{\gamma/3}\left(2+\left(\frac{\tilde\varrho}{\tilde\varsigma}\right)^{\gamma}\right)+\mathcal{\tilde E}_2\mathcal{\tilde K}^{-\frac{1}{3}\left(\frac{3-2\gamma}{\gamma-1}\right)}{\tilde\varsigma}^{2\gamma/3}\left(1+2\left(\frac{\tilde\varrho}{\tilde\varsigma}\right)^{\gamma}\right),
\end{equation}
provided that we define
\begin{equation}
    \mathcal{\tilde K} = f^\gamma \mathcal{K}, \qquad \mathcal{\tilde E}_1 = f^{\frac{2\gamma}{3(\gamma-1)}} \mathcal{E}_1, \qquad \mathcal{\tilde E}_2 = f^{\frac{\gamma}{3(\gamma-1)}} \mathcal{E}_2.
\end{equation}
So by changing the reference state baryonic mass density we move from the material with parameters $(\mathcal{K},\mathcal{E}_1,\mathcal{E}_2)$ to the material with parameters $(\mathcal{\tilde K},\mathcal{\tilde E}_1,\mathcal{\tilde E}_2)$, while leaving the energy density unchanged. This means that, up to a scale, these two materials have the same behaviour. If we set 
\begin{equation}
    \mathcal{E}_1 = \mathcal{E}^2 \cos \theta, \qquad \mathcal{E}_2 = \mathcal{E} \sin \theta,
\end{equation}
with $\mathcal{E}\geq 0$ and $\theta \in [0,2\pi]$, then the new adimensional elastic parameters $(\mathcal{E},\theta)$ transform under the rescaling as
\begin{equation}
    \mathcal{\tilde E} = f^{\frac{\gamma}{3(\gamma-1)}} \mathcal{E}, \qquad \tilde \theta = \theta.
\end{equation}
Therefore, up to a scale, the behaviour of LCS materials depends only on the parameter $\theta$.
Indeed, written in terms of the invariants 
\begin{equation}
    \bar{\varrho}=\frac{\mathcal{K}^{1/(\gamma-1)}}{\mathcal{E}^{3/\gamma}}\varrho, \qquad     \bar{\varsigma}=\frac{\mathcal{K}^{1/(\gamma-1)}}{\mathcal{E}^{3/\gamma}}\varsigma,
\end{equation}
the energy density becomes
\begin{equation}
    \widehat{\rho}^{(\mathrm{ps})}(\bar{\varrho},\bar{\varsigma})=\frac{\mathcal{E}^3}{\mathcal{K}^{1/(\gamma-1)}}\left[\frac{\bar{\varrho}^{\gamma}}{\gamma-1}+\cos{\theta}\bar{\varsigma}^{\gamma/3}\left(2+\left(\frac{\bar{\varrho}}{\bar{\varsigma}}\right)^{\gamma}\right)+\sin{\theta}\bar{\varsigma}^{2\gamma/3}\left(1+2\left(\frac{\bar{\varrho}}{\bar{\varsigma}}\right)^{\gamma}\right)\right].
\end{equation}
The radial and tangential pressures are given by
\begin{equation}
    \widehat{p}^{(\mathrm{ps})}_\mathrm{rad}(\bar{\varrho},\bar{\varsigma})=(\gamma-1) \widehat{\rho}^{(\mathrm{ps})}(\bar{\varrho},\bar{\varsigma})-\frac{\mathcal{E}^3}{\mathcal{K}^{1/(\gamma-1)}}\left(2\gamma\cos{\theta}\bar{\varsigma}^{\gamma/3}+\gamma\sin{\theta}\bar{\varsigma}^{2\gamma/3}\right),
\end{equation}
\begin{equation}
    \widehat{p}^{(\mathrm{ps})}_\mathrm{tan}(\bar{\varrho},\bar{\varsigma}) = (\gamma-1) \widehat{\rho}^{(\mathrm{ps})}(\bar{\varrho},\bar{\varsigma})-\frac{\mathcal{E}^3}{\mathcal{K}^{1/(\gamma-1)}}\left(\gamma\cos{\theta}\bar{\varsigma}^{\gamma/3}\left(1+\left(\frac{\bar{\varrho}}{\bar{\varsigma}}\right)^{\gamma}\right)+\gamma\sin{\theta}\bar{\varsigma}^{2\gamma/3}\left(\frac{\bar{\varrho}}{\bar{\varsigma}}\right)^{\gamma}\right),
\end{equation}
and the isotropic and anisotropic pressures are
\begin{equation}
    \widehat{p}^{(\mathrm{ps})}_\mathrm{iso}(\bar{\varrho},\bar{\varsigma}) =  (\gamma-1) \widehat{\rho}^{(\mathrm{ps})}(\bar{\varrho},\bar{\varsigma})-\frac{\mathcal{E}^3}{\mathcal{K}^{1/(\gamma-1)}}\left[\frac{2\gamma}{3}\cos{\theta}\bar{\varsigma}^{\gamma/3}\left(2+\left(\frac{\bar{\varrho}}{\bar{\varsigma}}\right)^{\gamma}\right)+\frac{\gamma}{3}\sin{\theta}\bar{\varsigma}^{2\gamma/3}\left(1+2\left(\frac{\bar{\varrho}}{\bar{\varsigma}}\right)^{\gamma}\right)\right],
\end{equation}
\begin{equation}
    \widehat{q}^{(\mathrm{ps})}(\bar{\varrho},\bar{\varsigma}) = \frac{\mathcal{E}^3}{\mathcal{K}^{1/(\gamma-1)}}\left(\gamma\cos{\theta}\bar{\varsigma}^{\gamma/3}+\gamma\sin{\theta}\bar{\varsigma}^{2\gamma/3}\right)\left(1-\left(\frac{\bar{\varrho}}{\bar{\varsigma}}\right)^{\gamma}\right).
\end{equation}
The transverse wave speeds are
\begin{equation}
c^{2}_\mathrm{T}(\bar{\varrho},\bar{\varsigma})= \frac{\left(\cos{\theta}\bar{\varsigma}^{\gamma/3}+\sin{\theta}\bar{\varsigma}^{2\gamma/3}\right)\left(1-\left(\frac{\bar{\varrho}}{\bar{\varsigma}}\right)^{\gamma}\right)}{\left[\frac{\bar{\varrho}^{\gamma}}{\gamma-1}+\cos{\theta}\bar{\varsigma}^{\gamma/3}+\sin{\theta}\bar{\varsigma}^{2\gamma/3}\left(1+\left(\frac{\bar{\varrho}}{\bar{\varsigma}}\right)^{\gamma}\right)\right]\left(1-\left(\frac{\bar{\varrho}}{\bar{\varsigma}}\right)^{2}\right)},
\end{equation}
\begin{equation}
\tilde{c}^{2}_\mathrm{T}(\bar{\varrho},\bar{\varsigma})=\left(\frac{\bar{\varrho}}{\bar{\varsigma}}\right)^{2-\gamma}\frac{\left(\cos{\theta}\bar{\varsigma}^{\gamma/3}+\sin{\theta}\bar{\varsigma}^{2\gamma/3}\right)\left(1-\left(\frac{\bar{\varrho}}{\bar{\varsigma}}\right)^{\gamma}\right)}{\left[\frac{\bar{\varrho}^{\gamma}}{\gamma-1}+\cos{\theta}\bar{\varsigma}^{\gamma/3}+2\sin{\theta}\bar{\varsigma}^{2\gamma/3}\right]\left(1-\left(\frac{\bar{\varrho}}{\bar{\varsigma}}\right)^{2}\right)},
\end{equation}
and
\begin{equation}
\tilde{c}^2_\mathrm{TT}(\bar{\varrho},\bar{\varsigma})=\frac{\gamma}{2}\frac{\cos{\theta}\bar{\varsigma}^{\gamma/3}+\sin{\theta}\bar{\varsigma}^{2\gamma/3}\left(\frac{\bar{\varrho}}{\bar{\varsigma}}\right)^{\gamma}}{\frac{\bar{\varrho}^{\gamma}}{\gamma-1}+\cos{\theta}\bar{\varsigma}^{\gamma/3}+\sin{\theta}\bar{\varsigma}^{2\gamma/3}\left(1+\left(\frac{\bar{\varrho}}{\bar{\varsigma}}\right)^{\gamma}\right)}.
\end{equation}
%

\subsubsection{Center of symmetry}%
For static solutions with a regular center of symmetry, $\eta(0)=\delta(0)=\delta_\mathrm{c}>0$, and so we have for the invariants $\bar{\varsigma}(\delta_c)=\bar{\varrho}(\delta_\mathrm{c})=\bar{\varrho}_\mathrm{c}$. The central energy density $\rho_\mathrm{c}=\widehat{\rho}(\delta_\mathrm{c},\delta_\mathrm{c})$ and the central pressure $p_\mathrm{c} =\widehat{p}_\mathrm{rad}(\delta_\mathrm{c},\delta_\mathrm{c})=\widehat{p}_\mathrm{tan}(\delta_\mathrm{c},\delta_\mathrm{c})$ are given by
\begin{equation}
\rho^{(\mathrm{ps})}_\mathrm{c} = \frac{\mathcal{E}^3\bar{\varrho}^{\gamma/3}_\mathrm{c}}{\mathcal{K}^{1/(\gamma-1)}}\left[\frac{\bar{\varrho}^{2\gamma/3}_\mathrm{c}}{\gamma-1}+3\sin{\theta}\bar{\varrho}^{\gamma/3}_\mathrm{c}+3\cos{\theta}\right],
\end{equation}
\begin{equation}
p^{(\mathrm{ps})}_\mathrm{c} = \frac{\mathcal{E}^3\bar{\varrho}^{\gamma/3}_\mathrm{c}}{\mathcal{K}^{1/(\gamma-1)}}\left[\bar{\varrho}^{2\gamma/3}_\mathrm{c}-(3-2\gamma)\sin{\theta}\bar{\varrho}^{\gamma/3}_\mathrm{c}-(3-\gamma)\cos{\theta}\right].
\end{equation}
From equation~\eqref{cTCenter}, the transverse wave speeds at the center of symmetry are
\begin{subequations}
	\begin{align}
	c^2_\mathrm{T}(\bar{\varrho}_\mathrm{c},\bar{\varrho}_\mathrm{c}) &= \frac{\gamma}{2}(\gamma-1)\frac{\sin{\theta}\bar{\varrho}^{\gamma/3}_\mathrm{c}+\cos{\theta}}{\bar{\varrho}^{2\gamma/3}_\mathrm{c}+2(\gamma-1)\sin{\theta}\bar{\varrho}^{\gamma/3}_\mathrm{c}+(\gamma-1)\cos{\theta}},
	\end{align}
\end{subequations}
so that
\begin{equation}
c^2_\mathrm{L}-\frac{4}{3}c^2_\mathrm{T}(\bar{\varrho}_\mathrm{c},\bar{\varrho}_\mathrm{c}) = (\gamma-1)\left[\frac{\bar{\varrho}^{2\gamma/3}_\mathrm{c}-\frac{2}{3}\left(3-2\gamma\right)\sin{\theta}\bar{\varrho}^{\gamma/3}_\mathrm{c}-\frac{1}{3}(3-\gamma)\cos{\theta}}{\bar{\varrho}^{2\gamma/3}_\mathrm{c}+2(\gamma-1)\sin{\theta}\bar{\varrho}^{\gamma/3}_\mathrm{c}+(\gamma-1)\cos{\theta}}\right].
\end{equation}

From the Definition~\ref{IDstatic} of physically admissible initial data for equilibrium configurations, we deduce the following necessary conditions for the existence of physically admissible regular ball solutions:
\begin{proposition}
	A set of necessary conditions for the existence of physically admissible radially compressed ball solutions is given by 
 \begin{equation}
  \mathcal{K}>0   ,\quad \mathcal{E}>0,
 \end{equation}
 and
	\begin{equation}\label{eq:nec_conds_ultrarigid_prestressed}
	 \begin{cases}
	&\frac{1}{2^{3/\gamma}}\left[(3-2\gamma)\sin{\theta}+\sqrt{(3-2\gamma)^2\sin^2{\theta}+4(3-\gamma)\cos{\theta}}\right]^{3/\gamma}<\bar{\varrho}_\mathrm{c}<\infty ,\quad1<\gamma\leq2, \quad  \theta\in[0,\frac{\pi}{2}];\\
    &(-\cot{\theta})^{3/\gamma}<\bar{\varrho}_\mathrm{c}<\infty ,\quad\frac{3}{2}\leq\gamma\leq2,\quad\theta\in(\frac{\pi}{2},\pi) ;\\
    &(-\cot{\theta})^{3/\gamma}<\bar{\varrho}_\mathrm{c}<\infty ,\quad 1<\gamma<\frac{3}{2},\quad 
    \theta\in(\pi-\arccos{a},\pi), \quad a=\frac{3(3-\gamma)}{2(3-2\gamma)^2}\left(1+\sqrt{1+\frac{4(3-2\gamma)^4}{9(3-\gamma)^2}}\right) ;\\
    & \frac{1}{2^{3/\gamma}}\left[(3-2\gamma)\sin{\theta}+\sqrt{(3-2\gamma)^2\sin^2{\theta}+4(3-\gamma)\cos{\theta}}\right]^{3/\gamma}<\bar{\varrho}_\mathrm{c}\leq(-\cot{\theta})^{3/\gamma},\quad \theta\in(\frac{3\pi}{2},2\pi).
	\end{cases} 
	\end{equation}
	%
	%
	%
\end{proposition}
%

\subsubsection{Numerical results}
%
\begin{figure}[ht!]
	\centering
	\includegraphics[width=0.45\textwidth]{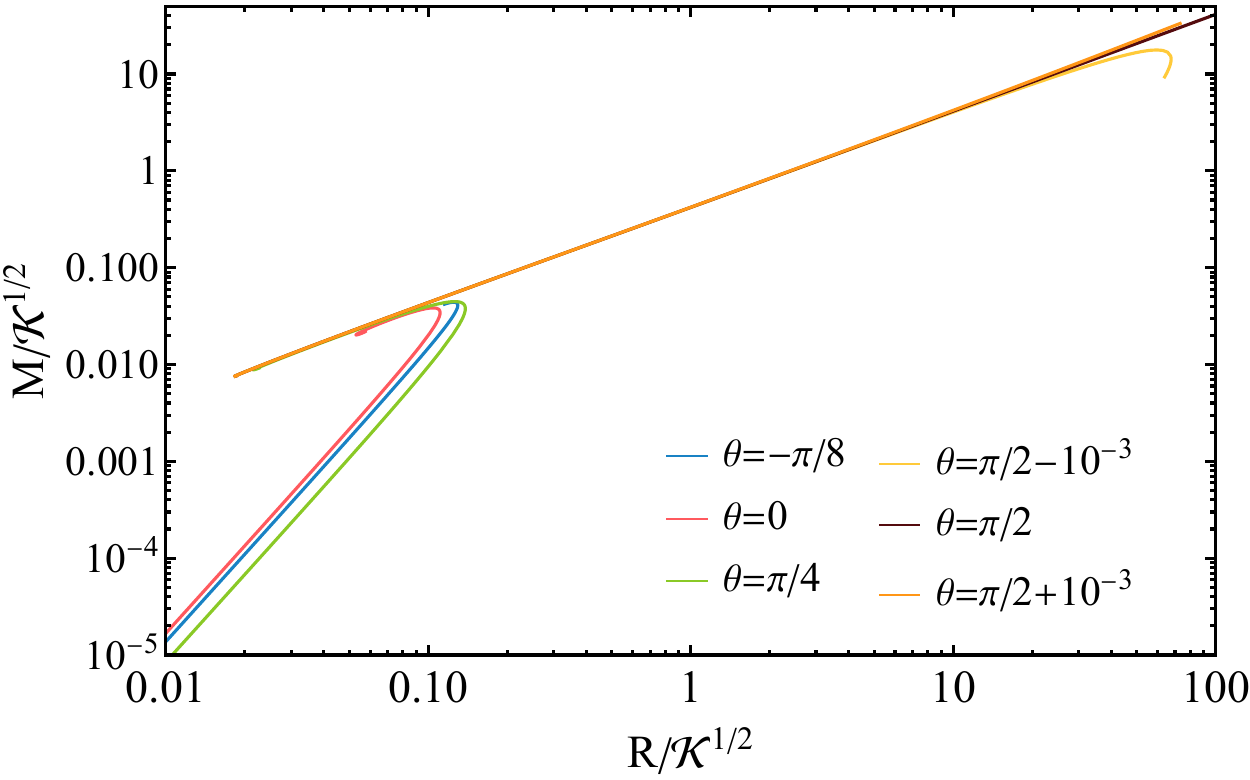} 
 \includegraphics[width=0.45\textwidth]{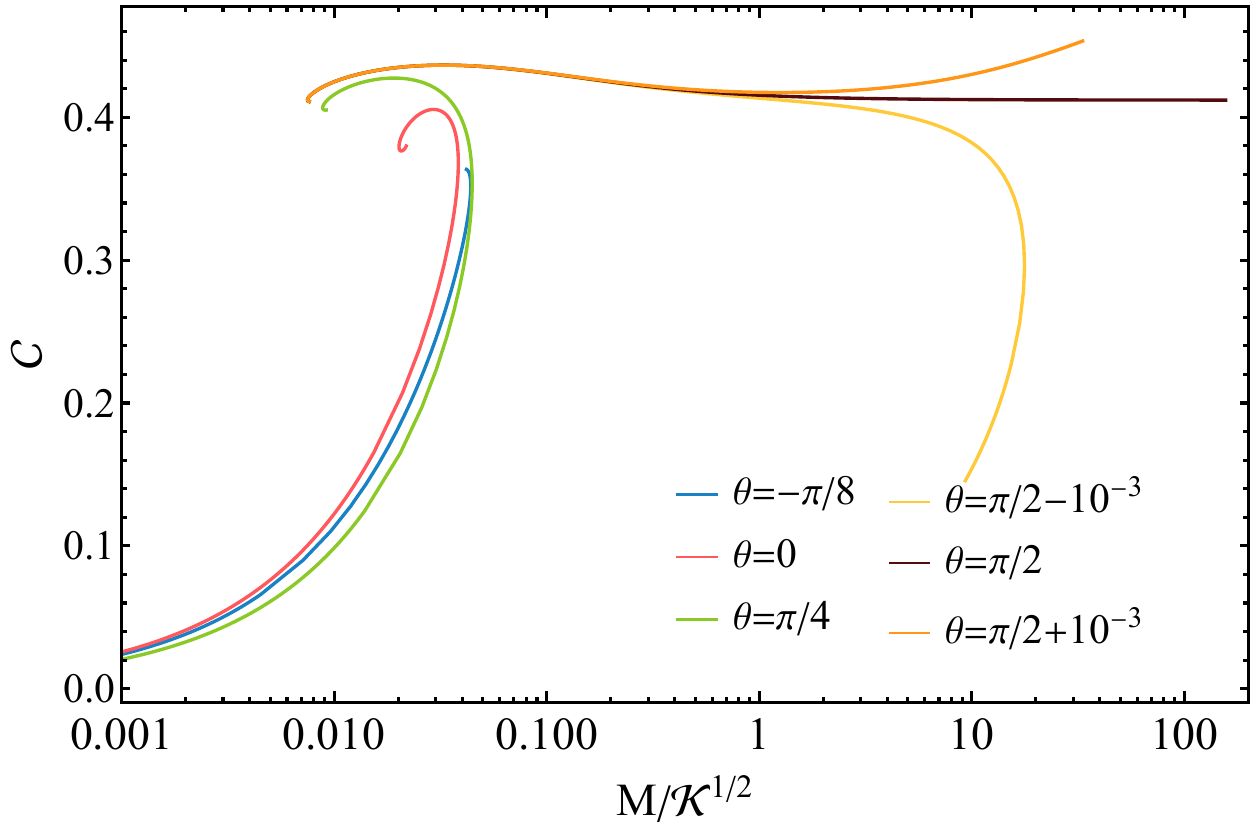} \\
	\caption{Mass-radius (left) and compactness-mass (right) diagram for the LCS material with $\gamma=2$, ${\cal E} =1$ and some representative values of $\theta$. \label{fig:MR_ultrarigid_prestressed}
	}
\end{figure}
In this section we are interested in the analysis of self-gravitating configurations for the LCS material. Interestingly, also in this case there are no solutions with a well-defined radius in the fluid limit (${\cal E}=0$), but when elasticity is turned on we can find bounded configurations. Some representative results are shown in the mass-radius and compactness-mass diagrams of Fig.~\ref{fig:MR_ultrarigid_prestressed}. 

Let us consider first the case where $0\leq\theta\leq \pi/2$, corresponding to both elastic parameters ${\cal E}_1$ and ${\cal E}_2$ being nonnegative (red, green, yellow and brown curves in Fig.~\ref{fig:MR_ultrarigid_prestressed}). From the first condition in~\eqref{eq:nec_conds_ultrarigid_prestressed}, we note that there is a (necessary) minimum central density required to have  physically admissible solutions, but there is no maximum. Increasing the central density from this minimum value up to a sufficiently large central density yields the curves in Fig.~\ref{fig:MR_ultrarigid_prestressed}, and our results suggest that within this range the solutions satisfy all conditions for physical viability. The matter profiles and non-trivial wave speeds for the two solutions in the boundary of this region are presented in the top panel of Fig.~\ref{fig:matter_LCS}. It is interesting to note that for $\theta=0$ (red curve) the solution exhibits a more standard mass-radius diagram, with a stable (right) branch where the mass increases with increasing central density up until a maximum mass, and a unstable (left) branch with decreasing mass for increasing central densities, whereas for  $\theta=\pi/2$ (brown curve) the stable branch vanishes and the solutions are all radially unstable. In general, solutions within this parameter range will behave similarly to $\theta=0$ (see green curve), unless $\theta$ is sufficiently close to $\pi/2$ (see yellow curve).

Let us now consider the region where $-\pi/2<\theta<0$. From conditions~\eqref{eq:nec_conds_ultrarigid_prestressed}, we find that physical viability imposes upper and lower limits on the central density. Another consequence is that there is a (necessary) critical minimum value of $\theta=\theta_1$ that allows for physically viable solutions. A representative solution in this region is shown as a blue curve in Fig.~\ref{fig:MR_ultrarigid_prestressed}, and its matter and sound speed profiles are in the bottom right panel of Fig.~\ref{fig:matter_LCS}. Interestingly, our analysis suggests that in this region conditions~\eqref{eq:nec_conds_ultrarigid_prestressed} are not only necessary but also sufficient to guarantee the physical viability of the solutions. 

In contrast with the previous behaviour, solutions with $\pi/2<\theta<\pi$ (orange curve) are more exotic and more complex to analyze. First, the necessary conditions for physical viability are no longer sufficient, since some conditions that are satisfied at the center can be violated elsewhere within the star (see e.g. dominant energy condition in the inset in the bottom right panel of Fig.~\ref{fig:matter_LCS}).  Our analysis also suggest that there is a critical value $\theta=\theta_2$ beyond which there are no physically viable solutions, and thus the range $\theta_1<\theta<\theta_2$ is the only physically relevant region to explore. The physically viable solutions are all unstable, similarly to case $\theta=\pi/2$. 
\begin{figure}[H]
	\centering
	\includegraphics[width=0.49\textwidth]{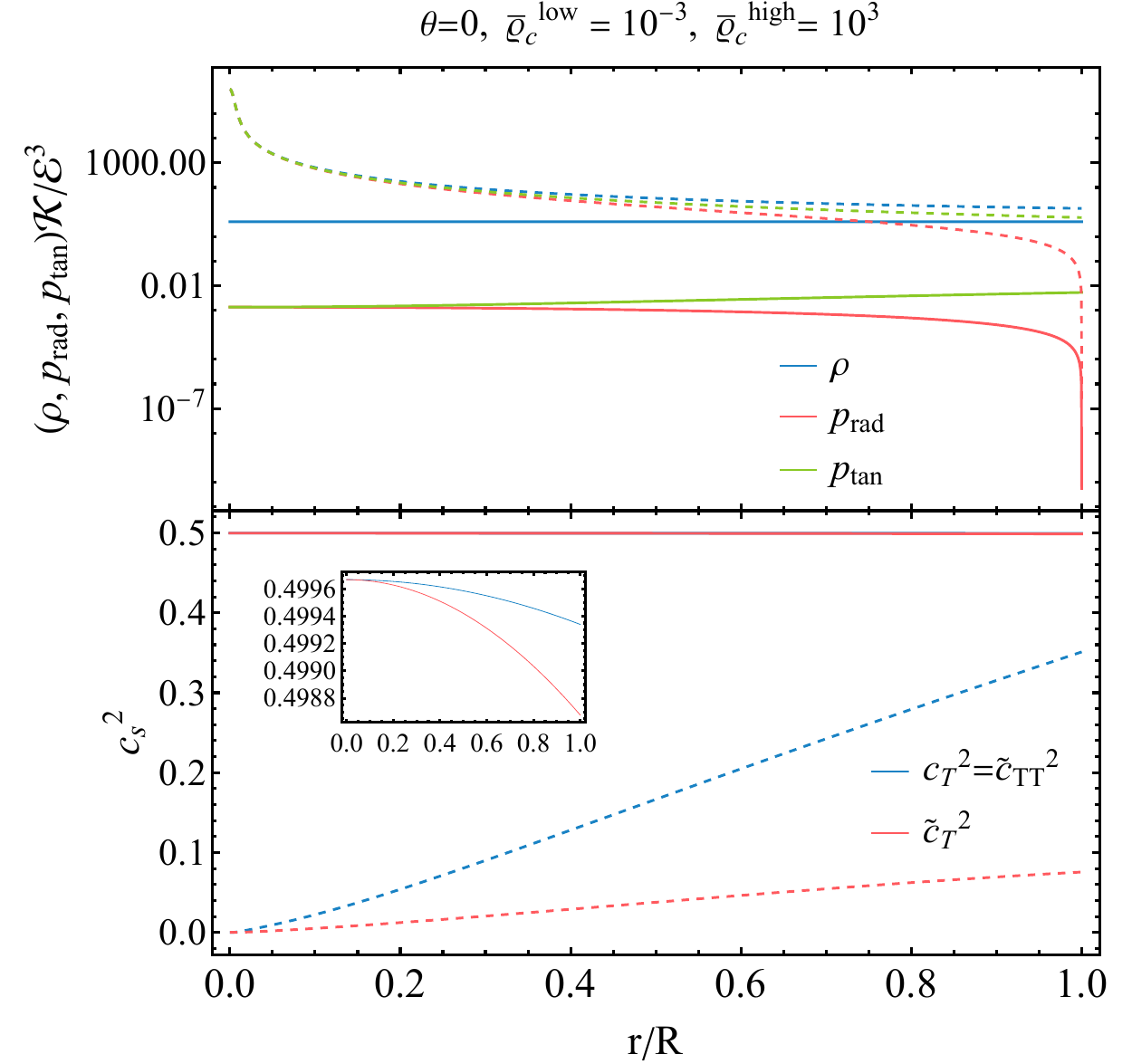}
 \includegraphics[width=0.49\textwidth]{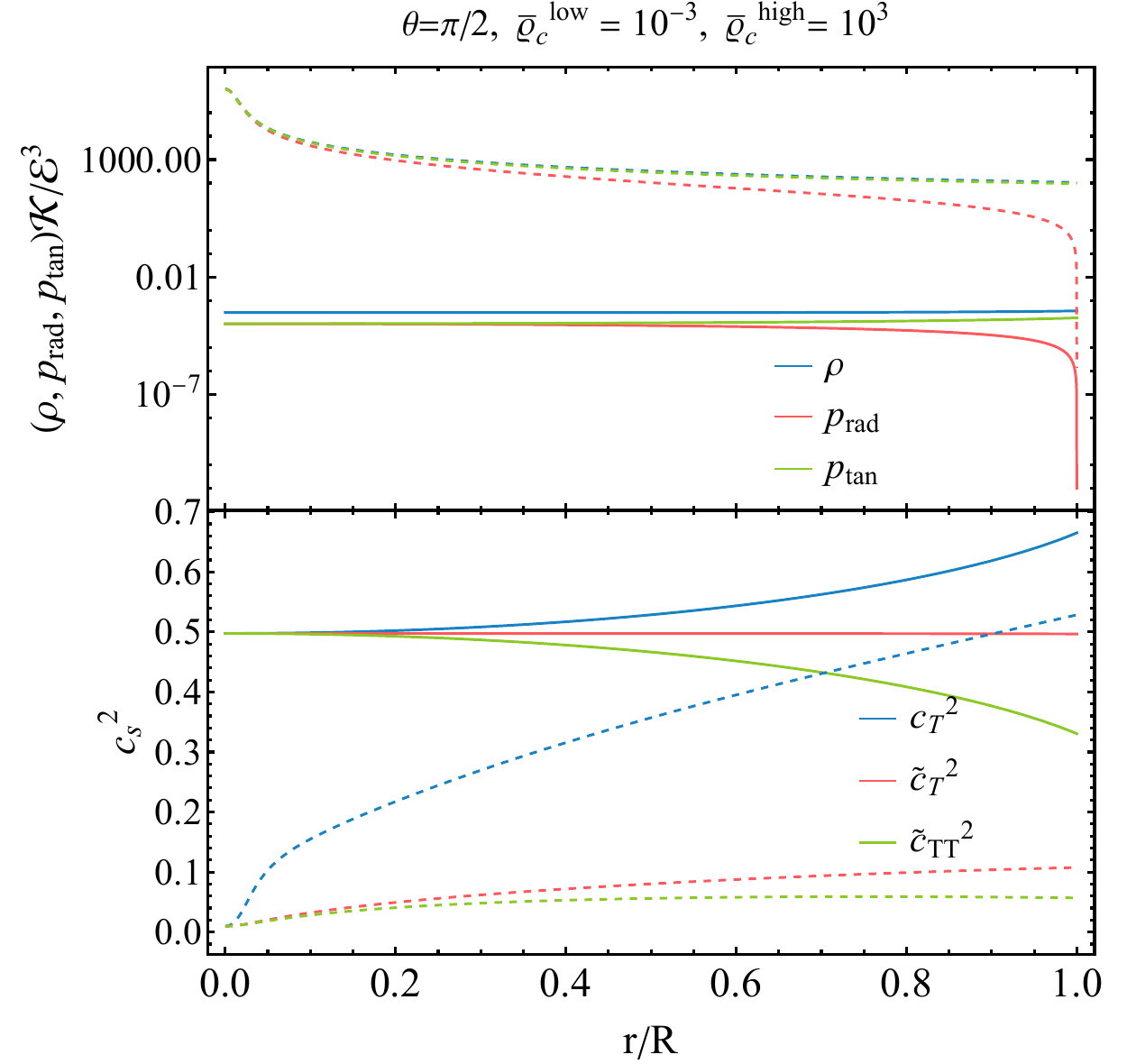}\\
\includegraphics[width=0.49\textwidth]{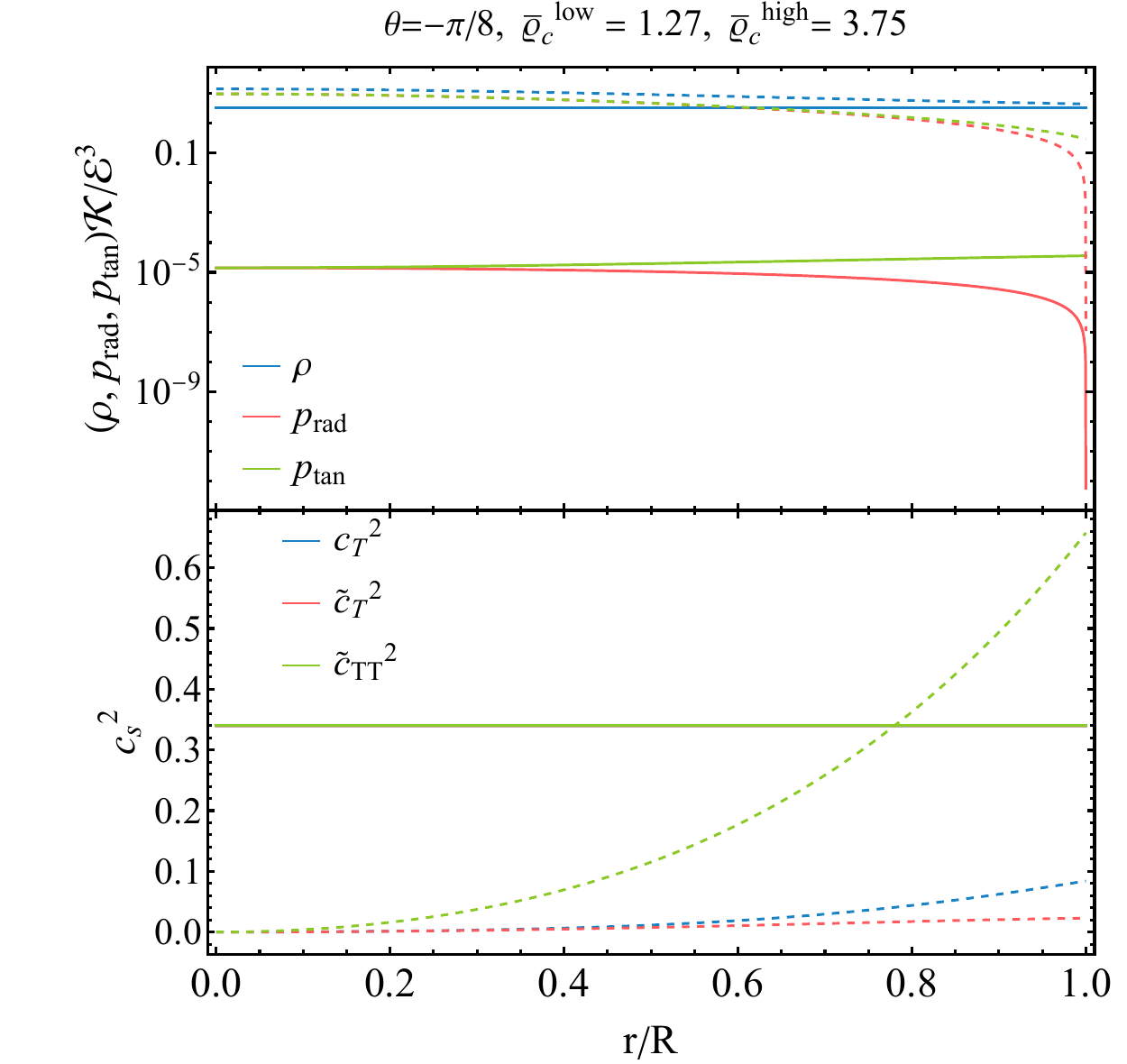}
  \includegraphics[width=0.49\textwidth]{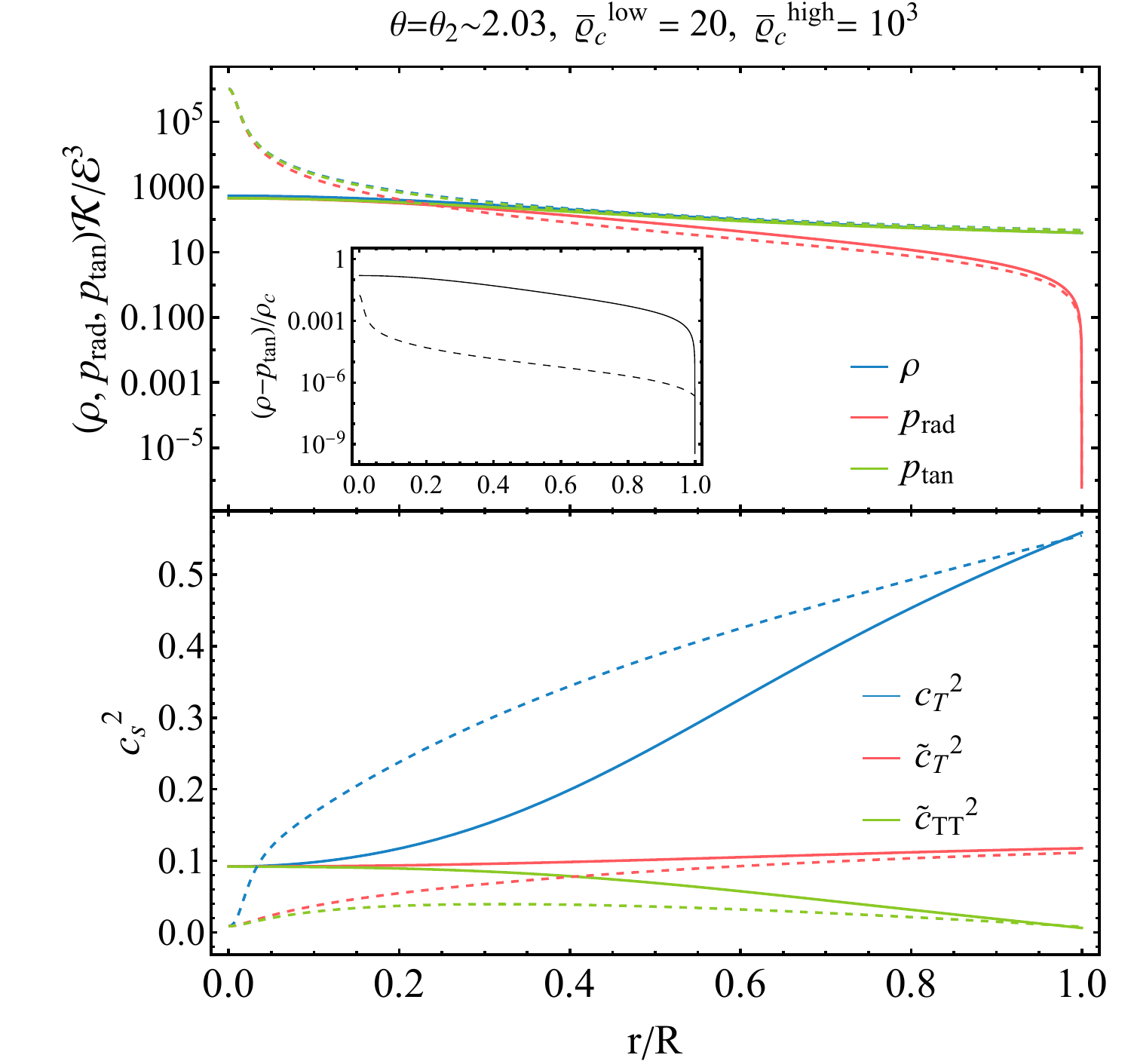}
	\caption{Matter and sound speed profiles of some characteristic configurations of the LCS model. Solid lines represent the profiles of a low central density configuration with $\varrho_c=\varrho_c^{\rm low}$, whereas the dashed lines represent a configuration with higher-central density~$\varrho_c=\varrho_c^{\rm high}$. The values of $\varrho_c^{\rm low}$ and  $\varrho_c^{\rm high}$ are, respectively, the lowest and highest values that lead to physically viable configurations (bottom row), or to the lowest and highest value of central density used in the numerical method (top row). \label{fig:matter_LCS}}
\end{figure}
%



We analyzed the maximum compactness of self-gravitating configurations of LCS material. The results are shown in Fig.~\ref{fig:C_max_prestressed}, where we focus exclusively on physically admissible (dashed curves) and radially stable physically admissible (solid line) configurations. As expected, we observe that (in general) the maximum compactness increases with the sound speed, as can be seen by comparing the red and blue curves in Fig.~\ref{fig:C_max_prestressed}.  In conclusion, we find the following bounds for the LCS model:
\begin{align}
 {\cal C}_{\rm PA}&\sim  0.451, \quad \theta \sim 1.67\,,\\
 {\cal C}_{\rm PAS}&\sim 0.368, \quad \theta \sim 0\,.
\end{align}
\begin{figure}[ht!]
	\centering
	\includegraphics[width=0.65\textwidth]{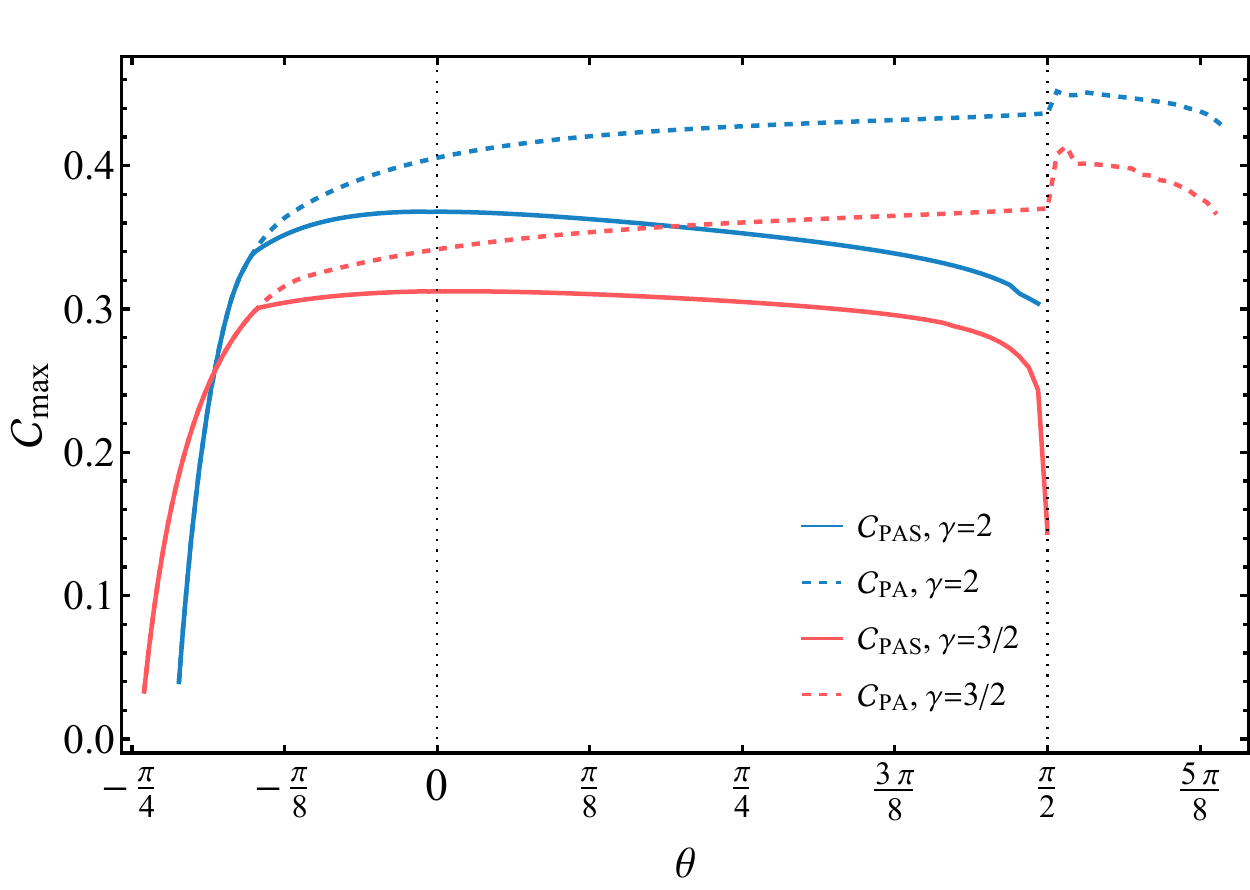} \\
	\caption{Maximum compactness of stars described by the LCS model for $\gamma =2$ (blue) and $\gamma=3/2$ (red). Dashed curves represent ${\cal C}_{\rm max}$ corresponding to physically admissible configurations, while solid lines represent physically admissible and stable configurations.\label{fig:C_max_prestressed}
	}
\end{figure}

\newpage

\section{Conclusions and outlook}
We have given a renewed overview of the problem of spherically symmetric elastic objects within GR.
Besides laying down the general formalism in all details, we have expanded the results presented in our recent papers~\cite{Alho:2021sli,Alho:2022bki,Alho:2023mfc}, and presented various new models (see Table~\ref{table} for a summary).
In particular, we have studied: 
\begin{enumerate}
    \item[i)] two models which deform the polytropic fluid EoS, namely QP and NSI, the latter allowing for a scale invariance in the Newtonian limit;
    \item[ii)] two models with constant wave speeds in the isotropic state, LIS and AIS, which are deformations of the linear and affine fluid EoS, respectively;
    \item[iii)] two models with a (single) constant longitudinal wave speed, ACS and LCS, which are deformations of the affine and linear fluid EoS, respectively.
\end{enumerate}
For all these models, we investigated the existence and physical viability (including the fulfillment of the energy conditions, absence of superluminal wave propagation, and radial stability) of static, spherically symmetric, self-gravitating elastic balls.

Our main results can be summarized as follows:
\begin{enumerate}
    \item Elasticity can allow for self-gravitating balls also when the latter do not exist in the corresponding fluid limit. As an example, we have found elastic stars which are deformations of fluid polytropic models with index $\mathrm{n}\geq5$;
    \item To the best of our knowledge, we found the first model of a relativistic star enjoying scale invariance, and found numerical evidence for its physical viability. Scale invariance implies absence of a maximum mass and a linear mass-radius relation, akin to the black hole case~\cite{Alho:2023mfc};
    \item Overall, for each model we have numerically determined the maximum compactness, which can generically exceed Buchdahl's bound and even reaches the black hole limit. However, imposing causality of wave propagation and the energy conditions reduces the maximum compactness of physically viable solutions, and radial stability reduces the maximum compactness even further;
    \item The maximum compactness of physically admissible, stable solutions was found for the ACS model, and reads ${\cal C}_{\rm PAS}^{\rm max}\approx 0.384$. Even if such model has an external light ring, its compactness is not high enough to support quasi-trapped modes and produce echoes in the ringdown~\cite{Cardoso:2016rao,Cardoso:2017cqb,Pani:2018flj}.
\end{enumerate}
It would be interesting to obtain existence theorems for static ball solutions for the present models as well as (dis)prove that the above compactness is the maximum allowed by any physically admissible elastic material in spherical symmetry.

We can envisage several important extensions of our work. First, a simple follow-up is to consider elastic deformations of the fluid case beyond the single polytropic EoS, for example by considering piecewise polytropes that approximate (tabulated) nuclear-physics based EoS~\cite{Read:2008iy}. Along a similar direction, one can attempt to build multilayer solutions, e.g. with a perfect fluid interior continuously connected with an outer elastic crust~\cite{Chamel:2008ca,Suleiman:2021hre,Raposo:2020yjy}.

A natural and important extension is to go beyond spherical symmetry, for example by allowing for rotating and/or deformed objects not allowed in the fluid limit~\cite{Raposo:2020yjy}, in an attempt to build more realistic models and investigate their phenomenology. 
For example, it would be interesting to assess whether elastic stars feature approximately universal (i.e., EoS independent) relations among their moment of inertia and multipole moments, as in the case of perfect fluid neutron stars~\cite{Yagi:2016bkt}, since these relations have important astrophysical consequences.
In this context, an important open problem consists in constructing a relativistic stored energy function without symmetry assumptions that reduces to the polytropic models and the models with constant wave speeds in the isotropic state, introduced in this work under spherical symmetry. In Appendix~\ref{existencenosym} it is shown that such stored energy functions do exist.

Finally, another outstanding extension concerns the nonlinear time evolution and the initial value problem for elastic objects, which would also be relevant to check if ultracompact elastic stars are unstable due to the presence of a inner photon sphere~\cite{Cardoso:2014sna,Cunha:2017qtt,Ghosh:2021txu}, as was recently shown for other specific models of exotic compact objects~\cite{Cunha:2022gde}. 
We have shown that the initial-value problem for elastic stars in spherical symmetry is strictly hyperbolic, at least in Minkowski spacetime. This provides a promising starting point to study the full nonlinear evolution of elastic stars within our framework, for example the coalesce of a binary system.
A well-posed time evolution beyond spherical symmetry would allow studying the tidal perturbations of elastic stars and their ringdown, both of which have direct consequences for gravitational wave astronomy~\cite{Cardoso:2019rvt,Maggio:2021ans}.

Among time-dependent solutions, the relativistic scale-invariant model introduced in this work allows for self-similar solutions. An interesting question concerns the formation of naked singularities arising from self-similar gravitational collapse, which are known to exist in the fluid case for soft equations of state, i.e. $\gamma\approx 1$~\cite{PhysRevLett.59.2137,PhysRevD.42.1068}. 
More generally, it would also be interesting to investigate the critical collapse of the elastic material models introduced in this work. This would generalize the previous studies of the collapse of perfect fluids, where the critical solution has continuous self-similarity. \cite{2007LRR....10....5G}.

\section*{Acknowledgments}
We thank Diogo Silva for carefully reading through the manuscript and suggesting several corrections. A.A.\ and J.N.\ were partially supported by FCT/Portugal through CAMGSD, IST-ID,
projects UIDB/04459/2020 and UIDP/04459/2020, and also by
the H2020-MSCA-2022-SE project EinsteinWaves, GA no.~101131233.
P.P.\ acknowledges financial support provided under the European Union's H2020 ERC, Starting 
Grant agreement no.~DarkGRA--757480, and under the MIUR PRIN and FARE programmes (GW-NEXT, CUP:~B84I20000100001), and support from the Amaldi Research Center funded by the MIUR program ``Dipartimento di Eccellenza" (CUP:~B81I18001170001). G.R.\ was supported by the Center for Research and Development in Mathematics and Applications (CIDMA) through the Portuguese Foundation for Science and Technology (FCT - Fundação para a Ciência e a Tecnologia), references UIDB/04106/2020 and UIDP/04106/2020, and by the 5th Individual CEEC program, reference 10.54499/2022.04182.CEECIND/CP1720/CT0004.

\newpage

\appendix

\section{Relativistic homogeneous and isotropic elastic bodies}\label{App:A}
There are several formulations of relativistic elasticity in the literature, starting with the foundational work of Carter \& Quintana~\cite{CarterQuintana}, and more recently by Kijowski \& Magli~\cite{KijMag92} and Beig \& Schmidt~\cite{Beig:2002pk}. Here we follow closely the formalism put forward by Beig \& Schmidt, 
to which we refer the reader for a more detailed exposition.

The configuration of a relativistic elastic body is described by a projection map (that is, a submersion) $\bm{\Pi}:\mathcal{S}\to\mathcal{B}$, mapping the (oriented, time-oriented) spacetime $(\mathcal{S},\bm{g})$ to the (oriented) 3-dimensional  Riemannian material space $(\mathcal{B},\bm{\gamma})$. The corresponding push-forward map $d\bm{\Pi}: \mathcal{T}\mathcal{S}\rightarrow \mathcal{T}\mathcal{B}$ is called the \emph{configuration gradient}. In local coordinates, $\bm{\Pi}$ is given by $X^{I}=\Pi^{I}(x^{\mu})$, and the configuration gradient by
\begin{equation}\label{DefT}
f^{\,I}_{\mu}=\partial_\mu \Pi^I,
\end{equation}
where Greek letters $\mu,\nu,\sigma,\lambda=0,1,2,3$ denote spacetime indices and capital Latin letters $I,J,K,L=1,2,3$ denote material indices.

The inverse images by $\bm{\Pi}$ of the points in $\mathcal{B}$ are assumed to form a congruence of timelike curves in $\mathcal{S}$, generated by a future-directed timelike vector field $\bm{u}$ satisfying $\bm{g}(\bm{u},\bm{u})=-1$ (the $4$-velocity of the particles making up the body), so that
\begin{equation}\label{Orthof}
u^{\mu}\partial_\mu \Pi^I =0.
\end{equation}  
Hence, $u^{\mu}$ spans the kernel of $f^{\,I}_{\mu}$, which has then full rank, determining an isomorphism between the orthogonal complement of $\bm{u}$ in $\mathcal{T}_{x}\mathcal{S}$ and $\mathcal{T}_{\Pi(x)}\mathcal{B}$ for each $x\in\mathcal{S}$.

We assume that there exists a volume form $\Omega_{IJK}$ on $\mathcal{B}$ which gives the number of particles in each subset of $\mathcal{B}$. It is related to the volume form $\eta_{IJK}$ induced by the material metric $\gamma_{IJ}$ through $\Omega_{IJK}=n_0 \eta_{IJK}$, where the function $n_0>0$ represents the \emph{number of particles density} in the reference state. The pull-back of $\Omega_{IJK}$,
\begin{equation}\label{STvol}
\omega_{\mu\nu\sigma}=f^{\,I}_{\mu}f^{\,J}_{\nu}f^{\,K}_{\sigma}\Omega_{IJK},
\end{equation}
defines a volume form on the tangent subspaces orthogonal to $u^\mu$.  
%
The dual vector field to $\omega_{\mu\nu\sigma}$ is proportional to $u^{\mu}$, and it is automatically conserved:
\begin{equation}\label{Current}
J^{\mu}=\frac{1}{3!}\epsilon^{\mu\nu\sigma\lambda}\omega_{\nu\sigma\lambda} = n u^\mu, \qquad \qquad  \nabla_{\mu}J^\mu=\nabla_\mu(n u^\mu)=0,
\end{equation}
where $\epsilon_{\mu\nu\sigma\lambda}$ is the spacetime volume form associated with $\bm{g}$. This conservation law, which is equivalent to the identity $d\bm{\omega} = d\bm{\Pi}^*\bm{\Omega} = \bm{\Pi}^*d\bm{\Omega} = \bm{0}$, justifies the interpretation of the function $n>0$ as the \emph{number of particles density}. 

A central role is played in relativistic elasticity by the push-forward of the inverse spacetime metric $g^{\mu\nu}$, denoted by $\bm{H}$, whose components 
\begin{equation}\label{DefH}
H^{IJ}=f^{\,I}_{\mu}f^{\,J}_{\nu}g^{\mu\nu}
\end{equation}
form a symmetric, positive-definite matrix, characterizing the deformations of the elastic body\footnote{Equivalently, this role can be played by the {\em strain}, defined as
 $E^{IJ}=\frac{1}{2} \left(H^{IJ}-\gamma^{IJ}\right)$.}. We denote by $\bm{H}^{-1}$ the inverse of this push-forward metric, with components $H_{IJ}$ such that $H_{IJ} H^{JK}=\delta^{K}_{I}$. The pull-back $\bm{\Pi}^*\bm{H}^{-1}$ is a Riemannian metric on the tangent subspaces orthogonal to $\bm{u}$, with components
\begin{equation}
h_{\mu\nu} = f^{\,I}_{\mu}f^{\,J}_\nu H_{IJ} =g_{\mu\nu}+u_\mu  u_\nu .
\end{equation}
As is usual, we can construct the orthogonal projection operator $h^{\mu}_{\nu}=g^{\mu\sigma}h_{\sigma\nu}$ on the tangent subspaces orthogonal to $\bm{u}$, i.e. $h^{\mu}_{\sigma}h^{\sigma}_\nu=h^{\mu}_\nu$. Moreover, since $f^{\,I}_\mu$ is an isomorphism when restricted to these subspaces, it has an inverse, given by $f^{\mu}_{\,I}=H_{IJ} g^{\mu\nu} f^{\,J}_{\nu}$. From the orthogonality condition~\eqref{Orthof} it follows that 
$H^{IJ}$ can in fact be written as
\begin{equation}
H^{IJ}=f^{\,I}_\mu f^{\,J}_\nu h^{\mu\nu} .
\end{equation}  
Defining
\begin{equation}
\mathcal{H}^{I}_{J}=H^{IK}\gamma_{KJ}, 
\end{equation}
it follows from~\eqref{STvol} and~\eqref{Current} that the number of particles density depends only on the configuration map $X^I=\Pi^I(x^\mu)$ via the number of particles density in the material space, $n_0$, and $\mathcal{H}^{I}_J$:
\begin{equation}
n=\sqrt{-J_\mu J^\mu}=n_{0}(\bm{\Pi})\sqrt{\det(\bm{\mathcal{H}})}.
\end{equation}
%
%
%
%
The dynamics of relativistic elastic bodies can be obtained from an action principle: variation of the action
\begin{equation}
S[\bm{\Pi}] =\int_{\mathcal{S}}\rho(\bm{\Pi},d\bm{\Pi})\sqrt{-\det(\bm{g})}\,d^4x,
\end{equation}
for a given Lagrangian density $\rho=\rho(\bm{\Pi},d\bm{\Pi})$, yields the Euler-Lagrange equations
\begin{equation}
\frac{\partial\rho}{\partial\Pi^I} -\frac{1}{\sqrt{-\det(\bm{g})}}\partial_\mu \left(\sqrt{-\det(\bm{g})}\frac{\partial \rho}{\partial f^{\,I}_{\mu}}\right)=0.
\end{equation}
The stress-energy tensor is then given by
\begin{equation}\label{SET}
T_{\mu\nu}=2\frac{\partial\rho}{\partial g^{\mu\nu}}-\rho g_{\mu\nu},
\end{equation}
and its divergence is
\begin{equation}\label{consenmom}
\nabla_{\mu}T^{\,\mu}_{\nu}=-\left(\frac{\partial\rho}{\partial\Pi^I} -\frac{1}{\sqrt{-\det(\bm{g})}}\partial_\mu \left(\sqrt{-\det(\bm{g})}\frac{\partial \rho}{\partial f^{\,I}_{\mu}}\right)\right) f^{\,I}_{\nu}.
\end{equation}
Therefore, the elastic field equations are satisfied if the stress-energy tensor is divergence-free. On the other hand,~\eqref{consenmom} implies that $u^{\nu}\nabla_\mu T^{\mu}_{\nu}=0$ always holds, and so the elastic field equations are actually equivalent to energy-momentum conservation.

The relativistic analogue of the material frame indifference in Newtonian elasticity theory consists in the assumption that $\rho$ is \emph{covariant under spacetime diffeomorphisms}, meaning that it should not change under arbitrary transformations of the spacetime coordinates. This assumption entails that $\rho$ should not depend on the configuration gradient explicitly, but rather we should have
\begin{equation}
\rho=\rho(\bm{\Pi},\bm{H}).
\end{equation}
Since $n$ also depends only on $\bm{\Pi}$ and $\bm{H}$, it is convenient to decompose $\rho$ as
\begin{equation}
\rho=n e,
\end{equation}
where $e$ is the \emph{total energy per particle}. We also define the \emph{relativistic stored energy function}
\begin{equation}
\epsilon = n_0 e,
\end{equation}
so that
\begin{equation}
\rho=\frac{n}{n_0} \epsilon.
\end{equation}

An important consequence of diffeomorphism invariance is that 
\begin{equation}\label{OrthoT}
\frac{\partial \rho}{\partial g^{\mu\nu}}u^{\mu}=0.
\end{equation}
From~\eqref{SET} and~\eqref{OrthoT} it then follows that
\begin{equation}
T_{\mu\nu}u^{\nu}=-\rho u_\mu,
\end{equation}
and so the stress-energy tensor can be uniquely decomposed as
\begin{equation}\label{DecT}
T_{\mu\nu}=\rho u_\mu u_\nu +\sigma_{\mu\nu}, \qquad \qquad \sigma_{\mu\nu}u^{\nu}=0,
\end{equation}
where $\sigma_{\mu\nu}$ is the symmetric \emph{Cauchy tensor}. Using the relation $g_{\mu\nu}=h_{\mu\nu}-u_\mu u_\nu$, it follows from~\eqref{SET} that
%
\begin{equation}\label{EOS}
\sigma_{\mu\nu}= 2 \frac{\partial\rho}{\partial g^{\mu\nu}}-\rho h_{\mu\nu}.
\end{equation}
We can also write the stress tensor as
\begin{equation}
\sigma_{\mu\nu}=2\frac{n}{n_0}\frac{\partial \epsilon}{\partial h^{\mu\nu}}= \frac{n}{n_0} \Sigma_{IJ}f^{\,I}_{\mu}f^{\,J}_{\nu},
\end{equation}
where
\begin{equation}
\Sigma_{IJ}= 2 \frac{\partial \epsilon}{\partial H^{IJ}}, 
\end{equation}
is the \emph{second Piola-Kirchhoff tensor}. Projecting the divergence of the stress-energy tensor in the directions parallel and orthogonal to $u^\mu$ yields the following equations of motion:
%
%
\begin{subequations}\label{ConsEqs}
	\begin{align}
	&u^\mu\partial_\mu\rho+\rho\theta+\sigma^{\sigma\lambda}\theta_{\sigma\lambda}=0, \\
	&\rho a^{\nu}+\sigma^{\nu}_{\mu}a^{\mu}+h^{\mu}_{\sigma}h^{\nu}_{\lambda}\nabla_\mu \sigma^{\sigma\lambda}=0,
	\end{align}	
\end{subequations}
where $a^{\nu}=u^\mu\nabla_\mu u^{\nu}$ is the $4$-acceleration, $\theta_{\mu\nu}=\nabla_{(\mu}u_{\nu)}$ is the shear and $\theta=g^{\mu\nu}\theta_{\mu\nu}=h^{\mu\nu}\theta_{\mu\nu}$ is the expansion. Given the equation of state~\eqref{EOS}, the first equation is identically satisfied, reducing the elastic field equations to the components orthogonal to $u^\mu$.

A relativistic elastic material is said to be \emph{homogeneous} if neither $n$ nor $\rho$ depend on the position $X^I$, but solely on the deformation $\mathcal{H}^{I}_J$. Homogeneity therefore implies that $n_0$ is a positive constant.

A material is said to be \emph{isotropic} if $\rho$ depends on the deformation $\bm{\mathcal{H}}$ via its principal invariants, 
\begin{equation}
\rho=\rho (i_1(\bm{\mathcal{H}}),i_2(\bm{\mathcal{H}}),i_3(\bm{\mathcal{H}})),
\end{equation}
where the invariants are given in terms of the eigenvalues $h_1$, $h_2$ and $h_3$ of $\bm{\mathcal{H}}$ as
\begin{subequations}\label{PrincInv}
	\begin{align}
	i_1(\bm{\mathcal{H}}) &=h_1+h_2+h_3,
	\\
	i_2(\bm{\mathcal{H}}) &=h_1h_2+h_1h_3+h_2h_3, 
	\\
	i_3(\bm{\mathcal{H}}) &=h_1h_2h_3.
	\end{align}
\end{subequations}
These eigenvalues are positive, and can be seen as the squares of the normalized linear particle densities $n_1$, $n_2$ and $n_3$ along the principal directions spanned by the eigenvectors $e^{J}_{(i)}$, $i=1,2,3$,  since
\begin{equation}\label{PLD}
\frac{n}{n_0}=\sqrt{\det(\bm{\mathcal{H}})}=\sqrt{h_1 h_2 h_3}=n_1 n_2 n_3.
\end{equation}
Moreover, normalizing the eigenvectors $e^{J}_{(i)}$ with respect to the metric $H_{IJ}$, we can write
\begin{equation}
\begin{split}
\Sigma_{IJ}&=2\frac{\partial \epsilon}{\partial\mathcal{H}^{(I}_{K}}\gamma_{J)K}=2\sum_{i=1}^{3}e_{(i)K}e_{(i)(I}\sum_{j=1}^{3}e_{(j)J)}e^{K}_{(j)}\frac{\partial \epsilon}{\partial h_j} \\
&=2\sum_{i=1}^{3} h_i e_{(i)I}e_{(i)J}\frac{\partial \epsilon}{\partial n_i}\frac{\partial n_i}{\partial h_i}=\sum_{i=1}^{3}n_i\frac{\partial \epsilon}{\partial n_i}e_{(i)I}e_{(i)J}
    \end{split}
\end{equation}
(where we used round brackets to indicate symmetrization over the indices $I$ and $J$). Denoting by $e_{(i)\mu}$ the pull-back of the reference orthonormal coframe $e_{(i)I}$, with components
\begin{equation}
e_{(i)\mu}=f^{I}_{\mu}e_{(i)I},
\end{equation}
we see that the Cauchy tensor is diagonal,
\begin{equation}\label{CSPP}
\sigma_{\mu\nu} = \sum^{3}_{i=1} p_{i} e_{(i)\mu}e_{(i)\nu},\qquad \sigma^{\mu}_\nu e^{\nu}_{(i)}=p_i e^{\mu}_{(i)},
\end{equation}
where the principal pressures $p_i=p_i(n_1,n_2,n_3)$ are given explicitly by
\begin{equation}\label{pp}
p_i=\frac{n}{n_0} n_i \frac{\partial \epsilon}{\partial n_i}.
\end{equation}
In terms of $\rho$, this equation becomes
\begin{equation}\label{pprho}
p_i=n_i \frac{\partial\rho}{\partial n_i}-\rho.
\end{equation}
Due to the symmetries of the principal invariants on the principal linear densities, the energy density $\rho(n_1,n_2,n_3)$ is invariant under all permutations of their variables $n_i$, and the principal pressures obey the symmetry conditions
\begin{equation}
\begin{split}
&p_1(n_1,n_2,n_3)=p_2(n_2,n_1,n_3)=p_3(n_3,n_2,n_1)=p_1(n_1,n_3,n_2), \\
&p_2(n_1,n_2,n_3)=p_3(n_1,n_3,n_2)=p_1(n_2,n_1,n_3)=p_2(n_3,n_2,n_1), \\
&p_3(n_1,n_2,n_3)=p_2(n_1,n_3,n_2)=p_1(n_3,n_2,n_1)=p_3(n_2,n_1,n_3).
\end{split}
\end{equation}
%
%
As shown by Karlovini \& Samuelsson~\cite{Karlovini:2002fc} (see also~\cite{Car73}), there are exactly 9 independent wave speeds, corresponding to \emph{longitudinal} waves in $i$-th direction,
\begin{equation}
c^2_{\mathrm{L}i} =\frac{\displaystyle n_i \frac{\partial p_i}{\partial n_i}}{\rho+p_i}=\frac{\displaystyle n^2_i \frac{\partial^2\rho}{\partial n^2_i}}{\rho+p_i},
\end{equation} 
and to \emph{transverse} waves in $i$-th direction, oscillating in the $j$-th direction:
\begin{equation}\label{TWaves}
c^2_{\mathrm{T}ij} = 
\begin{cases}
\displaystyle\frac{n^2_j }{\rho+p_j}\frac{(p_j-p_i)}{n^2_j-n^2_i}, & \text{if}\ n_i\neq n_j ,\\ \\
\displaystyle\frac{\frac{1}{2}n_j}{\rho+p_j}\left(\frac{\partial p_i}{\partial n_i}-\frac{\partial p_j}{\partial n_i}\right), & \text{if}\ n_i = n_j .
\end{cases}
\end{equation}

\subsection{The degenerate case of two equal principal linear densities}\label{Deg}
If two of the principal linear densities are equal, say $n_2=n_3=n_\mathrm{t}$, then we have
\begin{equation}
p_2(n_1,n_\mathrm{t},n_\mathrm{t})=p_3(n_1,n_\mathrm{t},n_\mathrm{t})=p_1(n_\mathrm{t},n_1,n_\mathrm{t}) .
\end{equation}
In addition, the derivatives must satisfy the symmetry conditions
\begin{equation}
\left.\frac{\partial p_1}{\partial{n_2}}\right|_{(n_1,n_\mathrm{t},n_\mathrm{t})}=\left.\frac{\partial p_1}{\partial{n_3}}\right|_{(n_1,n_\mathrm{t},n_\mathrm{t})} ,\quad
\left.\frac{\partial p_2}{\partial{n_2}}\right|_{(n_1,n_\mathrm{t},n_\mathrm{t})}=\left.\frac{\partial p_3}{\partial{n_3}}\right|_{(n_1,n_\mathrm{t},n_\mathrm{t})},\quad \left.\frac{\partial p_2}{\partial{n_3}}\right|_{(n_1,n_\mathrm{t},n_\mathrm{t})}=\left.\frac{\partial p_3}{\partial{n_2}}\right|_{(n_1,n_\mathrm{t},n_\mathrm{t})}
\end{equation}
and similarly for higher-order derivatives. In this case there are only two independent longitudinal wave speeds,
\begin{subequations}\label{longdeg}
\begin{align}
&c^2_{\mathrm{L}1}(n_1,n_\mathrm{t},n_\mathrm{t})=\left[\frac{n_1}{\rho+p_1}\frac{\partial p_1}{\partial{n_1}}\right]_{(n_1,n_\mathrm{t},n_\mathrm{t})} ,\label{longdeg1}\\ &c^2_{\mathrm{L}2}(n_1,n_\mathrm{t},n_\mathrm{t})=\left[\frac{n_2}{\rho+p_2}\frac{\partial p_2}{\partial{n_2}}\right]_{(n_1,n_\mathrm{t},n_\mathrm{t})}=\left[\frac{n_3}{\rho+p_3}\frac{\partial p_3}{\partial{n_3}}\right]_{(n_1,n_\mathrm{t},n_\mathrm{t})}=c^2_{\mathrm{L}3}(n_1,n_\mathrm{t},n_\mathrm{t}),\label{longdeg2}
\end{align}
\end{subequations}
and only 3 independent transverse wave speeds,
\begin{subequations}
\begin{align}
&c^2_{\mathrm{T}12}(n_1,n_\mathrm{t},n_\mathrm{t})=\left[\frac{n^2_2}{\rho+p_2}\frac{(p_2-p_1)}{n^2_2-n^2_1}\right]_{(n_1,n_\mathrm{t},n_\mathrm{t})}=\left[\frac{n^2_3}{\rho+p_3}\frac{(p_3-p_1)}{n^2_3-n^2_1}\right]_{(n_1,n_\mathrm{t},n_\mathrm{t})} =c^2_{\mathrm{T}13}(n_1,n_\mathrm{t},n_\mathrm{t}), \label{cT12}\\
&c^2_{\mathrm{T}21}(n_1,n_\mathrm{t},n_\mathrm{t})=\left[\frac{n^2_1}{\rho+p_1}\frac{(p_1-p_2)}{n^2_1-n^2_2}\right]_{(n_1,n_\mathrm{t},n_\mathrm{t})}=\left[\frac{n^2_1}{\rho+p_1}\frac{(p_1-p_3)}{n^2_1-n^2_3}\right]_{(n_1,n_\mathrm{t},n_\mathrm{t})}=c^2_{\mathrm{T}31}(n_1,n_\mathrm{t},n_\mathrm{t}), \label{cT21}\\
&c^2_{\mathrm{T}23}(n_1,n_\mathrm{t},n_\mathrm{t})=\left[\frac{\frac{1}{2}n_3}{\rho+p_3}\left(\frac{\partial p_2}{\partial n_2}-\frac{\partial p_3}{\partial n_2}\right)\right]_{(n_1,n_\mathrm{t},n_\mathrm{t})}=\left[\frac{\frac{1}{2}n_2}{\rho+p_2}\left(\frac{\partial p_3}{\partial n_3}-\frac{\partial p_2}{\partial n_3}\right)\right]_{(n_1,n_\mathrm{t},n_\mathrm{t})}=c^2_{\mathrm{T}32}(n_1,n_\mathrm{t},n_\mathrm{t}). \label{cT23}
\end{align}
\end{subequations} 

\subsection{Isotropic states, hydrostatic stress, and perfect fluids}
In an \emph{isotropic state} $n_1=n_2=n_3=n_\mathrm{c}=(n/n_0)^{1/3}$, the principal pressures satisfy
\begin{equation}
p_{1}(n_\mathrm{c},n_\mathrm{c},n_\mathrm{c})=p_2(n_\mathrm{c},n_\mathrm{c},n_\mathrm{c})=p_3(n_\mathrm{c},n_\mathrm{c},n_\mathrm{c}),
\end{equation}
and their derivatives must satisfy the symmetry conditions
\begin{equation} \label{sym_isotropic_state}
\begin{split}
&\left.\frac{\partial p_1}{\partial{n_1}}\right|_{(n_\mathrm{c},n_\mathrm{c},n_\mathrm{c})}=\left.\frac{\partial p_2}{\partial{n_2}}\right|_{(n_\mathrm{c},n_\mathrm{c},n_\mathrm{c})}=\left.\frac{\partial p_3}{\partial{n_3}}\right|_{(n_\mathrm{c},n_\mathrm{c},n_\mathrm{c})} ,\\
&\left.\frac{\partial p_1}{\partial{n_2}}\right|_{(n_\mathrm{c},n_\mathrm{c},n_\mathrm{c})}=\left.\frac{\partial p_2}{\partial{n_1}}\right|_{(n_\mathrm{c},n_\mathrm{c},n_\mathrm{c})}=\left.\frac{\partial p_2}{\partial{n_3}}\right|_{(n_\mathrm{c},n_\mathrm{c},n_\mathrm{c})}=\left.\frac{\partial p_3}{\partial{n_2}}\right|_{(n_\mathrm{c},n_\mathrm{c},n_\mathrm{c})}=\left.\frac{\partial p_3}{\partial{n_1}}\right|_{(n_\mathrm{c},n_\mathrm{c},n_\mathrm{c})}=\left.\frac{\partial p_1}{\partial{n_3}}\right|_{(n_\mathrm{c},n_\mathrm{c},n_\mathrm{c})}.
\end{split}
\end{equation}
It follows that in an isotropic state all longitudinal wave speeds coincide,
\begin{equation}
c^2_{\mathrm{L}}=\left[\frac{n_1}{\rho+p_1}\frac{\partial p_1}{\partial{n_1}}\right]_{(n_\mathrm{c},n_\mathrm{c},n_\mathrm{c})}=\left[\frac{n_2}{\rho+p_2}\frac{\partial p_2}{\partial{n_2}}\right]_{(n_\mathrm{c},n_\mathrm{c},n_\mathrm{c})}=\left[\frac{n_3}{\rho+p_3}\frac{\partial p_3}{\partial{n_3}}\right]_{(n_\mathrm{c},n_\mathrm{c},n_\mathrm{c})},
\end{equation}
and all transverse wave speeds also coincide:
\begin{equation}
\begin{split}
&c^2_{\mathrm{T}}=\left[\frac{\frac{1}{2}n_1}{\rho+p_1}\left(\frac{\partial p_2}{\partial n_2}-\frac{\partial p_1}{\partial n_2}\right)\right]_{(n_\mathrm{c},n_\mathrm{c},n_\mathrm{c})}= \left[\frac{\frac{1}{2}n_1}{\rho+p_1}\left(\frac{\partial p_3}{\partial n_3}-\frac{\partial p_1}{\partial n_3}\right)\right]_{(n_\mathrm{c},n_\mathrm{c},n_\mathrm{c})} \\
&\quad=\left[\frac{\frac{1}{2}n_2}{\rho+p_2}\left(\frac{\partial p_1}{\partial n_1}-\frac{\partial p_2}{\partial n_1}\right)\right]_{(n_\mathrm{c},n_\mathrm{c},n_\mathrm{c})}= \left[\frac{\frac{1}{2}n_2}{\rho+p_2}\left(\frac{\partial p_3}{\partial n_3}-\frac{\partial p_2}{\partial n_3}\right)\right]_{(n_\mathrm{c},n_\mathrm{c},n_\mathrm{c})} \\
&\quad=\left[\frac{\frac{1}{2}n_3}{\rho+p_3}\left(\frac{\partial p_1}{\partial n_1}-\frac{\partial p_3}{\partial n_1}\right)\right]_{(n_\mathrm{c},n_\mathrm{c},n_\mathrm{c})}= \left[\frac{\frac{1}{2}n_3}{\rho+p_3}\left(\frac{\partial p_2}{\partial n_2}-\frac{\partial p_3}{\partial n_2}\right)\right]_{(n_\mathrm{c},n_\mathrm{c},n_\mathrm{c})}.
\end{split}
\end{equation}
Moreover, the stress is just the \emph{hydrostatic stress} $\sigma_{\mu\nu}=p_\mathrm{iso}(n)h_{\mu\nu}$, where $p_\mathrm{iso}(n)=p_i(n_\mathrm{c},n_\mathrm{c},n_\mathrm{c})$ is the \emph{isotropic pressure}. From $n=n_\mathrm{c}^3$ we obtain
\begin{equation}
n\frac{dp_\mathrm{iso}}{dn}(n)=\frac{n_\mathrm{c}}{3}\left.\left(\frac{\partial p_i}{\partial n_1}+\frac{\partial p_i}{\partial n_2}+\frac{\partial p_i}{\partial n_3}\right)\right|_{(n_\mathrm{c},n_\mathrm{c},n_\mathrm{c})},
\end{equation}
and using~\eqref{pprho} we find
\begin{equation}
n\frac{d\rho}{dn}(n)=\rho(n)+p_\mathrm{iso}(n),
\end{equation}
so that
\begin{equation}
\frac{dp_\mathrm{iso}}{d\rho}=c^2_{\mathrm{L}}-\frac{4}{3}c^2_{\mathrm{T}}.
\end{equation}
Assuming that $dp_\mathrm{iso}/d\rho>0$, we must have $c^2_{\mathrm{L}}>\frac{4}{3}c^2_{\mathrm{T}}$. 
In the case of a perfect fluid we have $c^2_\mathrm{T}=0$, and $c^2_{\mathrm{L}}=dp_\mathrm{iso}/d\rho$ is simply the squared speed of sound.

\subsection{The reference state, elastic moduli, and invariance under renormalization}
The reference state corresponds to the absence of deformation, $n_1=n_2=n_3=1$. We assume that
\begin{equation}
\rho_0 = \rho(1,1,1) = \epsilon(1,1,1)
\end{equation}
is a positive constant. The energy due to the deformation of the body can then be encoded in a \emph{deformation potential energy density function} $w(n_1,n_2,n_3)$. 
%
%
\begin{definition}[Relativistic and ultra-relativistic materials]\label{DefURM} 
A material is said to be \emph{relativistic} if the stored energy function $\epsilon$ consists of the sum of a positive constant reference state mass density $\varrho_0=\mathfrak{m}n_0$, where $\mathfrak{m}$ is the rest mass of the particles making up the body, and the deformation potential density:
\begin{equation}
\epsilon(n_1,n_2,n_3)= \varrho_0 + w(n_1,n_2,n_3).
\end{equation}
A material is said to be \emph{ultra-relativistic} if the stored energy function $\epsilon(n_1,n_2,n_3)$ consists only in the deformation potential density:
\begin{equation}
\epsilon(n_1,n_2,n_3)= w(n_1,n_2,n_3).
\end{equation}
\end{definition}
Therefore, the energy density of a relativistic material is given by
\begin{equation}\label{GenRED}
\rho =\varrho + \frac{n}{n_0} w(n_1,n_2,n_3),
\end{equation}
where $\varrho = \mathfrak{m}n$ is the \emph{baryonic mass density}, whereas the energy density of a ultra-relativistic material is given by
\begin{equation}\label{GenURED}
\rho =\frac{n}{n_0}w(n_1,n_2,n_3).
\end{equation}
\begin{definition}[Stress-free natural reference state] 
The reference state is said to be a \emph{stress-free natural reference state} if the principal pressures satisfy the \emph{stress-free reference state condition}
\begin{equation}
p_i(1,1,1)=0, \qquad i=1,2,3,
\end{equation}
and the deformation potential energy density satisfies the \emph{natural reference state condition} 
\begin{equation}
w(1,1,1)=0.
\end{equation}
\end{definition}
The natural reference state condition asserts that the undeformed reference state is a state of zero potential energy. Therefore, these materials are never ultra-relativistic, and so the reference state energy density $\rho_0$ is just the baryonic mass density $\varrho_0$. Since~\eqref{pprho} is equivalent to 
\begin{equation}
p_i=\frac{n}{n_0} n_i \frac{\partial w}{\partial n_i},
\end{equation}
we see that a stress-free reference state is a stationary point of $w(n_1,n_2,n_3)$, and compatibility with linear elasticity yields the well-known conditions
\begin{equation}
\frac{\partial^2 w}{\partial n_i\partial n_j}(1,1,1) 
=\lambda+2\mu\delta_{ij} \, , \qquad i,j=1,2,3,
\end{equation}
where $\lambda$ and $\mu$ are the \emph{Lam\'e constants}. It is easy to check that the reference state is a nondegenerate local minimum of $w$ if and only if $\mu>0$ and $\lambda + \frac23 \mu > 0$. 

It follows that the first order derivatives of the principal pressures satisfy the reference state compatibility conditions
\begin{equation}\label{PPLames}
\frac{\partial p_i}{\partial n_j}(1,1,1)=\lambda+2\mu\delta_{ij},
\end{equation}
and the wave speeds satisfy the reference state conditions
\begin{equation}\label{RSws0}
c^2_{\mathrm{L}i}(1,1,1)=\frac{\lambda+2\mu}{\rho_0},\qquad c^2_{\mathrm{T}ij}(1,1,1)=\frac{\mu}{\rho_0},
\end{equation}
and
\begin{equation}
\frac{dp_\mathrm{iso}}{d\rho}(1,1,1)=c^2_{\mathrm{L}i}(1,1,1)-\frac{4}{3}c^2_{\mathrm{T}ij}(1,1,1)=\frac{\lambda+\frac{2}{3}\mu}{\rho_0}.
\end{equation}
\begin{remark}\label{modulicond}
	Linear elasticity is fully characterized by two elastic constants. Other elastic constants of interest are the \emph{p-wave modulus} $L$, the \emph{Poisson ratio} $\nu$, the \emph{Young modulus} $E$, and the \emph{bulk modulus} $K$, given in terms of the Lam\'e parameters by, respectively,
	\begin{equation}
	L=\lambda+2\mu,\qquad \nu=\frac{\lambda}{2(\lambda+\mu)},\qquad E=\frac{(3\lambda+2\mu)\mu}{\lambda+\mu},\qquad K=\lambda+\frac{2}{3}\mu.
	\end{equation} 
  The standard physical admissibility conditions 
  \begin{equation}
      \mu\geq0,\qquad K>0
  \end{equation}
  imply
  \begin{equation}
      L>0,\qquad E\geq0,\qquad \nu\in\biggl(-1,\frac12\,\biggr],
  \end{equation}
  where $\nu=\frac{1}{2}$ and $E=0$ for $\mu=0$, while the lower bound on $\nu$ arises from the condition $K>0$.
\end{remark}
\begin{definition}[Pre-stressed reference state]
The reference state is said to be \emph{pre-stressed} if the principal pressures satisfy the \emph{pre-stressed reference state condition}
\begin{equation}
p_i(1,1,1)=p_0, \qquad i=1,2,3,
\end{equation}
for some constant $p_0\neq0$, and the potential energy density satisfies
\begin{equation}
    w(1,1,1)=w_0
\end{equation}
for some constant $w_0\neq 0$. 
\end{definition}
For relativistic materials, the reference constant energy density consists of the sum $\rho_0= \varrho_0+w_0$, and so $w_0$ can be see as the difference between the rest energy density and the rest mass density; in particular, the stress-free case corresponds to $w_0=0$. For ultra-relativistic materials, where the rest baryonic mass density does not contribute to the stored energy function, we simply have $\rho_0=w_0$. 
In the pre-stressed case, $\lambda$ and $\mu$ do not have a well defined meaning, since the reference state is arbitrary (it is not an equilibrium state, and so there is no particular property singling it out). Nevertheless, one can still take the relation~\eqref{PPLames} to be valid, independently of $p_0$, so that the wave speeds satisfy the reference state conditions
\begin{equation}\label{RSws}
c^2_{\mathrm{L}i}(1,1,1)=\frac{\lambda+2\mu}{\rho_0+p_0},\qquad c^2_{\mathrm{T}ij}(1,1,1)=\frac{\mu}{\rho_0+p_0},
\end{equation}
and
\begin{equation}
\frac{dp_\mathrm{iso}}{d\rho}(1,1,1)=c^2_{\mathrm{L}i}(1,1,1)-\frac{4}{3}c^2_{\mathrm{T}ij}(1,1,1)=\frac{\lambda+\frac{2}{3}\mu}{\rho_0+p_0}=\frac{K}{\rho_0+p_0}.
\end{equation}
\begin{remark}[Invariance under renormalization of the reference state]
If we choose a new reference state which is compressed by a linear factor $f$ with respect to an original reference state, then the new linear particle densities $\tilde{n}_i$ are related to the original linear particle densities $n_i$ by
\begin{equation}
n_i=f\tilde{n}_i .
\end{equation}
The expression $\tilde\rho$ of the Lagrangian as a function of the new linear densities must of course satisfy
\begin{equation}
    \tilde{\rho}(\tilde{n}_1,\tilde{n}_2,\tilde{n}_3) = \rho(n_1,n_2,n_3)=\rho(f\tilde{n}_1,f\tilde{n}_2,f\tilde{n}_3).
\end{equation}
The new elastic parameters $\tilde\lambda$ and $\tilde\mu$ (and others) will in general differ from $\lambda$ and $\mu$, as they are calculated with respect to a different reference state. However, it is usually possible to construct combinations of these parameters which are invariant under the renormalization of the reference state. It is these combinations, rather than the standard elastic parameters, that carry physical meaning.
\end{remark}
 \begin{remark}
     For ultra-relativistic materials (Definition~\ref{DefURM}), there is an extra symmetry due to the invariance of $\rho$ under rescalings of the rest baryonic mass density,
     \begin{equation}
      \tilde{\varrho}_0    = f \varrho_0,
     \end{equation}
     while keeping $n_i$ fixed. Notice that this is quite different from renormalizing the reference state, which is simply a change in the description of a \emph{fixed} material: here we are modifying the material itself by adding more baryons per unit volume, so that its (gauge-invariant) elastic parameters will change.
 \end{remark}
%
\subsection{Some remarks about stress-free and pre-stressed reference state materials}


    %
      %
     %
     %

\begin{proposition}
The transformation
    \begin{equation}\label{Tw}
        w^{(\mathrm{ps})}(n_1,n_2,n_3)=w^{(\mathrm{sf})}(n_1,n_2,n_3)-\alpha_0-p_0(n_1n_2n_3)^{-1},
      \end{equation}
      where $\alpha_0=-(w_0+p_0)$ (with $p_0\neq0$ and $w_0=\rho_0-\varrho_0^{(\mathrm{ps})}\neq0$), takes any given stress-free natural reference state material to a 1-parameter family (parameterized by $p_0$) of pre-stressed reference state materials (or vice-versa). This changes the energy density according to
     \begin{equation}\label{RhoTransf}
\rho^{(\mathrm{ps})}(n_1,n_2,n_3)=\rho^{(\mathrm{sf})}(n_1,n_2,n_3)+\left(-\alpha_0+\varrho^{(\mathrm{ps})}_0-\varrho^{(\mathrm{sf})}_0\right)(n_1 n_2 n_3)-p_0,
\end{equation}
     where $\varrho^{(\mathrm{ps})}_0=\varrho^{(\mathrm{sf})}_0$ for relativistic materials, and $\varrho^{(\mathrm{ps})}_0=0$ for ultra-relativistic materials.
\end{proposition}

%
      %
%
%
%
%
\begin{proof}
By~\eqref{pprho}, the principal pressures satisfy
\begin{equation}
p^{(\mathrm{ps})}_i (n_1,n_2,n_3) = p^{(\mathrm{sf})}_i(n_1,n_2,n_3)+p_0.
\end{equation}
\end{proof}
\begin{definition}[The natural choice of the reference state pressure]
The natural choice of $p_0$ is the minimum value needed in order for the isotropic state condition
\begin{equation}
\widehat{p}^{(\mathrm{ps})}_{\mathrm{iso}}(n_\mathrm{c},n_\mathrm{c},n_\mathrm{c})>0\quad \text{for}\quad n_\mathrm{c}>0
\end{equation}
to hold.   
\end{definition}
%

\subsection{Examples of elastic materials}\label{ApEx}
\begin{examples}[Beig-Schmidt materials]\label{BS}
	The relativistic stored energy function of the Beig-Schmidt materials~\cite{Beig:2002pk} is
	\begin{equation}
	\begin{split}
	\epsilon(n_1,n_2,n_3)=\rho_0&+\frac{3}{8}(3\lambda+2\mu)+\frac{1}{8}(\lambda+2\mu)(n^2_1+n^2_2+n^2_3)^2-\frac{1}{4}(3\lambda+2\mu)(n^2_1+n^2_2+n^2_3) \\
	& -\frac{\mu}{2}(n^2_1n^2_2+n^2_1 n^2_3 +n^2_2 n^2_3).
	\end{split}
	\end{equation}
It is a stress-free natural reference state material, with $p_0=w_0=0$. 
\end{examples}
\begin{examples}[Quasi-Hookean materials]\label{QH}
An elastic material is said to be \emph{quasi-Hookean} if its energy density is given by
\begin{equation}\label{QHEq}
\rho(n_1,n_2,n_3)=\check{\rho}(n/n_0)+\check{\mu}(n/n_0)S(n_1,n_2,n_3),
\end{equation}
where $S$ is a \emph{shear scalar}, that is, a symmetric function of $(n_1,n_2,n_3)$ which is invariant under rescalings and vanishes on isotropic states. One possible choice for the shear scalar was given by Tahvildar-Zadeh in~\cite{Tah98}:
\begin{equation}\label{Tahshear}
S(n_1,n_2,n_3)=\frac{n^2_1+n^2_2+n^2_3}{(n_1 n_2 n_3)^{2/3}}-3.
\end{equation}
A different choice was introduced by Karlovini \& Samuelsson in~\cite{Karlovini:2002fc}:
\begin{equation}\label{KSshear}
S(n_1,n_2,n_3)=\frac{1}{12}\left[\left(\frac{n_1}{n_2}-\frac{n_2}{n_1}\right)^2+\left(\frac{n_1}{n_3}-\frac{n_3}{n_1}\right)^2+\left(\frac{n_2}{n_3}-\frac{n_3}{n_2}\right)^2\right].
\end{equation}

\begin{itemize}
    \item Unsheared relativistic polytropes:
\end{itemize}
A popular choice for the unsheared quantities is to consider the relativistic polytropic fluid,
\begin{equation}
    \check{\rho}\left(\frac{n}{n_0}\right)=\left(\rho_0-\frac{K}{\gamma(\gamma-1)}\right)\left(\frac{n}{n_0}\right)+\frac{K}{\gamma(\gamma-1)}\left(\frac{n}{n_0}\right)^{\gamma}
\end{equation}
and
\begin{equation}
    \check{\mu}\left(\frac{n}{n_0}\right)=k p_\mathrm{iso}=k\frac{K}{\gamma}\left(\frac{n}{n_0}\right)^{\gamma},
\end{equation}
where $\gamma\in(0,\infty)\setminus\{1\}$ is the \emph{adiabatic index}, $k$ is a dimensionless constant, given by
\begin{equation}
k_\mathrm{Tah}=\frac{\mu}{2p_0} \qquad \text{or} \qquad k_{\mathrm{KS}}=\frac{\mu}{p_0}
\end{equation}
for the Tahvildar-Zadeh or the Karlovini \& Samuelsson choices of the shear scalar, respectively, and
\begin{equation}
    p_0=\frac{K}{\gamma}.
\end{equation}

\begin{itemize}
\item Unsheared linear EoS:
\end{itemize}
    Other popular models consists in taking the limit
    \begin{equation}
     \varrho_0\rightarrow 0\quad \Leftrightarrow\quad   \rho_0\rightarrow\frac{K}{\gamma(\gamma-1)}
    \end{equation}
    for the unsheared polytropes, which belong to the class of ultra-relativistic materials with
    \begin{equation}
        \rho_0=w_0, \qquad w_0 =\frac{p_0}{\gamma-1}.
    \end{equation}
    Independently of the choice of the shear scalar $S(\delta,\eta)$, all these models are characterized by having a linear EoS in the isotropic state, 
    \begin{equation}
        p_\mathrm{iso}(n_\mathrm{c},n_\mathrm{c},n_\mathrm{c})=(\gamma-1)\rho(n_\mathrm{c},n_\mathrm{c},n_\mathrm{c}),
    \end{equation}
    and
    \begin{equation}
\frac{dp_\mathrm{iso}}{d\rho}=c^2_\mathrm{L}(n_\mathrm{c},n_\mathrm{c},n_\mathrm{c})-\frac{4}{3}c^2_\mathrm{T}(n_\mathrm{c},n_\mathrm{c},n_\mathrm{c})=\gamma-1.    
\end{equation}
Using the reference state identity
\begin{equation}\label{identity475}
    \rho_0+p_0 = \frac{K}{\gamma-1},
\end{equation}
the longitudinal and transverse wave speeds in the isotropic state are constant and satisfy
\begin{equation}
    c^2_\mathrm{L}(n_\mathrm{c},n_\mathrm{c},n_\mathrm{c})=\frac{L}{K}(\gamma-1),\qquad c^2_\mathrm{T}(n_\mathrm{c},n_\mathrm{c},n_\mathrm{c})=\frac{\mu}{K}(\gamma-1).
\end{equation}

\begin{itemize}
\item Unsheared affine EoS:
\end{itemize}
For the material with a stress-free reference state ($p_0=0$), the identity~\eqref{identity475} yields
\begin{equation}
    \rho_0=\frac{K}{\gamma-1},
\end{equation}
and the corresponding EoS in the isotropic state is
    \begin{equation}
        p_\mathrm{iso}(n_\mathrm{c},n_\mathrm{c},n_\mathrm{c})=(\gamma-1)\bigl(\rho(n_\mathrm{c},n_\mathrm{c},n_\mathrm{c})-\rho_0\bigr).
    \end{equation}
\end{examples}
\begin{remark}\label{ShearPol}
    Independently of the choice of the shear scalar $S$, in the limit $\mu\rightarrow 0^{+}$ the relativistic quasi-Hookean polytropes reduce to a fluid with polytropic EoS of Example~\ref{Ex1}, with the identifications 
\begin{equation}
 \mathrm{K}=\frac{\lambda}{\gamma n^{\gamma}_0},\qquad    \mathrm{C}=\left(\frac{\gamma }{\lambda}\right)^{1/\gamma}\left(\rho_0-\frac{\lambda}{\gamma(\gamma-1)}\right),
\end{equation}
and the quasi-Hookean material with unsheared linear EoS reduces to the fluid with linear EoS of Example~\ref{Ex2}, which has $\mathrm{C}=0$, that is, $\rho_0=\lambda/(\gamma(\gamma-1))$.
\end{remark}

\begin{examples}[Materials with constant longitudinal wave speeds]\label{KarlSam}
Materials with constant longitudinal wave speeds,
\begin{equation}
 c^2_{\mathrm{L}i}=\gamma-1,\qquad \gamma\in(1,2],
\end{equation}
were introduced by Karlovini \& Samuelsson in~\cite{Karlovini:2004ix}. For this family of materials, it follows from~\eqref{RSws} that
\begin{equation}\label{cLconst}
\rho_0+p_0 = \frac{\lambda+2\mu}{\gamma-1}.
\end{equation}
\begin{itemize}
    \item{Natural stress-free reference state materials:}
\end{itemize}
    The family of materials with constant longitudinal wave speeds and natural stress-free reference state has
\begin{equation}
 p_0=0, \qquad \rho_0=\frac{\lambda+2\mu}{\gamma-1}.   
\end{equation}
Its deformation potential is given by
\begin{equation}
\begin{split}
    w^{(\mathrm{sf})}=&-\frac{\lambda+2\mu}{\gamma-1}+\left(\frac{\lambda+2\mu}{\gamma}-\frac{4\mu}{\gamma^2}+C_\mathrm{KS}\right)(n_1 n_2 n_3)^{-1}+\left(\frac{\lambda+2\mu}{\gamma(\gamma-1)} - \frac{2\mu}{\gamma^2}-C_{\mathrm{KS}}\right) (n_1 n_2 n_3)^{\gamma-1} \\
    & + \left(\frac{2\mu}{\gamma^2}-C_{\mathrm{KS}}\right) \left(\frac{n_1^{\gamma-1}}{n_2 n_3} + \frac{n_2^{\gamma-1}}{n_1 n_3} + \frac{n_3^{\gamma-1}}{n_1 n_2}\right) + C_{\mathrm{KS}} \left[ \frac{(n_1 n_2)^{\gamma-1}}{n_3} + \frac{(n_2 n_3)^{\gamma-1}}{n_1} + \frac{(n_1 n_3)^{\gamma-1}}{n_2} \right],
\end{split}
\end{equation}
and its Lagrangian (energy density) is given by
\begin{equation}\label{KSSF}
\begin{split}
\rho^{(\mathrm{sf})}(n_1,n_2,n_3) =  & \,\, \frac{\lambda+2\mu}{\gamma}-\frac{4\mu}{\gamma^2}+C_{\mathrm{KS}}+ \left(\frac{\lambda+2\mu}{\gamma(\gamma-1)} - \frac{2\mu}{\gamma^2}-C_{\mathrm{KS}}\right) (n_1 n_2 n_3)^\gamma \\ 
& + \left(\frac{2\mu}{\gamma^2}-C_{\mathrm{KS}}\right) (n_1^\gamma + n_2^\gamma + n_3^\gamma) + C_{\mathrm{KS}} \left[ (n_1 n_2)^\gamma + (n_2 n_3)^\gamma + (n_3 n_1)^\gamma \right],
\end{split}
\end{equation}
where $C_\mathrm{KS}$ is a third elastic constant.

\begin{itemize}
    \item{The natural choice of pre-stressed reference state materials:}
\end{itemize}

The pre-stressed reference state family of materials with constant longitudinal wave speeds belongs to the class of ultra-relativistic materials, and can be obtained using the transformation~\eqref{Tw} with $w_0=\rho_0$, together with~\eqref{cLconst}, which results in
\begin{equation}
\begin{split}
    w^{(\mathrm{ps})}=&\left(\frac{\lambda+2\mu}{\gamma}-\frac{4\mu}{\gamma^2}+C_\mathrm{KS}-p_0\right)(n_1 n_2 n_3)^{-1}+\left(\frac{\lambda+2\mu}{\gamma(\gamma-1)} - \frac{2\mu}{\gamma^2}-C_{\mathrm{KS}}\right) (n_1 n_2 n_3)^{\gamma-1} \\
    & + \left(\frac{2\mu}{\gamma^2}-C_{\mathrm{KS}}\right) \left(\frac{n_1^{\gamma-1}}{n_2 n_3} + \frac{n_2^{\gamma-1}}{n_1 n_3} + \frac{n_3^{\gamma-1}}{n_1 n_2}\right) + C_{\mathrm{KS}} \left[ \frac{(n_1 n_2)^{\gamma-1}}{n_3} + \frac{(n_2 n_3)^{\gamma-1}}{n_1} + \frac{(n_1 n_3)^{\gamma-1}}{n_2} \right].
\end{split}
\end{equation}
The natural choice of $p_0$ corresponds to
\begin{equation}
 p_0=\frac{\lambda+2\mu}{\gamma}-\frac{4\mu}{\gamma^2}+C_{\mathrm{KS}},
\end{equation}
and from~\eqref{cLconst}
\begin{equation}
    \rho_0 = \frac{\lambda+2\mu}{\gamma(\gamma-1)}+\frac{4\mu}{\gamma^2}-C_\mathrm{KS}.
\end{equation}
The Lagrangian (energy density) is then given by 
\begin{equation}\label{KSPS}
\begin{split}
\rho^{\mathrm{(ps)}} =  & \,\, \left(\frac{\lambda+2\mu}{\gamma(\gamma-1)} - \frac{2\mu}{\gamma^2}-C_{\mathrm{KS}}\right) (n_1 n_2 n_3)^\gamma \\ & + \left(\frac{2\mu}{\gamma^2}-C_{\mathrm{KS}}\right) (n_1^\gamma + n_2^\gamma + n_3^\gamma) + C_{\mathrm{KS}} \left[ (n_1 n_2)^\gamma + (n_2 n_3)^\gamma + (n_3 n_1)^\gamma \right].
\end{split}
\end{equation}
%

\begin{itemize}
    \item Wave speeds:
\end{itemize}

For both families the transverse wave speeds~\eqref{TWaves} are given by
\begin{equation}\label{CLws}
    c^{2}_{\mathrm{T}12}=\frac{n^{2-\gamma}_2(n^{\gamma}_2-n^{\gamma}_1)\left(\frac{2\mu}{\gamma^2}-C_\mathrm{KS}+C_\mathrm{KS}n^{\gamma}_3\right)}{(n^2_2-n^2_1)\left[\left(\frac{\lambda+2\mu}{\gamma(\gamma-1)}-\frac{2\mu}{\gamma^2}-C_\mathrm{KS}\right)(n_1n_3)^\gamma+\left(\frac{2\mu}{\gamma^2}-C_\mathrm{KS}\right)+C_\mathrm{KS}(n^\gamma_1+n^\gamma_3)\right]}
\end{equation}
and similar expressions obtained by permuting $(1,2,3)$.
\end{examples}
\begin{remark}
Taking the limit $\mu\rightarrow 0^{+}$ and $C_\mathrm{KS}\rightarrow 0$, the first family reduces to the fluid with affine EoS of Example~\ref{Ex3}, and the second family reduces to the fluid with linear EoS of Example~\ref{Ex2}, with the identification
\begin{equation}
\mathrm{K}=\frac{\lambda}{\gamma n^{\gamma}_0}\, .   
\end{equation}
The linear EoS can also be obtained from the relativistic polytropic fluid by setting $\mathrm{C}=0$, that is, $\rho_0=\lambda/(\gamma(\gamma-1))$.
\end{remark}

%
%
%
%
%

\newpage

\section{Einstein-elastic equations in spherical symmetry}
\label{App:B}
The Einstein equations in spherical symmetry have been treated in full generality by e.g. Christodoulou in~\cite{Christodoulou:1995}.
Under this symmetry, the spacetime manifold $\left(\mathcal{S},\bm{g}\right)$ admits a $SO(3)$ action by isometries, whose orbits are either fixed points or $2$-spheres. The orbit space $\mathcal{Q}=\mathcal{S}/SO(3)$ is a 2-dimensional Lorentzian manifold with boundary, corresponding to the set of fixed points in $\mathcal{S}$, which form timelike curves 
(necessarily geodesics). The {\em radius function}, defined by $r(p):=\sqrt{\text{Area}({\mathcal O}_p)/4\pi}$ (where ${\mathcal O}_p$ is the orbit through $p$), is monotonically increasing along the generators of the future null cones of fixed points, at least initially. In appropriate local coordinates, the metric $\bm{g}$ on $\mathcal{S}$ has the general form
\begin{equation}\label{SSMetric}
 \bm{g}=g_{ab}(x)dx^{a}dx^{b}+r^{2}(x)\cancel{g}_{AB}(y)dy^{A}dy^{B},
\end{equation}
where $g_{ab}(x)dx^{a}dx^{b}$ (with $a,b,=0,1$) is the Lorentzian metric on $\mathcal{Q}$, and $\cancel{g}_{AB}(y) dy^{A}dy^{B}$ (with $A,B=2,3$) is the metric on the unit 2-sphere. We shall use spherical coordinates $y^{A}=(\theta,\varphi)$, with $\theta\in(0,\pi)$ and $\varphi\in(0,2\pi)$, on the unit 2-sphere, so that
\begin{equation}
\cancel{g}_{AB}(y)dy^{A}dy^{B}=d\theta^2+\sin^2{\theta}d\varphi^2.
\end{equation}
In spherical symmetry, the 4-velocity of matter can be seen as a unit vector field on $\mathcal{Q}$,
\begin{equation}
\bm{u}=u^{a}(x)\partial_a,\qquad g_{cd}u^{c}u^{d}=-1,
\end{equation}
and the Riemannian metric $\bm{h}$ on the orthogonal space to $\bm{u}$ has nonvanishing components
\begin{equation}\label{SSh}
h_{ab}(x)=u_a(x) u_b(x)+g_{ab}(x), \qquad h_{\theta\theta}=r^2(x),\qquad h_{\varphi\varphi}=r^2(x)\sin^2{\theta}.
\end{equation}
The $4$-velocity $\bm{u}$ is the eigenvector of $\bm{h}$ associated with the zero eigenvalue, while the remaining normalized eigenvectors $\bm{e}_{(i)}$ are given by
\begin{equation}
\bm{e}_{(1)}=e_{(1)}^a(x)\partial_a,\qquad \bm{e}_{(2)} =\frac{1}{r}\partial_\theta,\qquad 
\bm{e}_{(3)} =\frac{1}{r\sin{\theta}}\partial_\varphi
\end{equation}
(with $g_{ab}u^{a}e^{b}_{(1)}=0$ and $g_{ab}e^{a}_{(1)}e^{b}_{(1)}=1$). The stress-energy tensor inherits the spacetime symmetry, so that we have
\begin{equation}
 \bm{T}=T_{ab}(x)dx^{a}dx^{b}+r^2(x)p_\mathrm{tan}(x)\cancel{g}_{AB}(y)dy^{A}dy^{B}.
\end{equation}
Here
\begin{equation}
T_{ab}(x)=\rho(x) u_{a}u_{b}+\sigma_{ab}(x), \qquad \sigma_{ab}(x)=p_\mathrm{rad}(x) h_{ab}(x),
\end{equation}
or equivalently
\begin{equation}\label{MatSS}
T_{ab}(x) = (\rho(x)+p_\mathrm{rad}(x))u_a u_b + p_\mathrm{rad}(x) g_{ab}(x), \qquad g^{ab}(x)T_{ab}(x)=p_\mathrm{rad}(x)-\rho(x),
\end{equation}
where $\rho(x)$ is the \emph{energy density}, and
\begin{equation}
p_\mathrm{rad}(x)=p_1(x), \qquad p_\mathrm{tan}(x)=p_2(x)=p_3(x)
\end{equation}
denote the \emph{radial} and \emph{tangential  pressures}, respectively. On $\mathcal{Q}$ we define the {\it Hawking (or Misner-Sharp) mass} as
 \begin{equation}\label{mass}
  m(x)=\frac{r(x)}{2}\left(1-g^{cd}(x)\nabla_{c}r(x)\nabla_{d}r(x)\right).
 \end{equation}
The nonvanishing connection coefficients for the metric~\eqref{SSMetric} are
\begin{align}
 \Gamma^{a}_{bc},\qquad\Gamma^{a}_{BC}=-g^{ab}\frac{\nabla_b r}{r} r^2\cancel{g}_{AB}, \qquad \Gamma^{A}_{BC},\qquad \Gamma^{A}_{aB}=\frac{\nabla_a r}{r} \delta^{A}_{B},
\end{align} 
where $\Gamma^{a}_{bc}$ and $\Gamma^{A}_{BC}$ are the connection coeficcients  of of the Lorentzian metric on $\mathcal{Q}$ and of the round metric on the unit 2-sphere.
%
%
The Ricci tensor is then
\begin{equation}
R_{ab}=Kg_{ab}-\frac{2}{r}\nabla_{a}\nabla_{b}r ,\qquad 
R_{aA}=0,\qquad 
R_{AB}=\left(\frac{2m}{r}-rg^{ab}\nabla_{a}\nabla_{b}r\right)\cancel{g}_{AB},
\end{equation}
where $K$ is the Gauss curvature of $g_{ab}$
 and $\nabla$ is the Levi-Civita connection of the Lorentzian metric on $\mathcal{Q}$, so that the Ricci scalar is
\begin{equation}
    R=g^{ab}R_{ab}+\frac{\cancel{g}^{AB}}{r^{2}}R_{AB}=2K+\frac{2}{r^2}\left(\frac{2m}{r}-2rg^{ab}\nabla_{a}\nabla_{b}r\right).
\end{equation}
%
%
%
%
%
%
Taking the trace of the Einstein equations~\eqref{Einstein} and substituting back, the Einstein equations take the form
\begin{equation}
    R_{\mu\nu}=8\pi \left[T_{\mu\nu}-\frac{1}{2}\left(g^{\alpha\beta}T_{\alpha\beta}\right)g_{\mu\nu}\right].
\end{equation}
In spherical symmetry, these split into the two equations
\begin{subequations}\label{EinTrace}
 \begin{align}
  -\frac{2}{r}\nabla_{a}\nabla_{b}r+Kg_{ab}&=8\pi\left[(\rho+p_\mathrm{rad}) u_a u_b+\frac{1}{2}\Bigl((\rho-p_\mathrm{rad})-2(p_\mathrm{tan}-p_\mathrm{rad})\Bigr)g_{ab}\right] ,\label{EinstenQuotient}\\
 -rg^{ab}\nabla_{a}\nabla_{b}r+\frac{2m}{r}&=4\pi r^{2}(\rho-p_\mathrm{rad})  .\label{EinsteinAngular} 
 \end{align}
\end{subequations}
The trace of~\eqref{EinstenQuotient} gives
\begin{equation}\label{TraceQuotient}
 K-\frac{1}{r}g^{ab}\nabla_{a}\nabla_{b}r=-8\pi p_\mathrm{tan}.
\end{equation}
Substituting~\eqref{EinsteinAngular} into~\eqref{TraceQuotient}, and then the resulting equation into~\eqref{EinstenQuotient}, gives the equivalent system
\begin{subequations}
    \begin{align}
        \nabla_{a}\nabla_{b}r &=\frac{m}{r^2}g_{ab}-4\pi r\left[ (\rho+p_\mathrm{rad})u_a u_b+\rho g_{ab}\right], \label{Hess} \\
         K & =\frac{2m}{r^{3}}-4\pi\left(\rho-p_\mathrm{rad}+2p_\mathrm{tan}\right). \label{integrability}
    \end{align}
\end{subequations}
In spherical symmetry, the conservation of energy-momentum~\eqref{T} becomes 
\begin{equation}\label{diffeqSS}
\frac{1}{r^2}\nabla_{a}(r^2T^{a}_{\,\,\, b}) - \frac{2}{r}p_\mathrm{tan} \nabla_b r=\nabla_a T^{a}_{\,\,\,b}+\frac{2}{r}(\rho+p_\mathrm{rad})u_b u^{a}\nabla_{a}r-\frac{2}{r}(p_\mathrm{tan}-p_\mathrm{rad}) \nabla_b r =0 .
\end{equation}
Projecting in the directions parallel and orthogonal to $\bm{u}$ we obtain equations~\eqref{ConsEqs} in spherical symmetry:
\begin{subequations}\label{SSCE}
    \begin{align}
        & u^{a}\nabla_a\rho+\frac{1}{r^2}(\rho+p_\mathrm{rad})\nabla_a (r^2 u^{a})+\frac{2}{r}(p_\mathrm{tan}-p_\mathrm{rad})u^{a}\nabla_a r =0, \\
        &(\rho+p_\mathrm{rad})h^{b}_{c}a_b+h^{b}_{c}\left(\nabla_b p_\mathrm{rad}-\frac{2}{r}(p_\mathrm{tan}-p_\mathrm{rad})\nabla_b r\right)=0,
    \end{align}
\end{subequations}
where the components of the acceleration 4-vector are given by $a^{b}=u^{a}\nabla_a u^{b}$. 
It can be shown (see e.g.~\cite{Christodoulou:1995,Karlovini:2004ix}) that equation~\eqref{integrability} is the integrability condition of the Hessian system~\eqref{Hess}, and is implied by~\eqref{diffeqSS}. 
From the Hessian systen one can extract two equations for the mass function by contracting equation~\eqref{Hess} with $\nabla^b r$: 
\begin{equation}\label{massEqs}
  \nabla_{a}m(x)=4\pi r^{2}\left[(\rho+p_\mathrm{rad})u_a u_b+\rho g_{ab}\right]\nabla^{b} r.
\end{equation}
These can be also projected in the directions parallel and orthogonal to $\bm{u}$:
\begin{subequations}
\begin{align}
    u^{a}\nabla_a m &= -4\pi r^2 p_\mathrm{rad} u^a \nabla_{a} r, \\
    h^{b}_{a}\nabla_b m &=4\pi r^2 \rho h^{b}_{a}\nabla_{b} r.
\end{align}
\end{subequations}
%
%
%
Finally, given an elastic EoS, we have the conservation of particle number density:
\begin{equation}
\frac{1}{r^2}\nabla_{a}(r^2 nu^{a})\equiv u^{a}\nabla_a n+n\nabla_a u^{a}
+\frac{2}{r}nu^{a}\nabla_a r=0.
\end{equation}
%

\subsection{The material metric}\label{sec:materialmetric}
We assume that the material body $(\mathcal{B},\bm{\gamma})$, which is a Riemannian 3-dimensional manifold, is also spherically symmetric. Let $X^I=(R,\Theta,\mathcal{G})$ be spherical coordinates on the material body such that the group orbits have area $4\pi R^2$.  The material metric takes the form
\begin{equation}
\gamma_{IJ}dX^{I}dX^{J} = e^{2\Psi_0(R)}dR^2 +R^2(d\Theta^2+\sin^2\Theta \, d\mathcal{G}^2).
\end{equation}
%
%
Assuming that the projection map $\bm{\Pi}$ is equivariant with respect to the spherical symmetry, it can be expressed as
\begin{equation}\label{SSConfMap}
R=R(x),\qquad \Theta=\theta, \qquad \mathcal{G}=\varphi,
\end{equation}
so that the nonvanishing components of the configuration gradient~\eqref{DefT} are 
\begin{equation}\label{Qf}
f^{R}_{\,\,a}=\partial_a R,  \qquad f^{\Theta}_{\,\,\theta}=1,\qquad f^{\mathcal{G}}_{\,\,\varphi}=1.
\end{equation}
The nonvanishing components of $H^{IJ}$ are then given by
\begin{equation}\label{QHcomponents}
H^{RR}=g^{ab}(\partial_a R)(\partial_b R), \qquad H^ {\Theta\Theta}=\frac{1}{r^2},\qquad H^{\mathcal{G}\mathcal{G}}=\frac{1}{r^2\sin^2{\theta}},
\end{equation}
with $g^{ab}(\partial_a R)(\partial_b R)>0$, and the pull-back of $\bm{H}^{-1}$ under $\bf{\Pi}$ has components given in~\eqref{SSh}, with
\begin{equation}
h_{ab}=\frac{(\partial_{a}R)(\partial_{b}R)}{g^{cd}(\partial_{c}R)(\partial_{d}R)}.
\end{equation}

\subsection{Comoving coordinates}
Comoving coordinates consist in taking as coordinates on the quotient manifold the configuration radius $R$ and a comoving time $\tau$, i.e., the level curves of $\tau$ are the simultaneity curves of the elastic matter particles,\footnote{This defines $\tau$ up to rescaling; for asymptotically flat spacetimes, the standard normalization condition is $\lim_{r\to\infty}{\bf g}(\partial_\tau,\partial_\tau)=-1$.} and the level curves of $R$ are the corresponding flow lines. We assume that the center of symmetry in material space, $R = 0$, coincides with the center of symmetry in physical space, i.e., we assume that $R = 0$ is mapped under the configuration map to a flow line on the boundary of $\mathcal{Q}$, so that $r(\tau,0)=0$.

The metric of $\mathcal{Q}$ is given by
\begin{equation}\label{ComovingTimedep}
g_{ab}dx^adx^b =- e^{2\Phi(\tau,R)}d\tau^2+e^{2\Psi(\tau,R)}dR^2 ,
\end{equation} 
and the orthonormal frame on $\mathcal{Q}$ is of the form
\begin{equation}
u^{a}(\tau,R)\partial_a = e^{-\Phi(\tau,R)}\partial_\tau ,\qquad    e^{a}_{(1)}(\tau,R)\partial_a = e^{-\Psi(\tau,R)}\partial_R.
\end{equation}
The nontrivial components in~\eqref{Qf} are simply
\begin{equation}
f^{R}_{\,\,\tau}=0, \qquad f^{R}_{\,\,R}=1,
\end{equation}
while the nontrivial component in~\eqref{QHcomponents} is
\begin{equation}
H^{RR}=e^{-2\Psi},
\end{equation}
so that
\begin{equation}
    h_{\tau\tau}=0,\quad h_{\tau R} = 0, \quad h_{RR}=e^{2\Psi}.
\end{equation}
The Christoffel symbols $\Gamma^{a}_{bc}$ are given by
\begin{subequations}
	\begin{align}
	&\Gamma^{\tau}_{\tau\tau}=\partial_\tau \Phi,\qquad \Gamma^{\tau}_{\tau R}=\partial_R \Phi,\qquad \Gamma^{\tau}_{RR}=e^{-2(\Phi-\Psi)}\partial_\tau\Psi, \\
   &\Gamma^{R}_{\tau\tau}=e^{-2(\Psi-\Phi)}\partial_R\Phi ,\qquad \Gamma^{R}_{\tau R}=\partial_\tau\Psi,\qquad \Gamma^{R}_{RR}=\partial_R\Psi.
	\end{align}
\end{subequations}
From these, the acceleration vector on $\mathcal{Q}$ is given by
\begin{equation}
    a^{a}\partial_{a}=e^{-2\Psi}(\partial_{R}\Phi )\partial_R.
\end{equation}
The Hessian system~\eqref{Hess} results in a system of three equations for the areal radius function 
\begin{subequations}
    \begin{align}
        e^{-2\Phi}\left(\partial^2_\tau r-(\partial_\tau\Phi)\partial_\tau r\right) -e^{-2\Psi}(\partial_R\Phi)\partial_R r&= -\frac{m}{r^2} -4\pi r p_\mathrm{rad}, \\
        \partial_\tau\partial_R r-(\partial_R\Phi)\partial_\tau r -(\partial_\tau\Psi)\partial_R r &=0, \\
        e^{-2\Psi}\left(\partial^2_R r-(\partial_R\Psi)\partial_Rr\right)- e^{-2\Phi}(\partial_\tau\Psi)\partial_\tau r  &=\frac{m}{r^2}-4\pi r \rho,
    \end{align}
\end{subequations}
where the Hawking mass~\eqref{mass} is given by
\begin{equation}\label{masscom}
1-\frac{2m}{r}=-e^{-2\Phi}(\partial_\tau r)^2+e^{-2\Psi}(\partial_R r)^2.
\end{equation}
The divergence of the energy-momentum tensor~\eqref{SSCE} yields
\begin{subequations}
\begin{align}
&\partial_\tau \rho+\frac{2}{r}(\rho+p_\mathrm{rad})\left(\partial_\tau r+\frac{r}{2}\partial_\tau\Psi\right)=-\frac{2}{r}(p_\mathrm{tan}-p_\mathrm{rad}) \partial_\tau r, \\
&\partial_Rp_\mathrm{rad}=\frac{2}{r}(p_\mathrm{tan}-p_\mathrm{rad})\partial_Rr-(\rho+p_\mathrm{rad})\partial_R\Phi, \label{CauchySSCom}
\end{align}
\end{subequations}
and the equations for the mass~\eqref{massEqs} become
\begin{equation}\label{massevcom}
   e^{-\Phi}\partial_\tau m = -4\pi r^2 p_\mathrm{rad} v,\qquad \partial_R m = 4\pi r^2 \rho (\partial_Rr).
\end{equation}
The second equation for the mass can be integrated to yield
\begin{equation}
m(\tau,R)=4\pi \int^{R}_{0} \rho(\tau,S) \partial_R r(\tau,S)r^2(\tau,S)dS.
\end{equation}
The above system can be written as a first order system by introducing the radial velocity, defined as
\begin{equation}\label{DefvCom}
v(\tau,R)=e^{-\Phi(\tau,R)}\partial_\tau r(\tau,R),
\end{equation}
and a new variable
\begin{equation}
    u(\tau,R)=e^{-\Psi(\tau,R)}\partial_R r(\tau,R).
\end{equation}
From the expression~\eqref{masscom} for the Hawking mass, we can write 
\begin{equation}
    u^2=\left(1-\frac{2m}{r}\right)\langle v \rangle^2, \qquad \langle v\rangle=\left(1+\frac{v^2}{1-\frac{2m}{r}}\right)^{1/2}.
\end{equation}
The first equation of the Hessian system, together with~\eqref{CauchySSCom}, gives an evolution equation for the radial velocity:
\begin{equation}
\frac{(\rho+p_\mathrm{rad})e^{-\Phi}}{\left(1-\frac{2m}{r}+v^2\right)}\partial_\tau v +\frac{\partial_R p_\mathrm{rad}}{\partial_Rr}= \frac{2}{r}(p_\mathrm{tan}-p_\mathrm{rad})-\frac{(\rho+p_\mathrm{rad})}{\left(1-\frac{2m}{r}+v^2\right)}\left(\frac{m}{r^2}+4\pi rp_\mathrm{rad}\right).
\end{equation}
From the second equation of the Hessian system we obtain an evolution equation for $u$:
\begin{equation}
        e^{-\Phi}\partial_\tau u =e^{-\Psi} v \partial_R\Phi.
\end{equation}
%
%
%
%

\subsection{Schwarzschild coordinates}
The Schwarzschild (or Eulerian) coordinate system is $x^{a}=(t,r)$, where $r$ is area radius and $t$ is defined so that the curves of constant $t$ are orthogonal to the curves of constant $r$.\footnote{This defines $t$ up to rescaling; for asymptotically flat spacetimes, the standard normalization condition is $\lim_{r\to\infty}{\bf g}(\partial_t,\partial_t)=-1$.} In these coordinates, the line element on the orbit space reads
\begin{equation}\label{SchwaTimedep}
g_{ab}(t,r)dx^{a}dx^{b}=- e^{2\phi(t,r)}dt^2+e^{2\psi(t,r)}dr^2,
\end{equation} 
and the orthonormal frame on $\mathcal{Q}$ is of the form
\begin{equation}
u^{a}\partial_a=e^{-\phi}\langle v\rangle\partial_t +v\partial_r,\qquad e_{(1)}^a\partial_a =e^{\psi-\phi}v\partial_t+e^{-\psi}\langle v\rangle\partial_r,\qquad \langle v\rangle=\sqrt{1+e^{2\psi}v^2},
\end{equation}
where $v$ is the \emph{radial velocity}.
In Schwarzschild coordinates, it follows from~\eqref{SSConfMap} that
\begin{equation}
R=R(t,r),
\end{equation}
and so the nonvanishing components of the configuration gradient~\eqref{DefT} are 
\begin{equation}
f^{R}_{\,\,t}=\partial_t R, \qquad f^{R}_{\,\,r}=\partial_r R.
\end{equation}
The nonvanishing components of $H^{IJ}$ are then given by
\begin{equation}
H^{RR}=e^{-2\psi}(\partial_r R)^2-e^{-2\phi}(\partial_t R)^2,
\end{equation}
so that the pull-back of $\bm{H}^{-1}$ has the nonvanishing components
\begin{equation}\label{hSch}
	h_{tt} =
	e^{2(\psi+\phi)}v^2,\quad h_{tr}= 
	-e^{2\psi+\phi}\langle v\rangle v,\quad
	h_{rr} =
	e^{2\psi}\langle v\rangle^2,
\end{equation}
with the identifications
\begin{equation} \label{vand<v>}
v=-\frac{e^{-\phi}\partial_t R/\partial_r R}{\left(1-e^{2(\psi-\phi)}\left(\partial_t R/\partial_r R\right)^2\right)^{1/2}}, \qquad \langle v\rangle=\frac1{\left(1-e^{2(\psi-\phi)}\left(\partial_t R/\partial_r R\right)^2\right)^{1/2}}.
\end{equation}
The nonvanishing Christoffel symbols are given by
\begin{subequations}
	\begin{align}
	&\Gamma^{t}_{tt}=\partial_t\phi,\qquad \Gamma^{t}_{tr}=\partial_r\phi,\qquad \Gamma^{t}_{rr}=e^{2(\psi-\phi)}\partial_t\psi,      \\
    &\Gamma^{r}_{tt}= e^{2(\phi-\psi)}\partial_r\phi,\qquad \Gamma^{r}_{tr}=\partial_t\psi,\qquad \Gamma^{r}_{rr}=\partial_r\psi.
    \end{align}
\end{subequations}
%
%
%
%
%
%
%
%
Moreover, the nonvanishing components of the $4$-acceleration $a^{a}$ are~\footnote{Here we make use of the relations $\partial_t\langle v\rangle=\frac{v e^{2\psi}}{\langle v\rangle}(\partial_t v+v\partial_t\psi)$ and $\partial_r \langle v\rangle=\frac{v e^{2\psi}}{\langle v \rangle}(\partial_r v+v\partial_r\psi)$.}
\begin{subequations}
	\begin{align}
	a^{t} &= e^{-\phi}\langle v\rangle \left[v\partial_r\phi +e^{2\psi}\frac{v^2}{\langle v\rangle^2}(\partial_r v+v\partial_r\psi)+e^{2\psi-\phi}\frac{v}{\langle v \rangle}(\partial_t v+2v\partial_t\psi)\right],\\
	a^{r} &=e^{-\phi}\langle v\rangle \left(\partial_t v+2v\partial_t\psi\right)+\left(v\partial_r v+v^2\partial_r\psi\right)+e^{-2\psi}\langle v\rangle^2 \partial_r\phi.
	\end{align}
\end{subequations}
%
%
%
%
%
%
%
%
%
The Hessian system~\eqref{Hess} yields
\begin{subequations}
\begin{align}
    e^{-2\psi}\partial_r\phi &=\frac{m}{r^2}+4\pi r \left(p_\mathrm{rad} +(\rho+p_\mathrm{rad})e^{2\psi}v^2\right), \label{EE0}\\
    e^{-\phi}\langle v\rangle\partial_t\psi &= -4\pi r (\rho+p_\mathrm{rad})e^{2\psi}\langle v \rangle^2 v, \label{EE1}\\
    e^{-2\psi}\partial_r\psi&=-\frac{m}{r^2}+4\pi r\left(\rho +(\rho+p_\mathrm{rad})e^{2\psi}v^2\right), \label{EE2}
\end{align}
\end{subequations}
and equations~\eqref{SSCE} for the vanishing of the divergence of the stress-energy tensor~\eqref{ConsEqs}
become
%
%
%
\begin{subequations}\label{ConsSS}
	\begin{align}
	&e^{-\phi}\langle v\rangle\partial_t\rho+v\partial_r\rho+\frac{2}{r}(p_\mathrm{tan}-p_\mathrm{rad})v \nonumber\\
	&+(\rho+p_\mathrm{rad})\left[e^{-\phi}\langle v\rangle\left(\frac{e^{2\psi}v}{\langle v\rangle^2}\partial_t v+\left(1+\frac{e^{2\psi}v^2}{\langle v\rangle^2}\right)\partial_t\psi\right)+\partial_r v+v\left(\frac{2}{r}+\partial_r\phi+\partial_r\psi\right) \right] =0, \label{EMC0}\\
	& \frac{e^{2\psi}v}{\langle v\rangle^2}e^{-\phi}\langle v\rangle \partial_t p_\mathrm{rad}+\partial_r p_\mathrm{rad}-\frac{2}{r} (p_\mathrm{tan}-p_\mathrm{rad}) \nonumber \\ &+(\rho+p_\mathrm{rad})\left(\frac{e^{2\psi}}{\langle v\rangle^2}\left(e^{-\phi}\langle v\rangle\partial_t v+2ve^{-\phi}\langle v\rangle\partial_t \psi+v\partial_r v+v^2\partial_r\psi\right)+\partial_r\phi\right)=0. \label{EMC1}
	\end{align}
\end{subequations}
%
Finally, equation~\eqref{Current} for the conservation of the particle current is
\begin{equation}\label{PCSS}
\begin{split}
 e^{-\phi}\langle v\rangle \partial_t n+v\partial_r n +n\partial_r v + & nv\frac{e^{2\psi}}{\langle v \rangle^2} e^{-\phi}\langle v\rangle\partial_t v  \\
 &+n\left(\left(1+\frac{e^{2\psi}v^2}{\langle v\rangle^2}\right)e^{-\phi}\langle v\rangle\partial_t\psi+v\left(\frac{2}{r}+\partial_r\psi+\partial_r\phi\right)\right)=0.
\end{split}
\end{equation}
The metric function $\psi(t,r)$ is related to the Hawking mass function $m(t,r)$ by
\begin{equation}
e^{-2\psi(t,r)}=1-\frac{2m(t,r)}{r},\qquad \partial_t \psi = \frac{\partial_t m}{r\left(1-\frac{2m}{r}\right)},\qquad \partial_r\psi = \frac{\partial_r m-\frac{m}{r}}{r\left(1-\frac{2m}{r}\right)}.
\end{equation}

When written in terms of the mass function $m(t,r)$, equations~\eqref{EE1}-\eqref{EE2}  yield~\eqref{HwkMass}. Moreover, equations~\eqref{EE0} result in equation~\eqref{RelPot}. 
When written in terms of the mass function, equations~\eqref{EMC0}-\eqref{EMC1} lead to equations~\eqref{Cons1}-\eqref{Cons2}, while equation~\eqref{PCSS} for the number of particles density results in equation~\eqref{Consn}.

\subsection{Principal linear densities}\label{APpld}
For homogeneous and isotropic elastic materials, the constitutive functions depend on the principal linear densities, which are obtained by taking the square root of the eigenvalues of the linear operator $\mathcal{H}^{I}_{\,\,\,J}=H^{IK}\gamma_{KJ}$. In the spherically symmetric setting, two of the principal linear densities coincide, $n_2=n_3=n_\mathrm{t}$, corresponding to the degenerate case~\ref{Deg}, and are given by
\begin{equation}
    n_1(x)=e^{\Psi_0(R(x))}\left(g^{ab}(\partial_a R)(\partial_b R)\right)^{1/2}, \qquad n_\mathrm{t}(x)=\frac{R(x)}{r(x)}.
\end{equation}
In the special case of a spherically symmetric perfect fluid, $\rho(x)=\widehat{\rho}\left(\delta(x)\right)$, where
\begin{equation}
 \delta(x)\equiv\frac{n(x)}{n_0}=n_1(x) n^2_\mathrm{t}(x).
\end{equation}
In the spherically symmetric elastic setting, in addition to $\delta(x)$ we define the variable
\begin{equation}
    \eta(x)\equiv n^{3}_\mathrm{t}(x),
\end{equation}
so that
\begin{equation}
    \rho(x)=\widehat{\rho}\left(\delta(x),\eta(x)\right).
\end{equation}
\subsubsection{Comoving coordinates}
In comoving coordinates, the principal linear densities are given by
\begin{equation}
n_1(\tau,R) = e^{\Psi_0(R)-\Psi(\tau,R)}, \qquad 
n_\mathrm{t}(\tau,R) =\frac{R}{r(\tau,R)} .
\end{equation}
Therefore,
\begin{equation}\label{Comdefdelta}
\delta(\tau,R)=
e^{\Psi_0(R)-\Psi(\tau,R)}\left(\frac{R}{r(\tau,R)}\right)^2
\end{equation}
and
\begin{equation}\label{Comdefeta}
\eta(\tau,R)=
\left(\frac{R}{r(\tau,R)}\right)^3 .
\end{equation}
To obtain an explicit form for the function $\eta(\tau,R)$, one integrates
\begin{equation}
\partial_R(r^3\eta)=3e^{\Psi-\Psi_0} \delta r^2,
\end{equation}
which follows directly from the definitions of $\eta$ and $\delta$. This leads to 
\begin{equation}
\eta(\tau,R)=\frac{3}{r^3}\int^{R}_ {0} e^{\Psi(\tau,S)-\Psi_0(S)} \delta(\tau,S) r^2(\tau,S) dS.
\end{equation}
Differentiating~\eqref{Comdefeta} and using~\eqref{DefvCom}, as well as~\eqref{Comdefdelta}, we find that
\begin{equation}
\partial_\tau \eta=-\frac{3}{r}\eta e^{\Phi} v, \qquad \partial_R\eta = -\frac{3}{r}\left(\eta-\frac{e^{\Psi-\Psi_0}\delta}{\partial_R r}\right)\partial_R r.
\end{equation}
\subsubsection{Schwarzschild coordinates}
In Schwarzschild coordinates, the principal linear densities are given by
\begin{equation}
n_1(t,r) = \frac{e^{\Psi_0(R(t,r))-\psi(t,r)}}{\langle v\rangle(t,r)}\partial_r R(t,r) , \qquad 
n_\mathrm{t}(t,r) =\frac{R(t,r)}{r} .
\end{equation}
Therefore,
\begin{equation}\label{defdelta}
\delta(t,r)=\frac{n(t,r)}{n_0}=
\frac{e^{\Psi_0(R(t,r))-\psi(t,r)}}{\langle v\rangle(t,r) }\left(\frac{R(t,r)}{r}\right)^2 \partial_r R(t,r)
\end{equation}
and
\begin{equation}\label{defeta}
\eta(t,r)=
\left(\frac{R(t,r)}{r}\right)^3 .
\end{equation}
To obtain an explicit form for the function $\eta(t,r)$ in terms of the variables $(\phi,\psi,v,n)$, one needs to integrate the equation
\begin{equation}
e^{\Psi_0(R)}\partial_r (r^3\eta)=3e^{\Psi_0(R)}R^2 \partial_r R = 3 e^{\psi}\langle v\rangle \delta r^2,
\end{equation}
which follows directly from the definitions of $\eta$ and $\delta$. For ball solutions, and under the assumption that $R(t,0)=0$, i.e. that a ball in physical space must originate from a ball in material space, we integrate from $r=0$:
\begin{equation}
P\left(r\eta^\frac13(t,r)\right)~\equiv P(R(t,r))=3\int^{r}_0 e^{\psi(t,s)}\langle v\rangle(t,s) \delta(t,s) s^2 ds ,
\end{equation}
where
\begin{equation}
P(u)=\int^{u}_0 3e^{\Psi_0(s)} s^2 ds .
\end{equation}
Using~\eqref{vand<v>} and~\eqref{defdelta}, we find that
\begin{equation}
    \partial_t P =-3e^{\phi+\psi}v\delta r^2, \qquad \partial_r P= 3 e^{\psi}\langle v\rangle \delta r^2,
\end{equation}
that is,
\begin{equation}
\partial_t \eta=-\frac{3}{r} e^{\phi+\psi-\Psi_0}v\delta, \qquad \partial_r\eta = -\frac{3}{r}\left(\eta-e^{\psi-\Psi_0}\langle v\rangle \delta\right),
\end{equation}
where it should keep in mind that $\Psi_0=\Psi_0(r\eta^{1/3})$. For a flat material metric, $\Psi_0\equiv 0$, we have $P(u)=u^3$, and so we obtain the explicit formula
\begin{equation}
\eta(t,r)=\frac{3}{r^3}\int^{r}_0 e^{\psi(t,s)}\langle v\rangle(t,s) \delta(t,s) s^2 ds .
\end{equation}
%
%
%
%
%
%
%
%
%
\subsection{Derivatives of the energy density and the principal pressures}
%
%
A choice of $\widehat\rho=\widehat\rho(\delta,\eta)$ does not determine $\rho(n_1,n_2,n_3)$, but only
\begin{equation}
\rho(n_1,n_2,n_2) = \widehat\rho(n_1n_2^2, n_2^3).
\end{equation}
From this formula, and the symmetry of $\rho(n_1,n_2,n_3)$ with respect to $n_2$ and $n_3$, we easily deduce, by taking derivatives with respect to $n_1$ and $n_2$, the following relations:
\begin{subequations}
\begin{align}
n_1\frac{\partial \rho}{\partial n_1} &= \delta \partial_\delta \widehat\rho, \label{firstorder1}\\
n_2\frac{\partial \rho}{\partial n_2} &= \delta \partial_\delta \widehat\rho + \frac32 \eta \partial_\eta \widehat\rho , \label{firstorder2}\\
n_1^2\frac{\partial^2 \rho}{\partial n_1^2} &= \delta^2 \partial^2_\delta \widehat\rho,  \label{2ndn1}\\
n_2^2\frac{\partial^2 \rho}{\partial n_2^2} + n_2^2\frac{\partial^2 \rho}{\partial n_2 \partial n_3} &= 2\delta^2 \partial^2_\delta \widehat\rho + 6 \delta \eta \partial_\delta \partial_\eta \widehat\rho + \frac92 \eta^2 \partial^2_\eta \widehat\rho + \delta \partial_\delta \widehat\rho + 3 \eta \partial_\eta \widehat\rho.\label{2ndn2}
\end{align}
\end{subequations}
However, it is not possible to obtain the second order partial derivatives in~\eqref{2ndn2} separately. From~\eqref{pprho} and~\eqref{firstorder1}-\eqref{firstorder2}, the principal pressures,
\begin{equation}
        \widehat{p}_\mathrm{rad}(\delta,\eta)=p_1(n_1,n_2,n_2),\qquad \widehat{p}_\mathrm{tan}(\delta,\eta)=p_2(n_1,n_2,n_2)=p_3(n_1,n_2,n_2),
\end{equation}
are written respectively as
\begin{subequations}
	\begin{align}
	\widehat{p}_\mathrm{rad}(\delta,\eta) &=\delta\partial_\delta\widehat{\rho}(\delta,\eta)-\widehat{\rho}(\delta,\eta), \label{prad} \\
	\widehat{p}_\mathrm{tan}(\delta,\eta) &=\widehat{p}_\mathrm{rad}(\delta,\eta)+\frac{3}{2}\partial_\eta\widehat{\rho}(\delta,\eta) \label{ptan}.
	\end{align}
\end{subequations}
From~\eqref{prad} and~\eqref{ptan} we obtain
\begin{subequations}
\begin{align}
\delta \partial_\delta \widehat{p}_{\rm rad} &= \delta^2 \partial^2_\delta \widehat\rho, \\
\delta \partial_\delta \widehat{p}_{\rm tan} &= \delta^2 \partial^2_\delta \widehat\rho + \frac32 \delta \eta \partial_\delta \partial_\eta \widehat\rho, \\
\eta \partial_\eta \widehat{p}_{\rm rad} &= \delta \eta \partial_\delta \partial_\eta \widehat\rho - \eta \partial_\eta \widehat\rho,  \\
\eta \partial_\eta \widehat{p}_{\rm tan} &= \delta \eta \partial_\delta \partial_\eta \widehat\rho + \frac32 \eta^2 \partial^2_\eta \widehat\rho + \frac12 \eta \partial_\eta \widehat\rho.
\end{align}
\end{subequations}
Differentiating the identities
\begin{equation}
p_1(n_1,n_2,n_2) = \widehat{p}_\mathrm{rad}(n_1n_2^2,n_2^3),\qquad p_2(n_1,n_2,n_2)=\widehat{p}_\mathrm{tan}(n_1n_2^2,n_2^3),
\end{equation}
we obtain, in the isotropic state $n_1=n_2=n_\mathrm{c}$,
\begin{subequations}
\begin{align}
\delta\partial_\delta \widehat{p}_{\rm rad}(\delta,\delta) &= n_\mathrm{c} \frac{\partial p_1}{\partial n_1}(n_\mathrm{c}, n_\mathrm{c}, n_\mathrm{c}), \label{dpiso1}\\
2\delta\partial_\delta \widehat{p}_{\rm rad}(\delta,\delta) + 3\delta\partial_\eta \widehat{p}_{\rm rad}(\delta,\delta) &= 2n_\mathrm{c} \frac{\partial p_1}{\partial n_2}(n_\mathrm{c}, n_\mathrm{c}, n_\mathrm{c}), \\
\delta\partial_\delta \widehat{p}_{\rm tan}(\delta,\delta) &= n_\mathrm{c} \frac{\partial p_1}{\partial n_2}(n_\mathrm{c}, n_\mathrm{c}, n_\mathrm{c}), \\
2\delta\partial_\delta \widehat{p}_{\rm tan}(\delta,\delta) + 3\delta\partial_\eta \widehat{p}_{\rm tan}(\delta,\delta) &= n_\mathrm{c} \frac{\partial p_1}{\partial n_1}(n_\mathrm{c}, n_\mathrm{c}, n_\mathrm{c}) + n_\mathrm{c} \frac{\partial p_1}{\partial n_2}(n_\mathrm{c}, n_\mathrm{c}, n_\mathrm{c}),\label{dpiso4}
\end{align}
\end{subequations}
where we used~\eqref{sym_isotropic_state}. In particular, we have the identities
\begin{subequations}
\begin{align}
& 2\partial_\delta \widehat{p}_{\rm rad}(\delta,\delta) + 3\partial_\eta \widehat{p}_{\rm rad}(\delta,\delta) = 2\partial_\delta \widehat{p}_{\rm tan}(\delta,\delta), \\
& \partial_\delta \widehat{p}_{\rm tan}(\delta,\delta) + 3\partial_\eta \widehat{p}_{\rm tan}(\delta,\delta) = \partial_\delta \widehat{p}_{\rm rad}(\delta,\delta).
\end{align}
\end{subequations}
If we particularize~\eqref{dpiso1}-\eqref{dpiso4} to the reference state $n_\mathrm{c}=1$, we easily obtain from~\eqref{PPLames} the relations
\begin{subequations}
\begin{align}
\partial_\delta \widehat{p}_{\rm rad}(1,1) = \lambda + 2\mu, \\
\partial_\eta \widehat{p}_{\rm rad}(1,1) = - \frac43\mu, \\
\partial_\delta \widehat{p}_{\rm tan}(1,1) = \lambda, \\
\partial_\eta \widehat{p}_{\rm tan}(1,1) = \frac23\mu.
\end{align}
\end{subequations}
\subsection{The wave speeds}\label{WSetadelta}
Using~\eqref{2ndn1}, we obtain from~\eqref{longdeg} that the longitudinal wave speed in the radial direction is given by
\begin{equation} \label{cL}
c_\mathrm{L}^2(\delta,\eta) =  \frac{\delta \partial_\delta \widehat{p}_{\rm rad}(\delta,\eta)}{\widehat\rho (\delta,\eta)+ \widehat{p}_{\rm rad}(\delta,\eta)}=\frac{\delta^2 \partial^2_\delta \widehat\rho(\delta,\eta)}{\widehat\rho(\delta,\eta) + \widehat{p}_{\rm rad}(\delta,\eta)} .
\end{equation}
The speed of transverse waves in the radial direction, given by equation~\eqref{cT12}, reads
\begin{equation}
c_\mathrm{T}^2(\delta,\eta) = \frac{\widehat{p}_{\rm tan}(\delta,\eta) - \widehat{p}_{\rm rad}(\delta,\eta)}{(1 - \delta^2/\eta^2)(\widehat\rho(\delta,\eta) + \widehat{p}_{\rm tan}(\delta,\eta))},
\end{equation}
while the speed of transverse waves in the tangential direction oscillating in the radial direction, obtained from~\eqref{cT21}, is given by
\begin{equation}
\tilde{c}_\mathrm{T}^2(\delta,\eta) = \frac{\widehat{p}_{\rm rad}(\delta,\eta) - \widehat{p}_{\rm tan}(\delta,\eta)}{(1 - \eta^2/\delta^2)(\widehat\rho(\delta,\eta) + \widehat{p}_{\rm rad}(\delta,\eta))}.
\end{equation}
We cannot obtain the speed of longitudinal waves in the tangential direction, given by
\begin{equation} \label{cannot1}
(\widehat\rho(\delta,\eta) + \widehat{p}_{\rm tan}(\delta,\eta))\tilde{c}_\mathrm{L}^2(\delta,\eta) = n_2^2\frac{\partial^2 \rho}{\partial n_2^2},
\end{equation}
because of the degeneracy in the second order partial derivatives. The speed of transverse waves in the tangential direction oscillating in the tangential direction, given in equation~\eqref{cT23}, can be written as
\begin{equation}\label{cannot2}
\begin{split}
(\widehat\rho(\delta,\eta) + \widehat{p}_{\rm tan}(\delta,\eta))\tilde{c}_{\mathrm{TT}}^2 &=  - \frac12 n_2 \frac{\partial p_2}{\partial n_3} + \frac12 n_2 \frac{\partial p_2}{\partial n_2}  \\
& =  \frac12 n_2^2 \frac{\partial^2 \rho}{\partial n_2^2} - \frac12 n_2^2 \frac{\partial^2 \rho}{\partial n_2\partial n_3} + \frac12 n_2 \frac{\partial \rho}{\partial n_2}.
\end{split}
\end{equation}
Combining Equations~\eqref{2ndn2},~\eqref{cannot1}, and~\eqref{cannot2}, we obtain the expression
\begin{equation}
\begin{split}
\tilde{c}^2_\mathrm{L}(\delta,\eta)-\tilde{c}^2_\mathrm{TT}(\delta,\eta) 
&=\frac{\frac{3}{4}\eta\partial_{\eta}\widehat{\rho}+\delta^2\partial^2_\delta \widehat{\rho}+3\delta\eta\partial^2_{\eta\delta}\widehat{\rho}+\frac{9}{4}\eta^2\partial^2_\eta\widehat{\rho}}{\widehat{\rho}(\delta,\eta)+\widehat{p}_\mathrm{tan}(\delta,\eta)}. 
\end{split}
\end{equation}
\subsection{Examples of elastic materials in spherical symmetry}\label{ApExSS}
\begin{examples}[Beig-Schmidt materials]
	In spherical symmetry, the Lagrangian (energy density) for the Beig-Schmidt materials of Example~\ref{BS} takes the form 
\begin{align}
\widehat{\rho}(\delta,\eta) &= \delta \Bigg[\rho_0  +\frac{3}{8}(3\lambda+2\mu)+\frac{1}{8}(\lambda+2\mu)\eta^{\frac{4}{3}}\left(1+4\left(\frac{\delta}{\eta}\right)^2+\left(\frac{\delta}{\eta}\right)^4\right) \nonumber \\
& \qquad \qquad -\frac{1}{4}(3\lambda+2\mu)\eta^{\frac{2}{3}}\left(2+\left(\frac{\delta}{\eta}\right)^2\right)-\frac{\mu}{2}\eta^{\frac{4}{3}}\left(2+\left(\frac{\delta}{\eta}\right)^2 \right)\Bigg]. \label{BSSS}
		\end{align}
These materials were analyzed by Kabobel \& Fraundiener in~\cite{FK07}.
\end{examples}
\begin{examples}[Spherically symmetric quasi-Hookean materials]\label{QHSS}
In spherical symmetry, the shear scalar for the quasi-Hookean materials of Example~\ref{QH} takes the form 
\begin{equation}
S\left(\frac{\delta}{\eta}\right)=\left(\frac{\delta}{\eta}\right)^{-2/3}\left(2+\left(\frac{\delta}{\eta}\right)^2\right)-3
\end{equation}
for the choice~\eqref{Tahshear}, while for the choice~\eqref{KSshear} it reduces to
\begin{equation}
S\left(\frac{\delta}{\eta}\right)=\frac{1}{6}\left(\frac{\delta}{\eta}\right)^{-2}\left(1-\left(\frac{\delta}{\eta}\right)^2\right)^2.
\end{equation}
The cases in which the unsheared quantities are given as in Remark~\ref{ShearPol} were analyzed by Andreasson \& Calogero in~\cite{Andreasson:2014lka}, and by Karlovini \& Samuelsson in~\cite{Karlovini:2002fc}.
\end{examples}
%
\begin{examples}[Materials with constant longitudinal wave speeds]\label{KarlSamSS}
In spherical symmetry, the Lagrangian (energy density) for materials with constant longitudinal wave speeds and a natural stress-free reference state (equation~\eqref{KSSF}) reduces to
\begin{equation}
\begin{split}
\widehat{\rho}^{(\mathrm{sf})}(\delta,\eta) =  & \,\,\frac{\lambda+2\mu}{\gamma}-\frac{4\mu}{\gamma^2}+C_{\mathrm{KS}}+ \left(\frac{\lambda+2\mu}{\gamma(\gamma-1)} - \frac{2\mu}{\gamma^2}-C_{\mathrm{KS}}\right) \eta^\gamma \left(\frac{\delta}{\eta}\right)^{\gamma}   \\ 
& + \left(\frac{2\mu}{\gamma^2}-C_{\mathrm{KS}}\right) \eta^{\gamma/3}\left(2+\left(\frac{\delta}{\eta}\right)^{\gamma}\right)+C_{\mathrm{KS}} \eta^{2\gamma/3}\left(1+2\left(\frac{\delta}{\eta}\right)^{\gamma}\right),
\end{split}
\end{equation}
%
%
%
while the Lagrangian (energy density) for the pre-stressed reference state material (equation~\eqref{KSPS}) is given in spherical symmetry by
\begin{equation}\label{PSCLSSS}
\begin{split}
\widehat{\rho}^{(\mathrm{ps})}(\delta,\eta) =  & \left(\frac{\lambda+2\mu}{\gamma(\gamma-1)} - \frac{2\mu}{\gamma^2}-C_{\mathrm{KS}}\right) \delta^\gamma    \\ 
& + \left(\frac{2\mu}{\gamma^2}-C_{\mathrm{KS}}\right) \eta^{\gamma/3}\left(2+\left(\frac{\delta}{\eta}\right)^{\gamma}\right)+C_{\mathrm{KS}} \eta^{2\gamma/3}\left(1+2\left(\frac{\delta}{\eta}\right)^{\gamma}\right).
\end{split}
\end{equation}
%
%
For both materials, the transverse wave speeds given in eqwuation~\eqref{CLws} reduce, in spherical symmetry, to~\eqref{cT12},~\eqref{cT21}, i.e.,
\begin{equation}\label{cT1}
c^{2}_\mathrm{T}(\delta,\eta)=\frac{\left[\left(\frac{2\mu}{\gamma^2}-C_\mathrm{KS}\right)\eta^{\gamma/3}+C_\mathrm{KS}\eta^{2\gamma/3}\right]\left(1-\left(\frac{\delta}{\eta}\right)^{\gamma}\right)}{\left[\left(\frac{\lambda+2\mu}{\gamma(\gamma-1)}-\frac{2\mu}{\gamma^2}-C_\mathrm{KS}\right)\delta^{\gamma}+\left(\frac{2\mu}{\gamma^2}-C_\mathrm{KS}\right)\eta^{\gamma/3}+C_\mathrm{KS}\eta^{2\gamma/3}\left(1+\left(\frac{\delta}{\eta}\right)^{\gamma}\right)\right]\left(1-\left(\frac{\delta}{\eta}\right)^{2}\right)},
\end{equation}
\begin{equation}\label{cT2}
\tilde{c}^{2}_\mathrm{T}(\delta,\eta)=\left(\frac{\delta}{\eta}\right)^{2-\gamma}\frac{\left[\left(\frac{2\mu}{\gamma^2}-C_\mathrm{KS}\right)\eta^{\gamma/3}+C_\mathrm{KS}\eta^{2\gamma/3}\right]\left(1-\left(\frac{\delta}{\eta}\right)^{\gamma}\right)}{\left[\left(\frac{\lambda+2\mu}{\gamma(\gamma-1)}-\frac{2\mu}{\gamma^2}-C_\mathrm{KS}\right)\eta^{\gamma}+\left(\frac{2\mu}{\gamma^2}-C_\mathrm{KS}\right)\eta^{\gamma/3}+2C_\mathrm{KS}\eta^{2\gamma/3}\right]\left(1-\left(\frac{\delta}{\eta}\right)^{2}\right)}.
\end{equation}
Using~\eqref{cT23} we obtain for the degenerate transverse wave speed
\begin{equation}\label{cT3}
\tilde{c}^2_\mathrm{TT}(\delta,\eta)=\frac{\gamma}{2}\frac{\left(\frac{2\mu}{\gamma^2}-C_\mathrm{KS}\right)\eta^{\gamma/3}+C_\mathrm{KS}\eta^{2\gamma/3}\left(\frac{\delta}{\eta}\right)^{\gamma}}{\left(\frac{\lambda+2\mu}{\gamma(\gamma-1)}-\frac{2\mu}{\gamma^2}-C_\mathrm{KS}\right)\delta^{\gamma}+\left(\frac{2\mu}{\gamma^2}-C_\mathrm{KS}\right)\eta^{\gamma/3}+C_\mathrm{KS}\eta^{2\gamma/3}\left(1+\left(\frac{\delta}{\eta}\right)^{\gamma}\right)}.
\end{equation}
\end{examples}
\subsection{Existence of spherically symmetric energy density functions}\label{existencenosym}
One may wonder if any choice of spherically symmetric energy density function $\widehat{\rho}(\delta,\eta)$ is possible, that is, if any such function  arises from a general energy density function $\rho(n_1,n_2,n_3)$ which is invariant under permutation of its variables. While $\delta$ can readily be replaced by $n_1n_2n_3$ in the expression for $\widehat{\rho}$, it is not obvious that a similar substitution is available for $\eta$. The problem is then whether one can find a function $\eta(n_1,n_2,n_3)$, invariant under permutation of its variables, such that $\eta(n_1,n_2,n_2)=n_2^3$.

To show that such a function exists, we consider the triangles
\begin{equation}
\Delta_c = \{ (n_1,n_2,n_3) \in \R^+\times\R^+\times\R^+ : n_1+n_2+n_3=c\}
\end{equation}
for $c>0$. The action of the permutation group $S_3$ on these triangles is precisely the action of the dihedral group $D_3$; the function $\eta$ is known in the three heights of the triangle (where $n_2=n_3$, or $n_1=n_2$, or $n_1=n_3$), and it is invariant under this action. To construct $\eta$, we just have to define it on one of the six regions determined by the heights of the triangle, and then extend it by symmetry. This can be done by fixing a diffeomorphism $\Phi:\Delta_1\to\R^2$, equivariant for the action of $D_3$ on both $\Delta_1$ and $\R^2$; for concreteness, let us assume that $\Phi$ maps the center of $\Delta_1$ to the origin, and one of the heights of $\Delta_1$ to the $x$-axis. The choice of $\Phi$ induces polar coordinates $(r,\theta) \in \R^+ \times (-\pi,\pi)$ on $\Delta_1$, and when written in these coordinates the functions $\eta(r,0)\equiv f(r)$ and $\eta(r,\frac\pi3)\equiv g(r)$ are known. Choosing a smooth function $\alpha:\R \to \R$ such that $\alpha(t)\equiv 1$ for $t<\frac\pi9$ and $\alpha(t)\equiv 0$ for $t>\frac{2\pi}9$, we can define
\begin{equation}
\eta(r,\theta)=f(r)\alpha(\theta)+g(r)(1-\alpha(\theta))
\end{equation}
for $\theta\in[0,\frac\pi3]$. Extending this function to $\Delta_1$ by the action of the dihedral group $D_3$ results in a smooth function $\eta:\Delta_1\to\R$, invariant under permutation of its variables, and satisfying $\eta(n_1,n_2,n_2)=n_2^3$. Finally, the same procedure can be applied to all triangles $\Delta_c$ by using the diffeomorphism $\Phi_c:\Delta_c\to\R^2$ given by $\Phi_c=\Phi \circ h_c$, where $h_c:\Delta_c\to\Delta$ is the dilation by $c^{-1}$ with respect to the center of the triangle. The function $\eta:\R^+\times\R^+\times\R^+$ thus obtained is clearly smooth, invariant under permutation of its variables, and satisfying $\eta(n_1,n_2,n_2)=n_2^3$.

It is possible to show that this construction can be modified so that the second partial derivative 
\begin{equation}
\frac{\partial^2 \rho}{\partial n_2^2}(n_1,n_2,n_2)
\end{equation}
can be freely chosen, corresponding to the freedom in choosing the unspecified wave velocities for spherically symmetric materials. However, we will not present the (considerably longer) proof here.

\newpage

\bibliographystyle{utphys}
\bibliography{biblio}

\end{document}